\documentclass[11pt]{article}

\usepackage[margin=1in]{geometry}
\usepackage{amsmath,amssymb,amsthm,mathtools}
\usepackage[T1]{fontenc}
\usepackage{lmodern}
\usepackage{authblk}
\usepackage{enumitem}
\usepackage{xcolor}
\usepackage{appendix}
\usepackage{tabularx}
\usepackage{array}
\usepackage{colortbl}
\usepackage{framed}
\usepackage{tikz}
\usetikzlibrary{arrows.meta}
\usepackage{hyperref}
\usepackage{cleveref}
\usepackage{aliascnt}

\newenvironment{boxedstatement}[1]{%
  \begin{framed}\noindent\textbf{#1}\par\medskip
}{%
  \end{framed}
}

\usepackage[backend=biber,style=alphabetic,maxbibnames=99,giveninits=true]{biblatex}
\hypersetup{
  colorlinks=true,
  linkcolor=blue!60!black,
  citecolor=blue!60!black,
  urlcolor=blue!60!black,
  pdftitle={EFI Pairs Without One-Way Puzzles: Oracle Separations from Communication Complexity},
  pdfauthor={Atul Mantri},
  bookmarksnumbered=true,
}

\newtheorem{theorem}{Theorem}[section]
\newcommand{\sharedthm}[3]{%
  \newaliascnt{#1}{theorem}%
  \newtheorem{#1}[#1]{#2}%
  \aliascntresetthe{#1}%
  \crefname{#1}{#2}{#3}\Crefname{#1}{#2}{#3}}
\sharedthm{lemma}{Lemma}{Lemmas}
\sharedthm{proposition}{Proposition}{Propositions}
\sharedthm{corollary}{Corollary}{Corollaries}
\sharedthm{claim}{Claim}{Claims}
\sharedthm{conjecture}{Conjecture}{Conjectures}
\theoremstyle{definition}
\sharedthm{definition}{Definition}{Definitions}
\sharedthm{model}{Model}{Models}
\sharedthm{remark}{Remark}{Remarks}
\sharedthm{example}{Example}{Examples}
\sharedthm{problem}{Open Problem}{Open Problems}
\crefname{appendix}{Appendix}{Appendices}\Crefname{appendix}{Appendix}{Appendices}
\theoremstyle{plain}
\newtheorem*{theorem*}{Theorem}

\newcommand{\cH}{\mathcal H}\newcommand{\cO}{\mathcal O}\newcommand{\cR}{\mathcal R}
\newcommand{\cD}{\mathcal D}\newcommand{\cA}{\mathcal A}\newcommand{\cB}{\mathcal B}
\newcommand{\cK}{\mathcal K}\newcommand{\cX}{\mathcal X}\newcommand{\cT}{\mathcal T}
\newcommand{\cV}{\mathcal V}
\newcommand{\cG}{\mathcal G}
\newcommand{\Count}{\mathsf{Count}}
\newcommand{\OWPuzz}{\mathsf{OWPuzz}}
\newcommand{\EVOWPuzz}{\textsf{EV-}\OWPuzz}
\newcommand{\OnePRS}{\mathsf{1PRS}}
\newcommand{\OWSG}{\mathsf{OWSG}}
\newcommand{\QEFID}{\mathsf{QEFID}}
\newcommand{\QPT}{\mathsf{QPT}}
\newcommand{\VSP}{\mathsf{VSP}}
\newcommand{\poly}{\operatorname{poly}}
\newcommand{\polylog}{\operatorname{polylog}}
\newcommand{\cM}{\mathcal M}
\newcommand{\negl}{\operatorname{negl}}
\newcommand{\tr}{\operatorname{Tr}}
\newcommand{\td}{\operatorname{TD}}
\newcommand{\TV}{\operatorname{TV}}
\newcommand{\Adv}{\operatorname{Adv}}
\newcommand{\Herm}{\operatorname{Herm}}
\newcommand{\E}{\mathbb E}
\newcommand{\norm}[1]{\left\|#1\right\|}

\newcommand{\ket}[1]{|#1\rangle}
\newcommand{\bra}[1]{\langle #1|}
\newcommand{\ip}[2]{\langle #1 | #2\rangle}
\newcommand{\proj}[1]{|#1\rangle\!\langle #1|}
\newcommand{\Id}{I}
\newcommand{\eps}{\varepsilon}
\newcommand{\bbR}{\mathbb R}\newcommand{\bbC}{\mathbb C}
\newcommand{\opnorm}[1]{\norm{#1}_{\mathrm{op}}}
\newcommand{\Un}{\mathrm U(n)}\newcommand{\On}{\mathrm O(n)}
\newcommand{\PP}{\mathsf{PP}}\newcommand{\QCMA}{\mathsf{QCMA}}\newcommand{\BQP}{\mathsf{BQP}}
\newcommand{\PSPACE}{\mathsf{PSPACE}}\newcommand{\EXPTIME}{\mathsf{EXPTIME}}
\newcommand{\eqdef}{\vcentcolon=}
\newcommand{\Samp}{\mathsf{Samp}}
\newcommand{\Ver}{\mathsf{Ver}}

\title{EFI Pairs Without One-Way Puzzles:\\ Oracle Separations from Communication Complexity}
  \author{Atul Mantri~\thanks{atulmantri@vt.edu}}
  \affil{Department of Computer Science, Virginia Tech, USA 24061}
\date{}

\begin{document}
\maketitle

\begin{abstract}
EFI pairs (Brakerski, Canetti, and Qian, ITCS 2023) and one-way puzzles (Khurana and Tomer, STOC
2024) are the leading candidates for the minimal assumption of quantum cryptography.  The first
are efficiently preparable quantum states, statistically far yet computationally
indistinguishable; the second are classical puzzles, easy to sample and hard to solve.
One-way puzzles imply EFI pairs, and whether the converse holds is open.

We construct a single classical oracle relative to which one-way puzzles do not exist, even with an
unbounded verifier, while an EFI pair survives every distinguisher that queries the oracle
classically throughout and holds advice about it, making its one superposition query at the end.
The oracle answers every question about the output probabilities of quantum samplers, which removes
the puzzles, and hides a Haar-random half-dimensional subspace.

To prove security we reduce it to communication complexity.  An adversary whose knowledge of the
subspace arrives as classical query answers can be simulated inside a two-party protocol against the
party holding it, so it does no better than the best classical protocol for Vector-in-Subspace
(Klartag and Regev, STOC 2011), whatever the oracle computes.  That argument does not cover the
superposition query, which we bound instead using tools from random matrix theory.

The same attack gives a classical simulation of any quantum party in a classical-message protocol
with no entanglement shared in advance, so relative to the oracle there is no proof of quantumness
either.  Quantum polynomial time therefore offers no advantage on any task with classical inputs
and outputs, while the two quantum states stay indistinguishable.  We state conjectures on
removing the restriction on superposition queries.
\end{abstract}

\newpage
\begingroup
\small
\tableofcontents
\endgroup
\newpage

\section{Introduction}
\label{sec:intro}

Classical cryptography is organized around a single minimal primitive.  One-way functions are
necessary and sufficient for essentially all of Minicrypt, so that one assumption governs
commitments, symmetric encryption, signatures, and zero knowledge.  Quantum cryptography has no such
primitive.  Instead, a line of work beginning with Ji, Liu, and Song~\cite{JLS18} has shown that
quantum commitments, and through them secure two-party and multiparty computation, follow from
assumptions that are not known to imply one-way functions and are plausibly
weaker~\cite{MY22,AQY22,KT24,BCQ23}.

Out of that line two candidates for the minimal assumption have emerged.  The first is the
\emph{EFI pair} of Brakerski, Canetti, and Qian~\cite{BCQ23}: a
family of pairs of mixed states $(\rho_{0,\lambda},\rho_{1,\lambda})$ that are \emph{e}fficiently
generatable, statistically \emph{f}ar, and computationally \emph{i}ndistinguishable.  EFI pairs turn
out to be equivalent to quantum bit commitments, and are therefore necessary and sufficient for
oblivious transfer, for secure multiparty computation, and for quantum computational zero knowledge
for all of $\mathsf{QIP}$~\cite{BCQ23}.  They are also equivalent, up to $O(\log\lambda)$ bits of
non-uniformity, to single-copy pseudorandom states~\cite{MetaEFI25}.  The second candidate is the
\emph{one-way puzzle} of Khurana and Tomer~\cite{KT24}: a pair $(\Samp,\Ver)$ in which the sampler
$\Samp$ is an efficient quantum algorithm producing a classical key--puzzle pair $(k,s)$, while the
verifier $\Ver$ is computationally unbounded.  Since challenge and solution are both classical,
one-way puzzles are the natural notion of one-wayness for quantum computation with classical
communication~\cite{CGG24}.  They are implied by pseudorandom states, by one-way state generators,
and by quantum money, and they in turn imply quantum commitments~\cite{KT24,CGG24}.  Moreover, their
classically-secure variant characterizes inefficient-verifier proofs of quantumness against uniform
adversaries~\cite{MSY25}.

One-way puzzles imply EFI pairs, and whether the converse holds is open.  The two primitives ask for
hardness of different kinds.  A one-way puzzle asks for a classical search problem that is hard on
average, whereas an EFI pair asks only that two quantum states be hard to tell apart.  The question
is whether the second kind of hardness always brings the first with it.  Equivalently, it is the question of where the boundary lies
between the two lowest worlds in the classification of~\cite{CountCrypt}.  One asymmetry is
already known: one-way puzzles cannot exist if $\BQP=\PP$, so counting power removes them,
whereas EFI pairs have no known attack from any fixed classical complexity oracle.

The state of the art is due to Bostanci, Chen, and Nehoran~\cite{BCN24}, who show that in the common
Haar random state (CHRS) model, augmented by unitary oracles for large complexity classes, EFI pairs exist while one-way state
generators and \emph{sample-efficient} one-way puzzles do not.  That last restriction is deliberate
rather than technical, since one-way puzzles with an inefficient verifier are in fact constructed in
the same augmented model.  The reason is that a common Haar random state does not remove the
difficulty of estimating the output probabilities of quantum samplers, and one-way puzzles exist
precisely when that estimation problem is hard on average~\cite{CGGH25,HM24}.  The following question therefore remains.

\begin{quote}
\emph{Is there an oracle relative to which EFI pairs exist but one-way puzzles, including the
inefficiently verifiable notion, do not?}
\end{quote}

In other words, we are asking whether the minimal assumption of quantum cryptography can be separated
from classical average-case hardness.  An affirmative answer would say that, relative to an oracle,
commitments and everything built from them can rest on hardness that is not the hardness of any
classical search problem, and hence that no relativizing argument derives a one-way puzzle from an
EFI pair.

The characterization of~\cite{CGGH25} suggests how to proceed.  To remove one-way puzzles we should
hand every oracle-aided sampler its own output probabilities, so that estimating them becomes easy.
The difficulty is that those probabilities depend on the hidden object out of which the EFI pair is
built, so the oracle that removes the puzzles also leaks information about the object on which the
pair's indistinguishability rests.  The rest of the paper shows that the two can coexist.

\subsection{Main results}
\label{sec:results}

We answer the question affirmatively for distinguishers with \emph{coherent-last access}, a model
made precise below.  \Cref{def:efi-classical} names the resulting notion of security, and
\Cref{cor:barrier} states the black-box barrier it yields.  Throughout, $n=n(\lambda)=2^{\lambda}$
and $H\subseteq\bbR^n$ is a Haar-random subspace of dimension $n/2$.  The pair we hide consists of
the maximally mixed states on $H$ and on its orthogonal complement,
\[
  \rho_{0}=\frac{2P_H}{n},\qquad \rho_{1}=\frac{2P_{H^\perp}}{n} .
\]
The two states have orthogonal supports, so a party who knows $H$ tells them apart with certainty:
it measures $\{P_H,P_{H^\perp}\}$.  A party who does not know $H$ has no such measurement
available.  The question is how much of $H$ an oracle-aided adversary can learn.

We hide a subspace because of a property that Theorem~C isolates.  Deciding which of the two states
one holds is a two-party problem that is easy for quantum communication and hard for classical
communication, and the security proof turns that gap into a bound on the distinguisher's
advantage.

The oracle itself is a pair $\cO=(\cR,\Count)$.  The source $\cR$ returns, on input $(b,1^\lambda)$,
one fresh copy of $\rho_b$.  The counting oracle $\Count$ is classical, and when asked about an
oracle-aided quantum sampler $C$ with classical output it returns bits of the exact output
probability $\Pr[C^{\cO}(1^\lambda)=x]$.  Such a sampler may of course query $\cO$ in turn, so the
definition as just stated would be circular.  We break the circularity by giving each query a
\emph{rank}: a query of rank $r$ may ask only about samplers whose own calls have rank below $r$.

One thing has to be fixed before we can state the results.  In an oracle model the access a
distinguisher is granted is part of the security claim, so security must be stated against a named
class of distinguishers.  The class below grants every resource the EFI game requires, together with
a full polynomial budget of classical queries, and restricts only where superposition queries may
occur.

\begin{model}[Coherent-last access]
\label{model:access}
A distinguisher has \emph{coherent-last access} to an oracle if it may do the following, in this
order:
\begin{enumerate}[label=(\roman*),leftmargin=2.2em,itemsep=2pt,topsep=3pt]
\item[(o)] receive advice depending arbitrarily on the oracle, either a classical string of
  polynomial length or a quantum state on $(\tfrac13-\eps)\lambda$ qubits for a fixed $\eps>0$
  (\Cref{cor:hybrid-qadvice}; \Cref{cor:qadvice-sharp} raises this to $(1-\eps)\lambda$);
\item obtain polynomially many reference copies of both outputs, from honest executions of the
  generator run by the challenger, which returns only the designated output register;
\item make polynomially many adaptive \emph{classical} queries of polynomial total length,
  interleaved with arbitrary quantum processing; and
\item make one \emph{coherent} query on polynomially many qubits, that is, of width
  $M=2^{\poly(\lambda)}$, before measuring.
\end{enumerate}
Its computational power is otherwise unbounded.  Equivalently, every oracle interaction but the last
is classical.
\end{model}

\begin{boxedstatement}{Theorem A (EFI pairs without one-way puzzles; informal, \Cref{thm:single})}
There is a single deterministic Boolean oracle $\cO'$ on classical strings such that, relative to
$\cO'$:
\begin{enumerate}[label=(\roman*),leftmargin=2.2em]
\item one-way puzzles do not exist, even with a computationally unbounded verifier, and the attack is
  a classical polynomial-time algorithm making classical queries, so classically-secure puzzles and
  $\QEFID$ pairs are ruled out as well; and
\item an approximation of the half-subspace pair, stored in the oracle, is efficiently generatable
  with $\td\ge1-2^{-\Omega(\lambda)}$, and every $\QPT$ distinguisher with coherent-last access
  (\Cref{model:access}) has advantage $2^{-\Omega(\lambda)}$ against it.
\end{enumerate}
Part~(i) holds for every fixing of the oracle's randomness, and part~(ii) with probability one.
Part~(i) carries no restriction on access: the attack is classical, so it applies verbatim to
adversaries that query $\cO'$ in superposition.  Part~(ii) makes the pair a classical-advice EFI pair
(\Cref{def:efi-classical}), secure in addition against quantum advice of linearly many qubits, and it
yields a black-box barrier: no fully black-box construction of one-way puzzles from EFI pairs has a
security reduction in the class of \Cref{model:access} (\Cref{cor:barrier}).
\end{boxedstatement}

Part~(i) is elementary: because $\Count$ returns exact bits, the attack is a short
conditional-sampling argument.  All of the difficulty lies in part~(ii), and there it is imported
rather than manufactured: from the communication lower bound of Klartag and Regev for the classical
queries, and from matrix concentration for the coherent one.  We build and analyze the oracle in two stages.  \Cref{sec:oracle}
through \Cref{sec:coherent} work with an idealized source that returns a fresh sample on every call,
since that is the form the concentration arguments require, and \Cref{sec:single} then implements
that source inside a Boolean function, so that the final oracle $\cO'$ is classical.

\paragraph{Scope of the access model.}  Part~(i) carries no
restriction.  The attack is a classical polynomial-time algorithm making classical queries, so
it applies verbatim to adversaries with full coherent access to $\cO'$.

Only part~(ii) is restricted.  Two features define the class of \Cref{model:access} against which it
is proved.  The first is the placement of the coherent query.  In the standard quantum oracle model an algorithm may query
a classical oracle in superposition throughout, whereas here one such query comes last.  Running the
public generator oneself costs one coherent query per execution, so a distinguisher may do that once,
and further executions are then supplied to it as reference copies run by the challenger.  The second
feature is the advice, which may be classical of any polynomial length, or quantum on $O(\lambda)$
qubits (\Cref{cor:hybrid-qadvice}).  Against quantum advice alone, by contrast, the pair is secure at
every polynomial size (\Cref{cor:quantum-advice}).  \Cref{cor:barrier} draws the black-box
consequence: a reduction escapes the class only by making two or more adaptive coherent queries, or
else one coherent query followed by classical ones.

\paragraph{Comparison with known coherent-query bounds.}  The restriction to one coherent query is
not special to our construction.  We know
of no ensemble for which a lower bound is available against more than one adaptive coherent query to
an arbitrary Boolean function of the secret, at the width $M=2^{\poly(\lambda)}$ that this model
forces.  The state of the art is a single query together with polynomially many parallel ones for the
phase-state ensemble~\cite{LMW24}, explicit one-query separations for structured
unitaries~\cite{DLM26}, and a bound at ``one and a half'' for pseudorandom states~\cite{Huang25},
whose sharper form is conditional and whose unconditional form does not reach the width we need.  To these we
add two statements that reach past a single query (\Cref{prop:advice-one-query},
\Cref{prop:asym-two-query}), and coherent-last access grants in addition a full polynomial budget of
adaptive classical queries, which none of those settings addresses (\Cref{tab:positioning}).
Removing both features at once is what \Cref{conj:lift} asks for.  Resolving it affirmatively would
also admit quantum advice of any polynomial size (\Cref{sec:lift}), and
\Cref{sec:intro-lift} returns to the question.

Part~(i) is stronger than the nonexistence of one-way puzzles.  The estimator
behind it applies to any party in a protocol whose messages are classical, sent one outgoing bit at a
time, and it collapses quantum polynomial time to classical throughout that setting.

\begin{boxedstatement}{Theorem B (interactive collapse; informal, \Cref{thm:interactive}, \Cref{cor:no-poq})}
Relative to $\cO'$, every quantum polynomial-time party in a protocol whose messages are classical,
with no entanglement or other correlated setup shared in advance, is simulated by a classical
polynomial-time party making classical queries to $\cO'$, to within $2^{-\lambda}$ in total variation
on the whole transcript.  The simulation is uniform in the counterparty, which may have unbounded
computational power.  Consequently, relative to $\cO'$ there is no proof of
quantumness even against an unbounded verifier, and every quantum-samplable classical distribution is
classically samplable.
\end{boxedstatement}

The proof of part~(ii) is different in kind, and it uses nothing about $\Count$ beyond
the fact that its answers are classical bits.  Since hiding a subspace is only one way to apply it,
we state it separately.

\begin{boxedstatement}{Theorem C (security from communication complexity; informal, \Cref{thm:transcript-security})}
Let $X$ be a hidden random object, and let $\Pi_X$ be the two-party problem in which Bob holds $X$,
Alice holds a classical description of a sample from $\rho_0(X)$ or from $\rho_1(X)$, and Alice must
say which one after classical communication.  Write $\beta_\Pi(\ell)$ for the best advantage that a
classical protocol for $\Pi_X$ achieves with $\ell$ bits.  Suppose that, apart from the challenge and
declared classical side information, a distinguisher's dependence on $X$ runs through $L$ bits of
classical query-and-answer traffic with an oracle determined by $X$, together with $a$ bits of
classical advice.  Then its average advantage is at most $2\beta_\Pi(L+a+1)$.  The oracle need only be deterministic; it
may be arbitrarily powerful, and it need not be computable.
\end{boxedstatement}

Security thus reduces to a communication lower bound, and it does so no matter how powerful the
oracle is.  Notice what is being bounded here: not what the oracle can compute, but how much of $X$
actually reaches the adversary.  Adaptivity and the number of queries, in particular, cost nothing beyond what
they contribute to the total length $L$, because a two-way protocol charges only the length of the
conversation.  \Cref{sec:framework} makes the comparison with a one-query coherent bound
precise.  Advice, by contrast, we charge directly.  Advice about $X$ is itself a message from Bob, so $a$ bits
of it degrade the budget from $L$ to $L+a+1$, and they do so even when the advice is chosen after the
oracle has been fixed.

Theorem~C also constrains the construction, through two requirements.  The pair must be usable, so
the associated two-party problem has to be \emph{easy} for quantum communication: the holder of the
challenge sends her state to the holder of the secret, who measures it.  The pair must also not
leak, so the same problem has to be \emph{hard} for classical communication.  We need, then, a
problem whose quantum and classical communication complexities are far apart, and the size of that
gap is the security we obtain.  Of the two requirements, the classical one is built into the bound
itself, which certifies nothing once $\Pi_X$ admits a cheap classical protocol.  The quantum one is
automatic (\Cref{prop:converse}): the holder of $X$ can always perform the Helstrom measurement on a
state that Alice sends in $\lceil\log_2 d\rceil$ qubits.  Any use of Theorem~C rests on such a
gap.

The standard example of a problem with that gap is Vector-in-Subspace, $\VSP_n$, studied by Klartag
and Regev~\cite{KR11}.  Here Bob holds a half-dimensional subspace $H\subseteq\bbR^n$, Alice holds a
unit vector promised to lie either in $H$ or in $H^\perp$, and Alice must decide which.  Quantumly
the problem is easy.  An $O(\log n)$-qubit two-message protocol decides it with certainty: Alice
sends the vector encoded in amplitudes, and Bob, who knows $H$, measures $\{P_H,P_{H^\perp}\}$ and
returns the outcome.  Classically it is hard.  The randomized communication complexity is
$\Omega(n^{1/3})$, even with two-way interaction and shared randomness.  This is why the pair we hide
is built from a half-dimensional subspace.

Theorem~D turns to coherent queries.  A separation in the standard quantum oracle model has to handle
them, and Theorem~C provably cannot, since a coherent query behaves like quantum communication and
$\VSP$ is exponentially easy in that model.

\begin{boxedstatement}{Theorem D (coherent queries; informal, \Cref{thm:one-query}, \Cref{cor:parallel}, \Cref{cor:quantum-advice}, \Cref{thm:hybrid}, \Cref{prop:advice-one-query}, \Cref{prop:asym-two-query})}
Let $f_H\colon[M]\to\{\pm1\}$ be an \emph{arbitrary} Boolean function of $H$, accessed as a phase
oracle, and assume nothing about the adversary's computational power.  Against $(\rho_0,\rho_1)$:
\begin{enumerate}[label=(\roman*),leftmargin=2.2em]
\item one coherent query has average bias $O\big(\sqrt{\log(2M)/n}\big)$, and $t$ parallel queries
  have average bias $O\big(\sqrt{t\log(2M)/n}\big)$;
\item $m$ qubits of quantum advice give bias $O(\sqrt{m+1}/n)$, so $m=\Theta(n^2)$ is the exact
  threshold for constant bias;
\item $\poly(\lambda)$ adaptive classical $\Count$-queries followed by one coherent query, with
  reference copies and advice, classical or of linearly many qubits, give advantage $\negl(\lambda)$;
\item advice of dimension $n^{2-\eps}$ with one coherent query of any width, and two adaptive coherent
  queries whose first has width $n^{2-\eps}$ however wide the second, give bias $2^{-\Omega(\lambda)}$.
\end{enumerate}
Part~(iii) is the full game of Theorem~A(ii); the others take one challenge copy and no reference
copies, and part~(iv) is as far past one query as the argument reaches.
\end{boxedstatement}

Theorem~E asks the opposite question: not how little the oracle reveals, but how much.  Two rates
come out of it.  The first measures what a natural class of classical adversaries learns, and the
second measures what a message is worth in the one-way model of the communication problem behind it.

\begin{boxedstatement}{Theorem E (two rates; informal, \Cref{prop:exact-rate}, \Cref{prop:oneway-rate})}
\begin{enumerate}[label=(\roman*),leftmargin=2.2em]
\item An adversary that measures its challenge copy with any $H$-independent POVM, takes no reference
  copies, and is then told the \emph{exact} probability of every outcome, a relaxation of the bits $\Count$ supplies, has average advantage at most $2/\sqrt{\pi n}$.  Measuring in the computational basis with
  $O(\lambda)$ classical queries attains this to within a factor $1-O(n^{-1})$, so the rate in this
  model is $2/\sqrt{\pi n}\,(1+O(n^{-1}))$: sharp to the leading constant, not an exact finite-$n$
  optimum.
\item In the one-way model of $\VSP_n$ an $L$-bit message yields advantage $O(L/\sqrt n)$, while
  $\Omega(n^{-1/2})$ is available already at logarithmic length.  This holds at every message length,
  on the exact promise and the rotation-invariant distribution; the constant-advantage endpoint it
  recovers, $\Theta(\sqrt n)$ bits, is already known by a different route
  (\Cref{sec:oneway-rate}).
\end{enumerate}
\end{boxedstatement}

The upper and lower bounds of part~(i) agree asymptotically, at the rate $2/\sqrt{\pi n}$, and three
consequences follow.  First, the dimension must be superpolynomial rather than merely large, since
negligible security forces $n=\lambda^{\omega(1)}$.  Second, the security level is
$2^{-\Theta(\lambda)}$ on both sides, so that the gap between the $n^{-1/6}$ we prove and the
$n^{-1/2}$ an attack achieves is a gap in the exponent only.  That gap comes from converting a
constant-error communication bound into a small-advantage one, and not from the $\Omega(n^{1/3})$
bound itself (\Cref{sec:exact-rate}).  Third, the same rate is the most that a single bit of
information about $H$ can ever be worth (\Cref{lem:local-corr}), and each answer of the counting
oracle is exactly one such bit.  The scale $n^{-1/2}$ then appears once more in Theorem~D, as the
value of one coherent query, tight there up to $\sqrt{\log M}$, and \Cref{cor:one-query-tight} shows
that a one-query adversary attains it.

Part~(ii), finally, is a statement about $\VSP_n$ alone, proved by a spectral argument that is
independent of the rest of the paper.  \Cref{sec:oneway-rate} relates it to the program aiming at the same
bound for two-way protocols.

\paragraph{Primitives removed along with one-way puzzles.}  One-way puzzles do not exist relative to
$\cO'$, and so neither does any primitive known to imply them: \emph{multi-copy} pseudorandom states
(PRS) and unitaries, pure-output one-way state generators, and quantum money mini-schemes with pure
banknotes.  \Cref{cor:collateral} proves this for the source-model oracle of \Cref{sec:oracle}, and
the argument applies verbatim to $\cO'$, since \Cref{thm:single}(i) is the same attack, as is
Theorem~B.

The qualifier ``multi-copy'' is necessary, for the following reason.  By~\cite{MetaEFI25}, EFI pairs
are equivalent to single-copy pseudorandom states, up to $O(\log\lambda)$ bits of non-uniformity in
the state family.  Suppose that equivalence relativizes to the access model of \Cref{thm:single},
preserving its advice convention and query budget.  Then single-copy pseudorandomness survives
$\cO'$ in that model, and $\cO'$ separates multi-copy from single-copy pseudorandomness, which is the
separation studied in~\cite{CCS24}.  \Cref{sec:no-owpuzz} states precisely the condition on
relativization that this inference depends on.

\subsection{Technical overview}
\label{sec:overview}

Theorem~A combines two oracle layers whose requirements conflict.  Removing
one-way puzzles requires that every oracle-aided sampler be handed its own output probabilities,
because by~\cite{CGGH25,HM24} the average-case hardness of estimating those probabilities is exactly
what a one-way puzzle provides.  Those probabilities are, however, functions of the hidden subspace,
so the layer that removes the puzzle bears directly on the pair that is supposed to remain
indistinguishable.  The two are compatible because the counting oracle is queried classically and
answers in single bits, and because by the communication lower bound no polynomial number of
classical bits about $H$ suffices to distinguish the pair.

We sketch the argument here in the order in which the proof develops it; the formal treatment
occupies \Cref{sec:framework} through \Cref{sec:single}.  There are four parts.  We first define the
counting layer and show that it removes every form of classical-output hardness.  We then ask what a
bounded classical transcript about $H$ is worth, which turns out to be a question of communication
complexity, and this settles the classical part of the access model.  Next we ask what one
superposition query is worth; no communication argument can reach that question, so we treat it in
query complexity instead, and then combine the two bounds.  Finally we replace the state source by a
Boolean function, so that the oracle becomes a single classical object, and we draw the black-box
barrier.  Along the way we collect two exact rates, which measure what the oracle does leak.

\paragraph{The counting oracle.} \label{sec:overview-oracle} The samplers whose probabilities
$\Count$ reports may themselves query $\Count$, and without some restriction on this the definition
is inconsistent.  To see the difficulty, consider a sampler that requests the leading bit of
its own probability of outputting $0$ and then outputs that bit.  No answer to that query is
consistent with the sampler's behavior (\Cref{ex:liar}).  We therefore assign each query a
\emph{rank}, much as the polynomial hierarchy is stratified by alternation depth, and permit a
rank-$r$ query to refer only to samplers whose own calls have rank below $r$.  We can then define the whole
family by an ordinary recursion on $r$, and we need no fixed-point theorem to do so
(\Cref{lem:welldef}).

Two features of the definition matter later, and both concern how it is stated rather than what it
says.  The first is that the rank condition is checked against the description the oracle is handed,
and not against the behavior of the sampler that description defines.  This is necessary, because a
sampler may compute its rank at runtime, possibly in superposition, in which case there is no single
behavior to inspect (\Cref{sec:rank}).  The second is that $\Count$ returns exact bits of the
probability rather than an approximation.  Consequently we never have to control an accumulation of
approximation error across ranks, and the coherent form of the oracle remains an ordinary Boolean
function (\Cref{def:oracle}).

\paragraph{Removing classical-output hardness.} \label{sec:overview-attack} Because the bits are
exact, $t$ classical queries pin an output probability down to additive error $2^{-t}$, so that error
$2^{-\poly(\lambda)}$ is available at polynomial cost (\Cref{lem:pe-easy}).  Given that, the attack
on one-way puzzles is elementary.  What the adversary would like, on being given a puzzle $s$, is to
sample a key from the true conditional distribution of keys given $s$, since a key drawn that way is
accepted by the honest verifier just as often as an honest key is.  It can do exactly this, one bit
at a time.  Indeed, the conditional probability of the next key bit is a ratio of two output
probabilities, namely those of the samplers ``run $\Samp(1^\lambda)$ and output $(s,k_1\cdots k_i)$'',
and both are available from $\Count$.  A hybrid argument over the $m$ key bits then places the joint
distribution of $(s,\cA(s))$ within $2^{-\lambda}$ of the honest one (\Cref{lem:tv}), by separating
the prefixes that carry little probability from those on which the estimates are relatively accurate.
Both halves of that accounting -- the local one and the threshold -- are shared with the interactive
simulation below, so we state them once (\Cref{lem:ratio}, \Cref{lem:threshold}).  The unbounded verifier then accepts with probability $1-\negl(\lambda)$.  Notice that the
adversary never runs $\Ver$, so we need no assumption on the verifier.  Notice too that it
is a classical algorithm making classical queries, so classically-secure puzzles are ruled out as
well (\Cref{thm:no-owpuzz}).

The same estimates settle the corresponding decision problem.  With exact bits in hand the
likelihood-ratio test between two classical-output samplers becomes computable, so relative to $\cO$
no such pair is closer computationally than it is statistically (\Cref{prop:no-qefid}), and $\QEFID$
pairs do not exist either (\Cref{cor:no-qefid}).

The estimator applies round by round to an interactive party as well.  Simulating each outgoing bit
in turn replaces any $\QPT^{\cO}$ party in a classical-message protocol by a classical one, and does
so against every counterparty at once (\Cref{thm:interactive}), so that relative to $\cO$ there is no
proof of quantumness even with an unbounded verifier (\Cref{cor:no-poq}).  Taken together, these
three statements say what can survive $\cO$.  On any task with classical inputs and outputs, quantum
polynomial time gains nothing over classical polynomial time.  Whatever hardness survives $\cO$ must
be of a different kind altogether, namely the indistinguishability of two quantum states.

\paragraph{Security from a classical transcript.} \label{sec:overview-framework} Suppose that every
route from the hidden object $X$ to the adversary carries classical data.  Then we may run the
adversary inside a two-party game, with one player running it while a second player, who holds $X$
and is computationally unbounded, answers its queries.  Each query becomes a message and each answer
a message in return, so the advantage is bounded by the length of that conversation and by nothing
else, and in particular not by anything the oracle happens to be able to compute
(\Cref{thm:transcript-security}).  Adaptivity and the number of queries cost nothing here, since a
two-way protocol charges only the total length of the conversation.  Advice about $X$, on the other
hand, is itself a message from the holder of $X$, so $a$ bits of advice cost $a+1$ bits of budget,
and they do so even when the advice is chosen after the oracle has been fixed.

One requirement shapes the definition.  A communication protocol takes classical inputs, so the
challenge must be handed over as a description rather than as a state.  That makes the simulating
player stronger rather than weaker, which is precisely why the reduction is sound.  Unfortunately,
the same requirement excludes quantum side information depending on $X$.  That is why we treat the
reference copies separately below, and why coherent access falls outside the scope of the theorem
altogether, as \Cref{sec:framework} explains.

\paragraph{The associated two-party problem.} \label{sec:overview-vsp} Theorem~C placed two requirements
on the construction, one quantum and one classical, and Vector-in-Subspace meets both of
them (\Cref{prop:converse}, \Cref{thm:vsp-lb}).  We use the classical bound through a small-advantage
form, namely $\beta(n,L)\le C\sqrt{(L+1)/n^{1/3}}$, which we obtain by symmetrizing across and within
the two promise cases and then amplifying (\Cref{lem:vsp-small}).

\paragraph{Reference copies.} \label{sec:overview-reference} The generator is public, so a
distinguisher can do more than hold its challenge: it can draw polynomially many fresh copies of both
states besides.  Those copies are quantum side information about $H$, and quantum side information is
exactly what \Cref{thm:transcript-security} does not cover.  We therefore remove them, by handing
the adversary something classical that is at least as useful in their place, namely the \emph{spans}
$A\subseteq H$ and $B\subseteq H^\perp$ of the copies it would have drawn.  Given a basis of $A$, the
adversary can resample the copies for itself
with the correct joint distribution, so it loses nothing by the exchange.  Unlike the copies, the
spans are classical, which is exactly what the reduction needs.

Now condition on the spans.  The challenge then splits into two branches.  The first is supported on
the subspace the adversary already knows, and there we concede the entire advantage; since
that branch carries weight only $2q/n$, conceding it is cheap.  The second branch is a fresh
Vector-in-Subspace instance, this time inside $K=(A\oplus B)^\perp$, of dimension $n-2q$, in which
the residual subspace is once again Haar-random of half dimension (\Cref{lem:residual},
\Cref{lem:reference}).  Reference copies cost us two things, then, and nothing more: an additive
$2q/n$, and a loss of dimension (\Cref{cor:reference-vsp}).

One step then remains, and it is routine.  What we have bounded so far is an advantage averaged over
the random subspaces, whereas an oracle separation needs a single sequence $\{H_\lambda\}$ that
defeats every adversary at once.  A discretization of the gate set, Markov's inequality, and
Borel--Cantelli together make that passage, and we isolate it as \Cref{lem:fixing} because three
later arguments need it.  The one point to watch is that the maximization over advice strings has to
sit inside the expectation; the reduction accommodates this without loss, since the holder of $X$ can
send the best advice (\Cref{sec:security-thm}).

\paragraph{The two rates.} \label{sec:overview-rates} The bounds above say that the oracle leaks
little.  Two rates say how little, and we compute both directly rather than through the reduction.
Consider first a distinguisher that measures its challenge with a fixed POVM $\{E_k\}$ and is then
told every outcome probability exactly.  Its advantage is the likelihood-ratio quantity
\[
  \frac1n\sum_k\big|\tr(E_kR_H)\big| ,
\]
and each term in that sum is governed by the law of $\bra vP_H\ket v$,
which for a real half-dimensional subspace is $\mathrm{Beta}(n/4,n/4)$.  The mean absolute deviation
of that law gives the rate $2/\sqrt{\pi n}$, and measuring in the computational basis attains it
(\Cref{prop:exact-rate}).  For the one-way model of $\VSP_n$ the picture is different.  There the
second-moment operator $\E_H[(\int f\,d\mu_H)^2]$ is diagonalized exactly by the spherical harmonics,
with eigenvalue $\frac1{n-1}$ at degree two and rapid decay above it.  A level-$\ell$ inequality then
bounds what a single cell of the partition induced by Alice's message can leak, and summing over
cells gives $O(L/\sqrt n)$ (\Cref{prop:oneway-rate}).

\paragraph{The one-query bound.}
\label{sec:overview-coherent}
The communication argument charges the adversary for the classical bits it exchanges with the
oracle, and a superposition query is not a classical message, so that argument does not apply under
coherent access.  By \Cref{prop:converse} no communication bound can replace it, since the problem is
easy for quantum communication.  We turn to query complexity instead.  Write the pre-query processing of a one-query adversary as an isometry with
branches $A_x$, one branch for each label $x$ the oracle can be asked about, and write its post-query
measurement as an effect $\Pi$.  The bias is then a quadratic form in the unknown truth table,
\[
  Z_H(f)=\sum_{x,y}f_xf_y\,C_H(x,y),\qquad C_H(x,y)=\tr(\Pi_{xy}A_y\Delta A_x^\dagger),
  \qquad \Delta=\tfrac{R_H}{n},
\]
where $R_H=2P_H-\Id$ is the reflection in $H$.  Three objects enter this form, each in a different
way.  The adversary contributes the fixed data
$(A_x,\Pi_{xy})$.  The randomness of $H$ enters \emph{linearly}, through $\Delta$ alone.  The truth
table appears only through the rank-one sign pattern $f_xf_y$.  What we have to bound is the maximum
of $|Z_H(f)|$ over all $2^M$ sign patterns, since $f$ is an arbitrary function of $H$ and we are
given no control over it.

The one general tool for bounding a quadratic form is the operator norm, and applying it directly
costs a factor $M$ that the completeness relation shows is spurious.  Reweighting the labels by the
query mass $\tau_x=\frac1n\tr(A_x^\dagger A_x)$ removes it, because $D_\tau^{1/2}f$ is a unit vector
for \emph{every} sign pattern, so that the maximum over the $2^M$ patterns collapses to a single
operator norm $\opnorm{G_H}$ (\Cref{sec:onequery} does this carefully).  Neither the weighting nor
the collapse is new: both are the query-label form of the argument of Lombardi, Ma, and
Wright~\cite{LMW24}, and \Cref{sec:intro-lmw} compares the two settings.

What is left is to bound $\E_H\opnorm{G_H}$, and here our setting differs from the earlier one.  By
the linearity in
$R_H$, each entry of $G_H$ is a linear statistic $\tr(B_{xy}R_H/n)$ of one Haar-conjugated
reflection.  The scalar case is easy, and it already shows the scale we are after: such a statistic
is Lipschitz in the Haar unitary that defines $R_H$, hence subgaussian, so that a maximum over $N$
suitably normalized statistics is $O(\sqrt{\log(2N)}/n)$ (\Cref{lem:scalar-linear}).  An operator
norm, however, is not a maximum over finitely many scalars, and what we need is the matrix form of
that statement.  Had the entries been built from independent random signs, the classical matrix
Khintchine inequality would supply it, giving an operator norm of order $\sigma\sqrt{\log(2M)}$ for a
variance parameter $\sigma$ computed from the coefficient matrices.  A single Haar reflection,
though, offers no independence, so we take the matrix concentration from a curvature
hypothesis instead and compute the variance proxy it asks for (\Cref{lem:khintchine}), with the row
and column normalizations playing the role of the variance parameter.

The substitute for independence is curvature.  Haar measure on $\mathrm{SU}(n)$ and $\mathrm{SO}(n)$
has Ricci curvature $\Omega(n)$, which is precisely why scalar Lipschitz functions on those groups
concentrate at scale $n^{-1/2}$, and Huang and Tropp~\cite{HT21} show that the same hypothesis yields
subgaussian concentration for matrix-valued functions as well.  Their criterion asks for a bound on
the sum of squared derivatives along an orthonormal frame, and here that computation is short:
moving $U$ in a direction $K$ rotates $R_H$ by the commutator $i[K,D]$, and
$\norm{[K,D]}_2\le2\opnorm D\norm K_2$.  The conclusion is that
$\E_H\max_f|Z_H(f)|\le C\sqrt{\log(2M)/n}$ (\Cref{thm:one-query}).  This extends to parallel queries
with no further argument (\Cref{cor:parallel}) and, applied through Rosenthal's one-query state
synthesis~\cite{Rosenthal23}, it also excludes quantum advice of polynomial size
(\Cref{cor:quantum-advice}).

\paragraph{Combining classical and coherent queries.} \label{sec:overview-hybrid} The two arguments
above bound different resources by different methods, and \Cref{thm:hybrid} combines them into a
single statement, covering polynomially many adaptive classical queries together with advice,
followed by one coherent query.  The combination preserves negligibility, though not the sharp
constant of the pure one-query corner.  We split the advantage into two parts: the part that an
$H$-\emph{independent} final query would produce, which the communication bound controls, and the
increment that arises from making that query $H$-dependent, which the spectral bound controls.

For the increment we would like to fix the classical transcript and then apply the one-query bound to
whatever follows it.  The problem is that the transcript is correlated with $H$, so conditioning on
it changes the law of $H$, whereas the one-query bound is an average over that law.  Two facts
resolve it.  The first makes the conditioning explicit: write the answers into a fixed operator
$V_\pi$, and leave the question of whether they are the \emph{true} answers to an indicator
$\chi_\pi(H)$.  The transcript's $H$-dependence is then carried by a scalar, and the masses
$w_\pi=\frac1n\tr(V_\pi^\dagger V_\pi)$ satisfy $\sum_\pi\chi_\pi(H)w_\pi=1$ for every $H$
(\Cref{lem:record}), so that what follows is a convex combination.  The second fact is that the
quantity being conditioned, namely the relaxed spectral norm of the previous paragraph, is a
Lipschitz function of the Haar unitary defining $R_H$, and so concentrates.  Concentration is
exactly what survives conditioning on a classical transcript, since such a transcript is a condition
on that same unitary.  Removing the conditioning over the at most $2^{L+a+2}$ transcripts then costs
only a maximal-inequality term $C\sqrt{(L+a)/n}$.

Quantum advice is admitted on $O(\lambda)$ qubits, by charging a net of advice states to the
classical budget (\Cref{cor:hybrid-qadvice}).  Beyond that size, however, the same argument marks the
limit of what we can prove, since it requires a classical transcript \emph{between} the state source
and the coherent query.  There are two ways to violate that requirement.  A
coherent query placed first can synthesize arbitrary quantum advice about $H$, by
\Cref{cor:quantum-advice}, and two coherent queries already contain advice together with a query.  In
the first case the adversary holds $H$-dependent quantum side information before any transcript
exists, which the communication reduction cannot represent; in the second there is no transcript at
all (\Cref{sec:hybrid}, \Cref{sec:lift}).  Closing either case is the content of
\Cref{conj:lift}.

\paragraph{Removing the state source.} \label{sec:overview-single} \Cref{sec:oracle} through
\Cref{sec:coherent} work with the source $\cR$, which returns a fresh Haar sample on every call,
because that is the form the concentration arguments require.  The final oracle contains no such
source.  Instead we store samples inside a Boolean function, again through Rosenthal's one-query
synthesis: for each branch $b$ and each index $r$ we draw a vector $u_{b,r}$ from the appropriate
sphere, and the oracle carries the Boolean function that synthesizes an approximation of the rounded
state $\ket{\tilde u_{b,r}}$.  The generator picks $r$ with its own coins, runs the synthesis circuit
with its one query addressed to that section, and then discards $r$.

The construction works because $r$ is discarded, and the order of the two steps in the proof is
forced by that.  Suppose we tried to apply the communication reduction to the stored experiment
directly.  The holder of the hidden object would receive the entire list from which the challenge was
drawn, and the associated problem would collapse: Alice fingerprints her challenge vector in
$O(\lambda)$ bits, the other player finds the matching entry, and reads off its branch.  So we must
replace the stored samples by fresh ones \emph{before} applying the reduction, which is what
\Cref{lem:fresh} does.  Only then does the analysis of the earlier sections apply, and
it applies with the enlarged hidden object in place of $H$ (\Cref{lem:side}).  \Cref{sec:single} gives two
further reasons why the replacement is a precondition rather than a convenience.

\subsection{Comparison with Lombardi--Ma--Wright}
\label{sec:intro-lmw}

The closest prior security proofs are those of Lombardi, Ma, and Wright~\cite{LMW24} and,
concurrently with this manuscript, Huang~\cite{Huang25}.  Since then, Dong, Lombardi, and
Ma~\cite{DLM26} have proved explicit one-query lower bounds through oracle state search and Choi
state games, together with a sharper separation between one-query synthesis and quantum programs.
Their lower bounds concern one-query adversaries; moreover, some of their targets admit two-query
constructions, and their quantum-program result concerns advice without queries.  So our coherent-query
bounds do not improve on any of these, and the two settings are better compared by scope than by
strength.

For unitary synthesis the classical half is the easy one, which is not what one might expect.
Indeed, a classical algorithm making $T$ queries reads only $T$ bits of the oracle, so that once we
fix its
internal randomness the circuit can reach at most $2^T$ states, no matter how adaptively it chose
those queries.  A counting argument then caps the fidelity it can achieve with a generic state at a
constant (\Cref{obs:classical-synthesis}).  A single coherent query, by contrast, already suffices to
synthesize an arbitrary state~\cite{Rosenthal23}, and \Cref{sec:framework} makes the comparison
precise.

We adopt the plan of their one-query proof rather than rediscover it.  The
deterministic query-mass weighting that depends only on the isometry, which we call $\tau$, is their
weight vector $\ket{\mathrm{wt}_V}$ read on query labels instead of on workspace basis
vectors~\cite[\S2.4.1]{LMW24}.  Their $D_{V,h}$ is a different quantity, measuring the deviation of
amplitudes from typical.  The reduction of a maximum over truth tables to a single operator norm by
means of that weighting is also due to them, and it works in both settings for the same reason, namely
that a diagonal weighting commutes with the phase oracle.  A concentration statement for the
ensemble is likewise common to both proofs, in their case a Talagrand inequality on the Boolean cube,
which is essential to their conclusion about pseudorandom states.

What differs at this step is which concentration statement is available, and what has to be computed
to use it.  Their states are built from independent coordinates, so matrix concentration for
independent sums applies to them directly.  A single Haar-conjugated reflection provides no
independence, so we take the matrix concentration from a curvature hypothesis instead --- an
inequality due to Huang and Tropp rather than to us (\Cref{prop:herbst}) --- and what \Cref{lem:khintchine}
supplies is the variance proxy that hypothesis asks for, for a matrix whose entries are linear
statistics of one conjugation.

The one structural difference lies in how the hidden object enters the bias.  In~\cite{LMW24} the
challenge state $\ket{\psi_{R_k}}$ is itself random and oracle-dependent, so that the bias is a
\emph{quadratic} form in the same random vector, and decoupling is then required in order to separate
its two occurrences.  In our setting the bias is $Z_H(f)=\frac1n\tr\big(E_f\,R_H\big)$, which is
\emph{linear} in the hidden reflection, for every truth table $f$ and every query count $T$.
Consequently there is no decoupling step and no good event on which to condition, although the width
still enters through $\sqrt{\log(2M)}$.  The one-query proof accordingly reduces to three steps: a
weighted relaxation, a derivative computation, and an appeal to a matrix concentration inequality.
This is the one respect in which our setting is simpler than theirs.

Neither ensemble dominates the other, since the two are suited to different conclusions.  The phase
states of~\cite{LMW24} are efficiently preparable from a \emph{classical} random oracle.  A
polynomial-query theorem in their setting would yield a classical oracle relative to which
single-copy pseudorandomness survives while, as Lombardi, Ma, and Wright observe, quantum cryptography would not
black-box imply any hard language, which is the question raised in~\cite{KQST23}.  Our states, by
contrast, require Haar-random data, which \Cref{sec:single} stores in a Boolean function rather than
removes, so for that conclusion their ensemble is the better suited.  A random oracle, on the other
hand, does not eliminate one-way puzzles, whereas our counting oracle does, so for the purpose of
proving a separation ours is the better suited.  \Cref{tab:positioning} sets out the comparison.

\begin{table}[t]
\centering
\footnotesize
\renewcommand{\arraystretch}{1.3}
\setlength{\tabcolsep}{4pt}
\begin{tabularx}{\textwidth}{@{}>{\raggedright\arraybackslash}p{0.10\textwidth} >{\raggedright\arraybackslash}p{0.15\textwidth} >{\raggedright\arraybackslash}p{0.11\textwidth} >{\raggedright\arraybackslash}X >{\raggedright\arraybackslash}p{0.19\textwidth}@{}}
\hline
\textbf{Work} & \textbf{Coherent} & \textbf{Classical} & \textbf{Ensemble / target} & \textbf{Result type}\\
\hline
\multicolumn{5}{@{}l}{\emph{Prior work}}\\
\cite{LMW24} & one; poly parallel & advice only & random oracle; single-copy PRS, commitments & one-query bound (not an EFI-versus-$\OWPuzz$ separation)\\
\cite{Huang25}$^\dagger$ & one, then one classical & one (after) & random oracle; PRS & conditional bound (not a separation)\\
\hline
\rowcolor[gray]{.9}\textbf{This work} & one; poly parallel; one \emph{after} the classical phase; two at restricted width & poly adaptive ($+$ advice) & single classical oracle; Haar half-subspace pair; classical-advice EFI vs.\ $\OWPuzz$ and $\QEFID$ & separation in the classical-query model, extended by one coherent query\\
\hline
\emph{Open} & \emph{poly adaptive} & & \emph{any} & \emph{fully coherent access; \Cref{conj:lift}}\\
\hline
\end{tabularx}
\caption[Positioning against prior single-coherent-query bounds.]{Positioning. All three share a single-coherent-query base, and none reaches polynomially many adaptive coherent queries, the row that fully coherent access to a classical oracle would need. Classical queries are the easy case (\Cref{obs:classical-synthesis}), and the ensemble of~\cite{LMW24} is better suited to the complexity-independence conclusion. The third column is the resource their setting does not address and ours requires, since the counting oracle is queried classically. $^\dagger$Preprint; it gives an assumption-light route whose bound is not negligible at polynomial width, and a sharper route conditional on a block-orthogonality (sum-to-max) assumption that it states as such.}
\label{tab:positioning}
\end{table}

\paragraph{Single-copy security in both conjectures.}  \Cref{conj:lift} and the conjecture
of~\cite{LMW24} are both single-copy statements, and there is a reason for that.
Kretschmer~\cite{Kre21} breaks any \emph{multi-copy} PRS with a single $\PP$ query, using
$O(\lambda)$ copies, and so no unitary-synthesis lower bound can proceed through multi-copy security.
Single-copy security is therefore the only level at which the connection to synthesis remains
available, and~\cite{LMW24} say so explicitly.  Now, by \cite[Thm.~1.1]{MetaEFI25}, EFI pairs are
equivalent to non-uniform single-copy pseudorandom states ($\OnePRS$) with $O(\log\lambda)$ advice.
An EFI pair therefore sits at exactly the level where that connection is available and where the
counting attack does not apply.  This is why the two conjectures take the same form: both are
single-copy distinguishing questions against a secret-dependent classical oracle.  The ensembles
themselves, however, differ, and the relationship between the two conjectures is not a reduction.  An
equivalence between primitives does not transport an ensemble-level query lower bound from one to the
other, and we know of no reduction between the two ensembles in either direction.

\subsection{The question of fully coherent access}
\label{sec:intro-lift}

Access to $\cO'$ in \Cref{thm:single} is classical, followed by one coherent query.  Since $\cO'$ is
itself a Boolean function, the natural strengthening is the standard relativized quantum model, in
which the distinguisher queries $\cO'$ in superposition throughout.  That is \Cref{conj:lift}, stated
in \Cref{sec:lift}.  Both possible answers are informative.

If \Cref{conj:lift} holds at polynomial parameters, then the reductions of \Cref{sec:single} carry
\Cref{thm:single} over to fully coherent access, so the separation holds in the standard quantum
oracle model and admits quantum advice of any polynomial size (\Cref{thm:two-cases}(A)).  If instead
it fails, then coherent access to a classical function of the secret, with polynomial classical
advice, achieves a bias above the ceiling $O(\poly(\lambda)\cdot n^{-1/6})$ that Klartag--Regev
imposes on every classical transcript of polynomial length together with advice about the same
secret (\Cref{thm:two-cases}(B)).  That would separate coherent access from classical transcripts for
a concrete ensemble, though it would not resolve the Unitary Synthesis Problem of Aaronson and
Kuperberg~\cite{AK07}; \Cref{sec:lift} gives the three reasons.

At small width the question is already settled, by a union bound: enumerating truth tables and advice
strings settles \Cref{conj:lift} at every query count once $M+a\le n^{2-\eps}$ for a fixed $\eps>0$,
with polynomial classical advice and no reference copies (\Cref{prop:small-width}).  Two
further bounds reach past a single query, and both of them restrict a width rather than a query
count.  The first admits quantum advice of dimension $n^{2-\eps}$ together with one coherent query of
any width (\Cref{prop:advice-one-query}); the second admits two adaptive coherent queries whose first
has width $n^{2-\eps}$ (\Cref{prop:asym-two-query}).  What remains open is therefore $T\ge2$ at widths $n^{2-o(1)}$ and above.  Security against every $\QPT$ adversary has to cover that
regime, since such an adversary may ask about samplers of description length $\lambda^k$ for any $k$.

\subsection{Related work}
\label{sec:related}

\paragraph{Minimal primitives.} EFI pairs are due to~\cite{BCQ23} and $\OWPuzz$s to~\cite{KT24},
while pure-output and inefficiently verifiable one-way state generators ($\OWSG$s) sit above
$\OWPuzz$s and above EFI pairs respectively~\cite{BJ24}.  These primitives are organized
in~\cite{CountCrypt} into three classes.  \emph{QuantuMania} consists of those building
$\EVOWPuzz$s, \emph{CountCrypt} of those building $\OWPuzz$s but not $\EVOWPuzz$s, and
\emph{NanoCrypt} of those building EFI pairs but not $\OWPuzz$s.  In that language, ours is a separation of NanoCrypt from
CountCrypt in the model of \Cref{model:access}, that is, against adversaries whose queries to
$\Count$ are classical but for one at the end.

\paragraph{Meta-complexity.} $\OWPuzz$s are characterized by the average-case hardness of
probability estimation~\cite{CGGH25,HM24}.  That is precisely the hardness our counting oracle
removes, and our attack follows the conditional-sampling idea behind the forward direction of the
characterization.  There is a second characterization, by hardness of proper quantum distribution
learning~\cite{HHM25}, and since our oracle makes probability estimation easy it would make that task
easy relative to $\cO$ as well, provided the characterization relativizes.  The bare nonexistence of
$\OWPuzz$s under counting power is, we stress, not new: they cannot exist if
$\BQP=\PP$~\cite{CountCrypt}, and an explicit $\PP$-oracle attack is known~\cite{HardnessQOW}.

What is different here is that a bare $\PP$ oracle is a fixed language \emph{independent} of $H$,
whereas once a state source produces $H$-dependent states, a sampler may call it.  Our route to
eliminating $\OWPuzz$s, through probability estimation, therefore \emph{forces} an $H$-dependent
counting oracle, and such an oracle can in principle leak $H$ and break the very pair the
construction has to hide.  In the converse direction, $\OWPuzz$s can be built from
$\#\mathsf P$-hardness together with quantum advantage~\cite{KT25}, and our oracle removes exactly
the hardness those constructions rely on.

\paragraph{EFI, single-copy pseudorandomness, and $\mathsf P$ versus $\PSPACE$.} There is a
classical oracle relative to which $\mathsf P=\mathsf{NP}$ and single-copy pseudorandom states
exist, together with the explicit question of whether single-copy PRS existence implies $\mathsf
P\ne\PSPACE$~\cite{KQST23}.  Together with~\cite{MetaEFI25} that is, up to $O(\log\lambda)$ advice,
the EFI-versus-$\mathsf P=\PSPACE$ frontier.  The known $\PP$-oracle attack~\cite[Thm.~27]{Kre21}
uses $\poly(\lambda)$ copies, in fact $O(\lambda)$ in its proof, so it does not apply to single-copy
security, and copy amplification~\cite{BZ26} does not close the gap.  No implication from EFI to
$\mathsf P\ne\PSPACE$ is currently known, and our result does not give one; \Cref{sec:rank} returns
to this point.

\paragraph{Oracle separations and their lifting.} Pseudorandom states are separated from one-way
functions in~\cite{Kre21,KQST23}.  The CHRS model was introduced in~\cite{CCS24}, together with a
separation of single-copy from multi-copy pseudorandomness.  Both CHRS and a Haar random swap model
are used in~\cite{BCN24}; we rely only on the CHRS results there, since the swap-model section
carries an author's note retracting its EFI-versus-$\OWSG$ separation after a bug was found.  A
unitary oracle with EFI and $\QEFID$ pairs but no $\OWSG$s is given in~\cite{BMM24}, which is the
opposite orientation to ours, since it \emph{retains} $\OWPuzz$s.  $\OWPuzz$s are
separated from their efficiently verifiable variant in~\cite{CGG24}.  Finally, the current framework
for lifting CHRS separations to unitary ones is~\cite{GZ25}; our separation does not need it, since
\Cref{thm:single} already gives a single classical oracle.

\paragraph{Communication complexity.} Our security reduction is, to our knowledge, the first to route
the security of a quantum-state pair relative to a classical-query oracle through a communication
lower bound.  The hard problem is $\VSP$~\cite{KR11}, and Kerenidis et al.~\cite{KLLRX12} give a prior-free
information-complexity bound for a discretized robust variant.  On the upper-bound side, the
$O(\sqrt n)$ one-way protocol is stated by Raz~\cite{Raz99}, first published
in~\cite[App.~A]{Mon19}, and given a computationally efficient form in~\cite{GS19}.  That last paper
belongs to a line of work on compressed classical descriptions of quantum states, beginning
with~\cite{Aar04}, in which the description size is governed by the same one-way communication
bounds.  A matching $\Omega(\sqrt n)$ lower bound for one-way protocols on gapped instances follows
from the one-way complexity of Partial Matching~\cite{GKKRW07}, by a reduction recorded
in~\cite{GS19}; \Cref{sec:oneway-rate} compares it with \Cref{prop:oneway-rate}.

\addtocontents{toc}{\protect\setcounter{tocdepth}{1}}
\subsection{Organization}
\Cref{sec:overview} describes the whole argument informally.  \Cref{sec:prelims} then fixes notation and collects the background we use, including
the query model, the conventions of communication complexity, and the classical lower bound for
$\VSP$; the matrix concentration inequality is stated separately in \Cref{sec:concentration}, where
it is first used.  \Cref{sec:framework} proves the reduction to communication complexity, together
with the quantum protocol that bounds its reach.  \Cref{sec:oracle} defines the oracle and proves
that classical-output hardness, of the search, decision, and interactive kinds, does not survive it.
\Cref{sec:security} applies the reduction to our pair, proves EFI security, and states the two rates
of Theorem~E, whose proofs are deferred.  \Cref{sec:coherent} proves the coherent-query, quantum-advice, and hybrid
bounds.  \Cref{sec:single} replaces the state source by a Boolean function and gives the single
deterministic classical oracle of Theorem~A.  \Cref{sec:lift} states the conjecture that would place
the separation in the standard quantum oracle model, and \Cref{sec:discussion} collects the open
problems.  Three appendices follow: \Cref{app:attack} contains the conditional-sampling attack,
\Cref{app:rates} the proofs of the two rates, and \Cref{app:deferred} two results used only locally.

\Cref{fig:path} shows which parts of the paper the proof of \Cref{thm:single} uses.

\begin{figure}[t]
\centering
\begin{tikzpicture}[
  spine/.style={draw=black!65,fill=black!5,align=left,text width=58mm,inner sep=5pt,font=\small},
  side/.style={draw=black!30,dashed,align=left,text width=48mm,inner sep=5pt,font=\small,text=black!55},
  ar/.style={-{Latex[length=4.5pt]},black!65,line width=.5pt},
  arside/.style={-{Latex[length=3.5pt]},black!30,dashed,line width=.4pt}]

\node[spine] (s1) at (0,0)
  {\textbf{\Cref{sec:framework}}\\ Security from a bounded classical transcript
   (\Cref{thm:transcript-security})};
\node[spine] (s2) at (0,-2.15)
  {\textbf{\Cref{sec:rank}--\ref{sec:no-owpuzz}, \Cref{app:attack}}\\ The ranked counting oracle; no
   one-way puzzles (\Cref{thm:no-owpuzz})};
\node[spine] (s3) at (0,-4.3)
  {\textbf{\Cref{sec:reference}--\ref{sec:security-thm}}\\ Reference copies, then EFI security
   under classical access (\Cref{thm:EFI-security})};
\node[spine] (s4) at (0,-6.6)
  {\textbf{\Cref{sec:khintchine}--\ref{sec:onequery}}\\ The variance bound
   (\Cref{lem:khintchine}); one coherent query (\Cref{thm:one-query})};
\node[spine] (s5) at (0,-8.9)
  {\textbf{\Cref{sec:hybrid}}\\ Classical transcript, then one coherent query
   (\Cref{thm:hybrid})};
\node[spine] (s6) at (0,-11.05)
  {\textbf{\Cref{sec:single}}\\ One deterministic classical oracle (\Cref{thm:single}); the
   black-box barrier (\Cref{cor:barrier})};

\foreach \a/\b in {s1/s2,s2/s3,s3/s4,s4/s5,s5/s6} \draw[ar] (\a) -- (\b);

\node[side] (b1) at (7.4,-2.15)
  {\Cref{sec:no-qefid}--\ref{sec:interactive}: $\QEFID$ pairs and the interactive collapse
   (\Cref{thm:interactive})};
\node[side] (b2) at (7.4,-4.3)
  {\Cref{sec:exact-rate}--\ref{sec:oneway-rate}, \Cref{app:rates}: the two exact rates
   (Theorem~E)};
\node[side] (b3) at (7.4,-6.6)
  {\Cref{sec:advice-is-a-query}, \Cref{app:deferred}: what quantum advice and what reference
   copies are worth};
\node[side] (b4) at (7.4,-11.05)
  {\Cref{sec:lift}: what fully coherent access would need (\Cref{conj:lift})};

\draw[arside] (s2) -- (b1); \draw[arside] (s3) -- (b2);
\draw[arside] (s4) -- (b3); \draw[arside] (s6) -- (b4);
\end{tikzpicture}
\caption{Dependency structure. Solid boxes are used in the proof of \Cref{thm:single}; dashed boxes
are not.}
\label{fig:path}
\end{figure}
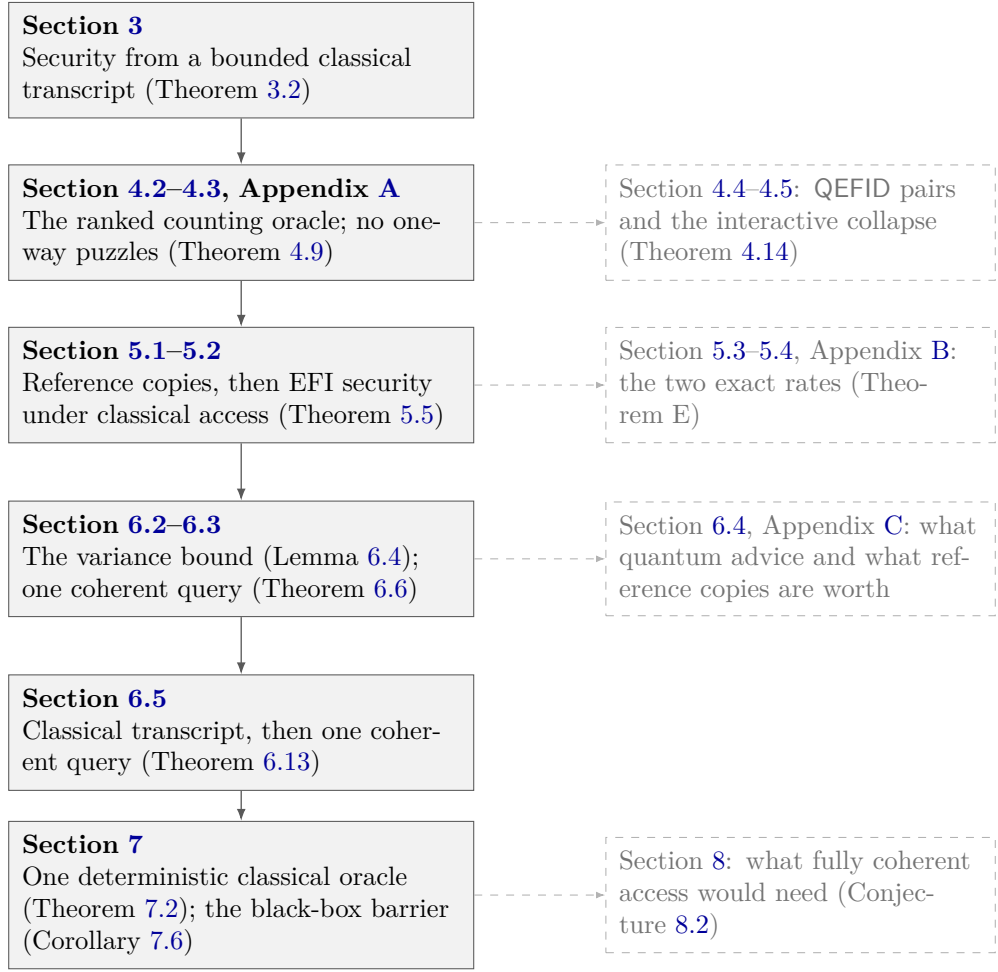
\addtocontents{toc}{\protect\setcounter{tocdepth}{2}}

\section{Preliminaries}
\label{sec:prelims}

This section fixes notation and collects the background the paper uses.  Three of its subsections
gather material from the areas the proof draws on.  \Cref{sec:prelim-query} fixes the form in which
we write a query algorithm, \Cref{sec:prelim-cc} recalls the conventions of classical communication
complexity, and \Cref{sec:prelim-vsp} states the communication lower bound we import.  There is a
fourth such import, the matrix concentration inequality, but we defer it to
\Cref{sec:concentration}, the first subsection of \Cref{sec:coherent}, where it is first used.  For
EFI pairs we follow~\cite{BCQ23}, and for one-way puzzles~\cite{KT24,CGG24,CountCrypt}.

\subsection{Notation}
\label{sec:prelim-notation}

We write $\QPT$ for quantum polynomial time.  A function is negligible, written $\negl(\lambda)$, if
it decays faster than every inverse polynomial.  The security parameter is $\lambda$, and the ambient
dimension is $n=n(\lambda)=2^{\lambda}$, a computable power of two, so that $\bbC^n$ is exactly
$\lambda$ qubits.  In fact the arguments use only that $n$ is even and $2^{\Theta(\lambda)}$, and
they use the two halves of that bound for different purposes.  Every security estimate uses the
lower bound $n\ge2^{\Omega(\lambda)}$.  The upper bound $n\le2^{O(\lambda)}$, on the other hand,
makes the dimension addressable: a $\QPT$ algorithm has to write indices into $[n]$, and so needs
$\log_2n=\poly(\lambda)$.  It also lets us report a bound of the form $n^{-\Omega(1)}$ as
$2^{-\Omega(\lambda)}$.  Following~\cite{BCQ23} we reserve $n$ for the dimension and use $\lambda$
for the security parameter.  Part of the $\OWPuzz$ literature writes $n$ for the security parameter
instead, which would collide with our usage.

We write $\ln$ for the natural logarithm and $\log_2$ for the base-two logarithm.  Inside a
concentration bound the base of a logarithm affects only the universal constant in front, and there
we write $\log$.  Finally, $\mathbf 1[\cdot]$ denotes the indicator of an event, and $\TV$
denotes total variation distance between distributions over a discrete set.

\subsection{Operators and effects}
\label{sec:prelim-operators}

For an operator $Y$ on a finite-dimensional Hilbert space we write $\norm Y_1$, $\opnorm Y$ and
$\norm Y_2=\sqrt{\tr(Y^\dagger Y)}$ for the trace, operator and Hilbert--Schmidt norms.  We use
$\opnorm Y\le\norm Y_2\le\norm Y_1$ throughout, and for a positive semidefinite $Y$ on $\bbC^d$ we
also use $\norm Y_2\le\sqrt{\opnorm Y\tr Y}\le\sqrt{\tr Y}$ when $\opnorm Y\le1$.  We write
$A\preceq B$ when $B-A$ is positive semidefinite, $\Herm_d$ for the $d\times d$ Hermitian matrices,
and $\Id$ for the identity.  The trace distance is $\td(\rho,\sigma)=\frac12\norm{\rho-\sigma}_1$.

An \emph{effect} is an operator $E$ with $0\preceq E\preceq\Id$, equivalently a single element of a
two-outcome POVM $\{E,\Id-E\}$.  The acceptance probability of a two-outcome measurement on a state
$\rho$ is then $\tr(E\rho)$, for the effect $E$ associated with the accepting outcome.  Effects are
how a computationally unbounded distinguisher enters the analysis, and they capture it completely:
any such distinguisher is a channel followed by a two-outcome measurement, and that composition is
again a two-outcome measurement.  When an operator $\Pi$ acts on a tensor product $\bbC^{\cX}\otimes\cK$, for
a finite index set $\cX$ and a finite-dimensional Hilbert space $\cK$, we write
$\Pi_{xy}\eqdef(\bra x\otimes\Id)\,\Pi\,(\ket y\otimes\Id)$ for its blocks, which are operators on
$\cK$.  If $\Pi$ is Hermitian then $\Pi_{xy}^\dagger=\Pi_{yx}$.

Three devices recur in \Cref{sec:coherent}, and we fix them here.  The first is the \emph{Hermitian dilation} of a
rectangular or non-Hermitian matrix $Y$,
\[
  \mathcal Y\ \eqdef\ \begin{pmatrix}0&Y\\ Y^\dagger&0\end{pmatrix},
\]
which is Hermitian and satisfies $\opnorm{\mathcal Y}=\opnorm Y$.  It converts a statement about
Hermitian matrices into one about arbitrary matrices, at the cost of doubling the dimension.  The
second device is purification.  For a state $\rho$ on $\bbC^d$, a \emph{purification} is a pure state
on a space of dimension at most $d^2$ whose reduced state is $\rho$, and handing a purification to a
receiver in place of $\rho$ only increases what that receiver can do.

The third is a net, which we record as a lemma because three separate arguments below reduce an
operator norm to a finite maximum in exactly this way.

\begin{lemma}[A net captures the norm of a Hermitian operator]
\label{lem:net}
Let $\cV$ be a complex Hilbert space of dimension $K$.  Its unit sphere carries a $\frac14$-net
$\mathcal N$ with $|\mathcal N|\le9^{2K}$, and for every such net and every Hermitian $Y$ on $\cV$,
\[
  \opnorm Y\ \le\ 2\max_{\gamma\in\mathcal N}\big|\bra\gamma Y\ket\gamma\big| .
\]
\end{lemma}

\begin{proof}
The unit sphere of $\cV$ is the unit sphere of a real space of dimension $2K$, which carries a
$\delta$-net of size at most $(1+2/\delta)^{2K}$ by the standard volume comparison; at
$\delta=\frac14$ that is $9^{2K}$.

For the inequality, recall that $\opnorm Y=\sup_{\norm x=1}|\bra xY\ket x|$ when $Y$ is Hermitian.
Take $x$ attaining that supremum and $\gamma\in\mathcal N$ with $\norm{x-\gamma}\le\frac14$.  Writing
the difference as $\bra{x-\gamma}Y\ket x+\bra\gamma Y\ket{x-\gamma}$ gives
\[
  \big|\bra xY\ket x-\bra\gamma Y\ket\gamma\big|\ \le\ 2\norm{x-\gamma}\,\opnorm Y\ \le\ \tfrac12\opnorm Y ,
\]
so that $\opnorm Y\le\max_{\gamma\in\mathcal N}|\bra\gamma Y\ket\gamma|+\frac12\opnorm Y$, which
rearranges to the claim.
\end{proof}

\subsection{Haar-random subspaces and the half-subspace pair}
\label{sec:prelim-haar}

We write $\Un$ and $\On$ for the unitary and real orthogonal groups, and $\mathrm{SU}(n)$ and
$\mathrm{SO}(n)$ for their subgroups of determinant one.  Each of these carries a unique
translation-invariant probability measure, its Haar measure.  On the sphere $S^{n-1}$ we write
$\sigma$ for the uniform probability measure, and $\nu_H$ for the uniform probability measure on the
unit sphere of a subspace $H$.

By a \emph{Haar-random subspace of dimension $n/2$} in $\bbR^n$, or in $\bbC^n$, we mean one drawn
from the rotation-invariant probability measure on the Grassmannian of $n/2$-dimensional subspaces.
Concretely, $H=Q\cdot\mathrm{span}(e_1,\dots,e_{n/2})$ for $Q$ Haar on $\On$, or on $\Un$ in the
complex case.  Rotation invariance means that $gH$ has the same law as $H$ for every fixed
$g\in\On$, and this is the only property of the measure that we use anywhere.  For such an $H$ we
write $P_H$ for the orthogonal projector onto it, and
\begin{equation}
  R_H \eqdef 2P_H-\Id
  \label{eq:reflection}
\end{equation}
for the associated reflection, which is a Hermitian involution with $R_H^2=\Id$ and $\tr R_H=0$.
Write $D$ for the diagonal matrix with $n/2$ entries $+1$ followed by $n/2$ entries $-1$.  Then
$R_H=QDQ^\top$ in the real case, and $R_H=UDU^\dagger$ with $U$ Haar on $\Un$ in the complex one, so
that the reflection is a Haar conjugate of one fixed traceless involution.  This is the form in which
every estimate in \Cref{sec:coherent} uses it.

The two half-subspace states are
\begin{equation}
  \rho_{0,H}=\frac{2P_H}{n}=\frac{\Id+R_H}{n},\qquad
  \rho_{1,H}=\frac{2P_{H^\perp}}{n}=\frac{\Id-R_H}{n},\qquad
  \rho_{0,H}-\rho_{1,H}=\frac{2R_H}{n} .
  \label{eq:pair}
\end{equation}
They have orthogonal supports, so that $\td(\rho_{0,H},\rho_{1,H})=1$ and the measurement
$\{P_H,P_{H^\perp}\}$ distinguishes them with certainty.  The third identity in \eqref{eq:pair} is
the reason the analysis stays linear in the hidden object.  Indeed, for any effect $E$ acting on a
single copy, the difference of the acceptance probabilities on $\rho_{0,H}$ and $\rho_{1,H}$ is
\begin{equation}
  \tr\big(E(\rho_{0,H}-\rho_{1,H})\big)\ =\ \frac2n\tr(E\,R_H),
  \label{eq:bias-linear}
\end{equation}
and since $\tr R_H=0$ only the traceless part of $E$ contributes.  We call a quantity of the form
$\tr(AR_H)$,
with $A$ fixed, a \emph{linear statistic} of the reflection.

Our construction uses the real ensemble, since that is the setting of the communication lower bound
in \Cref{sec:prelim-vsp}.  The arguments of \Cref{sec:coherent}, on the other hand, are stated for
the complex ensemble.  They transfer to the real one with a change of universal constant, and
\Cref{lem:su-so} makes that precise.

\subsection{The query model}
\label{sec:prelim-query}

Let $\cX$ be a finite set of query labels, write $M=|\cX|$ for the \emph{width}, and let
$f\colon\cX\to\{\pm1\}$ be a truth table, whose values we also write $f_x=f(x)$.  The \emph{phase
oracle} for $f$ is the unitary $D_f\otimes\Id$ on $\bbC^{\cX}\otimes\cK$ with
$D_f=\sum_xf_x\proj x$, where $\cK$ carries the rest of the algorithm's workspace.  A Boolean
function $g\colon\cX\to\{0,1\}$ is accessed instead in the \emph{bit-flip} form $\ket x\ket
c\mapsto\ket x\ket{c\oplus g(x)}$.  The two forms are interchangeable: preparing the answer qubit in
$\ket-$ turns a bit-flip query into a phase query, and conversely a phase query on the domain
$\cX\times\{0,1\}$ with truth table $(-1)^{c\,g(x)}$ implements a bit-flip query.  The two models are
equivalent up to a factor $2$ in the width, and we work with phase oracles throughout.
That conversion costs us nothing, since our bounds hold for \emph{every} truth table $f$ and depend
on the width only through $\log(2M)$.

Every algorithm below holds one copy of the challenge state, which is either $\rho_{0,H}$ or
$\rho_{1,H}$, and queries an oracle whose truth table is an arbitrary function of the hidden
subspace.  We describe such an algorithm in two ways, one general and one specific to a single query.

The general description comes first.  An algorithm making any number of queries, of arbitrary
computational power, is a channel followed by a two-outcome measurement.  Its acceptance probability
on a challenge $\rho$ is $\tr(E_f\rho)$ for an effect $E_f$ on $\bbC^n$ that depends on the
truth table $f$ but not on the challenge, and we measure its success by the \emph{bias}
\begin{equation}
  Z_H(f)\ \eqdef\ \tfrac12\tr\big(E_f(\rho_{0,H}-\rho_{1,H})\big)\ =\ \frac1n\tr\big(E_fR_H\big),
  \label{eq:bias-def}
\end{equation}
which by \eqref{eq:bias-linear} is half the distinguishing advantage.  Note that the dependence on
the number of queries is hidden inside $E_f$, and that for each fixed $f$ the bias is a linear
statistic of $R_H$.

The second description is specific to one query, and it is the form that the proof of
\Cref{thm:one-query} needs.  A one-query algorithm applies a unitary to the challenge together with
fresh ancillas, queries the phase oracle once, applies a second unitary, and measures.  Absorbing the
final unitary into the measurement, the pre-query processing is an isometry
\begin{equation}
  W\colon\bbC^n\to\bbC^{\cX}\otimes\cK,\qquad W\ket\psi=\sum_{x\in\cX}\ket x\,A_x\ket\psi,\qquad
  \sum_x A_x^\dagger A_x=\Id ,
  \label{eq:normal-form}
\end{equation}
where $A_x$ is the \emph{branch} on which the algorithm queries the label $x$.  The completeness
relation on the right is exactly the statement that $W$ preserves norms.  The post-query measurement
is then an effect $0\preceq\Pi\preceq\Id$ on $\bbC^{\cX}\otimes\cK$ with blocks $\Pi_{xy}$, and in
this notation the effect of \eqref{eq:bias-def} is
$E_f=W^\dagger(D_f\otimes\Id)\Pi(D_f\otimes\Id)W$.  Finally, we attach to the isometry the
\emph{query mass}
\begin{equation}
  \tau_x\ \eqdef\ \tfrac1n\tr(A_x^\dagger A_x)\ =\ \tr\big(A_x^\dagger A_x\cdot\tfrac{\Id}n\big),
  \label{eq:query-mass}
\end{equation}
which is the probability that the query lands on the label $x$ when the challenge is maximally
mixed.  By the completeness relation $\tau$ is a probability vector on $\cX$, and it depends on the
algorithm alone, not on $H$.

\subsection{EFI pairs and one-way puzzles}
\label{sec:prelim-primitives}

\begin{definition}[EFI pair,~\cite{BCQ23}]
\label{def:efi}
A family $\{(\rho_{0,\lambda},\rho_{1,\lambda})\}_\lambda$ of pairs of mixed states is an \emph{EFI pair} if:
\begin{enumerate}[label=(\roman*),leftmargin=2.2em]
  \item \textbf{Efficient generation:} a $\QPT$ algorithm $G$ satisfies $G(1^\lambda,b)=\rho_{b,\lambda}$ for both $b\in\{0,1\}$;
  \item \textbf{Statistical farness:} $\td(\rho_{0,\lambda},\rho_{1,\lambda})\ge\delta(\lambda)$ for some inverse polynomial $\delta$;
  \item \textbf{Computational indistinguishability:} every non-uniform $\QPT$ distinguisher $\mathsf D$ has advantage $\negl(\lambda)$ in distinguishing $\rho_{0,\lambda}$ from $\rho_{1,\lambda}$.
\end{enumerate}
Here \emph{non-uniform} is as in~\cite{BCQ23}: $\mathsf D$ receives an advice state of polynomial size, which
may be quantum.
\end{definition}

\begin{definition}[Classical-advice EFI pair]
\label{def:efi-classical}
A \emph{classical-advice EFI pair} satisfies (i) and (ii) of \Cref{def:efi} and, in place of (iii),
computational indistinguishability against every $\QPT$ distinguisher receiving a classical advice
string of polynomial length.
\end{definition}

What we construct is a classical-advice EFI pair, and in fact it is secure against rather more than
that definition asks.  It is secure against quantum advice of $O(\lambda)$ qubits together with
everything else the model grants (\Cref{cor:hybrid-qadvice}), and against quantum advice of any
polynomial size on its own (\Cref{cor:quantum-advice}).  Only one combination lies beyond these,
namely polynomial-size quantum advice together with classical queries, and that combination is what
separates \Cref{def:efi-classical} from \Cref{def:efi} here (\Cref{sec:security-thm}).  Note that
the convention is local to EFI distinguishers.  \Cref{def:owpuzz} below keeps the standard notion,
under which the attacker to be ruled out may carry quantum advice, and since our attacker is
classical, nothing is lost there.

Some later works, for instance~\cite{MetaEFI25}, require farness $1-\negl(\lambda)$ instead.  The
two requirements are equivalent up to standard polarization, and in any case our constructions
achieve $\td=1$ in the source model and $1-2^{-\Omega(\lambda)}$ relative to the single classical
oracle of \Cref{sec:single}, so the distinction is immaterial here.

One feature of the oracle setting needs care.  There both $G$ and $\mathsf D$ have oracle access, and
the third condition becomes meaningful only once we fix the access granted to $\mathsf D$: restricting $\mathsf D$ to
classical queries and forbidding it queries altogether are different security notions.  We therefore
always name the class, and ours is \Cref{model:access}.  Moreover, since $G$ is \emph{public}, a
security proof must in addition let $\mathsf D$ obtain polynomially many copies of both states, which we call
\emph{reference copies}.  Beyond this we use an EFI pair only as a pair of states that is hard to
distinguish.

\begin{definition}[One-way puzzle,~\cite{KT24,CGGH25}]
\label{def:owpuzz}
A \emph{one-way puzzle} is a pair $(\Samp,\Ver)$ where $\Samp$ is a uniform $\QPT$ algorithm
outputting a classical key--puzzle pair $(k,s)$, and $\Ver(k,s)\in\{0,1\}$ \emph{may be
computationally unbounded}, such that:
\begin{enumerate}[label=(\roman*),leftmargin=2.2em]
  \item \textbf{Correctness:} $\Pr_{(k,s)\gets\Samp(1^\lambda)}[\Ver(k,s)=1]\ge1-\negl(\lambda)$;
  \item \textbf{Security:} every non-uniform $\QPT$ adversary $\cA$ satisfies
    $\Pr_{(k,s)\gets\Samp(1^\lambda)}[\Ver(\cA(1^\lambda,s),s)=1]\le\negl(\lambda)$.
\end{enumerate}
If $\Ver$ is efficient the primitive is an \emph{efficiently verifiable} one-way puzzle ($\EVOWPuzz$).
\end{definition}

Inefficient verification is the default in~\cite{KT24}, because an efficient verifier would be
broken by a $\QCMA$ oracle.  Our target throughout is accordingly the unrestricted, inefficiently
verifiable notion.

A $\QEFID$ pair, introduced in~\cite[Def.~5]{CGG24} and stated in~\cite[Def.~2.1]{BMM24}, is the
classical-output analogue of an EFI pair: a pair of distributions over classical strings, samplable
by a $\QPT$ algorithm, statistically far, and computationally indistinguishable.  The two
classical-output primitives express hardness of different kinds.  A one-way puzzle expresses
classical-output hardness of \emph{search}, whereas a $\QEFID$ pair expresses classical-output
hardness of \emph{decision}.  They are nevertheless equivalent up to non-uniformity,
by~\cite[Lem.~8, Lem.~9, Cor.~14]{CGG24}.

\subsection{Classical communication complexity}
\label{sec:prelim-cc}

We recall the model in the form used in \Cref{sec:framework}, since it is not the form most familiar
to a reader coming from quantum cryptography.  In a two-party problem, Alice holds an input $u$ and
Bob holds an input $H$, and the two exchange classical messages according to functions fixed in
advance of the input, each message depending on the sender's own input and on the messages so far.
At the end one of them announces an output.  The \emph{cost} of the protocol is the maximum total
length of the messages exchanged, and this is the only resource charged: local computation is
unbounded on both sides, and neither party is required to be efficient, or even computable.  That
last point is what makes the model useful to us, since the party who will hold our hidden subspace
has to be able to answer arbitrary queries about it.

A protocol has \emph{shared randomness} if both parties see a common random string, drawn
independently of the inputs, before the protocol begins.  Equivalently, such a protocol is a
distribution over deterministic ones.  A \emph{one-way} protocol consists of a single message from
Alice to Bob, after which Bob announces the output; here we follow the definition of~\cite{KNR99}.
Finally, on a distribution over inputs together with a bit $b$ to be decided, the \emph{advantage} of
a protocol is $|\Pr[\text{correct}]-\frac12|$, the average being taken over the inputs, the bit, and
the randomness.

One structural fact is used twice below.  Fix a deterministic protocol and fix a full transcript
$\pi$.  Then the set of input pairs producing $\pi$ is a \emph{rectangle}, that is, a set of the form
$A\times B$ with $A$ a set of Alice's inputs and $B$ a set of Bob's.  The reason is that at each
round the sender's message depends only on its own input and on the transcript so far, so membership
in the set of inputs consistent with $\pi$ is decided separately on the two sides.  A protocol of
cost $L$ thus partitions the input space into at most $2^{L+1}$ rectangles.  For a one-way
protocol the partition is finer on Alice's side only: her message realizes a partition of her input
space into at most $2^{L+1}$ cells, and Bob then answers optimally given the cell and his own input.

We record last the convention relating the two notions of advantage that appear in the paper.  The
\emph{distinguishing advantage} of an algorithm $\mathsf D$ between two states is
$\Adv(\mathsf D)\eqdef|\Pr[\mathsf D(\rho_0)=1]-\Pr[\mathsf D(\rho_1)=1]|$, written $\Adv(\mathsf D;X)$ when the states depend on a
hidden object $X$.  With a uniform challenge bit, a distinguisher of advantage $\eps$ gives a
protocol of advantage exactly $\eps/2$, and we carry that factor $2$ explicitly rather than absorbing
it.

\subsection{The Vector-in-Subspace problem}
\label{sec:prelim-vsp}

\begin{definition}[Vector in Subspace,~\cite{KR11}]
\label{def:vsp}
In $\VSP_n$, Bob receives a subspace $H\subseteq\bbR^n$ of dimension $n/2$ and Alice a unit vector $u\in S^{n-1}$ promised to lie in $H$ or in $H^\perp$. The associated distribution draws $H$ from the rotation-invariant measure on the Grassmannian of
$n/2$-dimensional subspaces of $\bbR^n$, a uniform bit $b$, and $u$ uniformly from the unit sphere of $H$ (if $b=0$) or of $H^\perp$ (if $b=1$). The goal is to decide $b$ by classical randomized communication.
\end{definition}

The quantum upper bound is immediate.  Alice encodes $u$ in the amplitudes of $\lceil\log_2n\rceil$
qubits and sends them, and Bob applies $\{P_H,P_{H^\perp}\}$ and decides with certainty~\cite{KR11}.
It is the classical lower bound that we use.

\begin{theorem}[{Klartag--Regev~\cite[Thm.~4.2]{KR11}}]
\label{thm:vsp-lb}
There is a universal $c>0$ such that every classical randomized, two-way, shared-randomness protocol
whose messages are Borel functions of the inputs, and which decides $\VSP_n$ with error at most
$\frac13$ on every instance of the promise of \Cref{def:vsp}, communicates at least $c\,n^{1/3}$ bits.
\end{theorem}

Three features of~\cite{KR11} matter for the way we use it.  First, the bound is for \emph{two-way}
interactive protocols with shared randomness, and not only for one-way ones.

Second, the exact promise, that $u\in H$ or $u\in H^\perp$, is the case that Klartag and Regev call
$\VSP_0$, and it is exactly the case we need.  The Borel hypothesis in \Cref{thm:vsp-lb} is due to Klartag
and Regev, and it is needed only because the exact promise is supported on a measure-zero subset of
$S^{n-1}\times\mathrm{Gr}$, where $\mathrm{Gr}$ is the Grassmannian of $n/2$-dimensional subspaces,
so that a protocol has to realize a \emph{measurable} rectangle partition.  Our reduction satisfies
that hypothesis: for each fixing of the public randomness the protocol is a fixed algorithm acting on
the classical descriptions of $u$ and $H$, and so its message functions are Borel.  For the
robust promise $\VSP_\theta$ with $0<\theta<1/\sqrt2$ they remove the hypothesis entirely.

Third, their argument is a corruption bound, and the only use it makes of correctness is that the
errors under the two promise cases sum to at most $\frac23$.  What they prove is in fact
the distributional statement over the balanced distribution of \Cref{def:vsp}, which is the form we
use.

This lower bound does not rest on a single method.  The prior-free information complexity of the
discretized robust problem is also $\Omega(n^{1/3})$, by combining the main theorem
of~\cite{KLLRX12} with the same rectangle estimate of~\cite{KR11}.  That bound is prior-free, that
is, a maximum over input distributions rather than a statement about the distribution of
\Cref{def:vsp}, so we cite it as corroboration rather than as a component of the proof.

We use \Cref{thm:vsp-lb} only through the following small-advantage form.  Write $\beta(n,L)$ for the
supremum, over classical randomized two-way protocols of cost $L$ with shared randomness and Borel
message functions, of the advantage on the $\VSP_n$ distribution of \Cref{def:vsp}.  The complement,
symmetrization and majority-vote constructions in the proof below all preserve Borel measurability.

\begin{lemma}[Small-advantage bound]
\label{lem:vsp-small}
There is a universal $C>0$ such that $\beta(n,L)\le C\sqrt{(L+1)/n^{1/3}}$ for every even $n$ and every $L\ge0$. Moreover $\beta(n,0)=0$.
\end{lemma}

\begin{proof}
We amplify a protocol of small advantage into one of error $\frac13$ and apply \Cref{thm:vsp-lb}.
Let $P$ be an $L$-bit protocol of advantage $\eps$; replacing $P$ by its complement, we may assume
$\Pr[\text{correct}]=\frac12+\eps$.

We will show that $P$ can be replaced, at the same cost, by a protocol whose success probability is
$\frac12+\eps$ on \emph{every} instance of the promise.  Majority voting needs exactly that
uniformity over inputs, and $P$ guarantees only an average over the distribution.

The obvious device, rotating an instance, equalizes success \emph{within} each promise case but not
across the two, so we symmetrize the two cases against each other first.  If $(u,H,b)\sim\VSP_n$ then
so is $(u,H^\perp,1\oplus b)$.  Let $P'$ publicly flip a fair coin: on heads run $P(u,H)$; on tails Bob substitutes $H^\perp$ and
they output the complement of $P(u,H^\perp)$.  Then for each fixed $b$,
\[
  \Pr[P'\text{ correct}\mid b]=\tfrac12\Pr[P\text{ correct}\mid b]+\tfrac12\Pr[P\text{ correct}\mid1\oplus b]=\tfrac12+\eps,
\]
so the two conditional success probabilities are now equal.

Now we equalize within each case.  The $\VSP$ distribution is invariant under a common rotation, and
$\On$ acts transitively on pairs $(u,H)$ with $u\in H$, and on those with $u\in H^\perp$.  Sampling a
Haar $Q$ from public randomness and running $P'$ on $(Qu,QH)$ gives a protocol $P''$ of the
same cost whose success probability is $\frac12+\eps$ on every instance in the promise, which is the
uniformity we wanted.

Majority voting now applies.  Running $P''$ independently $t=\Theta(1/\eps^2)$ times and taking the
majority gives worst-case error at most $\frac13$, at cost $tL$.  \Cref{thm:vsp-lb} forces that cost
to be large, and the bound on $\eps$ follows:
\[
  tL\ \ge\ c\,n^{1/3},\qquad\text{so}\qquad \eps\ \le\ C\sqrt{\frac L{n^{1/3}}}\quad (L\ge1) .
\]

That leaves $L=0$.  With no communication the output is a function of the announcing
party's own input and the public randomness.  Bob's input $H$ is independent of $b$, and since $H$ is
uniform, Alice's marginal vector $u$ is sphere-uniform regardless of $b$; in either case the output is
independent of $b$, so $\beta(n,0)=0$.  Writing $L+1$ in place of $L$ absorbs this case at
the cost of a constant, and that is the stated bound.
\end{proof}

\section{Security from communication complexity}
\label{sec:framework}

Everything the oracle of \Cref{sec:oracle} will tell an adversary is a string of classical bits.
In this section we show that when that is so, security reduces to a communication lower bound,
however powerful the oracle may be.  We state the reduction for an arbitrary hidden object and an
arbitrary oracle, since the proof uses nothing about either beyond the fact that the answers are
classical.

Suppose, then, that the security of a primitive amounts to the claim that two states
$\rho_0(X),\rho_1(X)$ built from a hidden object $X$ are hard to distinguish, and suppose that every
route by which $X$ influences the adversary carries \emph{classical data}.  We may then place the
adversary in a two-party game, with Alice running it and holding the challenge while Bob holds $X$
and is computationally unbounded.  Each query becomes a message to Bob and each answer a message in
return, so the adversary's power is bounded by the length of that conversation and by nothing else,
and in particular not by any property of the oracle it queries.

One requirement shapes the definition that follows.  A communication protocol gives Alice a
\emph{classical} input, so the problem we associate to the pair has to hand her a classical
description of the challenge, from which she prepares the state herself.  Granting her that
description increases her power rather than restricting it, which is what makes the reduction
legitimate.  The same requirement, however, forces care over quantum side information depending on
$X$.  The theorem below does \emph{not} cover such information, and where it arises, namely for the
reference copies of \Cref{sec:reference}, we reduce it to classical data first.

\begin{definition}[The associated two-party problem]
\label{def:induced}
Let $X$ be a random hidden object taking values in a standard Borel space, and let
$(\rho_0(X),\rho_1(X))$ be a pair of states.  The \emph{associated two-party problem} $\Pi_X$ is the
following.  Bob receives $X$.  Alice receives, for a uniform bit $b$, a classical description $y$
drawn from a preparation kernel $\kappa_b(X)$; here $\kappa_b$ is a Borel map from $X$ to
distributions on a standard Borel space of descriptions, each description $y$ determines a pure state
$\proj{\psi_y}$ through a Borel map, and the mixture of those states under $\kappa_b(X)$ is
$\rho_b(X)$.  The kernels are part of the data of $\Pi_X$, and for the half-subspace pair they are
the uniform distributions on the unit vectors of $H$ and of $H^\perp$.  Alice also receives classical
side information $\omega(X)$, a Borel function of $X$.  The two then communicate classically, and
Alice outputs a guess for $b$.

Write $\beta_\Pi(L)$ for the supremum of the advantage over $L$-bit classical randomized protocols
for $\Pi_X$ on this distribution, with Borel message functions.
\end{definition}

\begin{theorem}[Security from a bounded classical transcript]
\label{thm:transcript-security}
Let $\cA_X\colon\{0,1\}^*\to\{0,1\}$ be \emph{any} deterministic Boolean truth table whose bits are measurable functions of $X$, not assumed efficiently computable, or computable at all. Let $\mathsf D$ be a quantum algorithm of arbitrary computational power whose dependence on $X$ enters \emph{only} through:
\begin{enumerate}[label=(\roman*),leftmargin=2.2em]
\item one challenge drawn from $\rho_0(X)$ or $\rho_1(X)$ with a uniform bit $b$, and classical side information $\omega(X)$, both as in \Cref{def:induced};
\item classical query strings to $\cA_X$, in a self-delimiting encoding, with total query-plus-answer length at most $L$;
\item classical advice $\alpha$ of length at most $a$, which may depend on $X$ through any Borel map;
\end{enumerate}
and whose instruments and message functions are Borel in the classical data it holds, and which may otherwise use arbitrary computation and arbitrary resources independent of $X$. Then
\[
  \E_X\big[\Adv(\mathsf D;X)\big]\ \le\ 2\,\beta_\Pi(L+a+1).
\]
\end{theorem}

\begin{proof}
We build a protocol for $\Pi_X$ in which Alice plays $\mathsf D$ and Bob holds $X$.  The advice is not a
query, so it must be paid for separately, and Bob supplies it before the simulation begins.

Bob holds $X$ and is unbounded, so he can compute an advice string $\alpha^*(X)$ that maximizes $\mathsf D$'s
advantage.  We take the lexicographically first among the finitely many maximizers, which makes
$X\mapsto\alpha^*(X)$ Borel.  He also computes a sign bit
\[
  \sigma(X)\ \eqdef\ \mathbf 1\big[\,\Pr[D_{\alpha^*}(\rho_0)=1]>\Pr[D_{\alpha^*}(\rho_1)=1]\,\big] ,
\]
and sends the pair $(\alpha^*,\sigma)$, at a cost of at most $a+1$ bits.

Alice now receives the classical description of the challenge and of $\omega(X)$, prepares the
corresponding states, and runs $D_{\alpha^*}$.  All of her quantum computation is local: unitaries,
weak measurements, and coherent reuse of a post-measurement state are all operations on her own
registers.

The queries are the only traffic.  Whenever $\mathsf D$ issues a classical query $z$ to $\cA_X$, Alice sends
$z$ to Bob; Bob evaluates the deterministic bit $\cA_X(z)$ and returns it.  This exchange is the only
place where we use the hypothesis that the answers are classical bits.  Alice then resumes $\mathsf D$, takes
its output as her guess for $b$, and flips that guess if $\sigma=1$.

Finally, a quantum Alice exchanging classical messages is simulated by a
classical one.  From the classical description of her state and the history so far, an unbounded
classical Alice computes the exact conditional distribution of $\mathsf D$'s next query and samples from it.
So the result is a classical randomized protocol, and we have only given Alice more power than she
would otherwise have.

It remains to check that the simulation is faithful, and to count the bits.  Alice's simulated view is
identical in distribution to $\mathsf D$'s real view, since the challenge is a genuine sample and every
oracle answer matches the true truth table.  Hence for every fixed $X$ the protocol is correct with
probability at least $\frac12+\frac12\Adv(\mathsf D;X)$, the sign bit having converted the absolute value in
$\Adv(\mathsf D;X)$ into a one-sided quantity.  The communication is the charged transcript plus the advice, at most
$L+a+1$ bits.  Averaging over $X$ and comparing with $\beta_\Pi$ gives
\[
  \frac12\E_X[\Adv(\mathsf D;X)]\le\beta_\Pi(L+a+1),
\]
as claimed.
\end{proof}

The hypotheses of \Cref{thm:transcript-security} deserve two comments, one on what they permit and
one on what they exclude.

\paragraph{Assumptions on the oracle.}  \Cref{thm:transcript-security} makes no
assumption on $\cA_X$ beyond its being a deterministic function of classical strings.  It may be
uncomputable, and it may encode $X$ in full.  What the theorem bounds is not the oracle but how much
of it reaches the adversary, and it is for this reason that the counting oracle of \Cref{sec:oracle}
may be computationally very powerful without weakening security, as \Cref{sec:rank} discusses.

The requirement is tight, in the sense that classical \emph{output} alone would not suffice.  To see
this, consider an oracle that takes the unknown challenge as quantum \emph{input}, measures
$\{P_H,P_{H^\perp}\}$, and returns one classical bit.  Its output would be classical, and yet it
would distinguish the pair with certainty.  Coherent access, quantum advice depending on the oracle,
and an oracle sharing entanglement with the adversary are likewise not classical transcripts; nor,
without further work, are quantum reference copies that depend on $X$.

\paragraph{Adaptivity and the number of queries.}  Our bound depends only on the
\emph{total transcript length}.  In a two-way protocol adaptivity is free, and $q$ queries of $\ell$
bits amount to $q(\ell+1)$ bits once the answers are counted.  Handling polynomially many fully
adaptive queries costs us no additional work: it follows from the form of the reduction alone.

Set this beside a one-coherent-query bound such as~\cite{LMW24}.  The comparison is not a question
of one result being stronger than the other.  The two resources are formally
incomparable.  A coherent query is depth one at unbounded width, whereas polynomially many adaptive
classical queries are depth polynomial at width one, and each is stronger than the other on some
task~\cite{AA15,CFHL21}.  For the purposes of synthesis, classical queries are the easy
case (\Cref{obs:classical-synthesis}).

\subsection*{The quantum upper bound}
\label{sec:converse}

\Cref{thm:transcript-security} converts a communication lower bound into security, and its bound
$2\beta_\Pi$ certifies nothing once $\Pi_X$ has a cheap classical protocol, so the method applies only
to pairs whose associated problem is classically hard.  The complementary requirement, that the holder
of the hidden object be able to use the pair, is automatic.

\begin{proposition}[Quantum communication always suffices]
\label{prop:converse}
Let $X$ be a hidden random object, let $(\rho_0(X),\rho_1(X))$ be states on $\bbC^d$, and let $\Pi_X$
be the associated problem of \Cref{def:induced} with no side information. If
$\E_X\td(\rho_0(X),\rho_1(X))\ge\delta$, then $\Pi_X$ has a protocol of advantage at least $\delta/2$
in which Alice sends $\lceil\log_2d\rceil$ qubits and Bob returns one classical bit.
\end{proposition}

\begin{proof}
Alice prepares from her classical description one sample of $\rho_b(X)$ on $\lceil\log_2d\rceil$
qubits and sends it.  Bob, who holds $X$ and is unbounded, applies the Helstrom measurement for the
pair $(\rho_0(X),\rho_1(X))$ and returns its outcome.  Since the success probability is affine in the
state received, the average over Alice's sampling randomness is the success probability on
$\rho_b(X)$, namely $\frac12+\frac12\td(\rho_0(X),\rho_1(X))$; averaging over $X$ gives advantage at
least $\delta/2$, as required.
\end{proof}

Together the two statements give a general criterion, which we state on its own since hiding a
subspace is only one way to apply it.  It mentions neither subspaces, nor the counting oracle, nor
Vector-in-Subspace.

\begin{corollary}[A communication gap yields a pair secure against classical transcripts]
\label{cor:recipe}
Let $X$ be a hidden random object and $(\rho_0(X),\rho_1(X))$ a pair of states on $\bbC^d$ with
$\E_X\td(\rho_0(X),\rho_1(X))\ge\delta$, and let $\Pi_X$ be the associated problem of
\Cref{def:induced}.  Suppose that $\beta_\Pi(\ell)\le\gamma(\ell)$ for a function $\gamma$.  Then for
\emph{every} deterministic oracle $\cA_X$, however powerful, every distinguisher whose dependence on
$X$ runs through at most $L$ bits of classical query traffic with $\cA_X$ and $a$ bits of classical
advice has average advantage at most $2\gamma(L+a+1)$; while the pair itself is usable, in the sense
that $\Pi_X$ has a $\lceil\log_2d\rceil$-qubit quantum protocol of advantage $\delta/2$.
\end{corollary}

\begin{proof}
The first assertion is \Cref{thm:transcript-security} with $\beta_\Pi$ replaced by its upper bound
$\gamma$; the second is \Cref{prop:converse}.
\end{proof}

So any two-party problem that is expensive for classical communication at small advantage, and whose
two cases are statistically far, gives a pair hidden from every classical-transcript adversary.  The
work in applying it is finding an $X$ for which the classical bound is both strong and
provable at small advantage, and \Cref{sec:prelim-vsp} explains why $\VSP$ is the example we know.

\paragraph{Limitations of the reduction.}  \Cref{prop:converse} is the reason no
communication argument can ever reach coherent queries, and it says more than the familiar
observation that quantum communication is sometimes cheaper.  For \emph{every} statistically far
pair, the associated problem is easy for quantum communication at logarithmic cost, with the holder
of the hidden object performing the measurement.  The gap we require is not, then, a fortunate
accident of $\VSP$.  Rather, the classical side has to be hard while the quantum side, which is
always easy, remains inaccessible.  \Cref{sec:lift} explains why the easy quantum protocol is
nevertheless not an attack, and \Cref{sec:coherent} supplies the query-complexity bounds that replace
the unavailable communication argument.

The reduction is also one-directional.  It hands Alice a classical description of her challenge, which the distinguisher it simulates does not
have, and a classical protocol may exploit that description in ways no holder of a single copy can
reproduce.  Indeed, if the descriptions of the two branches differ, a zero-communication protocol
reads $b$ straight off its input, whereas a one-copy distinguisher is bounded by the trace distance.
Classical hardness of $\Pi_X$ is thus sufficient for security against classical transcripts, but
it is not necessary, and the security of a pair can exceed what \Cref{thm:transcript-security}
certifies.  For the half-subspace pair the two receivers provably differ, as \Cref{sec:coherent}
shows.

\section{The oracle and the primitives it removes}
\label{sec:oracle}

The requirements on the oracle come from \Cref{thm:transcript-security}.  We want an oracle that
answers every probability-estimation question, so that one-way puzzles cannot survive it, and one
that answers in classical bits, so that the theorem applies to it.  This section builds such an
oracle and proves the first of the two halves of our result: relative to it, classical-output
hardness of search and of decision both disappear.

Fix $n=n(\lambda)=2^{\lambda}$.  For each $\lambda$, independently across $\lambda$, sample a
Haar-random half-dimensional real subspace $H_\lambda\subseteq\bbR^{n(\lambda)}$ and set
$\rho_{b,\lambda}$ as in \eqref{eq:pair}.  The supports are orthogonal, so that
$\td(\rho_{0,\lambda},\rho_{1,\lambda})=1$.

\subsection{The state source}
\label{sec:source}

\begin{model}[State source $\cR$]
\label{model:source}
On input $(b,1^\lambda)$, $\cR$ returns one \emph{fresh} copy of $\rho_{b,\lambda}$: it samples a Haar-random unit vector from $H_\lambda$ (if $b=0$) or $H_\lambda^\perp$ (if $b=1$) and returns that pure state, with the sampling randomness traced out. The input registers are measured in the computational basis before the sample is drawn, so that $\cR$ is the channel
\[
  \xi\ \longmapsto\ \sum_{b,\lambda}\bra{b,1^\lambda}\xi\ket{b,1^\lambda}\cdot\rho_{b,\lambda},
\]
a completely positive trace-preserving map on classical-input registers, extended to malformed inputs by returning a fixed state.
\end{model}

We shall use two properties of \Cref{model:source} later.  The first is that $\cR$ is a legitimate
channel even when its input register is entangled with the caller's workspace, because the dephasing
makes the index classical.  A call with a superposed index is a call with a classical index
drawn from the induced distribution.  Precisely, on a joint state $\Xi_{WI}$ of the caller's
workspace $W$ and the input register $I$,
\[
  (\mathrm{id}_W\otimes\cR)(\Xi_{WI})\ =\ \sum_i\big(\bra i_I\Xi_{WI}\ket i_I\big)\otimes\rho_i ,
\]
so the workspace stays correlated with the dephased index.  There is consequently no ambiguity of
Stinespring dilation to resolve, and no access mode for $\cR$ beyond sample access.  Note that
after $t$ adaptive calls the retained joint state need not be a mixture of products, since a later
index may depend on an earlier outcome; the domination below does not assume that it is.

The second property is that the output is a fixed state for each index, so that $t$ calls at security
parameter $\lambda$ yield at most $t$ copies of $\rho_{0,\lambda}$ and $\rho_{1,\lambda}$ in total.
This is dominated by granting $t$ independent copies of each in advance and then, after each dephased
index, swapping in the next unused copy of the corresponding state, which reproduces the joint state
exactly.  That is the form in which \Cref{sec:reference} charges them.

The source doubles as the public EFI generator: $G(1^\lambda,b)$ calls $\cR(b,1^\lambda)$.  A
distinguisher may make additional calls on both branches, and \Cref{sec:reference} is where
we control them.

\paragraph{Comparison with the CHRS model.}  Each call to $\cR$ re-samples, so $t$
calls yield the product $\rho_b^{\otimes t}$.  This differs from the CHRS
model~\cite{CCS24,BCN24,MNY23}, where a single pure $\ket\phi$ is fixed at initialization and
returned on every call, giving the permutation-invariant and strongly correlated
$\E_\phi[\proj{\phi}^{\otimes t}]$.  A swap test tells the two apart already at $t\ge2$.  The
distinction does not affect the separation, since \Cref{sec:single} replaces $\cR$ by a Boolean
function altogether.

\subsection{The counting oracle}
\label{sec:rank}

The second layer answers the probability-estimation questions that, by~\cite{CGGH25,HM24},
characterize $\OWPuzz$s.  It is an \emph{exact probability-bit} oracle: it returns individual binary
digits of an output probability, so that a polynomial number of requests yields an
inverse-exponentially accurate estimate.  This is considerably stronger than the approximate counting
of Stockmeyer~\cite{Sto83}, which it subsumes, and the strength is deliberate, since
\Cref{thm:transcript-security} is indifferent to it.

There is one obstacle to defining such an oracle, and it is not a technicality.  The samplers whose
probabilities we estimate are themselves oracle-aided, and so may call the very oracle we are
defining.  The following example shows that unrestricted self-reference has no consistent solution at
all.

\begin{example}[Self-reference is inconsistent]
\label{ex:liar}
Suppose a single oracle $\Count$ answered output-probability bits of \emph{all} oracle-aided samplers, including those querying $\Count$. Consider the sampler $C$ that queries $\Count$ for the leading bit of $\Pr[C(1^\lambda)=0]$ and outputs $1$ if that bit is $1$, else $0$. If $\Count$ reports $1$ (a value $\ge\frac12$), then $C$ outputs $1$ deterministically, so the true probability of output $0$ is $0$, whose leading bit is $0$, a contradiction. If $\Count$ reports $0$, then $C$ outputs $0$ deterministically, so the probability of output $0$ is $1$, whose leading bit is $1$, again a contradiction.
\end{example}

Two other remedies suggest themselves, and we take neither.  One is to
look for a fixed point rather than to forbid the self-reference, in the manner of Kleene's recursion
theorem; \Cref{ex:liar} rules that out, since the obstruction is not a matter of finding the right
construction but of there being no consistent assignment.  The other is to weaken the oracle
to approximate counting, which is consistent because an approximate answer need not be the exact one
it perturbs.  That works, but at the cost of an error parameter, which would then have to be tracked
across the levels of the recursion and through every downstream estimate.  We prefer the bookkeeping
of ranks and exact answers.

The remedy we take is a well-founded rank, much as the polynomial hierarchy is stratified by
alternation depth.  Each counting query carries a rank $r\in\mathbb N$, and a rank-$r$ query may ask only about
samplers whose own oracle calls have rank strictly below $r$.  We can then define the family by
ordinary recursion on $r$, and no fixed-point theorem is needed, because the relation ``query $p$ may
refer to query $q$'' requires $\mathrm{rank}(q)<\mathrm{rank}(p)$ and so admits no cycles and
no infinite descending chains.

One point about the rank condition needs care, and it determines how we state the condition.  The
oracle has to be able to decide whether a query is admissible, and it has to do so even when the
sampler it is asked about computes its own ranks while running, possibly in superposition, in which
case there is no single behavior to inspect.  We therefore impose the condition on the description
the oracle is handed, rather than on the behavior of the sampler that description defines.  Concretely,
a sampler is \emph{admissible} as the argument of a rank-$r$ query when every counting gate in its
description takes one of two forms, each of which can be recognized by reading the code.

The first form carries a rank $s<r$ written into the description, and admissibility is then decided
by reading it.  A sampler uses the second form when it has to choose a layer at runtime.  That
form carries a \emph{bound} $J<r$, also written into the description, and it receives the rank itself
as runtime data, which the computation may have produced adaptively or coherently.  Such a gate
forwards the query to $\Count_s$ when the rank $s$ it receives at runtime is well formed with $s\le
J$, and returns $0$ otherwise.  Since the bound is fixed in the description while the rank is not, a
gate of the second form accesses only $\cO_{<r}$ whatever the computation does, and this too can be
decided by inspecting the code.

A sampler that computes an index $i$ and needs the answer of the $i$-th layer is accordingly given a
gate of the second form rather than unrestricted access, and queries whose rank exceeds $J$ return
$0$.  Choosing $J$ above every rank the sampler can write costs it nothing, and that is the situation
in every use below.  Throughout, samplers are \emph{clocked}: a description carries an explicit
polynomial step bound after which the computation halts, so the description alone bounds every rank
its sampler can write.

Write $\mathsf{Norm}_J$ for the compiler that replaces each call to the joined oracle by a gate of
the second form with bound $J$.  When $J$ is at least every well-formed rank that the original
clocked computation can write, $\mathsf{Norm}_J$ preserves its induced channel exactly, including for
adaptive and coherent calls, and with only polynomial overhead.  The reason is that the two gates
agree as maps on every basis state the computation can produce: a well-formed rank at most $J$ is
forwarded unchanged, and by the choice of $J$ no other rank ever occurs, so the substitution changes
no amplitude.  The overhead is the comparison against $J$.

One further design choice remains.  We make the oracle \emph{bit-valued}, returning bits of the
\emph{exact} output probability rather than an approximation.  This has three consequences: it removes
a parameter; it removes any need to argue that approximation errors compound across ranks, since
exact bits feed into exact bits; and, as \Cref{sec:coherent} explains, it keeps the coherent version
of the oracle an ordinary Boolean oracle.

\begin{definition}[Bit convention]
\label{def:bit}
For $p\in[0,1)$ let $p=\sum_{k\ge1}\mathrm{bit}_k(p)2^{-k}$ be the binary expansion not ending in all $1$s; for $p=1$ set $\mathrm{bit}_k(1)=1$ for all $k$. Then $\widetilde p_t=\sum_{k\le t}\mathrm{bit}_k(p)2^{-k}$ satisfies $0\le p-\widetilde p_t\le2^{-t}$ for every $p\in[0,1]$, and $\mathrm{bit}_1(p)=1$ iff $p\ge\frac12$.
\end{definition}

\begin{definition}[The oracle $\cO=(\cR,\Count)$]
\label{def:oracle}
Let $\cO_{<r}\eqdef\cR\cup\bigcup_{0\le r'<r}\Count_{r'}$, so that $\cO_{<0}=\cR$. For $r\ge0$, on input $(\langle C\rangle,1^\lambda,x,k)$ where $C$ is a clocked $\QPT$ sampler with classical output \emph{whose oracle access is to $\cO_{<r}$ only} in the sense above,
\[
  \Count_r\big(\langle C\rangle,1^\lambda,x,k\big)\ =\ \mathrm{bit}_k\!\Big(\Pr\big[C^{\cO_{<r}}(1^\lambda)=x\big]\Big).
\]
The index $k$ is written in unary, so a query of polynomial length requests polynomially many bits.
The full oracle is $\cO=(\cR,\Count)$ with $\Count(\langle1^r,q\rangle)\eqdef\Count_r(q)$, the unary-tagged Boolean join. Malformed tags, unclocked descriptions, and descriptions violating the rank condition receive answer $0$, so $\Count$ is a total Boolean oracle on classical strings, and all unary tags are charged to the query cost.
\end{definition}

\begin{lemma}[Well-definedness]
\label{lem:welldef}
\Cref{def:oracle} determines a unique oracle, and every bit of $\Count$ is a Borel function of the
full sequence $\mathbf H=\{H_\lambda\}$, and of $\{(H_\lambda,S_\lambda)\}$ in \Cref{model:single}, a sampler at
one security parameter being permitted to query another.
\end{lemma}

\begin{proof}
Read without the rank condition, \Cref{def:oracle} would ask for a fixed point, and \Cref{ex:liar}
shows that unrestricted self-reference admits no such assignment.  With the rank in place there is no
fixed point to find, only an ordinary recursion, and we run it.

We argue by strong induction on $r$.  The base $\cO_{<0}=\cR$ is given.  Suppose, then, that
$\Count_{r'}$ is defined for all $r'<r$.  Then $\cO_{<r}$ is determined, and for any admissible $C$
the quantity $\Pr[C^{\cO_{<r}}(1^\lambda)=x]$ is the output probability of a fixed algorithm with a
fixed, already-defined oracle.  It is a well-defined real, and so are its bits.  Invalid
inputs have the specified answer $0$.  Crucially, no reference to rank $\ge r$ occurs in the clause
defining $\Count_r$, so the recursion never reaches back into the layer it is defining.  Every layer,
and hence the whole family, is determined uniquely.

That leaves measurability, which follows by the same induction, in three steps.

At the base, a call to $\cR$ returns the fixed state $\rho_{b,\lambda}=2P_{H}/n$ or its complement, so
the output probability of a fixed clocked circuit with sample access alone is a polynomial in the
entries of $P_H$, and hence continuous in $H$.

At rank $r$, suppose the oracle bits the circuit reads below $r$ are Borel in $H$.  Its output
probability is then a finite sum of products of those bits, with coefficients that are themselves
continuous in $H$, and hence Borel; and $\mathrm{bit}_k$ of a Borel function is again Borel.

For \Cref{model:single}, finally, the base layer is Borel in $(H,S)$ because the rounding map is
constant on the cells of a finite measurable partition, and the same induction applies.
\end{proof}

\paragraph{The computational power of $\Count$.}  The security proof does
not require that $\Count$ be weak, and the ranked join is in fact strong.  Because ranks are unary
and uniformly addressable, a decreasing rank can serve as the recursion counter of an alternating
computation, and we expect that for every fixing of the hidden subspaces the tagged Boolean join
satisfies $\mathsf P^{\Count}=\PSPACE^{\Count}$, whence every classical-input class in between
collapses to it.  With binary tags the same argument would reach $\EXPTIME$, which shows that this is
a property of the encoding.  We do not use the claim anywhere, and so do not prove it.

None of this matters for security.  \Cref{thm:transcript-security} already covers an \emph{arbitrary,
even uncomputable} truth table, so the computational power of $\Count$ is irrelevant to EFI security
\emph{by construction}.  What is scarce is not the power of the oracle, but the amount and the type
of information about $H$ that reaches the adversary.

We should stress what this does not show.  It is no evidence that EFI pairs can coexist with the
unrelativized equality $\mathsf P=\PSPACE$, nor with a single classical oracle: efficient generation
still uses the state source $\cR$, and \Cref{thm:EFI-security} permits only classical access to
$\Count$.  The question of~\cite{KQST23,MetaEFI25} remains open, and our techniques do not address
it, since replacing the Haar source by a short classical description of $H$ would invalidate the
proof outright, the holder of $H$ being free to send that description.

\subsection{Nonexistence of one-way puzzles}
\label{sec:no-owpuzz}

\begin{lemma}[Probability estimation is easy relative to $\cO$]
\label{lem:pe-easy}
Every $\QPT^{\cO}$-samplable classical distribution admits efficient probability estimation relative to $\cO$, using one integer, a bound on the ranks the sampler can write, that is read off its explicit step bound. The estimate is $2^{-t}$-additively accurate for any $t=\poly(\lambda)$, with success probability $1$, using $t$ classical queries.
\end{lemma}

\begin{proof}
Fix an oracle-aided sampler $C^{\cO}$ with classical output.  The only obstacle is that $C$ is not,
by itself, an admissible argument to any single layer of $\Count$, and the proof consists in making
it one.

On inputs of length $\lambda$ the sampler makes polynomially many oracle calls, each carrying a
unary rank; let $J_C(\lambda)$ bound all ranks it can write at input length $\lambda$.  That integer is at most the sampler's step bound, since a rank is written in unary during the
computation, so we can read it off the clocked description.  Let $C_\lambda$ be the
clocked specialization running $C^{\cO}(1^\lambda)$ and $\widehat
C_\lambda\eqdef\mathsf{Norm}_{J_C(\lambda)}(C_\lambda)$.

We will show that $\widehat C_\lambda$ is an admissible argument to $\Count_{J_C(\lambda)+1}$ whose
output distribution is exactly that of $C^{\cO}(1^\lambda)$, after which the estimate can be read off
the oracle one bit at a time.

Admissibility is by inspection: $\widehat C_\lambda$ has polynomial description length, and its code
certifies that it accesses only $\cO_{<J_C(\lambda)+1}$.  The output distribution is unchanged for a
different reason: $J_C(\lambda)$ bounds every rank $C$ can write, so the normalization preserves the
induced channel, by the property of $\mathsf{Norm}_J$ recorded in \Cref{sec:rank}.

Now we run the estimator.  On input $x$ it constructs $\langle\widehat C_\lambda\rangle$ and
queries
\[
  \Count_{J_C(\lambda)+1}\big(\langle\widehat C_\lambda\rangle,1^\lambda,x,k\big),\qquad k=1,\dots,t,
\]
outputting the truncation $\widetilde p_t$.  By \Cref{def:bit},
$|\widetilde p_t-\Pr[C^{\cO}(1^\lambda)=x]|\le2^{-t}$ for \emph{every} value, including $0$ and $1$,
which is what the lemma asserts.
\end{proof}

The strength of \Cref{lem:pe-easy} is that $t$ classical queries give additive error $2^{-t}$, so
\emph{exponentially} small error is available at polynomial cost.  So the next theorem admits a direct and
elementary attack, and we need not assume that the characterization of~\cite{CGGH25} relativizes.

Two arguments below replace a quantum party by a classical one that samples from these estimates,
one bit at a time: the puzzle attack of \Cref{app:attack} and the interactive simulation of
\Cref{sec:interactive}.  Both pay for the estimator in the same way, and we isolate that
accounting here, in two steps, so that neither has to repeat it.

The first step is local.  A simulator that must produce the next bit knows two truncated
probabilities and samples from their ratio, and the question is how far that is from the true
conditional law.

\begin{lemma}[Sampling from a truncated ratio]
\label{lem:ratio}
Let $p_0,p_1\ge0$ with $p\eqdef p_0+p_1\le1$, and let $\eta=2^{-t}>0$.  Suppose
$0\le\tilde p_y\le p_y$ and $p_y-\tilde p_y\le\eta$ for $y\in\{0,1\}$, write
$\tilde S\eqdef\tilde p_0+\tilde p_1$, and
let $\tilde\pi$ be the law on $\{0,1\}$ that gives $y$ probability $\tilde p_y/\tilde S$ when
$\tilde S>0$, rounded to a multiple of $2^{-2t}$.  If $p\ge4\eta$ then $\tilde S>0$, the true
conditional law $\pi(y)\eqdef p_y/p$ is defined, and
\[
  \TV(\pi,\tilde\pi)\ \le\ \frac{4\eta}p .
\]
\end{lemma}

\begin{proof}
Both laws live on two points, so their total variation distance is the deviation at $y=0$.  Put
$e_y\eqdef p_y-\tilde p_y\in[0,\eta]$.  Before the rounding, a direct computation gives
\[
  \Big|\frac{p_0}p-\frac{\tilde p_0}{\tilde S}\Big|
  \ =\ \frac{\big|e_0p_1-e_1p_0\big|}{p\,\tilde S}\ \le\ \frac{\eta}{\tilde S} ,
\]
the numerator being at most $\eta\max\{p_0,p_1\}\le\eta p$.  The rounding adds at most $2^{-2t}$, and
$2^{-2t}\le\eta/\tilde S$ because $\tilde S\le1$, so $\TV(\pi,\tilde\pi)\le2\eta/\tilde S$.

The denominator is itself close to $p$, since $\tilde S\ge p-2\eta$.  So $p\ge4\eta$ gives
$\tilde S\ge p/2>0$, and the displayed bound follows.
\end{proof}

The second step is global.  A history of small probability is cheap because it is
rare; a history of large probability is cheap because \Cref{lem:ratio} is accurate on it; and one
threshold separates the two.

\begin{lemma}[Threshold accounting]
\label{lem:threshold}
Let $\cH$ be a set of at most $K$ elements, and let $P,V\colon\cH\to[0,1]$ satisfy
$\sum_{h\in\cH}V(h)P(h)=1$.  Let $\eta>0$ with $K^{-1}\ge4\eta$, and let $\tau\colon\cH\to[0,1]$
satisfy
\[
  \tau(h)\ \le\ \frac{4\eta}{P(h)}\qquad\text{whenever }P(h)\ge4\eta .
\]
Then
\[
  \sum_{h\in\cH}V(h)\,P(h)\,\tau(h)\ \le\ 2\sqrt{4\eta K} .
\]
\end{lemma}

\begin{proof}
Put $\theta\eqdef\sqrt{4\eta/K}$.  Then $\theta\ge4\eta$, since
$\theta^2=4\eta/K\ge16\eta^2$ by hypothesis.  Split the sum at $\theta$.

A \emph{light} $h$, one with $P(h)<\theta$, contributes at most $V(h)P(h)\tau(h)\le P(h)<\theta$,
using only $V,\tau\le1$.  There are at most $K$ elements in all, so the light ones contribute at most
$K\theta$ together.

A \emph{heavy} $h$, one with $P(h)\ge\theta$, has $P(h)\ge4\eta$, so the hypothesis applies to it,
and the heavy ones contribute
\[
  \sum_{h\ \mathrm{heavy}}V(h)P(h)\cdot\frac{4\eta}{P(h)}\ =\ 4\eta\sum_{h\ \mathrm{heavy}}V(h)
  \ \le\ \frac{4\eta}\theta\sum_{h\ \mathrm{heavy}}V(h)P(h)\ \le\ \frac{4\eta}\theta .
\]
The choice of $\theta$ makes the two ranges equal, each being $\sqrt{4\eta K}$, and adding them gives
the bound.
\end{proof}

\begin{theorem}[No one-way puzzles relative to $\cO$]
\label{thm:no-owpuzz}
Fix any choice of $\{H_\lambda\}$. For every candidate one-way puzzle $(\Samp^{\cO},\Ver)$ satisfying correctness there is a polynomial-time adversary $\cA$ (a \emph{classical} algorithm making classical queries to $\Count$) such that
\[
  \Pr_{(k,s)\gets\Samp(1^\lambda)}\big[\Ver\big(\cA(1^\lambda,s),s\big)=1\big]\ \ge\ 1-\negl(\lambda).
\]
In particular $\OWPuzz$s do not exist relative to $\cO$, including inefficiently verifiable ones, for every fixing of the oracle's randomness.
\end{theorem}

\Cref{app:attack} gives the proof.  The adversary samples the key bit by bit, from an estimate of the
true conditional distribution of keys given the puzzle, using \Cref{lem:pe-easy} on the prefix
samplers ``run $\Samp$ and output $(s,k_1\cdots k_i)$''.  Because the estimates are exponentially
accurate, the joint distribution of $(s,\cA(s))$ lies within $2^{-\lambda}$ of the honest one, and so
the unbounded verifier accepts with probability $1-\negl(\lambda)$.  Notice that the attack never
runs $\Ver$; that is why we need no assumption on the verifier.

\begin{corollary}[Consequences for other primitives]
\label{cor:collateral}
Relative to $\cO$, the following do not exist: \emph{multi-copy} pseudorandom states and pseudorandom unitaries, pure-output one-way state generators, and quantum money mini-schemes with pure banknotes (each implies $\OWPuzz$s~\cite{KT24,BJ24,KT25,MY22}). Moreover, since the adversary of \Cref{thm:no-owpuzz} is a classical polynomial-time algorithm making classical queries, even \emph{classically-secure} $\OWPuzz$s do not exist relative to $\cO$.
\end{corollary}

\paragraph{The scope of these inferences.}  The word ``multi-copy'' in \Cref{cor:collateral} is
essential, for the following reason.  By~\cite[Thm.~1.1]{MetaEFI25}, EFI pairs
exist if and only if non-uniform $\OnePRS$ with advice size $O(\log\lambda)$ exist.  Since
\Cref{thm:separation-source} keeps an EFI pair in its access model, non-uniform $\OnePRS$ with
$O(\log\lambda)$ advice must survive relative to $\cO$ in that same model.  This is
consistent, because multi-copy PRS imply $\OWPuzz$s and are removed for that reason, whereas
single-copy pseudorandomness is not known to imply them, and here does not.  It follows that $\cO$
separates multi-copy from non-uniform single-copy pseudorandomness, which is the separation
of~\cite{CCS24}, in the classical-query model of \Cref{model:access}.

That inference is conditional on the equivalence of~\cite{MetaEFI25} relativizing to this access
model, in which the $\OnePRS$ it yields is secure.  The other half, by contrast, is robust.  The
nonexistence of multi-copy PRS survives the model restriction, since the chain
$\text{PRS}\Rightarrow\text{pure }\OWSG\Rightarrow\OWPuzz$ turns our classical-query $\OWPuzz$
attacker into a classical-query PRS attacker, while the candidate generator remains free to query
$\Count$ coherently, which \Cref{def:oracle} permits.

Two further consequences carry conditions of the same kind.  Quantum lightning does not survive
$\cO$ for the standard notion, since lightning yields a quantum money mini-scheme~\cite[\S3.2]{Zha19}
and mini-schemes imply $\OWPuzz$s~\cite[Cor.~1.1]{KT25}, and both links relativize.  The mini-scheme
of~\cite{KT25} has a \emph{pure} banknote, however, so the mixed-banknote notion of~\cite{AarChr12}
is not covered, and \Cref{sec:discussion} returns to that gap.  Unconditional, on the other hand, is
the statement that $\cO$ admits an EFI pair and no classical-output quantum advantage, by
\Cref{thm:EFI-security} together with \Cref{cor:no-poq}; \Cref{sec:interactive} compares it with the
route through~\cite[Thm.~1.1]{MSY25}.

\subsection{Nonexistence of \texorpdfstring{$\QEFID$}{QEFID} pairs}
\label{sec:no-qefid}

The attack of \Cref{thm:no-owpuzz} removes the classical-output \emph{one-wayness} in the $\OWPuzz$
definition, and the same oracle removes classical-output \emph{indistinguishability} as well.

The qualitative conclusion here is not new, so we first state the result that already implies it.
By \cite[Lem.~8]{CGG24}, restated as~\cite[Thm.~6.3]{BMM24}, $\QEFID$ pairs imply $\OWPuzz$s in a
black-box manner, so \Cref{thm:no-owpuzz} \emph{already} rules out $\QEFID$ pairs relative
to $\cO$ through that implication.  What the direct argument below adds is quantitative.  Its
distinguisher is classical polynomial time with classical queries; its advantage is the full
statistical distance up to $2^{-\lambda}$, rather than merely non-negligible; the conclusion holds
for every fixing of $\{H_\lambda\}$; and no appeal to the relativization of~\cite{CGG24} is needed.
The proposition says that relative to $\cO$ the computational and the statistical theories of
quantum-samplable classical distributions coincide.

\begin{proposition}[Statistical and computational indistinguishability coincide for classical-output samplers]
\label{prop:no-qefid}
Fix any choice of $\{H_\lambda\}$. Let $S_0,S_1$ be $\QPT$ oracle samplers with classical outputs of
length $\ell=\ell(\lambda)=\poly(\lambda)$, and write $\cD_b$ for the distribution of
$S_b^{\cO}(1^\lambda)$, suppressing $\lambda$. There is a classical polynomial-time adversary $\mathsf D$, making
$O(\ell+\lambda)$ classical queries to $\Count$, whose
distinguishing advantage between $\cD_0$ and $\cD_1$ satisfies, for all sufficiently large
$\lambda$,
\[
  \Adv(\mathsf D)\ \ge\ \TV\big(\cD_0,\cD_1\big)-2^{-\lambda} .
\]
\end{proposition}

\begin{proof}
An adversary that knew the two output probabilities exactly would apply the likelihood-ratio test and
would achieve the full statistical distance.  Our adversary knows them only to inverse-exponential accuracy,
and the proof shows that this costs it almost nothing.  Write
\[
  p_b(x)\ \eqdef\ \Pr\big[S_b^{\cO}(1^\lambda)=x\big],\qquad b\in\{0,1\},\ x\in\{0,1\}^{\ell},
\]
for those probabilities.  The exact test outputs $1$ precisely when $p_1(x)>p_0(x)$, and its advantage
is
\begin{equation}
  \sum_x\max\big\{p_1(x)-p_0(x),\,0\big\}\ =\ \TV\big(\cD_0,\cD_1\big),
  \label{eq:lrt}
\end{equation}
which no function of $x$ improves on.

Now fix the precision $t\eqdef\ell+\lambda+2$ and write $\eta\eqdef2^{-t}$.  On input $x$ the
adversary applies \Cref{lem:pe-easy} to each of $S_0$ and $S_1$, obtaining truncations
$\widetilde p_0(x)$ and $\widetilde p_1(x)$ that satisfy
\[
  0\ \le\ p_b(x)-\widetilde p_b(x)\ \le\ \eta\qquad(b=0,1),
\]
and it outputs
\[
  g(x)\ \eqdef\ \mathbf 1\big[\,\widetilde p_1(x)>\widetilde p_0(x)+\eta\,\big].
\]
The margin $\eta$ in the comparison is there to absorb the truncation error.  Computing the two truncations takes $2t=O(\ell+\lambda)$ classical queries, and we read the single
integer bounding the ranks of both samplers off their clocked descriptions, exactly as in
\Cref{lem:pe-easy}.

We will show that $g$ and the exact rule $g^*(x)\eqdef\mathbf 1[p_1(x)>p_0(x)]$ can disagree only
when the two probabilities lie within $2\eta$ of each other.  Suppose first that
$p_1(x)>p_0(x)+2\eta$.  Then
\[
  \widetilde p_1(x)\ \ge\ p_1(x)-\eta\ >\ p_0(x)+\eta\ \ge\ \widetilde p_0(x)+\eta ,
\]
so $g(x)=1=g^*(x)$.  Suppose instead that $p_1(x)\le p_0(x)$.  Then
\[
  \widetilde p_1(x)\ \le\ p_1(x)\ \le\ p_0(x)\ \le\ \widetilde p_0(x)+\eta ,
\]
so $g(x)=0=g^*(x)$.  The two rules agree, then, outside the set
\[
  B\ \eqdef\ \big\{\,x\ :\ 0<p_1(x)-p_0(x)\le2\eta\,\big\},
\]
as claimed.

That leaves $B$.  Each $x\in B$ costs at most $|p_1(x)-p_0(x)|\le2\eta$ of advantage
relative to the exact rule, and $B$ contains at most $2^{\ell}$ points.  Subtracting this loss from
\eqref{eq:lrt},
\[
  \Adv(\mathsf D)\ \ge\ \TV\big(\cD_0,\cD_1\big)\ -\ 2^{\ell}\cdot2\eta
  \ =\ \TV\big(\cD_0,\cD_1\big)-2^{\ell+1-t}
  \ =\ \TV\big(\cD_0,\cD_1\big)-2^{-\lambda-1},
\]
which is the stated bound.
\end{proof}

\begin{corollary}[Nonexistence of $\QEFID$ pairs and of classical-output EFI pairs]
\label{cor:no-qefid}
Relative to $\cO$, and for every fixing of $\{H_\lambda\}$, the following hold.
\begin{enumerate}[label=(\roman*),leftmargin=2.2em]
\item $\QEFID$ pairs do not exist, and the attacker is a classical polynomial-time algorithm making
  classical queries, so classically-secure $\QEFID$ pairs do not exist either.
\item No EFI pair relative to $\cO$ has both states diagonal in a common basis reachable by a $\QPT$
  isometry whose $\Count$-queries are made on registers in computational-basis states and which
  retains all its registers, so that it preserves trace distance. In particular the pair of \Cref{thm:EFI-security}
  cannot be replaced by any classical ensemble.
\end{enumerate}
\end{corollary}

\begin{proof}
\emph{(i)} Suppose $(\cD_0,\cD_1)$ is a $\QEFID$ pair with samplers $S_0,S_1$ and farness
$\TV(\cD_0,\cD_1)\ge\delta(\lambda)$ for an inverse polynomial $\delta$.  Then \Cref{prop:no-qefid}
gives a distinguisher of advantage $\delta(\lambda)-2^{-\lambda}$, which is not negligible.  That
distinguisher is a classical polynomial-time algorithm making classical queries, so classically-secure
pairs are excluded by the same computation.  (Alternatively, and without the quantitative conclusion,
combine \cite[Lem.~8]{CGG24} with \Cref{thm:no-owpuzz}.)

\emph{(ii)} Suppose $(\rho_{0,\lambda},\rho_{1,\lambda})$ is an EFI pair relative to $\cO$, and let
$V_\lambda$ be a $\QPT$ isometry that diagonalizes both branches, in the sense that
$V_\lambda\rho_{b,\lambda}V_\lambda^\dagger$ is diagonal in the computational basis for both $b$.  We
ask of $V_\lambda$ only that it query $\Count$ on registers in computational-basis states, and that it
retain all of its registers.

Let $S_b$ run the EFI generator on branch $b$, apply $V_\lambda$, and measure.  Then $S_0,S_1$ are
$\QPT^\cO$ samplers with classical output.  A measurement in a basis that diagonalizes both states
preserves trace distance, so
\[
  \TV(\cD_0,\cD_1)\ =\ \td(\rho_{0,\lambda},\rho_{1,\lambda})\ \ge\ \delta(\lambda) ,
\]
so part~\emph{(i)} applies.  The distinguisher that applies $V_\lambda$, measures, and
then runs the adversary of \Cref{prop:no-qefid} makes only classical queries, so it lies in the class
of \Cref{model:access} against which the pair is secure, and its advantage is not negligible.  The second assertion of~\emph{(ii)} is the case $V_\lambda=\Id$.  A classical ensemble that was an EFI
pair relative to $\cO$ would be a pair of $\QPT^\cO$-samplable classical distributions, statistically
far apart and yet indistinguishable by every distinguisher of \Cref{model:access} -- a class that
contains the attacker of part~\emph{(i)}.
\end{proof}

\subsection{Interactive protocols with classical messages}
\label{sec:interactive}

\Cref{thm:no-owpuzz} and \Cref{prop:no-qefid} remove classical-output hardness from two fixed
experiments, one of search type and one of decision type.  The same estimator, however, removes it
from every experiment of that kind at once, interactive ones included, and that is what we prove
here.  Relative to $\cO$, a quantum polynomial-time party has no advantage over a classical one,
provided that everything the party sends and receives is classical.

\begin{model}[Interactive party with classical messages]
\label{model:interactive}
A \emph{quantum party} $\mathsf Q$ is a $\QPT^{\cO}$ algorithm carrying an internal quantum state
across $R=\poly(\lambda)$ rounds. In round $j$ it receives a classical message $c_j$, applies a
$\QPT^{\cO}$ channel to its state together with that message, and measures a designated register to
produce the classical message $d_j$, which it sends. Its oracle calls inside a round may be coherent.
A \emph{counterparty} $\mathsf V$ is any process, of unbounded computational power and with an
internal state of any kind, whose messages are classical and whose $c_j$ is produced from
$c_1,d_1,\dots,c_{j-1},d_{j-1}$ together with its own private randomness.  The two parties share no
initial entanglement and no correlated setup, and $\mathsf V$ may itself query $\cO$.  Truncating
each incoming message to the prefix that $\mathsf Q$ reads leaves $\mathsf Q$'s behavior unchanged, so
we take all messages to have fixed polynomial lengths.  We write $A$ for the total number of incoming
bits $\mathsf Q$ reads; its clock bounds that number, which depends on $\mathsf Q$ alone.
We write $N$ for the total length of $\mathsf Q$'s messages. We write $\cT(\mathsf Q,\mathsf V)$ for the
distribution of the full transcript $(c_1,d_1,\dots,c_R,d_R)$.
\end{model}

\begin{theorem}[Interactive collapse]
\label{thm:interactive}
Fix any choice of $\{H_\lambda\}$. For every quantum party $\mathsf Q$ of \Cref{model:interactive}
there is a \emph{classical} probabilistic polynomial-time party $\widehat{\mathsf Q}$, making
$\poly(\lambda)$ classical queries to $\Count$, such that
for \emph{every} counterparty $\mathsf V$ and all sufficiently large $\lambda$,
\[
  \TV\big(\cT(\widehat{\mathsf Q},\mathsf V),\ \cT(\mathsf Q,\mathsf V)\big)\ \le\ 2^{-\lambda}.
\]
The bound is uniform in $\mathsf V$: neither a bound on its complexity nor access to its description
is used anywhere.
\end{theorem}

\begin{proof}
The simulator draws each of $\mathsf Q$'s outgoing bits from an estimate of the conditional law that
$\mathsf Q$ itself uses, and the work lies in showing that these laws are visible to $\Count$ and that
the estimation error does not accumulate over the $N$ bits.  First, each conditional is a ratio of two
output probabilities of admissible samplers, and so is available from $\Count$.  Second, the
truncated ratio we can sample from at polynomial cost.  Third, a hybrid over the outgoing bits, with
a threshold separating the histories of small probability from those on which the estimates are
relatively accurate, keeps the accumulated error below $2^{-\lambda}$.

\emph{Step 1: the conditional law of each outgoing bit is a ratio of two output probabilities that
$\Count$ reports.} Regard the concatenation of $\mathsf Q$'s
messages as a single string of $N$ bits, and for a bit index $i$ let $j(i)$ be the round it belongs
to. For an incoming sequence $c=(c_1,\dots,c_{j(i)})$ let $C_{c,i}$ be the sampler that runs
$\mathsf Q(1^\lambda)$, feeds it $c_1,\dots,c_{j(i)}$ as the incoming messages, and outputs the first
$i$ bits that $\mathsf Q$ sends.

Feeding fixed messages irrespective of what $\mathsf Q$ replies is legitimate, since the result is
again a fixed clocked $\QPT^{\cO}$ algorithm with classical output, of polynomial description length
because $|c|\le A=\poly(\lambda)$.  So, with $J$ bounding the ranks $\mathsf Q$ can write,
$\mathsf{Norm}_J(C_{c,i})$ is an admissible argument to $\Count_{J+1}$ by the reasoning of
\Cref{lem:pe-easy}, and we read $J+1$ off $\mathsf Q$'s clocked description.

Write $P_c(y)\eqdef\Pr[C_{c,|y|}^{\cO}(1^\lambda)=y]$.  Notice that these are consistent:
\[
  P_c(y)\ =\ P_c(y0)+P_c(y1) .
\]
When $y0$ and $y1$ lie in the same round as $y$, this is the marginal of a single sampler.  When
they open a new round it holds for a different reason: the first $|y|$ bits are sent before
$c_{j(|y|+1)}$ arrives, and so cannot depend on it.  Consequently, for a history whose incoming part
is $c$ and whose outgoing part is $y$, the true conditional law of the next outgoing bit is
\[
  \pi_{c,y}(\beta)\ =\ \frac{P_c(y\beta)}{P_c(y)}\qquad\text{whenever }P_c(y)>0 .
\]

This last identity is where we use the absence of a correlated setup.  By induction on the rounds, on
every fixed transcript branch the joint state is a product of $\mathsf Q$'s and $\mathsf V$'s
sub-normalized branch states, each party acting on its own registers and learning of the other only
through classical messages.  The probability of a history thus factorizes as
$\Pr[c,y]=V(c,y)\,P_c(y)$, where $V(c,y)$ is the product of $\mathsf V$'s conditional probabilities of
its own messages, and likewise $\Pr[c,y\beta]=V(c,y)\,P_c(y\beta)$.  The conditional law of the next
outgoing bit is the ratio of the two, and the counterparty's factor cancels.

After Step 1 every number the simulator needs is the answer to a query about a single admissible
sampler, and it can be obtained to any polynomial precision.

\emph{Step 2: the simulator sends every outgoing bit drawn from the truncated ratio, at polynomial
cost.} Fix the precision $t\eqdef A+N+3\lambda$ and put $\eta\eqdef2^{-t}$.
The party $\widehat{\mathsf Q}$ keeps the history in the clear.  To produce its $i$-th outgoing bit,
with incoming messages $c$ and own prefix $y$, it does the following.
\begin{itemize}[leftmargin=1.6em,itemsep=2pt,topsep=3pt]
\item It queries $\Count$ for the first $t$ bits of $P_c(y0)$ and of $P_c(y1)$, obtaining truncations
  with $0\le P_c(y\beta)-\widetilde P_c(y\beta)\le\eta$.
\item If both truncations vanish, it sends $0$.
\item Otherwise it sends $\beta$ with probability
  $\widetilde P_c(y\beta)/(\widetilde P_c(y0)+\widetilde P_c(y1))$, rounded to a multiple of $2^{-2t}$
  so that $2t$ fair coins realize it exactly.
\end{itemize}
Call that rule $\widetilde\pi_{c,y}$.  The rounding costs $2^{-2t}$ in total variation, which
\Cref{lem:ratio} already absorbs into its constant.  The whole interaction costs
$2Nt=\poly(\lambda)$ classical queries.

So far we have shown that $\widehat{\mathsf Q}$ is a classical probabilistic polynomial-time party of
the kind the statement asks for, and it remains to compare the transcript it produces with the real
one.

\emph{Step 3: the hybrid over the $N$ outgoing bits loses at most $2^{-\lambda}$ in all.} The rule of
Step~2 is the one \Cref{lem:ratio} analyses, at $p_\beta=P_c(y\beta)$, so whenever $P_c(y)\ge4\eta$,
\[
  \TV\big(\pi_{c,y},\widetilde\pi_{c,y}\big)\ \le\ \frac{4\eta}{P_c(y)} ,
\]
and trivially $\TV\le1$ elsewhere.  For $0\le i\le N$ let $\cD_i$ be the interaction in which $\mathsf Q$'s
first $i$ outgoing bits are drawn from the true conditionals and the remaining ones from
$\widetilde\pi$, with $\mathsf V$ unchanged throughout, so that $\cD_N=\cT(\mathsf Q,\mathsf V)$ and
$\cD_0=\cT(\widehat{\mathsf Q},\mathsf V)$.  Adjacent hybrids use the same rule at every bit except
the $(i+1)$-st, and the history before that bit is distributed as in the real interaction, so
\[
  \TV(\cD_i,\cD_{i+1})\ \le\ \E_{(c,y)}\big[\TV(\pi_{c,y},\widetilde\pi_{c,y})\big] ,
\]
the expectation being over the real law of the history $(c,y)$ before bit $i+1$.

That expectation is exactly what \Cref{lem:threshold} bounds, and the chain rule supplies its
hypotheses.  Write
\[
  \Pr[c,y]\ =\ V(c,y)\,P_c(y) ,
\]
where $V(c,y)$ collects $\mathsf V$'s conditional probabilities of the messages in $c$.  Then
$V(c,y)\le1$, and $\sum_{c,y}V(c,y)P_c(y)=1$ because the histories at a fixed step partition the
probability space.  There are at most $2^{A+N}$ histories, so we take $K=2^{A+N}$; and
$K^{-1}\ge4\eta$ holds comfortably, since $\eta=2^{-(A+N+3\lambda)}$.  The lemma therefore gives
$\E_{(c,y)}[\TV]\le2\sqrt{4\eta\,2^{A+N}}$.  Summing over the $N$ steps,
\[
  \TV\big(\cT(\widehat{\mathsf Q},\mathsf V),\cT(\mathsf Q,\mathsf V)\big)
  \ \le\ 2N\sqrt{4\eta\,2^{A+N}}\ =\ 4N\,2^{-3\lambda/2}\ \le\ 2^{-\lambda}
\]
for all large $\lambda$, since $N=\poly(\lambda)$.  Notice that the only properties of $\mathsf V$ used
anywhere are $V(c,y)\le1$ and the normalization above, so the bound holds for every counterparty,
which is what the theorem asserts.
\end{proof}

Everything $\mathsf V$ computes is a function of the transcript and of its own randomness, so the
conclusion transfers to any quantity read off at the end of the interaction.  We state the two cases
we use.

\begin{corollary}[No quantum advantage over classical communication]
\label{cor:no-poq}
Relative to $\cO$, and for every fixing of $\{H_\lambda\}$:
\begin{enumerate}[label=(\roman*),leftmargin=2.2em]
\item every classical distribution samplable in $\QPT^{\cO}$ is samplable to within $2^{-\lambda}$ in
  total variation by a classical probabilistic polynomial-time algorithm making classical queries to
  $\Count$; and
\item there is no inefficient-verifier proof of quantumness.  For every $\QPT^{\cO}$ prover and every
  verifier of unbounded computational power exchanging classical messages with it, there is a
  classical probabilistic polynomial-time prover, making classical queries to $\Count$, whom that
  verifier accepts with probability differing by at most $2^{-\lambda}$.
\end{enumerate}
\end{corollary}

\begin{proof}
\emph{(i)} Apply \Cref{thm:interactive} with no incoming messages, so that $A=0$ and the transcript is
the sampler's own output.

\emph{(ii)} Apply \Cref{thm:interactive} with $\mathsf V$ the verifier.  Write $a(\pi)\in[0,1]$ for
its acceptance probability conditional on a transcript $\pi$.  That is the same function in the real
and in the simulated experiment, since the verifier's own process, given the messages it sends and
receives, is identical in both.  So the two acceptance probabilities are
\[
  \sum_\pi\Pr\nolimits_{\mathrm{real}}[\pi]\,a(\pi)
  \qquad\text{and}\qquad
  \sum_\pi\Pr\nolimits_{\mathrm{sim}}[\pi]\,a(\pi) .
\]
Since $0\le a\le1$, their difference is at most
$\sup_{0\le f\le1}|\E_{\mathrm{real}}f-\E_{\mathrm{sim}}f|$, which is the total variation distance
between the two transcript laws, and the theorem bounds that by $2^{-\lambda}$.
\end{proof}

Part~(ii) is stated in general form because a proof of quantumness is only the simplest instance.
The simulation is uniform in the counterparty, so it applies to any interactive protocol whose
messages are classical.  Relative to $\cO$, then, an efficient
\emph{quantum} prover confers no advantage over an efficient \emph{classical} one in such a protocol:
whatever an unbounded verifier can be convinced of by a $\QPT^{\cO}$ prover exchanging classical
messages, it can be convinced of by a classical polynomial-time prover making classical queries, with
acceptance probabilities differing by at most $2^{-\lambda}$.  The classical-message restriction is
essential: a prover sending quantum messages is outside the scope of
\Cref{thm:interactive}, as is a prover sharing entanglement with the verifier in advance.

\paragraph{Scope of the collapse.}  \Cref{cor:no-poq}(i) contains the sampling half of
\Cref{sec:no-owpuzz} and \Cref{sec:no-qefid}, and in its own case \Cref{prop:no-qefid} sharpens it,
since that proposition achieves the full statistical distance rather than a simulation.
\Cref{thm:no-owpuzz} is not an instance, however.  The attack there has to sample a key
\emph{conditioned on} a puzzle that the honest sampler itself produced, which is the same estimator
applied to conditional rather than to joint distributions.  What \Cref{thm:interactive} adds beyond
both is the interactive case, and with it part~(ii), which no single-message statement reaches.

Part~(ii) can also be obtained along a second route, by combining the characterization of
inefficient-verifier proofs of quantumness by classically-secure $\OWPuzz$s~\cite[Thm.~1.1]{MSY25}
with \Cref{cor:collateral}.  We prefer the direct route here because it assumes nothing about
relativization, holds for every fixing of $\{H_\lambda\}$, and is not restricted to uniform
adversaries.

\medskip

\Cref{thm:no-owpuzz}, \Cref{cor:no-qefid} and \Cref{thm:interactive} are three faces of a single
statement.  A $\OWPuzz$ is classical-output hardness of \emph{search}, a $\QEFID$ pair is
classical-output hardness of \emph{decision}, and a proof of quantumness is classical-output hardness
of \emph{interaction}.  What we have shown is that $\Count$ removes all three, for every fixing of
the hidden subspaces and against unbounded verifiers.  What remains available relative to $\cO$ is
hardness of a different kind, the indistinguishability of two quantum states, and
\Cref{thm:EFI-security} shows that this remainder is not empty.  The two layers therefore fit together rather
than merely coexist: the counting oracle trivializes the classical-output theory, and the state
source supplies what is left, a pair of states whose distinguishing measurement is
itself the secret.

Nor are the two halves independent facts about $\cO$.  By the equivalence recalled in
\Cref{sec:prelim-primitives}, classical-output search hardness and classical-output decision hardness
cannot be separated, so an oracle removing one necessarily removes the other.  What our two proofs
add is that both fail \emph{directly}, by a classical polynomial-time attack, for every fixing of the
hidden subspaces, with no implication needed in between.  One further consequence follows.
Since $\QEFID$ pairs trivially yield EFI pairs, by regarding the distributions as diagonal density
operators, $\cO$ also separates EFI pairs from $\QEFID$ pairs, in the classical-query model of
\Cref{model:access} for the EFI direction and unconditionally for the $\QEFID$ direction.  By that
same equivalence this is a restatement of \Cref{thm:separation-source} rather than a new separation.
It is, finally, the opposite orientation to~\cite{BMM24}, who give a unitary oracle with $\QEFID$
pairs but no $\OWSG$s and hence \emph{retain} the classical-output primitive.

\section{EFI security}
\label{sec:security}

We now apply \Cref{thm:transcript-security} to the half-subspace pair.  One obstacle remains.
Since the generator is public, a distinguisher may obtain reference copies, and those copies are
quantum side information depending on $H$, which is exactly what the hypotheses of
\Cref{thm:transcript-security} do not cover.  \Cref{sec:reference} reduces them to
classical side information first.  After that the associated problem is $\VSP$, and the theorem
applies.

\subsection{Reference copies and the residual instance}
\label{sec:reference}

The idea is to give the adversary something classical that is at least as useful as the copies
themselves.  A basis of the span of $q$ Haar vectors in $H$ determines their joint law, so the copies
can be resampled from the span alone.  In other words, the span is at least as good to the adversary
as the copies were, and unlike them it is a classical object.

Having replaced the copies by the two spans, we then condition on those spans.  What is left is a
Vector-in-Subspace instance in the orthogonal complement, of dimension $n-2q$, together with one
branch of weight $2q/n$ on which we concede everything.  Reference copies cost us that
branch and a loss of dimension, and nothing else.

Fix $q<n/4$ and sample independent uniformly random $q$-dimensional $A\subseteq H$ and
$B\subseteq H^\perp$.  Set
\[
  K=(A\oplus B)^\perp,\quad \dim K=n-2q,\qquad H_A=H\cap A^\perp,\qquad H_B=H^\perp\cap B^\perp.
\]

\begin{lemma}[Geometry of the residual]
\label{lem:residual}
The decompositions $H=A\oplus H_A$ and $H^\perp=B\oplus H_B$ are orthogonal, both $H_A$ and $H_B$ lie
inside $K$ and have dimension $(n-2q)/2$, and within $K$ they are orthogonal complements of each
other. Moreover, conditioned on $(A,B)$, the subspace $H_A$ is Haar-random of dimension $(n-2q)/2$
inside $K$.
\end{lemma}

\begin{proof}
Since $A\subseteq H$, the orthogonal complement of $A$ inside $H$ is $H\cap A^\perp=H_A$, so that
$H=A\oplus H_A$ with $\dim H_A=n/2-q$. The same argument applied to $B\subseteq H^\perp$ gives
$H^\perp=B\oplus H_B$.

Next, both residual subspaces lie in $K$.  Indeed $H_A\subseteq A^\perp$ by definition, while
$B\subseteq H^\perp$ gives $B^\perp\supseteq H\supseteq H_A$, so that
\[
  H_A\ \subseteq\ A^\perp\cap B^\perp\ =\ K ,
\]
and symmetrically $H_B\subseteq K$. Since $H_A\subseteq H$ and $H_B\subseteq H^\perp$, the two are
orthogonal to each other, and $\dim H_A+\dim H_B=n-2q=\dim K$, so they are orthogonal complements
inside $K$.

What is left is to identify the conditional law of $H_A$ given $(A,B)$, and this is the only place where
we use the joint $\On$-invariance of $(H,A,B)$. Given $h\in\mathrm O(K)$, extend it to $g\in\On$ acting
as the identity on $A$ and on $B$ and as $h$ on $K$. Then $g$ stabilizes $(A,B)$ setwise, so
$(gH,A,B)$ has the same conditional law as $(H,A,B)$ given $(A,B)$. Since $g$ preserves $A^\perp$ and
$H_A\subseteq K$, we have $gH_A=gH\cap A^\perp=h(H_A)$. The conditional law of $H_A$ is thus
$\mathrm O(K)$-invariant, and hence Haar.
\end{proof}

\begin{lemma}[Reference-copy reduction]
\label{lem:reference}
Let an adversary be given the classical descriptions of $A$ and $B$, and then one challenge from
$\rho_0$ or $\rho_1$ with a uniform bit.  Consider any strategy with Borel message functions whose
remaining classical communication with the holder of $H$ is at most $L$ bits.  Its distinguishing
advantage, averaged over $H$ and over $(A,B)$, is at most $\frac{2q}{n}+2\beta(n-2q,L+1)$.  The
additional bit here is the orientation of the gap, which the holder of $H$ sends.
\end{lemma}

\begin{proof}
If $q=0$ there is nothing to condition on: the challenge is the residual branch itself, and the
argument given below for that branch yields $2\beta(n,L+1)$ directly.  We may therefore assume
$q\ge1$.

Consider first why handing the adversary the spans concedes at least as much as handing it the
copies.  The span of $q$ i.i.d.\ Haar vectors in $H$ is almost surely a uniformly random
$q$-dimensional subspace $A\subseteq H$.  By $\On$-invariance, the conditional law of the vectors
given their span depends on nothing but $A$: it is a fixed distribution on spanning $q$-tuples,
rotated into $A$.  So an adversary holding a basis of $A$ can resample the reference copies locally
with the correct joint distribution, adaptively and one at a time if they were requested that way.

Now condition on $(A,B)$.  By \Cref{lem:residual} and $P_H=P_A+P_{H_A}$,
\[
  \frac{2P_H}{n}=\frac{2q}{n}\cdot\frac{P_A}{q}+\Big(1-\frac{2q}{n}\Big)\cdot\frac{P_{H_A}}{(n-2q)/2},
\]
and symmetrically for $\rho_1$ with $(B,H_B)$.  On each branch the challenge is a mixture:
with probability $2q/n$ the maximally mixed state on the \emph{known} subspace, and with probability
$1-2q/n$ the maximally mixed state on the residual.

The advantage splits along that mixture, because for a fixed strategy and fixed randomness the
acceptance probability is affine in the input density operator, being a composition of fixed
channels, of measurements, and of the holder's fixed responses.  On the known branch we bound the
advantage by $1$, paying $2q/n$.

Now for the residual branch.  There the challenge is maximally mixed on $H_A$ or on
$H_B=K\cap H_A^\perp$, inside $K$ of dimension $n-2q$, with $H_A$ Haar of half dimension.  A pure
version of that state is a uniform unit vector in $H_A$ or in its complement within $K$, so deciding
which of the two, using $L$ bits of interaction with the holder of $H$, is an $L$-bit protocol for
$\VSP_{n-2q}$.

One bit separates that from $\beta$, namely the orientation: the distinguishing advantage is the
absolute gap between the two branches, whereas $\beta$ is one-sided.  As in the proof of
\Cref{thm:transcript-security}, we let the holder of $H$ send one further bit recording the sign of
that gap.  The strategy then becomes an $(L+1)$-bit protocol whose advantage is half the
residual-branch distinguishing advantage, by the conventions of \Cref{sec:prelim-cc}.  So this branch
contributes at most $2\beta(n-2q,L+1)$, which together with the $2q/n$ already paid is what the lemma
asserts.
\end{proof}

\begin{corollary}[The associated problem for the EFI game]
\label{cor:reference-vsp}
For $q<n/4$ there is a universal $C>0$ such that every strategy with at most $q$ reference copies per side and at most $L$ bits of classical interaction with the holder of $H$ has average distinguishing advantage at most
\[
  \frac{2q}{n}+C\sqrt{\frac{L+2}{(n-2q)^{1/3}}}.
\]
\end{corollary}

\begin{proof}
Give the adversary the spans, which only increases its power and replaces the quantum reference
copies by classical side information.  \Cref{lem:reference} then bounds the average advantage by
$\frac{2q}{n}+2\beta(n-2q,L+1)$, and \Cref{lem:vsp-small} bounds the second term, applied with
$n-2q>n/2$.  We make that application for each fixed $(A,B)$ separately, with $(A,B)$ hard-wired into
the protocol, and averaging over $(A,B)$ then preserves the bound.  Absorbing the factor $2$ into $C$
gives the stated estimate.
\end{proof}

\subsection{The security theorem}
\label{sec:security-thm}

Everything above bounds an advantage \emph{averaged} over the random subspaces, whereas an oracle
separation needs a \emph{single} sequence $\{H_\lambda\}$ that defeats every adversary at once.  The
passage from one to the other is the same in all three places where we make it, here and in
\Cref{cor:hybrid-fixed} and \Cref{thm:single}, so we record it once.  It has two ingredients: the
adversaries must form a countable family, and the maximization over advice must sit \emph{inside} the
expectation.

\begin{lemma}[From average security to a fixed object]
\label{lem:fixing}
Let $\mathbf X=\{X_\lambda\}$ be a sequence of random objects and $\{N_i\}_{i\in\mathbb N}$ a
countable family of machines.  Write $\Adv^*_{N_i}(\lambda;\mathbf X)$ for the advantage of $N_i$ at
security parameter $\lambda$, \emph{maximized over its advice strings}, and suppose that for every
$i$,
\[
  \mu_i(\lambda)\ \eqdef\ \E_{\mathbf X}\big[\Adv^*_{N_i}(\lambda)\big]\ =\ \negl(\lambda) .
\]
Then, with probability one over $\mathbf X$, every $i$ has a threshold beyond which
$\Adv^*_{N_i}(\lambda)\le\lambda^{2}\mu_i(\lambda)$, which is again negligible, and which is
$2^{-\Omega(\lambda)}$ whenever $\mu_i$ is.  For such
a fixed $\mathbf X$, every non-uniform adversary whose machine lies in the family has negligible
advantage, however its advice depends on $\mathbf X$.
\end{lemma}

\begin{proof}
Markov's inequality at the threshold $\lambda^{2}\mu_i(\lambda)$ gives
\[
  \Pr\nolimits_{\mathbf X}\big[\Adv^*_{N_i}(\lambda)>\lambda^{2}\mu_i(\lambda)\big]\ \le\ \lambda^{-2} ,
\]
which is summable in $\lambda$.  The first Borel--Cantelli lemma does not require independence, so
almost surely only finitely many $\lambda$ violate the threshold.  Intersecting these
probability-one events over the countably many $i$ leaves an event of probability one on which every
$N_i$ obeys it.

The last sentence follows because the maximization over advice sits inside $\mu_i$.  A non-uniform
adversary is a pair (machine, advice sequence) whose machine is some $N_i$, and its advantage at
every $\lambda$ is at most $\Adv^*_{N_i}(\lambda)$, whatever advice it uses and however that advice
was chosen once $\mathbf X$ was fixed.
\end{proof}

The polynomial threshold is not essential.  Any $g(\lambda)\mu_i(\lambda)$ with
$\sum_\lambda g(\lambda)^{-1}<\infty$ would serve.  We take $g$ polynomial rather than, say, $g=\mu_i^{-1/2}$, because a polynomial factor
leaves the exponent of the expected advantage intact, and that is the exponent \Cref{sec:exact-rate}
compares with the attack.

Countability is the other ingredient, and it is the reason we enumerate \emph{machines} rather than
(machine, advice-sequence) pairs, of which there are uncountably many.  That enumeration suffices
only because $\max_\alpha$ sits inside $\E_{\mathbf X}$, and the reduction of
\Cref{thm:transcript-security} accommodates this at no charge: the holder of the hidden object
sends the best advice, together with a sign bit.

\begin{theorem}[EFI security]
\label{thm:EFI-security}
With probability one over $\{H_\lambda\}$, the ensembles
$\{\rho_{0,\lambda}\},\{\rho_{1,\lambda}\}$ form an EFI pair relative to $\cO$.  They are secure
against every $\QPT$ distinguisher that has the public source $\cR$, adaptive classical access to
$\Count$, and classical advice of any polynomial length depending arbitrarily on the fixed oracle.
\end{theorem}

\begin{proof}
Efficient generation and farness are immediate: $G(1^\lambda,b)$ calls $\cR(b,1^\lambda)$, and
$\rho_{0,\lambda}\rho_{1,\lambda}=0$ gives $\td=1$.  For indistinguishability there is one point to
arrange.  \Cref{cor:reference-vsp} bounds an \emph{expected} advantage over the random subspaces,
whereas \Cref{lem:fixing} needs that expectation for each member of a countable family.  Producing
such a family, and bounding its expectations, is all that is left to do.

To begin, fix a finite universal gate set and replace each machine's gates by approximations over
it, which by the Solovay--Kitaev theorem costs only polynomial overhead.  Up to a negligible change
in advantage, every $\QPT$ machine is then captured by a \emph{countable} family $\{N_i\}$ of oracle machines with
explicit polynomial bounds $t_i(\lambda)$ on running time, advice length, sample calls, and
transcript length.

Fix such a machine $N_i$ and let $\Adv^*_{N_i}(\lambda;\mathbf H)$ denote the \emph{best-advice}
advantage, the maximum over advice strings of length $\le t_i(\lambda)$.  We will show that its
expectation over $\mathbf H$ is $2^{-\Omega(\lambda)}$.

Place $N_i$ in the two-party game of \Cref{thm:transcript-security} at security parameter $\lambda$,
with Alice running the machine and Bob holding $H_\lambda$, and write $L$ for the length of the
resulting transcript of classical $\Count$ queries and answers.  Subspaces at other security
parameters are independent of the challenge, so Alice and Bob sample those from shared public
randomness at no cost, which makes every $H_\lambda$-independent resource public.  Bob then sends the
maximizing advice string, at a cost of $a\le t_i(\lambda)$ bits; the sign bit that converts the
two-sided advantage into a one-sided one is the orientation bit already charged inside
\Cref{lem:reference}.

What results is a strategy of the kind \Cref{cor:reference-vsp} bounds, with $q\le t_i(\lambda)$
reference copies per side and $L+a\le2t_i(\lambda)$ bits of interaction, so that its budget there is
$L+a+1\le2t_i(\lambda)+1$.  Hence, with $n=2^{\Omega(\lambda)}$,
\[
  \mu_i(\lambda)\eqdef\E_{\mathbf H}\big[\Adv^*_{N_i}(\lambda)\big]\ \le\ \frac{2t_i}{n}+C\sqrt{\frac{2t_i+2}{(n-2t_i)^{1/3}}}\ =\ 2^{-\Omega(\lambda)} .
\]
That is the hypothesis of \Cref{lem:fixing}, with $\mathbf X=\mathbf H$, and the lemma supplies the
rest: a probability-one set of sequences $\{H_\lambda\}$ on which every $N_i$, with whatever advice,
has negligible advantage.
\end{proof}

\begin{theorem}[The separation in the source model]
\label{thm:separation-source}
There is an oracle $\cO=(\cR,\Count)$, consisting of a sample-access state source and a Boolean
counting oracle, such that
\begin{enumerate}[label=(\roman*),leftmargin=2.2em]
  \item one-way puzzles, including inefficiently verifiable and classically-secure ones, do not exist
    relative to $\cO$, and neither do $\QEFID$ pairs; and
  \item the half-subspace pair $(\rho_{0,\lambda},\rho_{1,\lambda})$ is a secure EFI pair relative to
    $\cO$ against every non-uniform $\QPT$ distinguisher that calls $\cR$ polynomially many times,
    makes $\poly(\lambda)$ adaptive \emph{classical} queries to $\Count$, then makes one
    \emph{coherent} query to $\Count$ of width $2^{\poly(\lambda)}$, and receives $\poly(\lambda)$ bits of
    classical advice depending arbitrarily on the fixed oracle.
\end{enumerate}
\end{theorem}

\begin{proof}
\emph{(i)} $\OWPuzz$s fail for \emph{every} fixing by \Cref{thm:no-owpuzz}, as do $\QEFID$ pairs by
\Cref{cor:no-qefid}.

\emph{(ii)} EFI security in the stated access model holds on a probability-one set of fixings.
Without the final coherent query this is \Cref{thm:EFI-security}, which we proved above.  With that
query it is \Cref{cor:hybrid-fixed}, whose proof we defer to \Cref{sec:hybrid}, since it needs the
one-query bounds of \Cref{sec:coherent}.

Choosing a fixing in the intersection, which again has probability one, gives an oracle for which
both parts hold.
\end{proof}

Part~(ii) permits at most one coherent query to $\Count$, so this is a separation in the
classical-query model (\Cref{model:access}) extended by a single coherent query, and the fully
coherent case is \Cref{conj:lift}. \Cref{thm:separation-source} is an intermediate statement rather
than the final one: its oracle still contains the state source $\cR$. \Cref{sec:single} replaces
$\cR$ by a Boolean function, leaving the access model unchanged, and the resulting \Cref{thm:single}
is the main separation of the paper and the formal form of Theorem~A.

\paragraph{Quantum advice.}  \Cref{thm:transcript-security} charges classical advice to the
transcript.  Quantum advice has no counterpart there, since it corresponds to a one-way
\emph{quantum} message from the holder of $H$, so two separate theorems cover it instead.
\Cref{cor:quantum-advice} rules out $\poly(\lambda)$ qubits of arbitrarily $H$-dependent advice
against a holder of one challenge copy; it does so by converting the advice into a single Boolean
phase query, through Rosenthal's synthesis~\cite{Rosenthal23}, and then applying
\Cref{thm:one-query}.  \Cref{cor:hybrid-qadvice}, meanwhile, admits $O(\lambda)$ qubits alongside the
full classical budget of \Cref{model:hybrid}, by charging a net of advice states to that budget.
What is left over is polynomial-size quantum advice together with classical queries, and under the
same conversion that becomes a coherent query placed \emph{before} the classical ones, which is
precisely the ordering \Cref{sec:hybrid} identifies as the limit of the method.

The natural limit of the question is one-way \emph{quantum} communication for $\VSP$ from the holder
of $H$ to a holder of the classical $u$, and that is settled at $\Theta(n)$ by
\Cref{prop:reverse-oneway}.  It is a different model from the one our advice statements need, and
much stronger than anything we use here.

\subsection{The optimal advantage of a classical adversary}
\label{sec:exact-rate}

\Cref{thm:EFI-security} gives advantage $2^{-\Omega(\lambda)}$, and it is natural to ask whether the
pair is in fact perfectly hidden from $\Count$.  It is not, and for the natural class of classical
attacks we can determine the rate exactly.  The following fixes it at $n^{-1/2}$ from both sides.

\begin{model}[Measure-then-count adversary]
\label{model:mtc}
The adversary fixes a finite POVM $\{E_k\}_{k\in\cK}$ on $\bbC^n$, of any size and \emph{not} depending on
$H$, measures its single challenge copy with it, and outputs a bit computed from the outcome $k$
together with the exact values $\{\tr(E_{k'}P_H)\}_{k'\in\cK}$. No reference copies are used, and no
bound is placed on the post-processing.
\end{model}

The values in \Cref{model:mtc} are a relaxation of what $\Count$ provides: for an efficiently implementable
POVM, the sampler that calls $\cR(0,1^\lambda)$, measures, and outputs the label $k$ has output
probability $\frac2n\tr(E_kP_H)$, so \Cref{lem:pe-easy} returns its bits.

\begin{proposition}[The exact advantage of a measure-then-count adversary]
\label{prop:exact-rate}
Let $H\subseteq\bbR^n$ be Haar-random of dimension $n/2$, with $n\ge4$.
\begin{enumerate}[label=(\roman*),leftmargin=2.2em]
\item \emph{(Upper bound, every POVM.)} Every measure-then-count adversary has
  $\E_H[\Adv]\le2/\sqrt{\pi n}$.
\item \emph{(Lower bound, one basis and $\poly(\lambda)$ queries.)} Put
  $t=2\lceil\log_2n\rceil+\lambda$.  There is an adversary achieving
  \[
    \E_H[\Adv]\ \ge\ 2(\pi n)^{-1/2}\sqrt{1-2/n}-4n^{-1}2^{-\lambda} ,
  \]
  namely the one that measures its challenge in the computational basis to obtain $j$, queries
  $\Count$ for the first $t$ bits of
  $\Pr[\cR(0,1^\lambda)\text{ measured in the computational basis}=j]$, and outputs $1$ if that
  truncation is below $\frac1n-2^{-t}$.  Apart from the single computational-basis measurement of its
  challenge it is classical polynomial time, and it makes $t=\poly(\lambda)$ classical queries, since
  $\log_2n=\poly(\lambda)$ for any source a $\QPT$ algorithm can call.
\end{enumerate}
In particular the two bounds coincide asymptotically: the exact rate is
$2(\pi n)^{-1/2}\big(1+O(n^{-1})\big)$, attained by the computational basis.  Of the two
deficits in the lower bound, the factor $\sqrt{1-2/n}$ contributes $O(n^{-1})$ relatively and the
finite precision contributes $O(2^{-\lambda}n^{-1/2})$, which at $n=2^{\lambda}$ is $O(n^{-3/2})$, so
the first of them is what the rate records.
\end{proposition}

Both parts rest on one computation, the mean absolute deviation of a single diagonal entry of $P_H$
about its mean $\frac12$: the upper bound sums it over the eigenvectors of an arbitrary POVM, and the
lower bound realizes it in the computational basis.  That computation is a Beta-integral, and it and
the finite-precision accounting are the whole of the proof; both are in \Cref{app:exact-rate}.

\begin{corollary}[A protocol of logarithmic length]
\label{cor:vsp-cheap}
For even $n\ge4$, $\beta\big(n,\lceil\log_2n\rceil+1\big)\ge(2\pi n)^{-1/2}$: Alice sends the index $j$ drawn with
probability $u_j^2$, and Bob replies with $\mathbf 1[(P_H)_{jj}<\frac12]$.
\end{corollary}

\begin{proof}
Alice's sampling step is exactly the computational-basis measurement of the challenge, carried out on
her classical description of it.  Indeed, for $u$ uniform on the unit sphere of $H$,
\[
  \E[u_j^2]\ =\ \frac{2(P_H)_{jj}}n ,
\]
and symmetrically in $H^\perp$.  So the index $j$ she sends has the law of the outcome in
\Cref{prop:exact-rate}(ii).  Bob, holding $H$, then applies the exact rule
$\mathbf 1[(P_H)_{jj}<\frac12]$ rather than a truncation of it, and the protocol reproduces that
distinguisher with no loss to finite precision.

One factor of two separates the two statements.  Part~(ii) computes the distinguishing advantage of
the exact rule as $2\,\E_H|\xi_{e_1}|$, whereas the advantage of a protocol is half of that: the
former measures the gap between the two branches, the latter the excess of the success probability
over $\frac12$.  The advantage here is thus $\E_H|\xi_{e_1}|$, which is at least $(2\pi n)^{-1/2}$ by
\Cref{lem:mad}(a).

Alice sends an index in $[n]$ and Bob a single bit, so the transcript carries $\lceil\log_2n\rceil+1$
bits, as claimed.
\end{proof}

Three consequences follow, on the dimension, on the security level, and on how much the
oracle is worth.

First, the dimension must be superpolynomial.  By \Cref{prop:exact-rate}(ii) a classical adversary
making $O(\lambda)$ queries already distinguishes the pair with advantage $\Omega(n^{-1/2})$, so
negligible security forces $n^{-1/2}=\negl(\lambda)$, that is, $n=\lambda^{\omega(1)}$.  Our choice
$n=2^{\Theta(\lambda)}$ is one convenient point in that range, and the proofs use nothing else
about it.  Indeed, every bound in this paper is negligible as soon as $n$ is superpolynomial.

Second, the security level is $2^{-\Theta(\lambda)}$, and the exponent is the quantity at stake.
\Cref{thm:EFI-security} proves $O(\poly(\lambda)\cdot n^{-1/6})$, the polynomial collecting the
adversary's own parameters together with the factor its Markov step charges for passing from the
expectation to the fixed oracle, and \Cref{prop:exact-rate} exhibits
$\Omega(n^{-1/2})$.  Both are $2^{-\Theta(\lambda)}$, so the separation itself is unaffected.  The residual gap does not come from the
imported bound.  A worst-case communication
bound of $\Omega(n^\alpha)$ enters \Cref{thm:EFI-security} only through \Cref{lem:vsp-small}, whose
amplification step costs a square root and returns $O(\sqrt L\,n^{-\alpha/2})$.  Matching
\Cref{prop:exact-rate} would require $\alpha=1$, and that is impossible, since Raz's
protocol decides $\VSP_n$ with $O(\sqrt n)$ bits~\cite{Raz99,Mon19}.  Even the optimal worst-case
bound $\Theta(\sqrt n)$ would give only $n^{-1/4}$.  The gap is thus a feature of the route rather
than of the pair, and it can be closed by leaving the route: \Cref{app:sharp} bounds the same
adversaries by $\poly(\lambda)\cdot n^{-1/2}$, using the linearity of $\rho_0-\rho_1$ in $R_H$ in
place of any communication bound.  What would close it \emph{within} the communication route is a
direct estimate of $\beta(n,L)$ in the small-advantage regime.  \Cref{prop:oneway-rate} supplies
such an estimate for one-way protocols, linear in the message length, and \Cref{cor:twoway-disc}
handles the two-way case for very short transcripts.  The general two-way question remains open, and
\Cref{sec:discussion} returns to it as a question about $\VSP$ rather than as one this construction
needs.

Third, the counting oracle is worth a single bit of information about $H$.  \Cref{lem:local-corr}
shows that a single $[-1,1]$-valued function of $H$ correlates with $R_H$ at scale $n^{-1/2}$, and
\Cref{prop:exact-rate}(i) shows that granting the adversary the \emph{exact} value of every outcome
probability of an arbitrary measurement does not exceed that scale.  The same $n^{-1/2}$ recurs in
\Cref{sec:coherent} as the value of one coherent query, and it is attained there as well:
\Cref{cor:one-query-tight} turns the attack of part~(ii) into a one-query adversary meeting
\Cref{thm:one-query} to within $\sqrt{\log n}$.

\subsection{The one-way rate for \texorpdfstring{$\VSP_n$}{VSP}}
\label{sec:oneway-rate}

\Cref{prop:exact-rate} bounds every measure-then-count adversary by $2/\sqrt{\pi n}$ with no bound
on what it is told, but its adversary holds one \emph{quantum} copy of the challenge, whereas Alice in
\Cref{def:vsp} holds $u$ as classical data and is strictly stronger (\Cref{sec:advice-is-a-query}); at
unbounded message length she can send $u$ itself and decide with certainty.  The next proposition
bounds the one-way rate for the classical-input sender.

\begin{proposition}[The one-way rate for $\VSP_n$]
\label{prop:oneway-rate}
Let $\beta_{\mathrm{ow}}(n,L)$ be the maximum advantage on the $\VSP_n$ distribution over one-way protocols
with shared randomness~\cite{KNR99}, in which Alice, holding the classical description of $u$, sends a
single message of at most $L$ bits and Bob, holding $H$, announces the guess. Then, for even $n\ge4$, there is a universal $C$ with
\[
  \frac1{\sqrt{2\pi n}}\ \le\ \beta_{\mathrm{ow}}\big(n,\lceil\log_2n\rceil+1\big),
  \qquad
  \beta_{\mathrm{ow}}(n,L)\ \le\ C\,\frac{L+1}{\sqrt n}\quad\text{for all }L\ge0 .
\]
Consequently the one-way randomized communication complexity of $\VSP_n$ at constant advantage is
$\Theta(\sqrt n)$: the bounds above give $\Omega(\sqrt n)$, and the upper bound is Raz's one-way
protocol~\cite{Raz99}, whose first published proof is~\cite[App.~A]{Mon19} and for which
Gosset and Smolin give a computationally efficient alternative~\cite{GS19}. The
constant-advantage threshold is not new; see \Cref{sec:oneway-rate}. What the display adds is an
upper bound at every message length, linear in $L$, on the exact promise and the Haar distribution.
\end{proposition}

The lower bound is \Cref{cor:vsp-cheap}, whose protocol is one-way.  We prove the upper bound in
\Cref{app:rates}, which also collects the harmonic analysis on the sphere that it uses.  Here we
describe only the mechanism, deferring the two estimates that carry it.

Fix the shared randomness, so that the protocol is deterministic and Alice's message realizes a
partition $\{A_m\}$ of the sphere $S^{n-1}$ into at most $2^{L+1}$ cells.  Write
$\mu_H\eqdef\nu_H-\nu_{H^\perp}$ for the difference of the uniform probability measures on the unit
spheres of $H$ and $H^\perp$.  Since the two challenge distributions have the same marginal on $H$
and differ only in which of $\nu_H,\nu_{H^\perp}$ produces $u$, Bob's optimal decision given the
cell $A_m$ has advantage governed by $|\mu_H(A_m)|$, and the whole quantity to bound is
$\E_H\sum_m|\mu_H(A_m)|$.  So the question is how much a single cell of the sphere can distinguish
the two spheres, and how those contributions add up over a partition.

For one cell the answer comes from a second-moment computation.  The operator behind $\E_H[(\int
f\,d\mu_H)^2]$ is diagonalized exactly in \Cref{lem:second-moment} by the spherical harmonics, and
its eigenvalue at degree two is $\frac1{n-1}$, with the higher degrees decaying rapidly.  An
indicator of measure $\alpha$ has weight $O\big(\alpha^2\ln^2(e/\alpha)\big)$ at degree two, by a
level-$\ell$ inequality on the sphere of \Cref{lem:level-ell}, proved there in the style of the
Boolean-cube inequalities of harmonic analysis, so a single cell of measure $\alpha$ contributes $O(\alpha\ln(e/\alpha)/\sqrt n)$.

Summing over cells is then an entropy bound.  The measures $\alpha_m$ of the cells sum to $1$, so
$\sum_m\alpha_m\ln(e/\alpha_m)$ is an entropy and is at most $O(L)$ for a partition into $2^{L+1}$
parts.  The dependence on the message length is linear for this reason.  What we have to control in addition is the contribution of the higher degrees.  That is a tail, and
for small cells it must be suppressed; \Cref{lem:cell-leak} does this, and \Cref{app:oneway-proof}
combines the two ranges of cell size.

For two-way protocols the same second-moment computation gives a bound that is strong only for very
short transcripts.

\begin{corollary}[Two-way protocols with very short transcripts]
\label{cor:twoway-disc}
For even $n\ge4$ there is a universal $C>0$ with $\beta(n,L)\le C\,2^{L/2}/\sqrt n$ for every $L\ge0$.
\end{corollary}

\begin{proof}
Write $\gamma$ for the law of $H$, the rotation-invariant measure on the Grassmannian $\mathrm{Gr}$.
Fix the shared randomness, so that the protocol is deterministic, and let $\{A_t\times B_t\}$ be the
rectangles into which its transcripts partition $S^{n-1}\times\mathrm{Gr}$, of which there are
$K\le2^{L+1}$ (\Cref{sec:prelim-cc}).  On each transcript the output is a fixed bit, and the two
challenge branches contribute masses $\frac12\E_H[\mathbf 1_{B_t}(H)\nu_H(A_t)]$ and
$\frac12\E_H[\mathbf 1_{B_t}(H)\nu_{H^\perp}(A_t)]$ to it, so answering optimally on each transcript
gives
\begin{equation}
  \beta\ \le\ \tfrac14\sum_t\big|\E_H\big[\mathbf 1_{B_t}(H)\,\mu_H(A_t)\big]\big| .
  \label{eq:twoway-rect}
\end{equation}

Consider a single rectangle first.  \Cref{lem:second-moment} with Parseval and
$\omega_{2j}\le\omega_2=\frac1{n-1}$ gives $\E_H[\mu_H(A)^2]\le4\omega_2\sigma(A)$, so Cauchy--Schwarz
in $H$ gives
\[
  \big|\E_H\big[\mathbf 1_{B_t}(H)\,\mu_H(A_t)\big]\big|
  \ \le\ \sqrt{\gamma(B_t)}\;\sqrt{\E_H\big[\mu_H(A_t)^2\big]}
  \ \le\ \frac2{\sqrt{n-1}}\,\sqrt{\gamma(B_t)\,\sigma(A_t)} .
\]
Discarding $\sqrt{\gamma(B_t)\sigma(A_t)}\le1$ and summing over the $K$ transcripts already gives
$\beta\le2^{L}/\sqrt{n-1}$, but the factor $2^{L}$ is spurious.  It would be paid only if every
rectangle were as large as the whole space, whereas the $A_t\times B_t$ partition that space and so
$\sum_t\sigma(A_t)\gamma(B_t)=1$.  Our remedy is to keep the two measures instead of discarding them.

Cauchy--Schwarz over the $K$ transcripts then gives $\sum_t\sqrt{\gamma(B_t)\sigma(A_t)}\le\sqrt K$,
and substituting this and the previous display into \eqref{eq:twoway-rect},
\[
  \beta\ \le\ \frac{\sqrt K}{2\sqrt{n-1}}\ \le\ \frac{2^{(L+1)/2}}{2\sqrt{n-1}} ,
\]
as claimed.
\end{proof}

The two bounds available for two-way protocols are complementary. The second-moment bound of
\Cref{cor:twoway-disc} is the stronger for $L\le\frac23\log_2n$, and the corruption bound of
\Cref{lem:vsp-small} beyond that.

\paragraph{Prior work on the constant-advantage bound.}  An $\Omega(\sqrt n)$ lower bound for one-way
protocols at constant advantage is already known, by a reduction that Gosset and Smolin record and
attribute to Kothari~\cite[\S1]{GS19}.  The $\tfrac14$-Partial Matching problem, a variant of the
Hidden Matching problem introduced in~\cite{BJK04}, has one-way classical complexity $\Theta(\sqrt
n)$~\cite{GKKRW07}, and its instances embed into $\VSP_n$ as a state
$\ket\psi=n^{-1/2}\sum_i(-1)^{x_i}\ket i$ together with a projector $\Pi$ built from Bob's matching,
for which $\bra\psi\Pi\ket\psi$ takes the two values $\tfrac14$ and $\tfrac34$.

That reduction differs in two ways from what \Cref{prop:oneway-rate} proves.  First,
its instances are \emph{gapped}, with $\bra\psi\Pi\ket\psi$ bounded away from $0$ and from $1$, so
the bound is for a promise strictly weaker than the exact one of \Cref{def:vsp}.  Since exact
instances satisfy any gapped promise, a lower bound on the exact problem is the stronger of the two.
Second, the hard instances are the structured ones inherited from Partial Matching, and not the
rotation-invariant distribution of \Cref{def:vsp}, which is the distribution the security proof
requires.  What \Cref{prop:oneway-rate} contributes is the rate $O((L+1)/\sqrt n)$ at every
message length, for the exact promise and on the Haar distribution.  The reduction above gives the
constant-advantage endpoint for the \emph{gapped} promise; since exact instances satisfy that
promise, it does not by itself give the endpoint for \Cref{def:vsp}.

\paragraph{Relation to the $\Omega(\sqrt n)$ program for $\VSP$.}  That $\VSP_n$ should require
$\Omega(\sqrt n)$ bits, matching Raz, is an established target.  Klartag and Regev suggest the rate
$e^{-c\sqrt n}$ in place of $e^{-cn^{1/3}}$, in a form using the arithmetic mean of $\nu_H(A)$ and
$\nu_{H^\perp}(A)$, and Grupel~\cite[\S1]{Grupel17} sharpens this to a conjecture in which the
\emph{geometric} mean $\sqrt{\nu_H(A)\nu_{H^\perp}(A)}$ is at least $0.9\,\sigma(A)$ except with
probability $e^{-c\sqrt n}$, for every $A$ of measure at least $e^{-c'\sqrt n}$.  The geometric mean
is the essential change here.  The spherical cap that makes the Klartag--Regev rate tight satisfies
the geometric form, so the conjecture escapes the very example behind their assessment that
$\Omega(\sqrt n)$ is ``probably impossible using the rectangle bound''.  Grupel proves the conjecture
for sets depending on $O(\sqrt n)$ coordinates, and deduces $\Omega(\sqrt n)$ for two-way protocols
of total rank $O(\sqrt n)$, the rank counting the linear functionals of Alice's input that the
protocol evaluates; a variant of Raz's protocol lies in that class~\cite[Thm.~1.2,
Cor.~1.3]{Grupel17}.  The conjecture itself remains open.

\Cref{prop:oneway-rate} is incomparable to that result, and it does not require the conjecture.  It
places no rank restriction on Alice, whose message induces an arbitrary measurable partition, so the
total rank of an $L$-bit one-way protocol may be as large as $Ln$.  What replaces the conjecture is a
weaker demand on the concentration statement.  The rectangle method needs each cell to be atypical
with probability below $2^{-L}$, which forces a high-probability statement.  A one-way protocol, by
contrast, is charged $\sum_m\E_H|\mu_H(A_m)|$, a first moment over cells, and for that a
\emph{second}-moment computation suffices; \Cref{lem:second-moment} performs it exactly.
High-probability concentration is an artifact of interaction rather than of the geometry,
and the same observation places \Cref{cor:twoway-disc} at the limit of what a per-rectangle second
moment yields.  Neither Klartag and Regev~\cite{KR11} nor Grupel~\cite{Grupel17} treats the
small-advantage regime, which is the one the security proof uses.

\section{Coherent queries}
\label{sec:coherent}

The main theorem queries $\Count$ classically, which is what makes the communication argument
available: that argument charges the adversary for the classical bits it exchanges with the oracle,
and a superposition query is not a classical message.  We now ask how much survives under coherent
access.

Throughout this section $H\subseteq\bbC^n$ is Haar-random of dimension $n/2$, and we write $P=P_H$,
$R=2P-\Id$ and $\Delta=R/n$.  In the representation $R=UDU^\dagger$ we take the conjugating unitary
Haar on $\mathrm{SU}(n)$ rather than on $\Un$, which changes nothing, as \Cref{lem:su-so} below
records, and which is what the concentration inequality of \Cref{sec:concentration} requires.  The
arguments transfer to the real ensemble of the main construction with a change of universal constant,
with $\mathrm{SO}(n)$ in place of $\mathrm{SU}(n)$, since the bounds below use only that
$R=UDU^\dagger$ for $U$ Haar on a group of curvature $\Omega(n)$ and $D$ a fixed traceless
involution.  Advantage is measured as the bias \eqref{eq:bias-def}, half the distinguishing
advantage.  The adversary of \Cref{sec:onequery} holds one challenge copy and no reference copies, and we treat
the full EFI game in \Cref{sec:hybrid}.

One point about the model has to be settled before the estimates begin.  An oracle returning an exact
real \emph{value} in superposition is unsafe.  Write $v_x(P)$ for the real number such an oracle
returns at the query label $x$ when the hidden projector is $P$, and suppose, as the counting oracle
would, that these values include the matrix entries of $P$.  Querying the uniform superposition of
labels produces the state $\ket{\Phi_P}=M^{-1/2}\sum_x\ket x\ket{v_x(P)}$, and for two distinct
projectors drawn from the ensemble the values differ at every label almost surely, so the
corresponding states are orthogonal.  One coherent query would then write an exact record of the
hidden subspace.  Of course that model is an idealization: a register holding an exact real, with orthogonal states for
distinct values, spans a non-separable Hilbert space.  Rounding to any finite precision returns a
Boolean oracle of width $2^{\poly(\lambda)}$, and \Cref{thm:one-query} covers that, since it holds for
\emph{every} Boolean function of $H$.

\subsection{Concentration of measure on the classical compact groups}
\label{sec:concentration}

This subsection collects the probabilistic input, in the order in which the proof needs it: first the
elementary scalar fact, which already settles a special case and fixes the scale of every estimate
below, then the matrix statement, which the maximum over truth tables
requires.

Both rest on the same classical fact, that Lipschitz functions on a compact group of large curvature
concentrate.  In the Hilbert--Schmidt metric,~\cite[Thm.~5.17]{Meckes19} gives, uniformly for
$\mathrm{SO}(n),\mathrm{SU}(n),\Un$ and $\mathrm{Sp}(2n)$,
\begin{equation}
  \Pr\big[F\ge\E F+t\big]\ \le\ e^{-(n-2)t^2/(24\,\mathrm{Lip}(F)^2)}
  \label{eq:scalar-conc}
\end{equation}
for every Lipschitz $F$; applying this to $F$ and to $-F$, which has the same Lipschitz constant,
yields the two-sided bound $\Pr[|F-\E F|\ge t]\le2e^{-(n-2)t^2/(24\,\mathrm{Lip}(F)^2)}$ that we quote
below.  The scale $n^{-1/2}$ in \eqref{eq:scalar-conc} is the source of every $n^{-1/2}$ in this
section.

\begin{lemma}[A single linear statistic]
\label{lem:scalar-linear}
Let $R=UDU^\dagger$ with $U$ Haar on $\mathrm{SU}(n)$ or $\mathrm{SO}(n)$ and $D$ a fixed Hermitian
involution with $\tr D=0$.  For a fixed Hermitian matrix $A$ write $X_A\eqdef\frac1n\tr(AR)$.  Then $\E X_A=0$,
the map $U\mapsto X_A$ is Lipschitz with constant $2\norm A_2/n$ in the Hilbert--Schmidt metric, and
$X_A$ is subgaussian with variance proxy $O(\norm A_2^2/n^3)$.  Consequently, for any family of at
most $N$ fixed matrices $A$ with $\norm A_2\le\sqrt n$,
\[
  \E\,\max_A|X_A|\ =\ O\Big(\frac{\sqrt{\log(2N)}}{n}\Big).
\]
\end{lemma}

\begin{proof}
First, $\E X_A=0$.  Indeed, Haar measure is invariant under translation, so $\E R$ is unchanged by
conjugation by any element of the group.  The standard representation of $\mathrm{SU}(n)$ on
$\bbC^n$, and of $\mathrm{SO}(n)$ on $\bbR^n$, is irreducible, so $\E R$ must be a multiple of the
identity; and its trace is $\tr D=0$, whence $\E R=0$.

Next we bound the Lipschitz constant, from which the variance proxy will follow.  Cauchy--Schwarz in
the Hilbert--Schmidt inner product, and then $\opnorm D\le1$, give
\[
  |X_A(U)-X_A(U')|\ \le\ \tfrac1n\norm A_2\,\norm{UDU^\dagger-U'DU'^\dagger}_2\ \le\ \tfrac2n\norm A_2\,\norm{U-U'}_2 ,
\]
the last step by writing $UDU^\dagger-U'DU'^\dagger=(U-U')DU^\dagger+U'D(U-U')^\dagger$.  Substituting
this into \eqref{eq:scalar-conc} gives the stated variance proxy, the factor $n-2$ in the exponent
there supplying the third power of $n$.

Finally, we take the maximum over the family.  Under the hypothesis $\norm A_2\le\sqrt n$ every
member has variance proxy $O(n^{-2})$, so the maximal inequality for subgaussian variables, applied to
these at most $N$ variables, gives the last display.
\end{proof}

\Cref{lem:scalar-linear} already handles one regime completely.  By \eqref{eq:bias-def} the bias of
\emph{any} adversary, at any number of queries, is $\frac1n|\tr(E_fR)|$ for an effect $E_f$ on
$\bbC^n$, and $\norm{E_f}_2\le\sqrt{\tr E_f}\le\sqrt n$.  Taking $N=2^M$ in the
lemma bounds the maximum over all $2^M$ truth tables by $O(\sqrt M/n)$, which is negligible whenever
$M=O(n^{2-\eps})$.  The difficulty lies at large width, a point to which \Cref{sec:lift}
returns.

At the widths our model forces, $M$ is far larger than $n^2$ and this estimate is useless, so the
maximum over truth tables must be taken by a means other than a union bound.  The standard means is an
operator norm: a maximum of $2^M$ quadratic forms sharing one coefficient matrix costs only the norm
of that matrix.  What we then need is a bound on the expected operator norm of a random matrix whose
entries are linear statistics of $R$.  In the classical setting, where the randomness consists of
independent signs, this is matrix Khintchine: for fixed Hermitian $d\times d$ matrices $A_1,\dots,A_k$ and independent Rademacher
$\eps_i$,
\begin{equation}
  \E\,\opnorm{\textstyle\sum_i\eps_iA_i}\ \le\ C\sigma\sqrt{\log(2d)},\qquad
  \sigma^2=\opnorm{\textstyle\sum_iA_i^2},
  \label{eq:classical-khintchine}
\end{equation}
so the operator norm costs only a $\sqrt{\log}$ factor over the variance parameter, rather than the
$\sqrt k$ that a union bound over sign patterns would give.  Our randomness is a single Haar-conjugated
reflection and provides no independence, so \eqref{eq:classical-khintchine} does not apply and
we prove the analogue we need in \Cref{sec:khintchine}.

What replaces independence is curvature.  The reason \eqref{eq:scalar-conc} holds is that the special groups $\mathrm{SO}(n)$ and $\mathrm{SU}(n)$ have Ricci curvature of order $n$, and $\Un$
inherits the inequality from $\mathrm{SU}(n)$ by \Cref{lem:su-so}, so that Haar measure satisfies a
logarithmic Sobolev inequality by the Bakry--\'Emery criterion.  Huang and Tropp~\cite{HT21} show that the same
hypothesis yields concentration for matrix-valued functions, with the variance parameter of
\eqref{eq:classical-khintchine} replaced by a sum of squared derivatives along an orthonormal frame.

\begin{proposition}[{Matrix concentration under a curvature lower bound; Huang--Tropp~\cite[Thm.~1.1 of the arXiv version]{HT21}}]
\label{prop:herbst}
Let $\cM$ be a compact Riemannian submanifold of Euclidean space whose Ricci curvature is bounded
below by $\rho>0$, let $\mu$ be the uniform measure on $\cM$, let $\{\partial_i\}$ be an orthonormal
frame, and let $F\colon\cM\to\Herm_d$ be differentiable with $\E_\mu F=0$. Put
$v_F=\sup_{x\in\cM}\opnorm{\sum_i(\partial_iF(x))^2}$. Then, for a universal constant $C>0$,
\begin{enumerate}[label=(\roman*),leftmargin=2.2em]
\item $\Pr_\mu\big[\opnorm{F}\ge t\big]\ \le\ 2d\,e^{-\rho t^2/(2v_F)}$ for every $t\ge0$; and
  consequently
\item $\displaystyle \E_\mu\opnorm{F}\ \le\ C\sqrt{\frac{v_F\,\log(2d)}{\rho}}$.
\end{enumerate}
\end{proposition}

Part~(ii) is the standard conversion of a subgaussian tail into a moment bound: integrate~(i) and cut
at $t_0=\sqrt{2v_F\log(2d)/\rho}$, where the exponential factor absorbs the dimension $2d$, leaving a
remainder of order $\sqrt{v_F/\rho}$ that $t_0$ dominates.

By~\cite[Prop.~5.13]{Meckes19}, in the Hilbert--Schmidt metric the Ricci curvature is exactly
$\frac{n-2}4$ on $\mathrm{SO}(n)$ and $\frac n2$ on $\mathrm{SU}(n)$, so in both cases $\rho=\Omega(n)$
and \Cref{prop:herbst} reads
\begin{equation}
  \E\,\opnorm{F}\ \le\ C\,\sqrt{\frac{v_F\,\log(2d)}{n}},
  \qquad v_F=\sup_{x\in\cM}\opnorm{\textstyle\sum_i(\partial_iF(x))^2}\ \text{ as in \Cref{prop:herbst}.}
  \label{eq:herbst-group}
\end{equation}
Applying \Cref{prop:herbst} to our ensemble requires care on two points, and the structure of the
hidden reflection resolves both.

\begin{lemma}[Reduction to $\mathrm{SU}(n)$ and $\mathrm{SO}(n)$]
\label{lem:su-so}
Let $D$ be a fixed diagonal $\pm1$ matrix with $\tr D=0$.  Then $UDU^\dagger$ has the same law when
$U$ is Haar on $\Un$ as when $U$ is Haar on $\mathrm{SU}(n)$.  Likewise $QDQ^\top$ has the same law
when $Q$ is Haar on $\On$ as when $Q$ is Haar on $\mathrm{SO}(n)$.
\end{lemma}

\begin{proof}
Consider first the unitary case.  Every $U\in\Un$ is $e^{i\theta}V$ with $V\in\mathrm{SU}(n)$, and the
central phase cancels in the conjugation, so the law of $UDU^\dagger$ is a conjugation-invariant
probability measure on the orbit $\{VDV^\dagger:V\in\mathrm{SU}(n)\}$.  Since $\mathrm{SU}(n)$ acts
transitively on that orbit, such a measure is unique, and both Haar measures induce it.

Consider next the orthogonal case.  Here $\mathrm O(n)$ is the disjoint union of $\mathrm{SO}(n)$ and
$\mathrm{SO}(n)J$, where $J=\mathrm{diag}(-1,1,\dots,1)$; as $D$ is diagonal it commutes with $J$, so
$(Q'J)D(Q'J)^\top=Q'DQ'^\top$.  The two components induce the same law, as claimed.
\end{proof}

We need the reduction for two reasons.  First, $\Un$ has a one-dimensional center, along which the
Ricci curvature vanishes, so that \Cref{prop:herbst} does not apply to $\Un$ directly.  Second, and more
awkwardly, $\On$ is disconnected, so it is not a connected Riemannian manifold.  \Cref{lem:su-so} transfers
both questions to $\mathrm{SU}(n)$ and $\mathrm{SO}(n)$, where the curvature hypothesis holds.  For the \emph{scalar} statements we need no such reduction, since \eqref{eq:scalar-conc} covers all
four families uniformly.

The second point concerns the choice of frame.  Its index is written $a$ throughout this section and
the next, where it is the only meaning that letter carries; the advice length of
\Cref{model:hybrid}, also written $a$, appears in neither.  On $\mathrm{SU}(n)$ an orthonormal frame
is $\partial_aF(U)=\frac{d}{dt}F(Ue^{itK_a})|_{t=0}$ for $\{K_a\}_{a=1}^{n^2-1}$ an orthonormal basis
of \emph{traceless} Hermitian matrices, and on $\mathrm{SO}(n)$ it is
$\partial_aF(Q)=\frac{d}{dt}F(Qe^{tX_a})|_{t=0}$ for $\{X_a\}$ an orthonormal basis of real
antisymmetric matrices.  In both cases the only property we use is the derivative bound
\eqref{eq:carre} below, and its proof is insensitive to which of the two frames is used. On
$\mathrm{SU}(n)$ the direction omitted relative to a full basis of $\Herm_n$ is $K=\Id/\sqrt n$,
which contributes nothing since $[\Id,D]=0$, and on $\mathrm{SO}(n)$ Parseval over the antisymmetric
subspace only decreases the sum.  So we write the argument once, on $\mathrm{SU}(n)$.

\subsection{A variance bound for linear statistics of a Haar-conjugated observable}
\label{sec:khintchine}

\begin{lemma}[Matrix concentration for linear statistics of $R$]
\label{lem:khintchine}
Let $R=UDU^\dagger$ with $U$ Haar on $\mathrm{SU}(n)$, $D=D^\dagger$ diagonal, $\tr D=0$, $\opnorm D\le1$. Let $B_{xy}\in\bbC^{n\times n}$, $x,y\in[M]$, satisfy
\begin{align}
  \sup_{\norm\alpha_2=1}\ \sum_x\Big\|\sum_y\alpha_yB_{xy}\Big\|_2^2\ &\le\ b^2, \tag{Row}\label{eq:row}\\
  \sup_{\norm\beta_2=1}\ \sum_y\Big\|\sum_x\beta_xB_{xy}\Big\|_2^2\ &\le\ b^2. \tag{Col}\label{eq:col}
\end{align}
Define the $M\times M$ random matrix $Y_{xy}=\tr(B_{xy}R/n)$. Then $\E\opnorm Y\le C\,b\,n^{-3/2}\sqrt{\log(2M)}$ for a universal $C$.
\end{lemma}

The two hypotheses are the analogue of the variance parameter $\sigma^2$ in
\eqref{eq:classical-khintchine}: they say that the coefficient matrices, viewed as a map on the
index set, do not amplify a unit vector by more than $b$ in Hilbert--Schmidt norm, in either of the
two directions.  The conclusion has the same shape as \eqref{eq:classical-khintchine}, with
$b\,n^{-3/2}$ in the role of $\sigma$; the factor $n^{-3/2}$ is the one already visible in
\Cref{lem:scalar-linear}, and the content of the lemma is that passing from a single linear
statistic to the operator norm of a whole matrix of them costs only $\sqrt{\log(2M)}$.

We should be precise about the attribution here.  The concentration is entirely
Huang--Tropp's: \Cref{prop:herbst} is a matrix inequality that already dispenses with independence,
asking curvature of it instead, and~\cite[Ex.~3.11 of the arXiv version]{HT21} already instantiates it for a Haar-conjugated
fixed symmetric matrix on $\mathrm{SO}(d)$, computing the variance proxy through the carr\'e du
champ.  The Hermitian dilation we use to reach a non-Hermitian $Y$ is also due to them.

What \Cref{lem:khintchine} adds is the variance proxy in the shape the relaxation of
\Cref{sec:onequery} consumes.  Their example is a $d\times d$ function of $d\times d$ data; ours is an
$M\times M$ matrix, with $M$ as large as $2^{\poly(\lambda)}$, whose \emph{entries} are linear
statistics of one $n\times n$ conjugation.  Bounding $\opnorm{\sum_a(\partial_a\mathcal Y)^2}$ for
that object is what the row and column hypotheses \eqref{eq:row} and \eqref{eq:col} are for, and it
is the step that lets a maximum over $2^M$ truth tables be paid for at rate $\sqrt{\log(2M)}$.

Two features of the statement matter later, and the half-subspace ensemble plays no part in
either.  It does not need $D$ to be a reflection: the hypotheses ask only that $D$ be Hermitian and
traceless with $\opnorm D\le1$, so the bound applies verbatim to a Haar-conjugated projector.  It also
transfers to $\mathrm{SO}(n)$, by the reduction recorded at the start of \Cref{sec:coherent}.

\begin{proof}
Everything here rests on one elementary fact: the map $K\mapsto i[K,D]$ has Hilbert--Schmidt
operator norm at most $2\opnorm D$.  The rest of the argument carries that fact through the
dilation and through the row and column sums.

To begin, we put $Y$ into the form \Cref{prop:herbst} accepts.  That proposition applies to
Hermitian-matrix-valued functions of mean zero, and $Y$ is neither, so we pass to its Hermitian
dilation
\[
  \mathcal Y\ \eqdef\ \begin{pmatrix}0&Y\\ Y^\dagger&0\end{pmatrix},
\]
which satisfies $\opnorm{\mathcal Y}=\opnorm Y$ and lives in dimension $2M$.  Its mean vanishes,
because $\E R=0$ as in \Cref{lem:scalar-linear}, and hence $\E\mathcal Y=0$.

Now we differentiate along the frame.  Let $\{K_a\}_{a=1}^{n^2-1}$ be the orthonormal basis of
traceless Hermitian matrices fixed in \Cref{sec:concentration}.  Differentiating
$R=Ue^{itK_a}De^{-itK_a}U^\dagger$ at $t=0$ gives $\partial_aR=U\,i[K_a,D]\,U^\dagger$, so
$(\partial_aY)_{xy}=\frac1n\tr(B_{xy}U\,i[K_a,D]\,U^\dagger)$.  We will show that for any
$A\in\bbC^{n\times n}$,
\begin{equation}
  \sum_a\big|\tr(A\,i[K_a,D])\big|^2\ \le\ 4\opnorm D^2\norm A_2^2\ \le\ 4\norm A_2^2 ,
  \label{eq:carre}
\end{equation}
which is the elementary fact above, in the form the frame delivers it.  To see \eqref{eq:carre},
note first that $\tr(A\,i[K_a,D])=\tr(K_a\,i[D,A])$ by cyclicity.  Write $i[D,A]=S+iT$ with $S$ and
$T$ Hermitian, and complete $\{K_a\}$ to an orthonormal basis of $\Herm_n$ by adjoining
$\Id/\sqrt n$; that extra direction contributes nothing, since $[\Id,D]=0$.  Parseval in this basis
then gives
\[
  \sum_a|\tr(K_a\,i[D,A])|^2=\norm S_2^2+\norm T_2^2=\norm{[D,A]}_2^2\le(2\opnorm D\norm A_2)^2 ,
\]
which is \eqref{eq:carre}.

That leaves the variance parameter.  Bounding $\opnorm{\sum_a(\partial_a\mathcal Y)^2}$
means bounding $\langle z,\sum_a(\partial_a\mathcal Y)^2z\rangle$ over unit vectors $z$ of the
doubled space, and we write such a $z$ as a pair $(u,v)$ with $\norm u_2^2+\norm v_2^2=1$.  By the
block form of the dilation,
\[
  \langle z,(\partial_a\mathcal Y)^2z\rangle=\norm{(\partial_aY)v}_2^2+\norm{(\partial_aY)^\dagger u}_2^2 ,
\]
so it suffices to bound $\sum_a\norm{(\partial_aY)v}_2^2$ by $4b^2n^{-2}\norm v_2^2$.  The
calculation for the second summand is identical, with \eqref{eq:col} in place of \eqref{eq:row}, and
we omit it.  So fix $v$ and set $G_x\eqdef\sum_yv_yB_{xy}$.  Then, by cyclicity of the trace,
\[
  \big((\partial_aY)v\big)_x\ =\ \tfrac1n\tr\big(U^\dagger G_xU\,i[K_a,D]\big) .
\]
Applying \eqref{eq:carre} with $A=U^\dagger G_xU$, and noting $\norm{U^\dagger G_xU}_2=\norm{G_x}_2$,
\[
  \sum_a\norm{(\partial_aY)v}_2^2\le\frac4{n^2}\sum_x\norm{G_x}_2^2=\frac4{n^2}\sum_x\Big\|\sum_yv_yB_{xy}\Big\|_2^2\le\frac{4b^2}{n^2}\norm v_2^2
\]
by \eqref{eq:row}.  Hence $\sum_a(\partial_a\mathcal Y)^2\preceq\frac{4b^2}{n^2}\Id_{2M}$, so
\eqref{eq:herbst-group} with $d=2M$ and $v_{\mathcal Y}\le4b^2/n^2$ gives the claim, the $\log(4M)$ that
appears at $d=2M$ being at most $2\log(2M)$.
\end{proof}

\subsection{The one-query bound}
\label{sec:onequery}

First, the baseline against which the bound should be read.  A single
$[-1,1]$-valued function of the hidden subspace, a single bounded statistic of $H$, already correlates with $R$ at scale $n^{-1/2}$, and no more.

\begin{lemma}[Local correlation with the hidden reflection]
\label{lem:local-corr}
For $H\subseteq\bbC^n$ Haar of dimension $n/2$ and any random variable $s(H)\in[-1,1]$,
$\opnorm{\E_H[s(H)R]}\le(n+1)^{-1/2}$.
\end{lemma}

\begin{proof}
Reduce the operator norm to a scalar second moment. The operator $\E_H[s(H)R]$ is Hermitian, so its
operator norm is attained on unit vectors:
\[
  \opnorm{\E_H[s(H)R]}\ =\ \sup_{\norm v=1}\big|\E_H\big[s(H)\,\bra vR\ket v\big]\big| .
\]
So it is enough to bound the right-hand side at a single fixed $v$.

Fix such a $v$. Since $|s|\le1$, and then by Cauchy--Schwarz,
\[
  \big|\E_H\big[s(H)\,\bra vR\ket v\big]\big|\ \le\ \E_H\big|\bra vR\ket v\big|
  \ \le\ \sqrt{\E_H\big(\bra vR\ket v\big)^2} .
\]
The weight $s$ has now disappeared, and only the fluctuation of a single diagonal entry is left.

Now we compute that second moment. By rotation invariance $X\eqdef\bra vP\ket v$ is distributed
as $\mathrm{Beta}(n/2,n/2)$, so $\bra vR\ket v=2X-1$ has mean $0$ and variance $1/(n+1)$, which gives
the claim.
\end{proof}

We use the one-query normal form fixed in \Cref{sec:prelim-query}: an isometry $W$ with branches
$A_x$, a phase query, and an effect $\Pi$ with blocks $\Pi_{xy}$.  Expanding
$W\Delta W^\dagger=\sum_{x,y}\ket x\!\bra y\otimes A_x\Delta A_y^\dagger$, the two copies of $D_f$
contribute $f_xf_y$, and tracing against $\Pi$ picks out a block, so the bias \eqref{eq:bias-def} is
\begin{equation}
  Z_H(f)=\sum_{x,y}f_xf_y\,C_H(x,y),\qquad C_H(x,y)=\tr\!\big(\Pi_{xy}A_y\Delta A_x^\dagger\big).
  \label{eq:quadform}
\end{equation}
Each object enters in a different way: the adversary contributes the fixed data $(A_x,\Pi_{xy})$, the
randomness of $H$ enters $\Delta$ \emph{linearly}, and the unknown truth table appears only in the
rank-one sign pattern $f_xf_y$.

\begin{theorem}[One-query bound]
\label{thm:one-query}
There is a universal $C>0$ such that for every isometry $W$ and every effect $0\preceq\Pi\preceq\Id$,
with $M=|\cX|$ the width of the oracle,
\[
  \E_H\max_{f\in\{\pm1\}^{\cX}}|Z_H(f)|\ \le\ C\sqrt{\frac{\log(2M)}{n}} .
\]
\end{theorem}

\begin{proof}
\emph{Step 1: the bias is a quadratic form in the sign pattern whose coefficients are linear in $R$.}
This is the identity \eqref{eq:quadform}, established above together with the account of how each
object enters it.  We use nothing else about the algorithm below.

\emph{Step 2: the maximum over the $2^M$ sign patterns collapses to a single operator norm.} The one
general tool for a quadratic form is the operator norm, and applying it directly to \eqref{eq:quadform}
gives $|f^\top C_Hf|\le\norm f_2^2\opnorm{C_H}=M\opnorm{C_H}$.  The factor $M$ is spurious.  It would
be paid only by an adversary able to place unit query weight on every label at once, whereas the
completeness relation in \eqref{eq:normal-form} says that the total weight is $1$.  Our remedy is to change the normalization so that the sign vector becomes a unit vector.

Recall from \eqref{eq:query-mass} the query mass $\tau_x=\frac1n\tr(A_x^\dagger A_x)$, a
probability vector on $\cX$ that depends on the adversary and not on $H$.  Labels with $\tau_x=0$
have $A_x=0$, and we discard them.  Define the \emph{weighted coefficient matrix}
\begin{equation}
  G_H(x,y)\ \eqdef\ \frac{C_H(x,y)}{\sqrt{\tau_x\tau_y}} ,
  \label{eq:weighted}
\end{equation}
and rewrite \eqref{eq:quadform} through the vector $D_\tau^{1/2}f$, where $D_\tau=\mathrm{diag}(\tau)$.
That vector has norm exactly $1$, since
\[
  \norm{D_\tau^{1/2}f}_2^2=\sum_x\tau_xf_x^2=\sum_x\tau_x=1 ,
\]
which holds for \emph{every} sign pattern $f$, because $f_x^2=1$.  Hence for every $H$ and every $f$,
\[
  |Z_H(f)|=\big|(D_\tau^{1/2}f)^\top\,G_H\,(D_\tau^{1/2}f)\big|\le\opnorm{G_H},
  \qquad\text{so}\qquad \max_f|Z_H(f)|\le\opnorm{G_H}.
\]
We have therefore traded a maximum over $2^M$ sign vectors for a single operator norm, of a matrix
whose rows and columns are normalized by the adversary's own query distribution.  Neither the
weighting nor the trade is new: both are due to Lombardi, Ma, and Wright, the weighting read here on
query labels rather than on workspace basis vectors, and \Cref{sec:intro-lmw} makes the comparison
precise.  What Step~3 supplies is the variance computation that the concentration input of
\Cref{prop:herbst} asks for, in the shape this relaxation produces.  Notice that after Step~2 the truth table
has left the problem altogether, and all that remains is the expected operator norm of a single
random matrix.

\emph{Step 3: that operator norm has expectation at most $C\sqrt{\log(2M)/n}$.} By cyclicity the
entries of $G_H$ are linear statistics of $R$:
\[
  G_H(x,y)=\tr\!\Big(B_{xy}\frac Rn\Big),\qquad B_{xy}\eqdef\frac{A_x^\dagger\Pi_{xy}A_y}{\sqrt{\tau_x\tau_y}} .
\]
We verify \eqref{eq:row} with $b=n$; the verification of \eqref{eq:col} is identical, using
$\Pi_{xy}^\dagger=\Pi_{yx}$, and is omitted.  Three facts drive it: the masses normalize each branch,
$\tr(A_y^\dagger A_y)=n\tau_y$; the measurement is a contraction, $0\preceq\Pi\preceq\Id$; and each
branch has small operator norm relative to its mass,
\[
  \opnorm{A_x^\dagger}^2=\opnorm{A_x^\dagger A_x}\le\tr(A_x^\dagger A_x)=n\tau_x ,
\]
since $A_x^\dagger A_x\succeq0$.

Fix $\alpha$ with $\norm\alpha_2=1$ and put $F_y\eqdef\frac{\alpha_y}{\sqrt{\tau_y}}A_y$. Stack these
into the single operator
\[
  F\ \eqdef\ \sum_y\ket y\otimes F_y\ \colon\ \bbC^n\longrightarrow\bbC^{\cX}\otimes\cK ,
  \qquad\text{so that}\qquad \norm F_2^2=\sum_y\norm{F_y}_2^2 ,
\]
the two norms agreeing because the $\ket y$ are orthonormal. Multiplication by $\Pi$ acts on $F$
blockwise, $\Pi F=\sum_x\ket x\otimes\Gamma_x$ with $\Gamma_x\eqdef\sum_y\Pi_{xy}F_y$, and since
$\opnorm\Pi\le1$ it does not increase the Hilbert--Schmidt norm. Hence
\begin{equation}
  \sum_x\norm{\Gamma_x}_2^2=\norm{\Pi F}_2^2\ \le\ \norm F_2^2=\sum_y\frac{|\alpha_y|^2}{\tau_y}\tr(A_y^\dagger A_y)=\sum_y|\alpha_y|^2\,n=n .
  \label{eq:contraction}
\end{equation}
This is the only place where we use the contraction property $\opnorm\Pi\le1$.  For each $x$ we have
$\sum_y\alpha_yB_{xy}=\frac1{\sqrt{\tau_x}}A_x^\dagger \Gamma_x$, so
\[
  \Big\|\sum_y\alpha_yB_{xy}\Big\|_2^2=\frac1{\tau_x}\norm{A_x^\dagger \Gamma_x}_2^2\le\frac{\opnorm{A_x^\dagger}^2}{\tau_x}\norm{\Gamma_x}_2^2\le n\norm{\Gamma_x}_2^2 .
\]
Summing over $x$ and applying \eqref{eq:contraction} gives $\sum_x\|\sum_y\alpha_yB_{xy}\|_2^2\le n^2$, i.e.\ \eqref{eq:row} with $b=n$.  \Cref{lem:khintchine} then gives the bound of the theorem,
\[
  \E_H\opnorm{G_H}\le C\,\frac n{n^{3/2}}\sqrt{\log(2M)}=C\sqrt{\frac{\log(2M)}{n}} . \qedhere
\]
\end{proof}

We make two observations about the bound.  First, $n^{-1/2}$ is the right scale.  A single Boolean
bit of information about $H$ correlates with $R$ at scale at most $n^{-1/2}$
(\Cref{lem:local-corr}), and the theorem says that a coherent query of width $M$ loses only
$\sqrt{\log(2M)}$ over that single-bit baseline.  The quantity the proof bounds is genuinely of the size the theorem reports, although the
relaxation of Step~2 need not be tight: at $\Pi=\Id$, $\cK=\bbC$, $A_x=\bra{e_x}$ we have $\opnorm{G_H}=\max_x|R_{xx}|$, and each diagonal
entry has mean absolute deviation of order $n^{-1/2}$ by \Cref{lem:mad}.  A union bound over the $n$ of
them gives $\E_H\opnorm{G_H}=O(\sqrt{\log n/n})$, while retaining the single term $x=1$ gives
$\E_H\opnorm{G_H}\ge\E_H|R_{11}|=\Omega(n^{-1/2})$, so $\E_H\opnorm{G_H}$ is of that size up to the
logarithm.  At $\Pi=\Id$, however, $Z_H(f)=0$ for every $f$, so the norm can exceed the bias.  The \emph{conclusion} is tight too, and by an explicit attack rather than an extremal
example, as \Cref{cor:one-query-tight} shows.  Second, the adversary enters only through the normalizations in \eqref{eq:row} and \eqref{eq:col};
we use no property of the branch operators beyond completeness.  That is why the bound holds for the
maximum over \emph{all} truth tables, and why it composes over parallel queries at no cost.

\begin{corollary}[\Cref{thm:one-query} is tight up to $\sqrt{\log n}$]
\label{cor:one-query-tight}
Let $n\ge4$ be even.  Then there are a Boolean function $f_H$ of width $M=2n$ and a one-query
adversary, holding a single challenge copy, whose average bias is at least $(4\pi n)^{-1/2}$.  Set
this against the upper bound $C\sqrt{\log(4n)/n}$ that \Cref{thm:one-query} gives at that width.
The statement holds both in the complex ensemble of this section and in the real ensemble of the main
construction.
\end{corollary}

\begin{proof}
Take $f_H(j)=\mathbf 1[(P_H)_{jj}<\frac12]$ for $j\in[n]$, a Boolean function of $H$ and hence one
of the truth tables \Cref{thm:one-query} quantifies over. The adversary measures its challenge in the
computational basis, obtaining $j$, queries $f_H$ at $j$, and outputs the answer.

This is the exact rule analyzed in \Cref{prop:exact-rate}(ii), whose average distinguishing advantage
is $2\,\E_H|\xi_{e_1}|$ and whose bias is therefore $\E_H|\xi_{e_1}|$, computed in whichever ensemble
is in force.  In the real ensemble \Cref{lem:mad} bounds that quantity below by $(2\pi n)^{-1/2}$,
through the left-hand inequality of \eqref{eq:beta-mad}; this is the one place where we use the
hypothesis $n\ge4$.  In the complex ensemble,
\[
  \E_H|\xi_{e_1}|\ \ge\ \sqrt{1-1/n}\,(2\pi n)^{-1/2}\ \ge\ (4\pi n)^{-1/2}
\]
for $n\ge2$. In both ensembles the bias is at least $(4\pi n)^{-1/2}$.

The last thing to check is that the attack costs a single query in the model of \Cref{thm:one-query}. The
adversary makes one \emph{bit-flip} query, which is one phase query on a domain of size $2n$ by the
conversion of \Cref{sec:prelim-query}. With the answer register prepared in $\ket-$, the phase oracle for the truth table
$g_{(j,c)}=(-1)^{c\,f_H(j)}$ on $[n]\times\{0,1\}$ implements that query.  The two Hadamard gates of
the conversion do not depend on $H$, so we absorb them into $W$ and $\Pi$.  The width is $M=2n$, as
the statement says.
\end{proof}

So the $n^{-1/2}$ scale in \Cref{thm:one-query} is exactly right, and the only question is the
$\sqrt{\log(2M)}$ factor, which at $M=\Theta(n)$ is $\sqrt{\log n}$. \Cref{cor:one-query-tight}
also bears on the model. The extremal one-query adversary is elementary; it
measures in the computational basis and asks one classical question, which is why
\Cref{prop:exact-rate} and \Cref{thm:one-query} give the same rate.

\begin{corollary}[Parallel coherent queries]
\label{cor:parallel}
Consider an adversary that prepares a state from its challenge copy by an $H$-independent isometry,
applies $D_f^{\otimes t}$ once for a truth table $f$ of $H$, and then measures.  Every such adversary
has average bias at most $C\sqrt{(t\log(2M)+1)/n}$.  Consequently, for $n=2^{\Omega(\lambda)}$,
polynomially many parallel coherent queries of width $2^{\poly(\lambda)}$ give negligible bias.
\end{corollary}

\begin{proof}
A $t$-fold parallel query is a single query to a wider oracle. Indeed, $D_f^{\otimes t}$ is the phase
oracle on $\bbC^{\cX^t}$ with truth table $f'_{(x_1,\dots,x_t)}=\prod_if_{x_i}$, another Boolean
function of $H$, on a domain of size $M^t$. A $t$-parallel adversary is a one-query
adversary against that oracle, so \Cref{thm:one-query} applies with $M^t$ in place of $M$, and
$\log(2M^t)\le t\log(2M)+1$, which is the stated bound. For polynomially many queries of polynomial
width the quantity under the root has numerator $\poly(\lambda)$, while $n=2^{\Omega(\lambda)}$, so the
bound is $\negl(\lambda)$, which is the second assertion.
\end{proof}

\subsection{Quantum advice}
\label{sec:advice-is-a-query}

This subsection settles the question of \emph{quantum advice} that \Cref{sec:security-thm} raised.
We give two arguments, which serve different purposes. The first is direct and sharp: it charges $m$
qubits of advice at rate $\sqrt m/n$ rather than $\sqrt{m/n}$, and so reaches
$m=n^{2-\Omega(1)}$, which is optimal. The second routes the advice through Rosenthal's one-query
state synthesis~\cite{Rosenthal23}; it is quantitatively weaker, but it is the statement that
\emph{identifies} advice with a query, as \Cref{sec:lift} requires.

\begin{proposition}[Advice is charged at rate $\sqrt m/n$]
\label{prop:advice-sharp}
Let $\cV$ be a Hilbert space of dimension $d_\cA$ and let $0\preceq\Pi\preceq\Id$ be any fixed,
$H$-independent effect on $\cV\otimes\bbC^n$. Then
\[
  \E_H\ \sup_{\ket\sigma\in\cV,\ \norm{\sigma}=1}\ \frac1n\Big|\tr\big[\Pi\,(\proj\sigma\otimes R)\big]\Big|
  \ \le\ C\,\frac{\sqrt{\log(2d_\cA)}}{n}
\]
for the universal constant $C$ of \Cref{lem:khintchine}. Consequently, a distinguisher of unbounded
computational power that receives $m$ qubits of advice depending arbitrarily on $H$, together with one
challenge copy, and makes no oracle queries, has average bias at most $C'\sqrt{m+1}/n$. This is
$\negl(\lambda)$ for every $m=n^{2-\Omega(1)}$.
\end{proposition}

Notice the order of the quantifiers: the supremum is over the entire unit sphere of $\cV$, so the bound
holds for an arbitrary advice family $\{\ket{\sigma_H}\}$ with no measurability or computability
hypothesis, and no maximization is left to exchange with the expectation.

\begin{proof}
Fix an orthonormal basis $\{\ket i\}_{i=1}^K$ of $\cV$ and write
\[
  \Pi_{ji}\ \eqdef\ (\bra j\otimes\Id)\,\Pi\,(\ket i\otimes\Id)
\]
for the blocks of $\Pi$, which are operators on $\bbC^n$; since $\Pi$ is Hermitian,
$\Pi_{ji}^\dagger=\Pi_{ij}$. Writing $\ket\sigma=\sum_ic_i\ket i$ and expanding the trace,
\[
  \tr\big[\Pi(\proj\sigma\otimes R)\big]=\sum_{i,j}c_i\bar c_j\,\tr(\Pi_{ji}R)
  \qquad\text{and so}\qquad
  \tfrac1n\tr\big[\Pi(\proj\sigma\otimes R)\big]=\bra cY\ket c,
\]
where $Y$ is the $K\times K$ matrix with entries $Y_{ji}\eqdef\frac1n\tr(\Pi_{ji}R)$.

Two features of $Y$ are what decide the proposition.  First, it is Hermitian, because
$\overline{\tr(\Pi_{ij}R)}=\tr(R^\dagger\Pi_{ij}^\dagger)=\tr(\Pi_{ji}R)$.  Second, it does not
depend on $\ket\sigma$.  The supremum over unit $\ket\sigma$ is exactly $\opnorm
Y$, and this identity holds \emph{pointwise in $H$}, so that no maximization is left to exchange with
the expectation.  Contrast this with the alternative route, which would fix an advice family
$\{\ket{\sigma_H}\}$ and run a net over such families: that route needs a measurability hypothesis on
the family, and it pays for the net.  Here neither is required, and it remains only to bound
$\E_H\opnorm Y$.

Apply \Cref{lem:khintchine} with $B_{ji}=\Pi_{ji}$ on an index set of size $K$. Both of its hypotheses
hold with $b=\sqrt n$. For \eqref{eq:row}, fix $\alpha$ with $\norm\alpha_2=1$ and notice that
$\sum_i\alpha_i\Pi_{ji}=(\bra j\otimes\Id)\Pi(\ket\alpha\otimes\Id)$. Hilbert--Schmidt norms add over
blocks, so
\[
  \sum_j\Big\|\sum_i\alpha_i\Pi_{ji}\Big\|_2^2\ =\ \norm{\Pi(\ket\alpha\otimes\Id)}_2^2
  \ =\ \sum_{l=1}^n\norm{\Pi(\ket\alpha\otimes\ket{e_l})}^2
  \ \le\ \sum_{l=1}^n\norm{\ket\alpha\otimes\ket{e_l}}^2\ =\ n ,
\]
with equality at $\Pi=\Id$. For \eqref{eq:col}, put $\ket w=\sum_j\overline{\beta_j}\ket j$, so that
$\sum_j\beta_j\Pi_{ji}=(\bra w\otimes\Id)\Pi(\ket i\otimes\Id)$, and then
\[
  \sum_i\big\|(\bra w\otimes\Id)\Pi(\ket i\otimes\Id)\big\|_2^2\ =\ \norm{(\bra w\otimes\Id)\Pi}_2^2
  \ =\ \tr\big[(\bra w\otimes\Id)\Pi^2(\ket w\otimes\Id)\big]\ \le\ \tr\Id_n\ =\ n ,
\]
using $\Pi^2\preceq\Pi\preceq\Id$. With $b=\sqrt n$ on both sides, \Cref{lem:khintchine} gives
$\E_H\opnorm Y\le C\sqrt n\cdot n^{-3/2}\sqrt{\log(2d_\cA)}$, which is the bound displayed in the
statement.

Now we charge $m$ qubits of advice. A mixed $m$-qubit advice state has a purification on $2m$
qubits, and handing the receiver the purification only increases its power, so we may take $\cV$ of
dimension $d_\cA\le2^{2m}$ and $\sqrt{\log(2d_\cA)}=O(\sqrt{m+1})$. Finally $\sqrt m/n=\negl(\lambda)$ whenever
$m=n^{2-\Omega(1)}$, since $n=2^{\Omega(\lambda)}$.
\end{proof}

\paragraph{Tightness of the rate, and comparison with a coherent query.}  Both ends of
\Cref{prop:advice-sharp} are attained.  At $m=0$ the bias of a fixed effect is $O(n^{-1})$ by
\Cref{lem:scalar-linear}, which matches the bound.  At the other end, $m=\Theta(n^2)$, the advice can
carry the index of the nearest point of a constant-accuracy net of the Grassmannian.  Such a net has
$\exp(O(n^2))$ points, since the Grassmannian has real dimension $\Theta(n^2)$, and the effect
$\Pi=\sum_p\proj p\otimes P_{H_p}$ then achieves bias at least $\frac12-\delta$ for a net of constant
accuracy $\delta$.  So $m=O(n^2)$ qubits suffice for constant bias, and \Cref{prop:advice-sharp}
shows that this order is also necessary.

Comparing with \Cref{thm:one-query} explains the difference between the two rates.  A coherent query
of width $M$ is charged $\sqrt{\log(2M)/n}$, which is a factor $\sqrt n$ worse per bit of index than
advice.  The whole difference sits in the row constant.  \Cref{thm:one-query} must pay
$\opnorm{A_x^\dagger}^2\le n\tau_x$, because a query branch may concentrate the entire challenge on a
single label, whereas an orthonormal advice basis admits no such concentration and pays only
$\opnorm\Pi\le1$.  In this sense a coherent query of width $M$ is comparable to $\Theta(n\log(2M))$
qubits of advice: the two upper bounds coincide there, with a ratio of $\sqrt n$ between the row
constants.  That is why \Cref{cor:quantum-advice} below, which is obtained by turning advice into a
query, is the weaker of the two statements.

The identification of advice with a query is nevertheless needed in \Cref{sec:lift}, so we state it.

\begin{corollary}[Security against quantum advice, via one-query synthesis]
\label{cor:quantum-advice}
Let $\{\sigma_H\}$ be an arbitrary family of $\poly(\lambda)$-qubit states depending arbitrarily on $H$. Every distinguisher (\emph{of unbounded computational power}) that receives $\sigma_H$ together with one challenge copy, and makes no oracle queries, has average bias $\negl(\lambda)$.

Equivalently: $\VSP_n$ resists $\polylog(n)$-qubit one-way quantum advice from the holder of $H$ to \emph{the holder of one copy of the challenge state}, even against an unbounded receiver. \Cref{sec:advice-is-a-query} explains why that is not the same as the receiver of \Cref{def:vsp}.
\end{corollary}

\begin{proof}
The idea is to trade the advice for a query.  Rosenthal's one-query state synthesis turns the advice
state into one query to a Boolean function of $H$, after which the entire advice adversary sits in
the normal form of \Cref{sec:prelim-query}.  \Cref{thm:one-query} then applies.  It applies because
that theorem maximizes over \emph{all} truth tables and constrains the adversary only to be an
isometry followed by an effect, so we need know nothing about how the truth table was
produced.  The synthesis is exact only up to an $\exp(-\poly(\lambda))$ error, and we charge that at
the end.

We may take $\sigma_H$ pure: purifying can only increase the receiver's power, and a purification of
a $\poly(\lambda)$-qubit state lives on $\poly(\lambda)$ qubits. Write $\ket{\psi_H}$ for it, on
$m=\poly(\lambda)$ qubits.

\emph{Step 1: the advice state is prepared by a fixed circuit making one query to a Boolean function
of $H$.} Pad the advice register with $\ket0$ qubits so that $m\ge\lambda$, which changes neither the
state nor the receiver's power, and fix the target error $\eps(m)=\exp(-m)\le\exp(-\lambda)$.
By~\cite[Thm.~4.1]{Rosenthal23} there is a \emph{uniform} sequence of $\poly(m)$-qubit circuits
$(C_m)$, each making \emph{one} query to a classical oracle, with the following property: for every
$m$-qubit state $\ket\psi$ there is an oracle $f_{\ket\psi}$ such that the reduced state on the first
$m$ qubits of $C_m^{f_{\ket\psi}}\ket{0\cdots0}$ is within $\eps(m)$ trace distance of $\psi$.

The query register here holds $\poly(m)$ qubits, so the truth table queried has $2^{\poly(m)}$
entries, and it can carry far more information than the $m$ qubits being prepared.  The
content of the theorem is that a single fixed circuit decodes it.  We use four features of that
statement, all of them explicit in~\cite{Rosenthal23}.
\begin{itemize}[leftmargin=1.6em,itemsep=2pt,topsep=3pt]
\item The circuit $C_m$ depends on $\eps$, which we have fixed in advance, but \emph{not} on
  $\ket\psi$; only the oracle does.
\item The oracle may be taken to have a \emph{single output bit}, queried as the phase oracle
  $\sum_x(-1)^{f(x)}\proj x$, which is precisely our $D_f$.
\item The guarantee is on the \emph{reduced} state, so no condition is needed on the other
  registers.
\item The query register holds $\poly(m)$ qubits, so the domain has size
  $M=2^{\poly(m)}=2^{\poly(\lambda)}$.
\end{itemize}
Apply this with $\ket\psi=\ket{\psi_H}$ and set $f_H\eqdef f_{\ket{\psi_H}}$, a Boolean function of
$H$, arbitrary, which is all \Cref{thm:one-query} ever asks of it.

After Step~1 the whole dependence on $H$ has been moved into the truth table $f_H$, and the circuit
that prepares the advice is fixed. We still have to assemble that circuit and the distinguisher into a
single adversary of the required shape.

\emph{Step 2: the advice distinguisher becomes a one-query adversary in the normal form of
\Cref{sec:prelim-query}.} Let $\cB$ be the advice distinguisher, whose two-outcome measurement on the
advice register together with the challenge may be arbitrary and even unbounded.  Assemble it with
$C_m$ as follows.  The pre-query isometry is
\[
  W\ \eqdef\ \big(\text{pre-query part of }C_m\big)\otimes\Id_{\text{challenge}} ,
\]
acting on fresh ancillas and leaving the challenge untouched.  This is a legitimate isometry
$\bbC^n\to\bbC^{\cX}\otimes\cK$, with $\cK$ holding the challenge together with Rosenthal's ancillas,
and its completeness relation is immediate, the challenge register being carried along unchanged.
The single query is to $D_{f_H}$, and the final effect $\Pi$ is the post-query part of $C_m$ followed
by $\cB$'s own measurement.

\emph{Step 3: the bias is negligible, and the synthesis error costs only $\exp(-\poly(\lambda))$.} The
adversary just built has truth table $f_H$ and query width $M$, so by \Cref{thm:one-query} its bias is at most
\[
  C\sqrt{\frac{\log(2M)}{n}}\;=\;C\sqrt{\frac{\poly(\lambda)}{n}}\;=\;\negl(\lambda),
\]
since $n=2^{\Omega(\lambda)}$.  By Step~1 the state it hands $\cB$ is within $\exp(-\poly(\lambda))$ of
$\sigma_H$ in trace distance, so $\cB$'s own bias differs from that display by at most
$2\exp(-\poly(\lambda))$, and is negligible as well.

The $\VSP$ restatement is the same statement read through \Cref{def:vsp}.  A one-way $m$-qubit message
from the holder of $H$ is such an advice state, and $m=\poly(\lambda)=\polylog(n)$, which is the
claim.
\end{proof}

The obvious protocol is consistent with the corollary.  The holder of $H$ sending $t$ copies of
$\rho_0=2P_H/n$ uses $t\log n$ qubits, and a swap test of one of them against the challenge has
signal $1/n$, since $\tr\rho_0^2=2/n$ and $\tr\rho_0\rho_1=0$.  The receiver holds only one challenge
copy, so it cannot run $t$ such tests independently, and the $t$ copies together are worth
$\widetilde\Theta(\sqrt t)/n$ rather than $t/n$ (\Cref{prop:copies-sqrt}); either way the bias is
negligible for $t=\poly(\lambda)$.

\paragraph{Advice against a challenge copy.}
\Cref{prop:advice-sharp} and \Cref{cor:quantum-advice} bound a receiver holding \emph{one quantum
copy} of $\rho_b$.  In \Cref{def:vsp}, by contrast, Alice holds a classical description of $u$, from
which she may prepare as many copies of $\ket u$ as she likes and compute with $u$ directly.  That is
the one-directionality noted in \Cref{sec:framework}, and here it is strict: the two models
separate, so neither statement bears on the classical-input model of \Cref{def:vsp}.

To see the separation, let Bob send a description of one unit vector $w\in H$ accurate to $1/n$, at a
cost of $m=O(n\log n)$ bits, and let Alice threshold $|\ip u{\tilde w}|$ at $n^{-3/4}$.  Suppose
first that $b=1$.  Then $u\in H^\perp$, so $\ip uw=0$ exactly and $|\ip u{\tilde w}|\le1/n$.  Suppose
instead that $b=0$.  Then $u$ is uniform on the sphere of $H$, so $\ip uw$ has the law of a
coordinate of a uniform vector on $S^{n/2-1}$, of standard deviation $\sqrt{2/n}$; hence
$|\ip uw|\ge2n^{-3/4}$ except with probability $O(n^{-1/4})$, and in that case
$|\ip u{\tilde w}|\ge2n^{-3/4}-1/n>n^{-3/4}$.  Alice succeeds with probability $1-o(1)$.
At the same $m$, meanwhile, \Cref{prop:advice-sharp} caps a receiver holding one copy of $\rho_b$ at
$C\sqrt{n\log n}/n=o(1)$.  So at $m=O(n\log n)$ the classical-input receiver is stronger by a factor
$\Omega(\sqrt{n/\log n})$.

The reverse direction is in fact settled, and at a much larger scale than a net argument gives.

\begin{proposition}[Reverse one-way communication for $\VSP$ is linear]
\label{prop:reverse-oneway}
Let $n$ be even.  Suppose Bob, holding a Haar-random half-dimensional $H\subseteq\bbR^n$, sends
Alice an $m$-qubit message, after which Alice, holding $u$ and computationally unbounded, decides
$\VSP_n$ with advantage $\frac13$ on the distribution of \Cref{def:vsp}.  Then
$m\ge\frac n2\big(1-h_2(\tfrac16)\big)=\Omega(n)$, where $h_2$ is the binary entropy.  Conversely
$O(n)$ classical bits suffice for constant advantage, so the complexity is $\Theta(n)$.
\end{proposition}

\begin{proof}
\emph{Worst case first.}  Advantage $\frac13$ means success $\frac56$ on average over
\Cref{def:vsp}.  Sample a Haar $Q\in\On$ and a uniform bit $c$ from public randomness, run the
protocol on $(Qu,QH)$ when $c=0$ and on $(Qu,QH^\perp)$ when $c=1$, complementing the output in the
second case.  Bob can form $QH$ or $QH^\perp$ from $H$ and $Q$, and Alice can form $Qu$, so this is
again a one-way protocol from Bob to Alice of the same length.  Fix any promised pair $(u,H)$ with
$u\in H$.  When $c=0$ the pair $(Qu,QH)$ has the law of \Cref{def:vsp} conditioned on $b=0$, since
$QH$ is Haar and, given it, $Qu$ is uniform on its unit sphere; when $c=1$ the pair
$(Qu,QH^\perp)$ has that law conditioned on $b=1$.  Averaging the two coins, the transformed
protocol succeeds with probability exactly $\frac56$ on \emph{every} promised input, and
symmetrically for $u\in H^\perp$.  This is the symmetrization of \Cref{lem:vsp-small}, read in the
present model.

\emph{Then embed the index problem.}  Put $d=n/2$ and, for $x\in\{0,1\}^d$, let
$H_x=\mathrm{span}\{e_{2j-1+x_j}:j\in[d]\}$, a subspace of dimension $d$.  For $i\in[d]$ take
$u_i=e_{2i-1}$.  Then $u_i\in H_x$ if $x_i=0$, and $u_i\in H_x^\perp$ if $x_i=1$, so each
$(u_i,H_x)$ is a promised input.  Bob's message on $H_x$ therefore lets Alice recover any single bit
$x_i$ of her choosing with probability at least $\frac56$: it is a quantum random access encoding of
$d$ bits into $m$ qubits.

\emph{Then count.}  Let $X$ be uniform on $\{0,1\}^d$, let $C$ be the public randomness, independent
of $X$, and let $M$ be Bob's message register.  Holevo gives $I(X:M\mid C)\le m$.  Since the $X_i$
are independent, the chain rule gives $I(X:M\mid C)\ge\sum_iI(X_i:M\mid C)$, and Fano bounds each
term below by $1-h_2(\frac16)$.  Hence $m\ge d\,(1-h_2(\frac16))$, which is the claim; the bound
allows mixed encodings and arbitrary decoding measurements~\cite{Nayak99}.

\emph{The upper bound.}  Bob rounds a unit vector $w\in H$ to a $\delta$-net of the sphere at a
fixed constant accuracy and sends its index, which is $O(n)$ bits since such a net has
$e^{O(n)}$ points.  Alice accepts when $|\ip u{\tilde w}|\ge a/\sqrt n$.  If $u\in H^\perp$ then
$\ip uw=0$, so $\E|\ip u{\tilde w}|^2\le2\delta^2/n$ and the false-positive probability is at most
$2\delta^2/a^2$ by Chebyshev.  If $u\in H$ then $\norm{P_H\tilde w}\ge1-\delta$, and the
small-ball estimate for a uniform vector on the sphere of $H$ bounds the false-negative probability
by $O(a/(1-\delta))$.  Choosing $\delta\ll a\ll1$ makes both errors small.
\end{proof}

So the reverse one-way question is not an obstacle.  What our advice statements need is the
one-copy model, and \Cref{prop:advice-sharp} settles that at $m=\Theta(n^2)$; the two thresholds
differ by the factor $n$ that the separation above exhibits.

One resource remains to be measured, and the security proof does not need it, since
\Cref{lem:reference} charges the adversary the classical \emph{spans} of its copies rather than the
copies themselves.  The most natural advice about $H$ is a supply of fresh samples of the $b=0$
state, and those accumulate as a random walk: each copy couples to the challenge only through a swap,
and each swap shifts the mean by $2/n$ against a fluctuation of unit size.  So $q$ of them are worth
$\widetilde\Theta(\sqrt q)/n$ rather than $q/n$, both as an attack and as an upper bound
(\Cref{prop:copies-sqrt}, stated and proved in \Cref{app:def-copies}).

Two things follow.  A \emph{classical} description of
the span of the same $q$ vectors achieves $2q/n$ for $q\le n/2$, so classical knowledge of the
reference states is quadratically stronger than the states themselves; that measures how much the
passage to spans in \Cref{lem:reference} concedes.  At $m=q\lceil\log_2n\rceil$ advice qubits the walk
attains the $\sqrt m$ rate of \Cref{prop:advice-sharp} up to logarithms, by the most elementary
advice available, so that rate is tight in the interior of its range as well as at its endpoints.

\subsection{Classical queries followed by one coherent query}
\label{sec:hybrid}

\Cref{thm:EFI-security} and \Cref{thm:one-query} bound two incomparable access modes by different methods. We combine them into a single statement containing both: security against an adversary that makes its full adaptive-classical budget of $\Count$-queries and then, as its final oracle interaction, one coherent query. This also closes a gap in the scope of \Cref{thm:one-query}, which assumed no reference copies while the EFI game always provides them.

The composition works because the relaxed one-query norm of \Cref{thm:one-query} is a Lipschitz
function of the Haar unitary defining $R$, and hence concentrated. The classical transcript
conditions on that same unitary, and concentration is exactly what survives such conditioning. The
feature is not peculiar to our ensemble, since Lombardi, Ma, and Wright~\cite{LMW24} prove an analogous
concentration for theirs; what the composition requires is the pairing of the concentration with the
transcript.

\begin{model}[Hybrid adversary with a final coherent query]
\label{model:hybrid}
This is \Cref{model:access} in the source model, with every parameter named.  The adversary holds a
challenge $\rho_b$ with $b$ uniform, and it may, in this order:
\begin{enumerate}[label=(\roman*),leftmargin=2.2em,itemsep=2pt,topsep=3pt]
\item[(o)] receive a classical advice string $\alpha$ of length at most $a$, depending arbitrarily
  on $H$;
\item call $\cR$ for up to $q$ reference copies per side;
\item make up to $L$ bits of adaptive \emph{classical} $\Count$-queries, interleaved with arbitrary
  quantum processing; and then
\item apply one coherent phase query $O_H$ of width $M$, or equivalently one parallel batch
  $O_H^{\otimes t}$, in which case $M$ is the width of the batch, and measure.
\end{enumerate}
Its computational power is otherwise unbounded.  All parameters are $\poly(\lambda)$, the width
entering only through $\log M$.
\end{model}

\begin{theorem}[Hybrid security against classical queries, advice, and one coherent query]
\label{thm:hybrid}
There is a universal $C>0$ such that for all $q<n/4$, all $L,a\ge0$ and all $M\ge1$, every
hybrid adversary of \Cref{model:hybrid} satisfies
\[
  \E_H\Big[\max_{\alpha\in\{0,1\}^{\le a}}\Adv\Big]\ \le\ \frac{6q}{n}
  \;+\;2\beta\big(n-2q,\,L+a+1\big)
  \;+\;C\sqrt{\frac{\log(2M)+L+a}{n-2q}} ,
\]
with $\beta$ as in \Cref{lem:vsp-small}; by that lemma the middle term is at most
$C\sqrt{(L+a+2)/(n-2q)^{1/3}}$. For $n=2^{\Omega(\lambda)}$ and all parameters $\poly(\lambda)$ the
bound is $2^{-\Omega(\lambda)}$. Setting $M=1$ recovers \Cref{cor:reference-vsp} up to the enumeration term and the constant in the copy term, the generic split paying $6q/n$ where a direct argument pays $2q/n$. At $L=a=q=0$ the third term is the bound of \Cref{thm:one-query}, up to the factor two between bias
and advantage, while the middle term $2\beta(n,1)$ is what the two-part split costs for a single sign
bit.  \Cref{thm:one-query} is the sharp statement in that corner, and it is not superseded here.
\end{theorem}

The obvious move is not available.  One would like to fix
the classical transcript and then apply \Cref{thm:one-query} to whatever follows it.  The
instrument the classical phase realizes depends on $H$ --- the answers it receives are $\Count$'s, and
those are functions of the secret --- so conditioning on a transcript reweights the law of $H$ itself,
while \Cref{thm:one-query} is an average over the unconditioned law.  Nothing forbids that
reweighting from concentrating on exactly the subspaces where the coherent query does well.

\Cref{lem:record} is what removes the difficulty, and it does so by relocating the $H$-dependence
rather than by controlling it.  The operators are made $H$-independent by hard-wiring the answers
into them, and the question of whether those answers are the \emph{true} ones is pushed into a scalar
indicator.  What is then left is a convex combination, because the masses of the consistent records
sum to one at every $H$ separately, and a convex combination of bounds is a bound.

Both lemmas are proved first.  A record here is a classical query-and-answer string, unrelated to
the recorded databases of the compressed-oracle technique.  Throughout, a \emph{record} is the
string
$\pi=(x_1,a_1,\dots)$ of query strings and answer bits produced by the classical phase; it has
length at most $L$, so there are at most $2^{L+1}$ records.

\begin{lemma}[Record instrument]
\label{lem:record}
Fix the classical phase of the adversary, and defer all measurements except the query read-outs. There is, for each record $\pi$, a fixed operator $V_\pi$ on the workspace (independent of $H$ and of the challenge bit), and an indicator $\chi_\pi(H)=\prod_i\mathbf 1[\Count(x_i,H)=a_i]$ such that the sub-normalized post-classical state on record $\pi$ is $\chi_\pi(H)V_\pi\rho V_\pi^\dagger$. Writing $w_\pi=\frac1n\tr(V_\pi^\dagger V_\pi)$,
\[
  \sum_\pi\chi_\pi(H)V_\pi^\dagger V_\pi=\Id\ \ \text{for every }H,\qquad\text{hence}\qquad \sum_\pi\chi_\pi(H)w_\pi=1\ \ \text{for every }H .
\]
\end{lemma}

\begin{proof}
The construction hard-wires the answers into the operator and puts the truth into the indicator, so
that the operator carries no $H$-dependence and, for the true $H$, the surviving records are exactly
the consistent ones.

Model each classical query as a measurement of the query register, with outcome $x_i$, followed by
the oracle writing the forced answer $a_i$ into the workspace, and then a fixed unitary.  All other
operations are purified: every ancilla is retained in the workspace, and every discard is deferred to
the final effect, which acts as the identity on the discarded registers.

Let $V_\pi$ be the product of these fixed unitaries with the outcome projectors and the
answer-writes for the hard-wired string $(x_i,a_i)$.  It carries no $H$-dependence, precisely because
the answers are hard-wired, and $\chi_\pi(H)$ is what records whether those answers match the truth.

Now fix the true $H$.  Only consistent records occur, so
\[
  \rho\ \longmapsto\ \sum_{\pi\,:\,\chi_\pi(H)=1}V_\pi\rho V_\pi^\dagger
\]
is exactly the trace-preserving instrument realized by the classical phase.  Iterating
$\sum_{x_i}\proj{x_i}=\Id$ over the queries then gives $\sum_\pi\chi_\pi(H)V_\pi^\dagger V_\pi=\Id$,
and taking normalized traces gives the mass identity.
\end{proof}

\begin{lemma}[Sub-normalized one-query norm, with concentration]
\label{lem:subnorm}
Let $V\colon\bbC^n\to\bbC^{\cX}\otimes\cK$ be \emph{any} nonzero linear map and let
$0\preceq\Pi\preceq\Id$; the case $V=0$ is trivial, with $\phi\eqdef0$.  Define $Z_H(f)$ from
$(V,\Pi)$ by \eqref{eq:quadform}, put $w=\frac1n\tr(V^\dagger V)$, and let $\phi(H)\eqdef\opnorm{G_H}$
with $G_H$ the weighted coefficient matrix \eqref{eq:weighted} formed from $(V,\Pi)$.
Then $\max_f|Z_H(f)|\le w\,\phi(H)$.  Moreover $\phi$ is a nonnegative function of the Haar unitary
$U$ defining $R$, it is Lipschitz with constant $2$ along geodesics, and it satisfies
\[
  \E_H\phi\le C\sqrt{\frac{\log(2M)}n},\qquad \Pr\nolimits_H\big[\phi\ge\E_H\phi+t\big]\le2e^{-cnt^2}\ (t\ge0).
\]
\end{lemma}

\begin{proof}
Set $A_x\eqdef(\bra x\otimes\Id)V$ and let $\tau_x$ be the query mass \eqref{eq:query-mass} attached to
the branch $A_x$, so that
\[
  \sum_xA_x^\dagger A_x\ =\ V^\dagger V,\qquad \sum_x\tau_x\ =\ w .
\]
The computation of \Cref{thm:one-query} goes through verbatim and gives
$Z_H(f)=(D_\tau^{1/2}f)^\top G_H(D_\tau^{1/2}f)$, the only change being that
$\norm{D_\tau^{1/2}f}_2^2=w$ rather than $1$; hence $\max_f|Z_H(f)|\le w\opnorm{G_H}$.

The hypotheses \eqref{eq:row} and \eqref{eq:col} still hold with row and column constant $b=n$, since their verification used
only $\opnorm{A_x^\dagger}^2\le\tr(A_x^\dagger A_x)=n\tau_x$, valid for any $A_x$, together with the
contraction $0\preceq\Pi\preceq\Id$ and $\norm\alpha_2=1$. Neither $\sum_x\tau_x=1$ nor
$V^\dagger V\preceq\Id$ enters at any point, which is why the present lemma covers a map whose
$V^\dagger V$ has operator norm as large as $n$. \Cref{lem:khintchine} now gives the stated bound on
$\E_H\phi$.

What is left is the tail, and we prove it in three steps: a derivative bound, a conversion between
the geodesic and the Hilbert--Schmidt metric, and an appeal to Haar concentration.

\emph{The derivative bound.}  Write $\mathcal Y$ for the Hermitian dilation of the proof of
\Cref{lem:khintchine}, so that $\opnorm{\mathcal Y(U)}=\phi$.  There the derivative bound reads
$\sum_a(\partial_a\mathcal Y)^2\preceq\frac{4b^2}{n^2}\Id=4\Id$, and this already makes
$U\mapsto\phi$ Lipschitz with constant $2$ along geodesics.  Indeed, each $\partial_a\mathcal Y$ is
Hermitian, so for $\norm c_2=1$, by Cauchy--Schwarz in the index $a$,
\[
  \Big\|\Big(\sum_a c_a\partial_a\mathcal Y\Big)z\Big\|\le\sum_a|c_a|\,\norm{(\partial_a\mathcal Y)z}\le\sqrt{\big\langle z,\textstyle\sum_a(\partial_a\mathcal Y)^2z\big\rangle}\le2\norm z ,
\]
and the operator norm is in turn $1$-Lipschitz in its argument.

\emph{From geodesic to Hilbert--Schmidt distance.}  Since $\phi$ depends on $U$ only through
$UDU^\dagger$, it is invariant under central phases, and so extends unchanged to $\Un$,
where the derivative bound holds trivially along the additional central direction.  On $\Un$, as on
$\mathrm{SO}(n)$, the geodesic distance is at most $\pi/2$ times the Hilbert--Schmidt distance; this
follows from the eigenphase bound $\sin(\theta/2)\ge\theta/\pi$, valid for $0\le\theta\le\pi$.  Hence
$\phi$ is Lipschitz with constant at most $\pi$ in the Hilbert--Schmidt metric.  Restricting back to
$\mathrm{SU}(n)$ keeps that constant, the chord metric there being the restriction of the one on
$\Un$, and we absorb the constant into $c$.

\emph{Concentration.}  Haar concentration at scale $n^{-1/2}$, which is \eqref{eq:scalar-conc}, now
gives the subgaussian tail, which is the last of the three assertions.
\end{proof}

\begin{proof}[Proof of \Cref{thm:hybrid}]
We write the advantage as two pieces: the value that an $H$-\emph{independent} final query would
produce, plus the increment caused by the query's dependence on $H$.  The first piece is a
classical-transcript quantity, so the communication argument bounds it.  The increment carries the
entire cost of the coherent query, and it is the two lemmas above that bound it.

The increment is handled in four moves.  First, conditioning
on the spans replaces the increment by a half-subspace instance in the orthogonal complement, of
dimension $n-2q$.  Second, the record decomposition of \Cref{lem:record} survives that restriction,
with its masses still summing to $1$.  Third, \Cref{lem:subnorm} then applies record by record.
Fourth, concentration removes the conditioning on the record, at a cost only logarithmic in the
number of records.

One convention holds throughout.  We fix an advice length budget $a$, and every object below depends
implicitly on the advice string $\alpha$.  We keep $\max_\alpha$ outside all the inequalities, which
is legitimate because each step holds for every fixed $\alpha$.

\emph{Step 1: the advantage splits into a term with an $H$-independent effect and an increment.} Let
$F_H$ be the accept effect of step (iii) and $F_\star$ the accept effect obtained with the trivial
oracle, all of whose signs are $+1$.  Both are effects, so $\Gamma(H)\eqdef F_H-F_\star$ satisfies
$\opnorm{\Gamma(H)}\le1$.  With $\mathcal E_H$ the classical-phase channel, the signed advantage
$\Delta_{\mathrm{adv}}(H)$ splits as
\[
  \Delta_{\mathrm{adv}}(H)=\underbrace{\tr[F_\star\,\mathcal E_H(\rho_0-\rho_1)]}_{(\mathrm I)}+\underbrace{\tr[\Gamma(H)\,\mathcal E_H(\rho_0-\rho_1)]}_{(\mathrm{II})} .
\]

\emph{Step 2: the term $(\mathrm I)$ contributes at most $\frac{2q}n+2\beta(n-2q,L+a+1)$ on
average.}  For each fixed $\alpha$ the final query is a fixed, $H$-independent operation.  Hence
$(\mathrm I)$ is the advantage of an adversary with $q$ reference copies per side, an $L$-bit
classical transcript, $a$ bits of advice, and otherwise $H$-independent processing.

That is an adversary of the kind \Cref{lem:reference} bounds, once we charge the advice to the
transcript.  We charge it as in the proof of \Cref{thm:transcript-security}, by having the holder of
$H$ send the maximizing advice; the sign bit needed there is the orientation bit already charged
inside \Cref{lem:reference}.  Applying that lemma with $L+a$ bits of interaction bounds
$\E_H\max_\alpha|(\mathrm I)|$ by $\frac{2q}n+2\beta(n-2q,L+a+1)$.

Nothing so far has charged the coherent query, since $F_\star$ carries no dependence on $H$.  The four
remaining steps bound $(\mathrm{II})$, where that query is paid for.

\emph{Step 3: conditioning on the spans reduces the increment to a half-subspace instance of dimension
$n-2q$, at a cost of $\frac{4q}n$.} Condition on the spans $(A,B)$ of \Cref{lem:residual}, which only increases the adversary's power.
The challenge then splits into two branches, and we treat them differently.

On the known branch we concede everything.  There the channel $\mathcal E_H$ is trace preserving and
$\opnorm{\Gamma(H)}\le1$, so $|(\mathrm{II})|\le\norm{\mathcal E_H(\varrho_0-\varrho_1)}_1\le2$.
Since that branch has weight $2q/n$, it contributes at most $2\cdot\frac{2q}n=\frac{4q}n$ in total.

On the residual branch the challenge is a half-subspace instance inside $K=(A\oplus B)^\perp$, of
dimension $n'\eqdef n-2q$ and with reflection $R'$.  Note that, given $(A,B)$, the truth table of the
final coherent query is a Boolean function of the residual subspace alone, which is all the spectral
bound requires.  From here until the end of Step~6 we work inside $K$, writing $n$ for $n'$ and $R$
for $R'$; only the term $\frac{4q}{n}$ already paid refers to the original $n$.

\emph{Step 4: the record masses still form a probability vector after the restriction to $K$.}
We shall apply \Cref{lem:subnorm} to operators supported on $K$, so it remains to check that the
record decomposition of \Cref{lem:record} survives the restriction.  It does.  Put
$w_\pi\eqdef\frac1{n}\tr\big(P_KV_\pi^\dagger V_\pi P_K\big)$.  Conjugating the identity
$\sum_\pi\chi_\pi(H)V_\pi^\dagger V_\pi=\Id$ by $P_K$, and then taking normalized traces inside $K$,
gives
\[
  \sum_\pi\chi_\pi(H)\,w_\pi\ =\ 1\quad\text{for every }H .
\]
So the masses of the restricted instrument again form a probability vector at each fixed $H$, and that
is the only property of them used below.

\emph{Step 5: record by record, the increment is at most $4\sum_\pi\chi_\pi(H)w_\pi\phi_\pi(H)$.} By
\Cref{lem:record} and $\rho_0-\rho_1=2R/n$,
\[
  (\mathrm{II})=\frac2n\sum_\pi\chi_\pi(H)\,\tr\!\big[\Gamma(H)\,V_\pi RV_\pi^\dagger\big].
\]
Fix a record $\pi$, and let $W_\pi$ be the isometry of step (iii)'s pre-query processing.  Then
$W_\pi V_\pi$ satisfies $(W_\pi V_\pi)^\dagger(W_\pi V_\pi)=V_\pi^\dagger V_\pi\preceq\Id$ and has
normalized trace $w_\pi$, so \Cref{lem:subnorm} applies to it.  Now each of the two effects making up
$\Gamma(H)=F_H-F_\star$ contributes a quantity of the form $\tr[F\,V_\pi RV_\pi^\dagger]=n\,Z_H(f)$,
taken at $f=s(H)$ for $F_H$ and at $f\equiv1$ for $F_\star$; the factor $n$ appears because $Z_H$ is
defined through $\Delta=R/n$.  Hence
\[
  \big|\tr[\Gamma(H)V_\pi RV_\pi^\dagger]\big|\ \le\ 2n\max_f|Z_H(f)|\ \le\ 2n\,w_\pi\,\phi_\pi(H),
\]
and therefore $|(\mathrm{II})|\le4\sum_\pi \chi_\pi(H)\,w_\pi\,\phi_\pi(H)$ pointwise in $H$.

That estimate holds at every $H$, but the record still depends on $H$ through $\chi_\pi$ and each
$\phi_\pi$ is itself random.  We still have to replace the $\phi_\pi$ by their means uniformly over records
and advice strings.

\emph{Step 6: passing to the means costs $C\sqrt{(\log(2M)+L+a)/n}$.}  Let
$m_\pi=\E_H\phi_\pi\le C\sqrt{\log(2M)/n}$, and recall that each $\phi_\pi$ is $2$-Lipschitz and
hence subgaussian with variance proxy $\sigma^2\le C'/n$.  By Step~4 the numbers $w_\pi\chi_\pi(H)$
form a probability vector at each fixed $H$ and each fixed advice string.  The quantity $\sum_\pi
w_\pi\chi_\pi(H)\phi_\pi$ is a convex combination of the $\phi_\pi$, and so is at most
$\max_\pi\phi_\pi$.  Splitting each term as $\phi_\pi\le m_\pi+|\phi_\pi-m_\pi|$ now gives
\[
  \max_\alpha\sum_\pi w_\pi\chi_\pi(H)\phi_\pi\ \le\ \max_{\alpha,\pi} m_\pi+\max_{\alpha,\pi}|\phi_\pi-m_\pi| .
\]
The first term is at most $C\sqrt{\log(2M)/n}$.  The second is a maximum of at most
$2^{a+1}\cdot2^{L+1}$ centered subgaussians of variance proxy $\sigma^2\le C'/n$, so the standard
maximal inequality gives
\[
  \E_H\max_{\alpha,\pi}\big|\phi_\pi-m_\pi\big|
  \ \le\ \sigma\sqrt{2\ln\big(2^{a+L+3}\big)}\ \le\ C\sqrt{\frac{a+L+3}n} .
\]
Notice that dependence among the $\phi_\pi$ is irrelevant here, since that inequality follows from a
union bound on subgaussian tails alone.  Putting the two terms together,
\[
  \E_H\max_\alpha\big|(\mathrm{II})\big|\ \le\ \frac{4q}n+C\sqrt{\frac{\log(2M)+L+a}n} .
\]

Adding the bounds of Steps~2 and~6, and restoring $n-2q$ for $n$ in the residual estimates, gives the
theorem with $\frac{2q}n+\frac{4q}n=\frac{6q}n$.
\end{proof}

\begin{corollary}[A fixed oracle for the hybrid access model]
\label{cor:hybrid-fixed}
With probability one over $\{H_\lambda\}$, the pair $(\rho_{0,\lambda},\rho_{1,\lambda})$ is an EFI
pair relative to $\cO$ against every non-uniform $\QPT$ distinguisher that
\begin{enumerate}[label=(\roman*),leftmargin=2.2em,itemsep=1pt,topsep=3pt]
\item receives $\poly(\lambda)$ bits of classical advice, depending arbitrarily on the fixed oracle;
\item calls $\cR$ polynomially many times;
\item makes $\poly(\lambda)$ adaptive classical $\Count$-queries; and then
\item makes one coherent $\Count$-query of width $2^{\poly(\lambda)}$.
\end{enumerate}
This is the access model of \Cref{thm:separation-source}(ii).
\end{corollary}

\begin{proof}
This is \Cref{thm:EFI-security} with \Cref{thm:hybrid} in place of \Cref{cor:reference-vsp}.
Discretizing the gate set leaves countably many machines $N_i$, now with polynomial bounds
$t_i(\lambda)$ on running time, advice length, sample calls, transcript length, and the logarithm of
the width $M$ of the final coherent query.  For each $i$, \Cref{thm:hybrid} bounds
$\E_{\mathbf H}[\max_\alpha\Adv_{N_i}(\lambda)]$ by $2^{-\Omega(\lambda)}$, with the maximum over advice
strings already inside the expectation, which is what \Cref{lem:fixing} asks for.  That lemma then
fixes a single sequence $\{H_\lambda\}$.

The one thing left to check is that the query the corollary allows is one the theorem covers.  That
query is to some Boolean function of $H$, which is all \Cref{thm:hybrid} assumes of it, so the
coherent form of $\Count$ at any rank and any precision (\Cref{model:bit}) falls under the bound.
\end{proof}

Polynomially many \emph{challenge} copies are covered as well, and for a reason particular to the EFI
game: the model supplies reference copies of \emph{both} states, which is exactly what a hybrid over
the copies needs.

\begin{corollary}[Many challenge copies]
\label{cor:multicopy}
Fix the oracle and let $\mathsf D$ be a distinguisher of \Cref{model:hybrid} that receives $k$ copies
of the challenge instead of one, with all other resources as there.  Then there is a
single-challenge distinguisher $\mathsf D'$ of \Cref{model:hybrid}, with the same query budget,
ordering and advice, and with $k-1$ further reference copies, such that
$\Adv(\mathsf D')=\Adv(\mathsf D)/k$.  Consequently the bounds of \Cref{thm:hybrid} and
\Cref{cor:hybrid-fixed} hold for $k$ challenge copies after multiplication by $k$, and remain
negligible for every $k=\poly(\lambda)$.
\end{corollary}

\begin{proof}
Write $\Omega_j\eqdef\rho_0^{\otimes j}\otimes\rho_1^{\otimes(k-j)}$ and
$p_j\eqdef\Pr[\mathsf D(\Omega_j)=1]$, so that $p_k-p_0$ is the signed advantage of $\mathsf D$.
Let $\mathsf D'$ draw $j$ uniformly from $[k]$, place its own challenge in position $j$, fill
positions $1,\dots,j-1$ with reference copies of $\rho_0$ and positions $j+1,\dots,k$ with reference
copies of $\rho_1$, and run $\mathsf D$.  On challenge $\rho_0$ the input is $\Omega_j$ and on
$\rho_1$ it is $\Omega_{j-1}$, so the signed gap of $\mathsf D'$ is
$\frac1k\sum_{j=1}^k(p_j-p_{j-1})=(p_k-p_0)/k$.  The reduction makes no oracle queries of its own
and inserts the copies before the classical phase, so $\mathsf D'$ lies in \Cref{model:hybrid} with
$q+k-1$ reference copies per side.  Everything here is pointwise in the oracle, so it applies to the
stored pair of \Cref{sec:single} as well as to the ideal one.
\end{proof}

The step that makes this work is the availability of both $\rho_0$ and $\rho_1$ as reference copies.
It is not the generic copy amplification of a pseudorandom state, where only one state is available
and the joint law across copies is different.

The advice in \Cref{model:hybrid} is classical.  Quantum advice of moderate size is covered by the
same theorem, through a device that costs nothing beyond a longer classical string.

\begin{corollary}[Quantum advice in the hybrid]
\label{cor:hybrid-qadvice}
Let the adversary of \Cref{model:hybrid} receive, in addition to its classical advice, a quantum
advice register $\cA$ of dimension $d_\cA$ in a state $\sigma_H$ that depends arbitrarily on $H$.  Then,
with $a'\eqdef a+\lceil2d_\cA\log_29\rceil$ and $C$ the constant of \Cref{thm:hybrid},
\[
  \E_H\Big[\sup_{\sigma_H}\ \max_{\alpha\in\{0,1\}^{\le a}}\Adv\Big]\ \le\
  2\Big[\frac{6q}{n}+2\beta\big(n-2q,\,L+a'+1\big)+C\sqrt{\frac{\log(2M)+L+a'}{n-2q}}\Big].
\]
Suppose now that $n=2^{\lambda}$, that $q,L,a$ and $\log M$ are all polynomial in $\lambda$, and
that $d_\cA\le n^{1/3-\eps}$ for a fixed $\eps>0$.  Then the bound is $2^{-\Omega(\lambda)}$.  In other
words, quantum advice of up to $(\tfrac13-\eps)\lambda$ qubits is admissible.
\end{corollary}

\begin{proof}
The device is to spend classical bits on a net of advice states, so that the quantum advice becomes
part of the adversary's own workspace and \Cref{thm:hybrid} applies unchanged.

Fix the adversary and a classical advice string $\alpha$.  Everything it does after receiving its
advice is a quantum instrument applied to
$\sigma\otimes\rho_b\otimes(\text{reference copies})$, and the acceptance probability is linear in
the input state.  So there is a Hermitian operator $Z_\alpha(H)$ on $\cA$ with
\[
  \Adv(\sigma,\alpha;H)\ =\ \big|\tr\big(\sigma Z_\alpha(H)\big)\big| .
\]
The supremum over states $\sigma$ is $\opnorm{Z_\alpha(H)}$, and by \Cref{lem:net} there is a
$\frac14$-net $\mathcal N$ of the unit sphere of $\cA$, of cardinality at most $9^{2d_\cA}$, on which that
norm is at most $2\max_{\gamma\in\mathcal N}|\bra\gamma Z_\alpha(H)\ket\gamma|$.

For a fixed $\gamma\in\mathcal N$ the state $\ket\gamma$ does not depend on $H$, so the adversary with
advice state $\proj\gamma$ is an adversary of \Cref{model:hybrid}: it prepares $\ket\gamma$ as part of
its own initial workspace, and its only $H$-dependent advice is $\alpha$.  Let $\mathsf D^*$ be the adversary
that reads a classical advice string of length $a'$ as a pair $(\alpha,j)$ with $j$ an index into
$\mathcal N$, prepares $\ket{\gamma_j}$, and runs the original adversary with advice $\alpha$.  Then
$\mathsf D^*$ is again an adversary of \Cref{model:hybrid}, with $a'$ in place of $a$, and for every $H$
\[
  \max_{(\alpha,j)}\Adv(\mathsf D^*;H)\ =\ \max_{\alpha}\ \max_{\gamma\in\mathcal N}\big|\bra\gamma Z_\alpha(H)\ket\gamma\big|
  \ \ge\ \tfrac12\,\sup_{\sigma}\max_\alpha\Adv(\sigma,\alpha;H).
\]
Taking expectations and applying \Cref{thm:hybrid} to $\mathsf D^*$ gives the display.  The final claim is
arithmetic: $a'=O(K)$, and the middle term is at most $C\sqrt{(L+a'+2)/(n-2q)^{1/3}}$ by
\Cref{lem:vsp-small}, which is $2^{-\Omega(\lambda)}$ once $d_\cA\le n^{1/3-\eps}$, while the other two
terms are $2^{-\Omega(\lambda)}$ as well.
\end{proof}

The reduction is oblivious to what the hidden object is, so the corollary applies wherever
\Cref{thm:hybrid} does, in particular to the enlarged object $X=(H,S)$ of \Cref{sec:single}.

\paragraph{The limit of the method.}  The proof needs a classical transcript \emph{between} the
state source and the coherent query.  The reason is that the decorrelation step is a conditioning on
$H$ through the records, and it requires the information separating the two oracle interactions to be
classical.  Two coherent queries leave no transcript to condition on.  A coherent query placed
\emph{before} the classical ones would be worse still, since it would put the $H$-dependence on the
effect side of the relaxation, where there is no mass normalization.  The method reaches a
hybrid whose coherent query comes last, and no further.  Full interleaving would require a
two-sided-normalized relaxation, and two full-width coherent queries would require \Cref{conj:lift}.

Both of the excluded orderings have a concrete meaning, which \Cref{cor:quantum-advice} supplies: a
coherent query placed first can synthesize arbitrary quantum advice about $H$, and so already
contains that case.  ``Coherent-first, then classical'' is thus the quantum-advice version of this
theorem (\Cref{sec:security-thm}), and ``coherent, then coherent'' is quantum advice together with
one coherent query (\Cref{sec:lift}).  In neither case is there a record left to condition on.

\section{The separation relative to a single classical oracle}
\label{sec:single}

This section proves the main separation.  Our oracle here is one deterministic Boolean function on
classical strings: a single function that answers every probability-estimation question and, at the
same time, carries the two states of the pair.  Relative to it, one-way puzzles do not exist,
including the inefficiently verifiable ones, and the half-subspace pair is nevertheless an EFI pair
in the access model of \Cref{model:access}.

We store samples of the pair inside the oracle, through Rosenthal's one-query synthesis, so that the
source $\cR$ of \Cref{model:source} disappears from the statement.  That source was always the
analysis model rather than a part of the result: \Cref{sec:oracle} through \Cref{sec:coherent} work
with fresh samples from a hidden $H$, because that is the form the concentration arguments require.
Bridging the two costs us one elementary step, which replaces the stored challenge by a fresh one,
and \Cref{sec:single-thm} explains why we cannot omit it.

\subsection{The oracle and the generator}
\label{sec:single-oracle}

Rosenthal's theorem, in the form quoted in the proof of \Cref{cor:quantum-advice}, provides for
every $m$-qubit state $\ket\psi$ and every error $\eps\ge\exp(-\poly(m))$ a Boolean function
$f_{\ket\psi}$ and a single uniform circuit $C_m=C_{m,\eps}$, making one query, such that the reduced state on
the first $m$ qubits of $C_m^{f_{\ket\psi}}\ket{0\cdots0}$ is within trace distance $\eps$ of
$\psi$.

We apply that theorem not to the sampled vectors themselves but to their roundings to a fixed finite
set of states.  Our reason is measurability.  The expectations over $S$ taken below, and the Borel
hypothesis of \Cref{thm:vsp-lb}, both require the oracle to be a measurable function of the hidden
object, and $\psi\mapsto f_{\ket\psi}$ is a selection made inside the proof of~\cite{Rosenthal23},
about which we prefer to assume nothing.  Restricted to a finite set of states the assignment is a
function on a finite set, hence measurable for that reason alone, and rounding costs one further
$\delta$ in trace distance, both $\eps$ and $\delta$ being fixed in \Cref{model:single}.  We therefore fix, for each $\lambda$, a finite $\delta$-net $\mathcal
N_\lambda$ of the unit sphere of $\bbR^{n}$, together with the rule sending each $u$ to the nearest
point of $\mathcal N_\lambda$, least in a fixed enumeration if several are nearest.  That rule is
constant on the cells of a fixed finite measurable partition and so is Borel, and the rounded vector
$\tilde u$ satisfies $\td(\proj{\tilde u},\proj u)\le\delta$, since $\norm{u-\tilde u}\le\delta$
gives $|\ip u{\tilde u}|\ge1-\delta^2/2$.

\begin{model}[The oracle $\cO'$]
\label{model:single}
Fix $m\eqdef\log_2n=\lambda$, $\kappa\eqdef m+4\lambda$, and $\eps\eqdef\delta\eqdef2^{-3\lambda}$.
For each $\lambda$, independently across $\lambda$ and in addition to the subspace $H=H_\lambda$,
draw a family $S=\{u_{b,r}\}_{b\in\{0,1\},\,r\in\{0,1\}^\kappa}$ of independent vectors.  Here
$u_{0,r}$ is uniform on the unit sphere of $H$, and $u_{1,r}$ on the unit sphere of $H^\perp$.  Write
$\tilde u_{b,r}$ for the rounding of $u_{b,r}$ to $\mathcal N_\lambda$, and set
\[
  \cG\big(1^\lambda,b,r,x\big)\ \eqdef\ f_{\ket{\tilde u_{b,r}}}(x),
\]
the Rosenthal function of the rounded $(b,r)$-th sample at error $\eps$.  On inputs not of this form
we set $\cG\eqdef0$, so that $\cG$ is a total Boolean function.

The oracle is then the tagged Boolean join $\cO'\eqdef(\cG,\Count)$, where $\Count$ is the oracle of
\Cref{def:oracle} with base layer $\cG$ in place of $\cR$.  That is, a rank-$j$ query may reference
clocked samplers whose oracle access is to $\cG$ and to ranks strictly below $j$; we write the rank
as $j$ here because $r$ already indexes the stored samples.  Thus $\cO'$ is one deterministic Boolean
function on classical strings.
\end{model}

The generator $G'(1^\lambda,b)$ samples $r\in\{0,1\}^\kappa$ uniformly with its own coins, runs
$C_m$ with the one query addressed to $\cG(1^\lambda,b,r,\cdot)$, outputs the designated $m$-qubit register,
and discards the rest, including $r$. Writing $\tilde\rho_{b,r}$ for the output at index $r$, the
output of $G'$ on branch $b$ is $\tilde\rho_b\eqdef2^{-\kappa}\sum_r\tilde\rho_{b,r}$, and by
construction, through the rounded vector, $\td(\tilde\rho_{b,r},\proj{u_{b,r}})\le\eps+\delta$.
Since the vectors $u_{0,r}$ lie in $H$ and the vectors $u_{1,r}$ in $H^\perp$, we get
$\td(\tilde\rho_0,\tilde\rho_1)\ge1-2(\eps+\delta)=1-2^{-3\lambda+2}$. The generator queries its one $\cG$-section coherently. A distinguisher with coherent-last access (\Cref{model:access}) is granted a coherent query as well, so it may run $G'$ itself, and its reference copies are supplied at no
cost.  That same access restriction applies to generator and distinguisher alike.

Well-definedness of $\Count$ is \Cref{lem:welldef} verbatim. For every fixing of $(H,S)$ the base
layer $\cG$ is a fixed total Boolean function, so the recursion on ranks closes exactly as before,
and the present case is in fact simpler, no channel semantics being needed for the base layer.

\subsection{The separation theorem}
\label{sec:single-thm}

\begin{theorem}[Main separation]
\label{thm:single}
With probability one over $(H,S)=\{(H_\lambda,S_\lambda)\}_\lambda$, the oracle $\cO'$ of
\Cref{model:single} satisfies both of the following.
\begin{enumerate}[label=(\roman*),leftmargin=2.2em]
\item One-way puzzles, including inefficiently verifiable and classically-secure ones, do not exist
  relative to $\cO'$, and neither do $\QEFID$ pairs; both statements hold for every fixing of
  $(H,S)$.
\item The pair $(\tilde\rho_{0,\lambda},\tilde\rho_{1,\lambda})$ generated by $G'$ is efficiently
  generatable relative to $\cO'$, is statistically far, with $\td\ge1-2^{-3\lambda+2}$, and is
  computationally indistinguishable, with advantage $2^{-\Omega(\lambda)}$, to
  every $\QPT$ distinguisher with coherent-last access to $\cO'$ (\Cref{model:access}).  It is
  therefore a classical-advice EFI pair (\Cref{def:efi-classical}), and \Cref{cor:barrier} draws the
  black-box consequence.
\end{enumerate}
\end{theorem}

Part~(i) is the attack of \Cref{app:attack} unchanged. The adversary of \Cref{thm:no-owpuzz} queries
$\Count$ on the candidate sampler and its prefix specializations, all of which are admissible
arguments by the reasoning of \Cref{lem:pe-easy} with $\cG$ in place of $\cR$, and it never queries $\cG$ itself; the same applies to \Cref{prop:no-qefid} and to \Cref{thm:interactive}. Nothing in that reasoning used the base
layer being a state source: what it used is that base-layer calls carry no rank, so that
$\mathsf{Norm}_J$ leaves them untouched exactly as it left the calls to $\cR$, and that the base
layer is a fixed object once $(H,S)$ is fixed. Both hold for $\cG$, and the sampler may query it
coherently.

Part~(ii) rests on the following lemma, and the order of the two steps below is forced.  Applied
directly to the stored samples, the communication reduction would hand Bob the list from which
Alice's challenge was drawn, and the associated problem would then be trivial, so the samples must be
replaced before the reduction rather than after, and \Cref{sec:single-thm} gives the argument.

\begin{lemma}[Listed samples may be replaced by fresh ones]
\label{lem:fresh}
Fix $\lambda$ and a distinguisher $\mathsf D$ in the access model of \Cref{thm:single}(ii), with $q$
reference copies per side, transcript length $L$, advice length $a$, and coherent width $M$, all
finite.  Write $\Adv(\mathsf D;H,S)$ for its advantage, maximized over advice strings.  Write $\Adv^{\circ}(\mathsf D;H,S)$ for the same quantity in the modified experiment.  In that experiment the
$2q$ reference copies and the challenge are replaced by fresh samples: copies of $\rho_0(H)$ and
$\rho_1(H)$, and a challenge drawn from $\rho_{b}(H)$, all independent of $S$.  Then for every $H$
and every $\lambda\ge2$,
\[
  \E_S\,\big|\Adv(\mathsf D;H,S)-\Adv^{\circ}(\mathsf D;H,S)\big|
  \   \le\ 2(2q+1)\Big(2(\eps+\delta)+\sqrt{n2^{-\kappa}/2}\Big)
  \ \le\ 2(2q+1)\cdot 2^{-2\lambda+1} .
\]
\end{lemma}

\begin{proof}
Everything the distinguisher does collapses into a single effect, and the two experiments then differ
only in the states loaded into its $2q+1$ input registers.  First, averaging over the indices that the
honest executions of $G'$ discard makes the joint input state an exact tensor product, so that effect
is applied to a product.  Second, replacing the registers one at a time bounds the change in
acceptance probability by the trace norms of the individual discrepancies.  Third, each of those
discrepancies is small, the rounding contributing $2(\eps+\delta)$ and the sampling error
$\sqrt{n2^{-\kappa}/2}$.

\emph{Step 1: the acceptance probability is one fixed effect applied to a state that factorizes across
the input registers.} Fix $(H,S)$ and an advice string $\alpha$.  Deferring all measurements, the
distinguisher is a fixed oracle algorithm, and for a fixed oracle its interaction is a fixed quantum
channel followed by a two-outcome measurement.  Its acceptance probability on challenge branch $b$ is
therefore $\tr[E\,\omega_b]$ for a single effect $E=E(H,S,\alpha)$ on the $2q+1$ input registers, with
$\omega_b$ their joint state.

Notice that $E$ does not depend on the indices drawn by the honest executions of $G'$.  Those indices
are internal coins of the generator and are discarded, so they reach the distinguisher only through
the states themselves.  This is exactly why the generator discards $r$.

The indices of the $2q+1$ honest executions are independent and uniform on $\{0,1\}^\kappa$.
Averaging over them, the joint input state factorizes exactly:
\[
  \omega_b\ =\ \tilde\rho_0^{\otimes q}\otimes\tilde\rho_1^{\otimes q}\otimes\tilde\rho_b .
\]
The independence of the indices is what makes the average of a tensor product equal the tensor
product of the averages.  This includes the terms in which two indices coincide, since a repeated
index still contributes a product.  In the modified experiment
$\omega_b^{\circ}=\rho_0^{\otimes q}\otimes\rho_1^{\otimes q}\otimes\rho_b$.

Both experiments are now one effect applied to a product state, and no further property of $\mathsf D$ enters
below.

\emph{Step 2: passing from one product state to the other costs $2q+1$ single-register trace norms.}
Replacing the registers one at a time, and using $|\tr[EX]|\le\norm X_1$ together with
$\norm{A\otimes\tau}_1=\norm A_1$ for a state $\tau$, each replaced register contributes the trace
norm of its own discrepancy.  Branch $b=0$ carries $q+1$ registers of type $0$ and $q$ of type $1$,
and branch $b=1$ the reverse, so
\[
  \big|\tr[E\,\omega_0]-\tr[E\,\omega_0^{\circ}]\big|\ \le\ (q+1)\norm{\tilde\rho_0-\rho_0}_1+q\norm{\tilde\rho_1-\rho_1}_1 ,
\]
with the two coefficients exchanged for $b=1$.  The advantage is the difference of the two challenge
branches, so adding the two displays gives
\[
  \big|\Adv(\mathsf D;H,S)-\Adv^{\circ}(\mathsf D;H,S)\big|\ \le\ (2q+1)\big(\norm{\tilde\rho_0-\rho_0}_1+\norm{\tilde\rho_1-\rho_1}_1\big),
\]
the maximum over advice passing through since the bound is uniform in $\alpha$.

So far we have shown that the entire discrepancy is carried by the two single-register distances, and
it remains to bound them.

\emph{Step 3: each stored mixture is within $2(\eps+\delta)+\sqrt{n2^{-\kappa}/2}$ of the true state
in expectation over $S$.} We bound $\norm{\tilde\rho_{b'}-\rho_{b'}}_1$ by splitting off the exact
empirical mixture, namely $\sigma_{b'}\eqdef2^{-\kappa}\sum_r\proj{u_{b',r}}$, formed from the
\emph{unrounded} vectors, which are the ones actually lying in $H$ and $H^\perp$:
\[
  \tilde\rho_{b'}-\rho_{b'}\ =\ \underbrace{\big(\tilde\rho_{b'}-\sigma_{b'}\big)}_{\text{rounding}}
  \;+\;\underbrace{\big(\sigma_{b'}-\rho_{b'}\big)}_{\text{sampling}} .
\]
The rounding term has trace norm at most $2(\eps+\delta)$, since
$\td(\tilde\rho_{b',r},\proj{u_{b',r}})\le\eps+\delta$ for every $r$ and the trace norm is convex.

The sampling term is where the number of stored vectors enters.  It is supported on the
$(n/2)$-dimensional subspace carrying $\rho_{b'}$, so we may pass to the Hilbert--Schmidt norm at a
cost of $\sqrt{n/2}$.  The summands $\proj{u_{b',r}}$ are i.i.d.\ with mean $\rho_{b'}$ and
$\E\norm{\proj{u_{b',r}}-\rho_{b'}}_2^2=1-\tfrac2n$, so
\[
  \E_S\norm{\sigma_{b'}-\rho_{b'}}_2^2\ =\ 2^{-\kappa}\Big(1-\frac2n\Big),
  \qquad\text{and hence}\qquad
  \E_S\norm{\sigma_{b'}-\rho_{b'}}_1\ \le\ \sqrt{n2^{-\kappa}/2}
\]
by Jensen.  Taking $\E_S$ and summing the two terms gives the first bound.

The second bound is arithmetic.  Here $2(\eps+\delta)=2^{-3\lambda+2}$ and
$\sqrt{n2^{-\kappa}/2}\le2^{-2\lambda}$, the latter because $\kappa=m+4\lambda$ and $2^m\ge n$; the
first of these is below the second for $\lambda\ge2$, so their sum is at most $2^{-2\lambda+1}$, as
stated.
\end{proof}

The fresh experiment is one to which the theorems of \Cref{sec:coherent} already apply, with the
enlarged hidden object $X=(H,S)$ playing the role of $H$.

\begin{lemma}[The fresh experiment is governed by $H$ alone]
\label{lem:side}
Fix $\lambda$, and let $\mathsf D$ be any distinguisher of \Cref{thm:single}(ii), with $q$ reference copies
per side, transcript length $L$, advice length $a$, and coherent width $M$.  Then
\[
  \E_{H,S}\big[\Adv^{\circ}(\mathsf D;H,S)\big]\ \le\ \frac{6q}{n}+2\beta\big(n-2q,\,L+a+1\big)+C\sqrt{\frac{\log(2M)+L+a}{n-2q}},
\]
where $C$ is the universal constant of \Cref{thm:hybrid} and $\beta$ is as in \Cref{lem:vsp-small}.
The right-hand side is exactly the bound \Cref{thm:hybrid} gives for the hidden object $H$.
\end{lemma}

\begin{proof}
In the fresh experiment the challenge and the $2q$ reference copies are drawn from $\rho_0(H),\rho_1(H)$
independently of $S$, while the oracle $\cO'=(\cG,\Count)$ is a deterministic Boolean function of
$X=(H,S)$. This is exactly the setting of \Cref{thm:transcript-security} and \Cref{thm:hybrid} with $X$
in place of $H$, since those theorems ask only that the truth table be an arbitrary Boolean function of
the hidden object. What is left is to reduce $X$ to $H$, and we do so one term at a time.

First, $S$ does not enter the spectral term.  The weighted norm $\phi_\pi(H)$ of \Cref{lem:subnorm} is
built from the adversary's fixed operators and the reflection $R_H$, so it depends on $H$ but not on
$S$.  As for the truth table $s(X)$, it enters \Cref{thm:hybrid} only through the maximum over all
Boolean functions, and that maximum already covers any function of $(H,S)$.

The conditioning behaves in the same way.  The record identity of \Cref{lem:record} holds pointwise in
$(H,S)$, and the decorrelation step of \Cref{thm:hybrid} bounds a maximum over record strings whose
summands are functions of $H$ alone, so its expectation is over $H$.

Second, handing Bob the extra $S$ does not raise the communication term. In the associated problem Bob
holds $X=(H,S)$ and Alice a fresh challenge drawn from $\rho_b(H)$, independent of $S$ given $H$.
The conditional law of $S$ given $H$ is explicit here: it is $2^{\kappa+1}$ independent uniform
vectors on the spheres of $H$ and $H^\perp$.  So from $H$ and private randomness Bob draws $S'$ with
that law and runs the protocol with $S'$ in place of $S$.  Since the challenge is independent of $S$
given $H$, the joint law of the transcript and of Alice's output is unchanged, and any $L$-bit
protocol for the $(H,S)$-problem yields an $L$-bit protocol for the $H$-problem of the same
advantage.

The same substitution is needed after conditioning on the spans, which is the form in which
\Cref{thm:hybrid} uses it.  There Bob holds $(H_A,S)$ with $(A,B)$ known to both.  Since
$H=A\oplus H_A$ he can still sample $S'$ from the conditional law, and the residual challenge is
independent of $S$ given $(H_A,A,B)$.

Hence $\beta$ for the associated problem is the $\beta$ of \Cref{lem:vsp-small}, and
\Cref{cor:reference-vsp} applies with hidden object $H$. The spectral term and the communication term
are the ones \Cref{thm:hybrid} supplies for $H$, which is the stated bound.
\end{proof}

\begin{proof}[Proof of \Cref{thm:single}\textup{(ii)}]
Efficient generation from $\cO'$ and farness $\td\ge1-2(\eps+\delta)$ were shown after
\Cref{model:single}, so it remains to prove indistinguishability.

Discretize the gate set as in the proof of \Cref{thm:EFI-security}, giving a countable family of
machines $\{N_i\}$ with the same polynomial bounds $t_i(\lambda)$ on their parameters. Fix $i$. By \Cref{lem:fresh}, which replaces the stored samples
by fresh ones, and then \Cref{lem:side}, which returns the fresh experiment to the hidden object $H$,
\[
  \E_{H,S}\big[\Adv_{N_i}\big]\ \le\ \E_{H,S}\big[\Adv^{\circ}_{N_i}\big]+2(2t_i+1)\cdot2^{-2\lambda+1}
  \ \le\ 2^{-\Omega(\lambda)},
\]
the last step because the \Cref{thm:hybrid} bound is $2^{-\Omega(\lambda)}$ for $n=2^{\Omega(\lambda)}$
and all parameters $\poly(\lambda)$.

The quantum-advice clause is the same argument with \Cref{cor:hybrid-qadvice} in place of
\Cref{thm:hybrid}.  Its bound is $2^{-\Omega(\lambda)}$ under the stated size restriction, and the
supremum over advice states lies inside the expectation, which is what \Cref{lem:fixing} requires.

Throughout, the base objects $(H_\mu,S_\mu)$ at security parameters $\mu\ne\lambda$ are independent
of $(H_\lambda,S_\lambda)$, and we fix them as shared public randomness, exactly as the proof of
\Cref{thm:EFI-security} fixes the subspaces at other parameters.  Every bit of $\cO'$ is then a
Borel function of $X=(H_\lambda,S_\lambda)$ alone, by the rank induction of \Cref{lem:welldef}.  This
holds for every section and every rank, and it holds even for a $\Count$ bit about a sampler at
another input length that calls $\cG$ at level $\lambda$.

Now let $i$ range.  \Cref{lem:fixing}, applied with $\mathbf X=(H,S)$, fixes a set of $(H,S)$ of
probability one on which every $N_i$ has negligible advantage.  Since part~(i) holds for
\emph{every} fixing of $(H,S)$, that same set witnesses both parts, which gives the theorem.
\end{proof}

\paragraph{The order of the replacement.}  \Cref{lem:fresh} is a precondition for the communication
argument, not a convenience.  Applied directly to the real experiment, \Cref{thm:transcript-security}
would hand Bob the whole object $X=(H,S)$, and $S$ contains the description
$u_{b^\ast,r^\ast}$ of the very challenge Alice prepares, which makes the associated problem trivial:
Alice fingerprints her challenge vector in $O(\kappa)=O(\lambda)$ bits, and Bob, holding $S$, finds
the matching entry and reads off its branch.  Applying \Cref{lem:fresh} first removes the challenge
from $S$ and restores the $\Omega(n^{1/3})$ bound.  The generator of \Cref{model:single} discards its
index $r$ for the same reason, and in the source model the same consideration is why $\cR$ traces out
its sampling randomness (\Cref{model:source}).

Two further reasons bind harder still, since they apply to the coherent half of the access model,
where no communication argument is in play.  The first concerns the spectral bounds.  Those of \Cref{sec:coherent}
rest on $\rho_0-\rho_1=2R_H/n$, a linear statistic of a Haar-conjugated reflection with $\E R_H=0$,
whereas the stored pair has difference $\tilde\rho_0-\tilde\rho_1$, an empirical average over the
list, to which \Cref{lem:khintchine} does not apply.  The second concerns the reference copies.
\Cref{lem:reference} replaces them by their spans, using the rotation-invariant conditional law of
$q$ independent Haar vectors given their span; the stored copies, however, are drawn \emph{with
replacement from a finite list}, so resampling from the span does not reproduce their joint law.
Both objections disappear once the samples are fresh.

The two demands come apart, even at the most permissive end of the $\OWPuzz$ definition: an EFI pair
asks that two quantum states be hard to distinguish, a one-way puzzle that a classical search problem
be hard.  A single Boolean function
collapses the meta-complexity hardness that characterizes $\OWPuzz$s~\cite{CGGH25,HM24,HHM25},
removing them entirely. The half-subspace pair survives it, because distinguishing that pair
requires implementing the hidden measurement $\{P_H,P_{H^\perp}\}$, and neither a polynomial-length
classical transcript~\cite{KR11} nor a single coherent query (\Cref{thm:one-query}) conveys enough
information about $H$ to do so.

\subsection{The black-box barrier}
\label{sec:barrier}

\Cref{thm:single} has a consequence for constructions.  We state it in the fully black-box
framework, the setting for black-box separations since Impagliazzo and Rudich~\cite{IR89}, restricted
to reductions whose oracle access lies in the class the theorem covers.

\begin{definition}[Fully black-box construction; $\mathcal C$-bounded reduction]
\label{def:fbb}
Let $\mathcal C$ be a class of oracle distinguishers.  A \emph{fully black-box construction of
one-way puzzles from EFI pairs} is a pair $(\mathsf{Con},\mathsf{Red})$ of oracle algorithms with the
following property, required relative to every oracle $\cO$ and for every $\QPT^{\cO}$ generator $G$
of a statistically far pair.  First, $\mathsf{Con}^{G,\cO}$ is the sampler of a one-way puzzle
relative to $\cO$, for some verifier.  Second, for every adversary $\cA$ that inverts
$\mathsf{Con}^{G,\cO}$ with non-negligible probability, $\mathsf{Red}^{G,\cA,\cO}$ distinguishes the
pair of $G$ with non-negligible advantage.

We call the construction \emph{$\mathcal C$-bounded} if the following holds whenever $\cA$ is a
classical polynomial-time algorithm making classical queries to $\cO$.  The distinguisher
$\mathsf{Red}^{\cA,\cO}$, with its reference copies of $G$'s outputs supplied as in
\Cref{model:access}(i), belongs to $\mathcal C$.
\end{definition}

The boundedness condition says that the reduction, run together with the puzzle attacker, is itself a
distinguisher of the class.  For the class of \Cref{model:access} it means that the reduction obtains
the generator's outputs as reference copies and makes its own coherent query last.

\begin{corollary}[No coherent-last black-box construction]
\label{cor:barrier}
Let $\mathcal C$ be the class of $\QPT$ distinguishers with classical advice and coherent-last access
(\Cref{model:access}).  There is no $\mathcal C$-bounded fully black-box construction of one-way puzzles
from EFI pairs.
\end{corollary}

\begin{proof}
Suppose $(\mathsf{Con},\mathsf{Red})$ were one.  Fix $(H,S)$ in the probability-one event of
\Cref{thm:single} and take $\cO=\cO'$, $G=G'$.  The construction is then fed a genuine EFI pair: by
\Cref{thm:single}(ii) the pair is statistically far, so $\mathsf{Con}^{G',\cO'}$ is a one-way puzzle
sampler relative to $\cO'$.

But puzzles do not survive $\cO'$.  By \Cref{thm:single}(i) there is a classical polynomial-time $\cA$
making classical queries that inverts that sampler with probability $1-\negl(\lambda)$.  Feed it to
the reduction.  By $\mathcal C$-boundedness, $\mathsf{Red}^{\cA,\cO'}$ with challenger-supplied
reference copies is a distinguisher of \Cref{model:access}, and it has non-negligible advantage
against the pair of $G'$.  That contradicts \Cref{thm:single}(ii).
\end{proof}

\section{Fully coherent access}
\label{sec:lift}

\Cref{thm:single} places the separation in the classical-query model for $\cO'$, extended by a single
coherent query.  In this section we ask what survives when $\cO'$ may instead be queried in superposition throughout,
which for a deterministic Boolean oracle is the standard relativized model.  The hidden object is now $X=(H,S)$ of \Cref{model:single}, and we read every statement below
with an arbitrary Boolean function of $X$ where it had one of $H$.  Nothing depends on which of the
two it is, since each bound quantifies over all Boolean functions of the hidden object.  That quantification makes the
conjecture below stronger than coherent security of $\cO'$ itself: it implies that security, but its
failure for some function of $X$ would not by itself yield an attack on $\cO'$.

\begin{model}[Coherent access to $\cO'$]
\label{model:bit}
Since $\cO'$ is a single deterministic Boolean function on classical strings, coherent access to it
is the ordinary relativized query $U_X\ket z\ket c=\ket z\ket{c\oplus\cO'(z)}$, with $X=(H,S)$ the
hidden object.  Here $z$ ranges over the query strings of \Cref{model:single}, padded to a fixed
length $\ell(\lambda)=\poly(\lambda)$, so that the query domain at security parameter $\lambda$ is
$\{0,1\}^{\ell(\lambda)}$.

A query addressed to the counting layer carries a sampler description, an output string, and a bit
index $k$, and it returns the $k$-th bit of the corresponding output probability.  We call the
largest index $k$ an adversary uses its \emph{precision}.  A query addressed to the base layer $\cG$ returns a
bit of a stored sample instead.  No statement below depends on which layer a query addresses, since
every bound quantifies over all Boolean functions of $X$.
\end{model}

\begin{conjecture}[Coherent multi-query security]
\label{conj:lift}
There are $c>0$ and $d<\infty$ such that the following holds for every sufficiently large even $n$, with $\lambda\eqdef\lceil\log_2n\rceil$.
Let $X$ be either $H$ alone, Haar of dimension $n/2$ in $\bbR^n$, or else $(H,S)$ drawn as in
\Cref{model:single}.  Let $\cA$ be a fixed circuit whose operations, apart from its oracle queries,
do not depend on $X$.  Suppose $\cA$ receives:
\begin{enumerate}[label=(\roman*),leftmargin=2.2em,itemsep=1pt,topsep=3pt]
\item one challenge copy drawn from $\rho_0(H)$ or $\rho_1(H)$, the maximally mixed states on $H$ and
  on $H^\perp$, with a uniform bit;
\item at most $\bar q$ fresh reference copies of each of $\rho_0(H)$ and $\rho_1(H)$; and
\item $a$ bits of classical advice $\alpha(X)$.
\end{enumerate}
Suppose further that $\cA$ makes $T$ adaptive coherent queries to $U_{f_X}$, for a Boolean function
$f_X$ of $X$, each query of width at most $M$.  The maps $X\mapsto f_X$ and $X\mapsto\alpha(X)$ are
assumed Borel.  Write
\[
  \mathrm{bias}_\cA(X)\ \eqdef\ \tfrac12\big(\Pr[\cA\text{ accepts}\mid\rho_0]-\Pr[\cA\text{ accepts}\mid\rho_1]\big)
\]
for its signed bias at a fixed $X$. Then
\[
  \E_X\big|\mathrm{bias}_\cA(X)\big|\ \le\ \big(1+T+\bar q+a+\log M\big)^{d}\cdot n^{-c} .
\]
Since $\lambda=\lceil\log_2n\rceil$, this is $\negl(\lambda)$ for all $T,\bar q,a$ and $\log M$
polynomial in $\lambda$.  Note that a $\Count$ query carries its precision inside the query string,
so precision is charged here through $\log M$.
\end{conjecture}

\paragraph{The conjecture and the advice restriction.}  \Cref{conj:lift} grants its adversary
only classical advice.  Resolving it affirmatively nevertheless removes the quantum-advice
restriction as well, so that the two features delimiting \Cref{model:access} have a single
resolution.

Let an adversary hold $m=\poly(\lambda)$ qubits of quantum advice $\sigma_X$ and make $T$ classical
$\Count$ queries, and fix a finite net of $m$-qubit states.  For each net point the bias is a Borel
function of $X$, so a lexicographically first pointwise maximizer over the net is Borel, as in the proof of
\Cref{thm:transcript-security}, and it lies within $\exp(-\poly(\lambda))$ in trace distance of the
best advice.  Rosenthal's one-query synthesis then supplies a Boolean function $g_X$, together with a
fixed circuit that prepares a purification of the rounded state with one query to $g_X$, exactly as
in \Cref{cor:quantum-advice}.  Tagging $g_X$ together with the counting layer gives a single Boolean
function $F_X$ of $X$.  Under $F_X$ the adversary is a fixed circuit making $T+1$ adaptive coherent
queries to $U_{F_X}$ and carrying no quantum advice, which is exactly the form \Cref{conj:lift}
bounds.  What makes this work is the quantification over \emph{every} Boolean function of $X$: a
statement about the counting truth table alone would not cover $g_X$.

We call the weaker statement that every such circuit with $T,\bar q,a$ and $\log M$ polynomial in
$\lambda$ has $\E_X|\mathrm{bias}_\cA(X)|=\negl(\lambda)$ the \emph{polynomial form} of the
conjecture.  The displayed bound implies it, and it is the polynomial form that we use in
\Cref{thm:two-cases}.

The single-challenge restriction costs us nothing, EFI security being a single-copy statement, and
we do need it for the analysis, since $\rho_0^{\otimes s}-\rho_1^{\otimes s}$ is not linear in $R$
for $s>1$.  The reference copies are inessential, by the following lemma, so it suffices to prove or refute the conjecture at $\bar q=0$; this is why the conjecture is stated for every even $n$ rather than for $n=n(\lambda)$, since \Cref{lem:copies-free} moves the instance to dimension $n-2\bar q$.  The case $T=1$ with no advice and no reference copies is \Cref{thm:one-query}; with advice and copies it is \Cref{thm:hybrid} at $L=0$, and the parallel case is \Cref{cor:parallel}.

\begin{lemma}[Reference copies cost only a dimension loss under coherent access]
\label{lem:copies-free}
Let $\cA$ be an adversary that receives one challenge copy, $\bar q<n/4$ reference copies of each of
$\rho_0$ and $\rho_1$, and $a$ bits of classical advice given by a Borel map $H\mapsto\alpha_H$.
Grant it arbitrary coherent access to an arbitrary Boolean function of $H$, and suppose it achieves
average bias $\eta$.  Then there is an adversary $\cA'$ with one challenge copy, \emph{no} reference copies, $a$
bits of advice, and coherent access to an arbitrary Boolean function of a Haar-random
half-dimensional subspace of a space of dimension $n'=n-2\bar q$.  It has the same query count and
width as $\cA$, and its average bias is at least $\eta-2\bar q/n$.
\end{lemma}

\begin{proof}
Give $\cA$ the classical spans $(A,B)$ of the copies it would have drawn.  This only increases its
power: from a basis of the spans it can resample the copies with the correct joint law, by the
argument that opens the proof of \Cref{lem:reference}.

Now condition on $(A,B)$.  By \Cref{lem:residual} the challenge becomes a mixture of two branches, of
weights
\[
  \underbrace{\tfrac{2\bar q}n}_{\text{known subspace}}\ \text{ and }\ \underbrace{1-\tfrac{2\bar q}n}_{\text{residual}} .
\]
The first is the maximally mixed state on the subspace the adversary already knows.  The second is
the maximally mixed state on $H_A$, or on $H_B=K\cap H_A^\perp$, inside $K$; here $\dim K=n'$ and
$H_A$ is Haar-random of dimension $n'/2$ in $K$.

The known branch can be discarded. The bias is affine in the challenge, and is at most $1$ on that
branch, so the residual branch carries bias at least $\eta-2\bar q/n$ on average.

All that is left is to see that the residual branch is already an instance with no reference copies.
Given $(A,B)$, the truth table is a Boolean function of $H_A$ alone.  Hard-wiring $(A,B)$ and the
resampled copies into $\cA$ leaves an adversary $\cA'$ that holds one challenge copy and
nothing else, and that queries a Boolean function of a Haar-random half-dimensional subspace of $K$,
at the same query count and width.  Its average bias is at least $\eta-2\bar q/n$, which is the
claim.
\end{proof}

One general statement is available, and it locates the regime in which the question lives.

\begin{proposition}[Small width is settled at every query count]
\label{prop:small-width}
Consider an adversary that holds one challenge copy, makes any number of coherent queries to any
Boolean function of $H$ of width $M$, receives $a$ bits of classical advice, and uses no reference
copies.  There is a universal $C$ such that every such adversary satisfies
\begin{equation}
  \E_H\max_{f\in\{\pm1\}^{\cX},\,\alpha}|Z_H(f,\alpha)|\ \le\ C\,\frac{\sqrt{M+a}}{n} .
  \label{eq:union}
\end{equation}
Here $\cX$, of size $M$, is the set of query labels, and $Z_H(f,\alpha)$ is the bias
\eqref{eq:bias-def} of the adversary with truth table $f$ and advice string $\alpha$.
In particular, for each fixed $\eps>0$ the conclusion of \Cref{conj:lift} holds unconditionally with
$c=\eps/2$, at every query count, whenever $M+a\le n^{2-\eps}$.
\end{proposition}

\begin{proof}
Whatever the number of queries, the algorithm ends in a single acceptance effect $E_{f,\alpha}$ on
the challenge register, one such effect for each truth table $f$ and advice string $\alpha$.  By
\eqref{eq:bias-def} its bias is $\frac1n|\tr(E_{f,\alpha}R)|$, which is a linear
statistic of $R$ with a fixed coefficient matrix. Those coefficient matrices are bounded uniformly in
Hilbert--Schmidt norm,
\[
  \norm{E_{f,\alpha}}_2\ \le\ \sqrt{\tr E_{f,\alpha}}\ \le\ \sqrt n ,
\]
the last step because $0\preceq E_{f,\alpha}\preceq\Id$ on $\bbC^n$.  That is exactly the
normalization \Cref{lem:scalar-linear} asks for.

So apply that lemma to the family $\{E_{f,\alpha}:f,\alpha\}$, one effect per truth table and
advice string.  The family has cardinality at most $2^{M+a}$, so its logarithmic factor is
$O(\sqrt{M+a})$, and the bound the lemma returns is \eqref{eq:union}.
\end{proof}

Two statements are available beyond a single query, and each restricts a width rather than a query
count.  The first admits arbitrary quantum advice.

\begin{proposition}[Quantum advice with one coherent query]
\label{prop:advice-one-query}
There is a universal $C>0$ such that the following holds.  Let $\cA$ be an advice register of
dimension $K$, let $W\colon\cA\otimes\bbC^n\to\bbC^{\cX}\otimes\cK$ be an isometry that does not
depend on $H$, let $0\preceq\Pi\preceq\Id$, and write $\tr_{\bbC^n}$ for the partial trace over the
challenge register.  For a truth table $f\colon\cX\to\{\pm1\}$ put
\[
  Y_H(f)\ \eqdef\ \tfrac1n\,\tr_{\bbC^n}\big[(\Id_\cA\otimes R_H)\,W^\dagger D_f\Pi D_fW\big],
\]
an operator on $\cA$.  Then
\[
  \E_H\max_{f\in\{\pm1\}^{\cX}}\opnorm{Y_H(f)}\ \le\ C\Big(\sqrt{\tfrac{\log(2M)}{n}}+\tfrac{\sqrt K}{n}\Big).
\]
Consequently, consider an adversary that holds one challenge copy together with quantum advice of
$m$ qubits depending arbitrarily on $H$, and that then makes one coherent query of any width.  Its
average bias is $2^{-\Omega(\lambda)}$ whenever $\log M=\poly(\lambda)$ and $m\le(2-\eps)\lambda$ for
a fixed $\eps>0$.
\end{proposition}

\begin{proof}
The advice $\sigma_H$ enters only through $|\tr(\sigma_HY_H(f))|\le\opnorm{Y_H(f)}$, so the
consequence follows from the displayed bound, and it remains to prove that.  There are three steps.
We bound the quadratic form at a single fixed advice direction; we show that what results
concentrates at scale $n^{-1}$; and we then pass from the quadratic form to the operator norm by a
net.

\emph{Step 1: at a fixed direction the advice disappears, and the old bound applies.}  Fix a unit
vector $\alpha\in\cA$ and set $W_\alpha\eqdef W(\ket\alpha\otimes\Id)$.  Since $W^\dagger W=\Id$ on
$\cA\otimes\bbC^n$,
\[
  W_\alpha^\dagger W_\alpha\ =\ (\bra\alpha\otimes\Id)\,\Id\,(\ket\alpha\otimes\Id)\ =\ \Id
  \qquad\text{on }\bbC^n ,
\]
so $W_\alpha$ is an isometry of the challenge register alone, and its query masses
$\bra\alpha Q_x\ket\alpha$ sum to $1$.  In other words $(W_\alpha,\Pi)$ is an ordinary one-query
strategy, and its bias is $\bra\alpha Y_H(f)\ket\alpha$.  \Cref{thm:one-query} therefore applies to
it unchanged:
\[
  \E_H\max_f\big|\bra\alpha Y_H(f)\ket\alpha\big|\ \le\ C\sqrt{\frac{\log(2M)}n} .
\]

\emph{Step 2: that quantity fluctuates at scale $n^{-1}$.}  Put
\[
  X_{f,\alpha}\ \eqdef\ (\bra\alpha\otimes\Id)\,W^\dagger D_f\Pi D_fW\,(\ket\alpha\otimes\Id) ,
\]
an effect on $\bbC^n$, so that $\bra\alpha Y_H(f)\ket\alpha=\frac1n\tr(R_HX_{f,\alpha})$.  Being an
effect, $X_{f,\alpha}$ satisfies $\norm{X_{f,\alpha}}_2^2\le\tr X_{f,\alpha}\le n$, and Cauchy--Schwarz
then makes that quantity $n^{-1/2}$-Lipschitz in $R$ in the Hilbert--Schmidt metric.  A maximum of
$n^{-1/2}$-Lipschitz functions is again $n^{-1/2}$-Lipschitz, and $U\mapsto R$ is $2$-Lipschitz by the
commutator estimate of \Cref{lem:khintchine}.  Hence
\[
  F_\alpha\ \eqdef\ \max_f\big|\bra\alpha Y_H(f)\ket\alpha\big|
\]
is $2n^{-1/2}$-Lipschitz in $U$, and \eqref{eq:scalar-conc} makes it subgaussian at scale $n^{-1}$.

\emph{Step 3: a net converts the quadratic form into the operator norm.}  Each $Y_H(f)$ is Hermitian,
so \Cref{lem:net} applies to it.  Fix the net $\mathcal N$ it provides, of cardinality at most
$9^{2d_\cA}$; then
\[
  \opnorm{Y_H(f)}\ \le\ 2\max_{\alpha\in\mathcal N}\big|\bra\alpha Y_H(f)\ket\alpha\big| ,
  \qquad\text{so}\qquad
  \max_f\opnorm{Y_H(f)}\ \le\ 2\max_{\alpha\in\mathcal N}F_\alpha .
\]
Now combine the two estimates.  Step~1 bounds each $F_\alpha$ in expectation, Step~2 makes it
subgaussian, and a maximal inequality over $\mathcal N$ then gives
\[
  \E_H\max_{\alpha\in\mathcal N}F_\alpha\ \le\ C\sqrt{\frac{\log(2M)}n}\;+\;\frac{C\sqrt{\log|\mathcal N|}}n .
\]
Since $\log|\mathcal N|=O(K)$, this is the bound we wanted.
\end{proof}

The second statement allows two genuine queries, at the cost of a width restriction on the first.

\begin{proposition}[Two coherent queries with a narrow first query]
\label{prop:asym-two-query}
Consider an adversary that holds one challenge copy and makes two adaptive coherent queries, the
first of width $M_1$ and the second of width $M_2$.  Their truth tables $f$ and $g$ are independent,
and each may depend arbitrarily on $H$.  There is a universal $C>0$ such that every such adversary
satisfies
\[
  \E_H\max_{f,g}|Z_H(f,g)|\ \le\ C\Big(\sqrt{\tfrac{\log(2M_2)}{n}}+\tfrac{\sqrt{M_1}}{n}\Big).
\]
In particular the bias is $2^{-\Omega(\lambda)}$ whenever $M_1\le n^{2-\eps}$ for a fixed $\eps>0$ and
$\log M_2=\poly(\lambda)$, however large $M_2$ may be.
\end{proposition}

\begin{proof}
The two queries are handled by different means: the second by \Cref{thm:one-query}, the first by a
union bound, which is affordable only because its width is restricted.

\emph{The second query, at $f$ fixed.}  Everything the adversary does before its second query is a
composition of $H$-independent isometries with the fixed unitary $D_f$, and so is itself an
$H$-independent isometry.  That isometry, together with the final effect, is an ordinary one-query
strategy for the second query, so \Cref{thm:one-query} applies to it and gives
\[
  \E_H\max_g\big|Z_H(f,g)\big|\ \le\ C\sqrt{\frac{\log(2M_2)}n} .
\]

\emph{The maximum over $f$.}  Write $Z_H(f,g)=\frac1n\tr\big(R_HE_{f,g}\big)$ as in
\eqref{eq:bias-def}, where $E_{f,g}$ is the acceptance effect on the challenge register.  As in the
proof of \Cref{prop:advice-one-query} this is $n^{-1/2}$-Lipschitz in $R$; the maximum over $g$
preserves that constant; and so $\max_g|Z_H(f,g)|$ is subgaussian at scale $n^{-1}$.  A maximal
inequality over the $2^{M_1}$ first truth tables then costs
\[
  C\,\frac{\sqrt{M_1}}n ,
\]
which is where the restriction on $M_1$ is spent.  Adding the two estimates gives the bound.
\end{proof}

What remains open is $M=n^{2-o(1)}$, a regime the model forces, since $\QPT$
quantifies over all polynomials and an adversary may ask about samplers of description length
$\lambda^{k}$ for any $k$, giving $M=2^{\lambda^{k}}\gg n^2$.  At $T\ge2$ no nontrivial bound is
available once \emph{both} widths exceed $n^2$, which is what \Cref{prop:asym-two-query} leaves open.

\paragraph{The first open case.}  The first open case of \Cref{conj:lift} is more structured than the
general statement suggests, and \Cref{cor:quantum-advice} is what exhibits the structure.

By Rosenthal's one-query synthesis, a single coherent query to a suitable Boolean function of $H$
prepares an \emph{arbitrary} $\poly(\lambda)$-qubit state $\sigma_H$.  Splitting the query domain in
two costs nothing here, since the truth table is arbitrary and $M=2^{\poly(\lambda)}$.  Two adaptive
coherent queries subsume arbitrary $\poly(\lambda)$-qubit quantum advice about $H$,
together with one coherent query.

That special case is already settled up to advice of dimension $n^{2-\eps}$, by
\Cref{prop:advice-one-query}, and the complementary ordering -- advice followed by classical queries
-- is settled on $O(\lambda)$ qubits by \Cref{cor:hybrid-qadvice}.  What remains in both is advice on
more than linearly many qubits, and the reason is the one \Cref{sec:hybrid} identifies: the
decorrelation step of \Cref{thm:hybrid} conditions on a \emph{classical} record, quantum advice is
not a record, and two coherent queries leave nothing to condition on.  Note that the case is
narrower than $T=2$, since a first query may act on the challenge register jointly with the
workspace, which no advice state models.  For unitary synthesis, Dong, Lombardi, and
Ma~\cite{DLM26} compare one-query synthesis with quantum programs, that is, with synthesis from
quantum advice, for phase unitaries.

Both answers to the conjecture are informative, and we record what each implies.

\begin{proposition}[Consequences of each answer; informal]
\label{thm:two-cases}
Both of the following hold.
\begin{itemize}[leftmargin=1.6em]
\item \textbf{(A)} Suppose \Cref{conj:lift} holds in its polynomial form: every adversary with $T,\bar q,a,\log M=\poly(\lambda)$ has $\E_X|\mathrm{bias}_\cA(X)|=\negl(\lambda)$. Then \Cref{thm:single} holds with fully coherent access, so the separation carries no restriction on the access mode, and by \Cref{sec:lift} none on the advice either. (The no-$\OWPuzz$ direction is immediate: the attack of \Cref{app:attack} is classical, a special case of coherent access.)
\item \textbf{(B)} Suppose instead the polynomial form of \Cref{conj:lift} fails.  Then for every $c>0$ there are, for infinitely many
$\lambda$, a query count $T=\poly(\lambda)$, a width $M$ with $\log M=\poly(\lambda)$, an advice length $a=\poly(\lambda)$, and a family of Boolean functions
$\{f_H\colon[M]\to\{\pm1\}\}$ with advice strings $\{\alpha_H\}$, admitting a $T$-query circuit with
acceptance effect $E_{f_H,\alpha_H}$ such that
\[
  \E_H\Big[\tfrac1n\big|\tr\big(E_{f_H,\alpha_H}R_H\big)\big|\Big]\ >\ n^{-c} ,
\]
for a Haar-random half-dimensional $H$ in dimension $n=2^{\Theta(\lambda)}$, with no reference copies. Taking any
$c<1/6$, coherent access to a classical function of $H$, with polynomial classical advice, then achieves a
bias exceeding the ceiling $C\sqrt{(L+a+2)/n^{1/3}}=O(\poly(\lambda)\cdot n^{-1/6})$ that
\Cref{lem:vsp-small} imposes on every classical transcript of length $L=\poly(\lambda)$ together with
$a$ bits of advice about the same secret.
\end{itemize}
\end{proposition}

\begin{proof}
The two cases run in opposite directions.  In case~(A) we push the conjecture forward through the
reductions of \Cref{sec:single}; in case~(B) we pull a hypothetical attack backwards through them,
into the standard position the statement describes.

Suppose first that the polynomial form of \Cref{conj:lift} holds.  Then the indistinguishability half
of \Cref{thm:single} under fully coherent access follows by the reductions of \Cref{sec:single}, none
of which depends on the access mode.  There are three of them.

\begin{enumerate}[label=(\arabic*),leftmargin=2.2em,itemsep=3pt,topsep=3pt]
\item \Cref{lem:fresh} collapses any interaction with a fixed oracle into one effect, and replaces
  the stored samples by fresh ones at the stated cost.
\item The resulting experiment is exactly the one the conjecture bounds: fresh copies and challenge,
  the oracle $\cO'$ as the Boolean function of $X=(H,S)$, and the distinguisher's advice as the advice
  parameter.  The pointwise maximizing advice is Borel here, because the advice set is finite and each
  fixed-advice bias is Borel, so the lexicographically first maximizer is a Borel map.
\item \Cref{lem:fixing}, over the same countable family of machines as in \Cref{cor:hybrid-fixed},
  fixes one oracle.
\end{enumerate}

The no-$\OWPuzz$ half needs nothing further, since the attack of \Cref{app:attack} is classical and
hence a special case of coherent access.  That is case~(A).

Suppose instead that the polynomial form fails.  Then some adversary with polynomially bounded
$T,\bar q,a,\log M$ has bias exceeding $\lambda^{-k}$ for some $k$ and infinitely many $\lambda$.  Since
$n=2^{\Theta(\lambda)}$, an inverse polynomial in $\lambda$ is larger than $n^{-c}$ for every $c>0$, so
that bias exceeds $n^{-c}$ for every $c>0$ as well.

Three reductions now bring the adversary to the displayed form, and their order matters.
\begin{enumerate}[label=(\arabic*),leftmargin=2.2em,itemsep=3pt,topsep=3pt]
\item \emph{The hidden object is reduced to $H$.}  The stored samples are generated from $H$ together
  with auxiliary randomness $\omega$ independent of $H$, as $S=S(H,\omega)$.  Since the bias averaged
  over $\omega$ exceeds the threshold, some fixed $\omega$ does at least as well.  Fixing it makes the
  truth table $f_{(H,S(H,\omega))}$ and the advice $\alpha_{(H,S(H,\omega))}$ Borel functions of $H$
  alone, which is the hypothesis \Cref{lem:copies-free} requires.
\item \emph{The reference copies are removed.}  \Cref{lem:copies-free} removes them and retains the
  advice map while doing so, at a cost of $2\bar q/n$.  That cost is at most $\frac12n^{-c}$ for every
  $c<1$ and all large $\lambda$, since $\bar q=\poly(\lambda)$ while $n=2^{\Omega(\lambda)}$.  The
  residual dimension is $n'=n-2\bar q\ge n/2$, so the bias is at least $(n')^{-c'}$ for any $c'>c$;
  and as $c$ was arbitrary, the stated conclusion follows after renaming.
\item \emph{The residual instance is put in standard position.}  The spans that
  \Cref{lem:copies-free} conditions on are fixed by the same averaging, and the residual space $K$ is
  identified with $\bbR^{n'}$ by a fixed isometry, under which $H_A$ remains Haar of half dimension.
  The advice survives as the strings $\alpha_H$ of the statement.
\end{enumerate}

That leaves putting the surviving adversary into the displayed form. Holding one challenge copy, its
bias is $\frac1n|\tr(E_{f_H,\alpha_H}R_H)|$ by \eqref{eq:bias-def}, which is the quantity displayed in
case~(B), and the comparison with the classical ceiling is then \Cref{lem:vsp-small} read at
$L=\poly(\lambda)$.
\end{proof}

Case~(A) is the separation we are after.  Case~(B) would not by itself break $\cO'$, whose coherent
security would remain open.  It would, however, contrast sharply with Klartag--Regev: with the
same secret and the same holder of $H$ available, the passage from classical to coherent access would
turn a provable impossibility into an attack.  It would not resolve the Unitary
Synthesis Problem of Aaronson and Kuperberg~\cite{AK07} positively, and we see three reasons why.
It gives inverse-polynomial bias where synthesis asks for inverse-exponential diamond error; it
concerns a single ensemble where the problem quantifies over all unitaries; and it is query-efficient
where the problem asks for circuit efficiency.  It would instead refute a non-synthesis statement for \emph{measurement} synthesis on
the half-subspace ensemble, which~\cite{LMW24} single out as the hard core of unitary synthesis.

\paragraph{The quantum protocol for $\VSP$.}  The $O(\log n)$-qubit quantum protocol
for $\VSP$ does not break the pair under coherent access.  The power of that protocol lies in the
measurement rather than in the message: Alice sends $\ket u$,
and it is Bob, who holds $H$, who applies $\{P_H,P_{H^\perp}\}$.  Neither of our oracles performs
that measurement on the challenge, since $\cR$ only emits states and $\Count$ returns classical bits.
The adversary, who plays the part of Alice, cannot make the oracle act as the measuring
Bob.

The difference is one of direction.  Our correspondence sends a query to a message, and that is the
direction the lower bound uses.  The reverse direction, from a cheap quantum protocol to an attack,
would require a reflection oracle, a quantum-input channel, or some other means of letting the holder
of $H$ act on the challenge; neither $\cG$ nor $\Count$ is such a means, and neither was $\cR$ in the
analysis model.  A cheap quantum protocol rules out a communication-based security
\emph{proof} under coherent access, which is the reason for the query-complexity method of
\Cref{sec:coherent}, but it yields no attack by itself.

If we grant the reflection oracle named there, the pair breaks at once, which makes the distinction
just drawn a substantive one.  Suppose an adversary can apply the \emph{controlled} unitary $R_H$
using $O(1)$ oracle queries.  Since $P_H^2=P_H$ and $P_HP_{H^\perp}=0$ give $R_H\rho_0=\rho_0$ and
$R_H\rho_1=-\rho_1$, the challenge $\rho_b$ lies entirely in the $(-1)^b$-eigenspace of $R_H$.
Applying $R_H$ alone is useless, since it multiplies $\rho_b$ by a global sign; the control is what
converts that sign into an observable bit.  Prepare an ancilla in $\ket+$, apply controlled-$R_H$ to
the challenge, Hadamard the ancilla and measure it: the outcome is deterministic and equals $b$.

Two oracles supply the controlled reflection.  The subspace reflection $\Id-2P_H=-R_H$, given in
controlled form, supplies it in one query, the extra $Z$ on the control merely flipping the outcome
label.  Oracle access to the defining representation $Q$ of $\On$ together with $Q^\top$ supplies it
in two, since $R_H=QDQ^\top$ with $D$ public: apply $Q^\top$, then the controlled public $D$, then
$Q$, and note that $QQ^\top=\Id$ on the control-zero branch.  Neither oracle is a Boolean function of
$H$ answering in classical bits, which is the only kind \Cref{conj:lift} concerns.

\section{Discussion and open problems}
\label{sec:discussion}

Every classical resource we have granted the distinguisher has a communication counterpart, and
security against that resource is the hardness of $\VSP$ against its counterpart.  Adaptive classical
$\Count$ queries are interactive classical communication, and oracle-dependent classical advice is a
one-way message from the holder of $H$.  Both lie below the classical $\VSP$ threshold, and both are
provably useless.  Not perfectly useless, however, since the oracle does leak: at the rate
$\Theta(n^{-1/2})$ for the measure-then-count adversaries of \Cref{prop:exact-rate}, and at most
$\poly(\lambda)\,n^{-1/6}$ in general.  Quantum advice, finally, is a one-way \emph{quantum} message
to the holder of one challenge copy, and there the threshold is $m=\Theta(n^2)$ qubits, which is both
necessary and sufficient (\Cref{prop:advice-sharp}, \Cref{sec:advice-is-a-query}).

Coherent queries, by contrast, have no useful communication counterpart.  Of course, one can simulate
such a query by sending the query register to the holder of $H$ and back, but that is quantum
communication, and for quantum communication $\VSP$ is exponentially easy.  No reduction of our kind reaches
them, and \Cref{sec:coherent} argues in query complexity instead.  Just one resource crosses the two
regimes, namely a polynomial classical transcript followed by a single coherent query, and that is
the access model of \Cref{thm:single}.

We conclude with four open problems.

First, fully coherent access.  \Cref{conj:lift} asks whether the pair survives polynomially many
adaptive coherent queries.  An affirmative answer would place the separation of \Cref{thm:single} in
the standard quantum oracle model, with no restriction on where the distinguisher's queries fall, and
none on its advice either (\Cref{sec:lift}).  Two facts delimit an attempt.  The first is that the
question lives at $T\ge2$ with every width at $n^{2-o(1)}$ and above, by \Cref{prop:small-width} read
at each fixed $\eps$ and by \Cref{prop:asym-two-query}.  That is also the regime the model forces, so
a union bound over truth tables is unavailable, and the maximum over them must instead be taken by an
operator norm, as it is at $T=1$.  The second fact is that the first open case is concrete.  Two coherent queries contain quantum advice
together with one coherent query (\Cref{sec:lift}), which is \Cref{thm:hybrid} with its classical
record replaced by advice.  The advice on its own is already handled (\Cref{cor:quantum-advice}), and
\Cref{prop:advice-one-query} settles the pair of them for advice of dimension $n^{2-\Omega(1)}$.  What
remains open is advice of larger dimension together with a query.  What limits the argument there is the conditioning step rather than the
construction.  \Cref{thm:hybrid} conditions on a classical record and averages against the masses
$w_\pi$ of the consistent records, which sum to one for every $H$ (\Cref{lem:record}).  Quantum
advice supplies no record, and two coherent queries supply none either, so the convex combination
that makes the conditioning harmless has no counterpart.

Second, the classical communication complexity of $\VSP$ itself.  Our proof uses only
$\VSP_n\in\Omega(n^{1/3})$, while the known window is $[n^{1/3},\sqrt n]$ with Raz's protocol at the
top~\cite{Raz99,Mon19}.  Both endpoints are $2^{\Omega(\lambda)}$, so the separation is unaffected,
but the true threshold would locate the exact exponential rate at which classical counting begins to
compromise quantum state hiding.  Klartag and Regev show that their sampling theorem is tight, a
spherical cap of measure $\exp(-n^{1/3})$ deviating with probability $\exp(-n^{1/3})$, and conclude
that reaching $\Omega(\sqrt n)$ ``is probably impossible using the rectangle bound''~\cite{KR11}.
That assessment is a heuristic drawn from the example rather than a proved barrier, and the example
is circumvented by the geometric-mean form of the statement, which the same cap satisfies
(\Cref{sec:oneway-rate}).  A second and related gap is quantitative.  The security proof lives in the
small-advantage regime, where \Cref{lem:vsp-small} gives $O(\sqrt{L+1}\,n^{-1/6})$ and
\Cref{cor:twoway-disc} gives $O(2^{L/2}n^{-1/2})$, and where by \Cref{prop:exact-rate} no bound can
improve on $n^{-1/2}$ at $L=\poly(\lambda)$.  A two-way analogue of \Cref{prop:oneway-rate}, linear in $L$ at rate $O((L+1)n^{-1/2})$, would
sharpen \Cref{thm:EFI-security} to its exact exponent and would in
particular give $\Omega(\sqrt n)$ at constant advantage, matching Raz.  Whether a spectral argument of that kind
survives interaction is open.  It is not a rectangle bound, so the obstruction identified above does
not immediately apply to it.

Third, the classical communication complexity of $k$-$\VSP_n$, in which Alice receives $k$ unit
vectors drawn independently from the sphere of $H$ or from that of $H^\perp$.  This is no longer a
question the construction needs: security of the pair against $k$ challenge copies follows from the
single-copy statement by the hybrid of \Cref{cor:multicopy}, at a factor $k$, because the access
model supplies reference copies of both states.  As a communication question it remains open.  The
obvious reduction runs in the unhelpful direction: from $k$ vectors Alice can produce one uniform
vector of the same law, by taking a uniformly random unit vector in their span, so $k$-$\VSP$ is at
least as easy as $\VSP$ and hardness does not transfer.  We are not aware of a lower bound for any
$k\ge2$.  The gap between the two statements is the one-directionality of \Cref{sec:framework}:
Alice holds classical descriptions, and a receiver of $k$ quantum copies does not.

Fourth, quantum advice.  Linearly many qubits are admitted alongside the full classical budget
(\Cref{cor:hybrid-qadvice}), and polynomial size is admitted when the advice stands alone
(\Cref{cor:quantum-advice}); what is open is polynomial size together with classical $\Count$
queries.  What limits the argument is the ordering restriction in \Cref{thm:hybrid} rather than a
missing communication bound (\Cref{sec:security-thm}), and \Cref{sec:lift} gives the one route that
closes the case.

The nearest communication question, a one-way \emph{quantum} message from the holder of $H$ to a
holder of the classical $u$, is settled: \Cref{prop:reverse-oneway} puts it at $\Theta(n)$, by a
reduction to random access coding.  What our advice statements need is the other model, in which the
receiver holds one quantum copy of the challenge rather than a classical description, and there
\Cref{prop:advice-sharp} gives the threshold $m=\Theta(n^2)$.  The two models separate by the factor
$n$ between these answers (\Cref{sec:advice-is-a-query}).  What remains open is neither of them, but
the mixed resource: quantum advice \emph{together with} adaptive classical $\Count$ queries, which
corresponds to an initial quantum message followed by two-way classical interaction.

Two further questions we leave aside.  Whether public-key quantum money with \emph{mixed} banknotes
survives the oracle is open, since the reduction of~\cite{KT25} produces a mini-scheme with a pure
banknote and the mixed notion of~\cite{AarChr12} is not covered (\Cref{sec:no-owpuzz}).  Whether
the counting oracle can be made uniform, rather than defined by a recursion on ranks, we have not
pursued; nothing in the security proof would change, since \Cref{thm:transcript-security}
is indifferent to how the truth table is specified.

\clearpage
\section*{Acknowledgements}
We thank Keshav Bhateja and Ezekiel Cochran for useful discussions on one-way puzzles and EFI pairs.
Large language models were used in preparing this paper: Claude Opus 4.8 and ChatGPT 5 throughout,
and Claude Opus 5 in the later stages of writing and to check proofs and computations.  Their role
was most substantial in the proofs of \Cref{sec:coherent} and \Cref{app:rates}, in particular in
adapting the concentration inequality of Huang and Tropp to our ensemble and in the harmonic analysis
on the sphere.  We verified every argument line by line, checked the originality of the results and
every reference against its source, and take full responsibility for all content.
This work is supported by the author's faculty startup grant from Virginia Tech.

\printbibliography

\newpage
\appendix
\crefalias{section}{appendix}\crefalias{subsection}{appendix}
\addtocontents{toc}{\protect\setcounter{tocdepth}{1}}

\section{The conditional-sampling attack}
\label{app:attack}

This appendix proves \Cref{thm:no-owpuzz}.  We use the conditional-sampling attack that underlies
the forward direction of the characterization of~\cite{CGGH25}: given a puzzle $s$, sample a key from
the true conditional distribution of keys given $s$, one bit at a time, computing each conditional
probability from output-probability estimates.  In the general setting that estimator is weak and the
bookkeeping is delicate.  Here, by contrast, the counting oracle provides \emph{exact bits}, so
additive error $2^{-t}$ costs $t$ queries and our analysis stays elementary.

Let $(\Samp,\Ver)$ be a candidate one-way puzzle relative to $\cO$ with correctness.  We may assume
without loss of generality that the key has fixed length $m=m(\lambda)$ and the puzzle length
$\ell=\ell(\lambda)$, both polynomial. Since $\Samp$ is a fixed polynomial-time machine, the ranks of its $\Count$ queries are bounded by an
integer $J=J(\lambda)$, and it accesses $\cO_{<J+1}$ only.  We read the integer $J+1$ off its clocked
description. For $0\le i\le m$ let $C_{i,\lambda}$ be the clocked specialization that runs
$\Samp(1^\lambda)$ to get $(k,s)$ and outputs $(s,k_1\cdots k_i)$, and put
$\widehat C_{i,\lambda}\eqdef\mathsf{Norm}_J(C_{i,\lambda})$. Each of these is an admissible argument to $\Count_{J+1}$, by the reasoning of \Cref{lem:pe-easy}.
Finally, we write $P(w)\eqdef\Pr[\widehat C_{i,\lambda}^{\cO}(1^\lambda)=w]$ for $w=(s,y)$ with
$y\in\{0,1\}^i$, so that $P(w)=P(w0)+P(w1)$.

The adversary is the following.  Fix the precision $t\eqdef m+\ell+3\lambda$ and put
$\eta\eqdef2^{-t}$.  On input $(1^\lambda,s)$ it produces the key one bit at a time.  Having built the
prefix $w=(s,k_1\cdots k_{i-1})$, it determines the next bit $k_i$ as follows.
\begin{itemize}[leftmargin=1.6em,itemsep=2pt,topsep=3pt]
\item It queries $\Count_{J+1}$ for the first $t$ bits of $P(wy)$ for $y\in\{0,1\}$, obtaining
  truncations with $0\le P(wy)-\widetilde P(wy)\le\eta$.
\item If $\widetilde P(w0)+\widetilde P(w1)=0$, it sets $k_i\eqdef0$.
\item Otherwise it sets $k_i\eqdef y$ with probability
  $\widetilde P(wy)/(\widetilde P(w0)+\widetilde P(w1))$, rounded to a multiple of $2^{-2t}$ so that
  $2t$ fair coins realize it exactly.
\end{itemize}
After $m$ rounds it outputs $k$.  This is a classical randomized algorithm making
$2mt=\poly(\lambda)$ classical queries.

\begin{lemma}[The attacker's key distribution is close to the honest one]
\label{lem:tv}
Let $\cD_{\mathrm{real}}$ be the law of $(s,k)\gets\Samp(1^\lambda)$ and $\cD_\cA$ the law of $(s,\cA(1^\lambda,s))$ with $s$ from the puzzle marginal, where $\cA$ is the adversary above with parameters $t=m+\ell+3\lambda$ and $\eta=2^{-t}$. Then $\TV(\cD_\cA,\cD_{\mathrm{real}})\le4m\sqrt{\eta\,2^{\ell+m}}\le2^{-\lambda}$ for all large $\lambda$.
\end{lemma}

\begin{proof}
Fix a prefix $w$, write $\pi_w(y)=P(wy)/P(w)$ for the true conditional law of the next key bit, and
write $\widetilde\pi_w$ for the adversary's rule.  These are the two laws of \Cref{lem:ratio}, at
$p_y=P(wy)$ and $\tilde p_y=\widetilde P(wy)$, so on prefixes carrying mass $P(w)\ge4\eta$ that lemma
gives
\[
  \TV(\pi_w,\widetilde\pi_w)\ \le\ \frac{4\eta}{P(w)} .
\]
Elsewhere we use only the trivial bound $\TV\le1$.

Now replace the adversary's rule by the true conditional law, one bit at a time.  For $0\le j\le m$
let $\cD_j$ draw $(s,k_1\cdots k_j)$ from the real process and the remaining bits by $\cA$'s rule, so
that $\cD_m=\cD_{\mathrm{real}}$ and $\cD_0=\cD_\cA$.  Adjacent hybrids differ only in the $(j+1)$-st
bit, conditioned on a prefix distributed as in the real process, so
\[
  \TV(\cD_j,\cD_{j+1})\ \le\ \sum_wP(w)\,\TV(\pi_w,\widetilde\pi_w) ,
\]
where $w$ ranges over at most $2^{\ell+j}$ prefixes.

That sum is what \Cref{lem:threshold} bounds, at $K=2^{\ell+j}$ and with $V\equiv1$.  Its two
hypotheses hold: $\sum_wP(w)=1$, and the displayed per-prefix estimate is the required one, while
$K^{-1}\ge4\eta$ comfortably, since $\eta=2^{-(m+\ell+3\lambda)}$.  Hence
\[
  \TV(\cD_j,\cD_{j+1})\ \le\ 2\sqrt{4\eta\,2^{\ell+j}}\ \le\ 4\sqrt{\eta\,2^{\ell+m}} .
\]
Summing over the $m$ steps,
\[
  \TV(\cD_\cA,\cD_{\mathrm{real}})\le4m\sqrt{\eta\,2^{\ell+m}}=4m\,2^{-3\lambda/2}\le2^{-\lambda}
\]
for all large $\lambda$, as claimed.
\end{proof}

\begin{proof}[Proof of \Cref{thm:no-owpuzz}]
The event ``$\Ver$ accepts'' is a fixed measurable, possibly uncomputable, function of $(k,s)$, and
$\cA$ never runs $\Ver$, so no assumption on the verifier enters anywhere. Under
$\cD_{\mathrm{real}}$ that event has probability at least $1-\negl(\lambda)$ by correctness, hence at
least $1-\negl(\lambda)-2^{-\lambda}$ under $\cD_\cA$ by \Cref{lem:tv}.

Two further properties of the argument are used by later statements.
First, it holds for every fixing of $\{H_\lambda\}$, because \Cref{lem:pe-easy} does.  Second, $\cA$
is classical and makes classical queries, which is what gives \Cref{cor:collateral}.

To summarize, a classical randomized polynomial-time adversary making $\poly(\lambda)$ classical
queries produces, from the puzzle alone, a key that $\Ver$ accepts with probability
$1-\negl(\lambda)$.  That contradicts security.
\end{proof}

\section{The two rates}
\label{app:rates}

This appendix proves \Cref{prop:exact-rate} and the upper bound of \Cref{prop:oneway-rate}, the two
statements of Theorem~E.  Neither is needed for the separation, and the two arguments are
independent of one another.

\subsection{The exact rate against a measure-then-count adversary}
\label{app:exact-rate}

We isolate the Beta computation first, in the three ensembles the paper needs it in, and then read
both halves of \Cref{prop:exact-rate} off it.

\begin{lemma}[Mean absolute deviation of a diagonal entry]
\label{lem:mad}
For a unit vector $v$ put $\xi_v\eqdef\bra vP_H\ket v-\frac12$, so that $\bra vR_H\ket v=2\xi_v$.
\begin{enumerate}[label=(\alph*),leftmargin=2.2em]
\item If $H\subseteq\bbR^n$ is Haar of dimension $n/2$ and $v\in\bbR^n$, then
  $\bra vP_H\ket v\sim\mathrm{Beta}(a,a)$ with $a=n/4$, and for $n\ge4$
  \begin{equation}
    \frac1{\sqrt{2\pi n}}\ \le\ \frac{1}{\sqrt{\pi n}}\sqrt{1-\tfrac2n}\ \le\ \E_H|\xi_v|\ \le\ \frac1{\sqrt{\pi n}} .
    \label{eq:beta-mad}
  \end{equation}
  The middle bound is the sharp one: $\E_H|\xi_v|=(\pi n)^{-1/2}\big(1-\tfrac1{2n}+O(n^{-2})\big)$,
  so the two outer bounds of~\eqref{eq:beta-mad} differ by a factor $\sqrt2$ while the true value lies
  at the top of the interval.
\item If $H\subseteq\bbC^n$ is Haar of dimension $n/2$ and $v\in\bbC^n$, the same holds with $a=n/2$,
  giving $\sqrt{1-1/n}\,(2\pi n)^{-1/2}\le\E_H|\xi_v|\le(2\pi n)^{-1/2}$.
\item In the real ensemble the upper bound of~\eqref{eq:beta-mad} holds for every \emph{complex} unit
  vector $v\in\bbC^n$ as well.
\end{enumerate}
\end{lemma}

\begin{proof}
Each part is the mean absolute deviation of a symmetric Beta variable about its mean, read at a
different parameter, so we compute that deviation once.  Throughout this subsection $a$ denotes that
Beta parameter and nothing else.  Let $X$ be Beta distributed with both
parameters equal to $a$, so that its density on $[0,1]$ is proportional to $x^{a-1}(1-x)^{a-1}$ and its
law is symmetric about $\frac12$.  Substitute $x=\frac{1+u}2$.  The density becomes proportional to $(1-u^2)^{a-1}$ on $[-1,1]$, with
normalizing constant $2\cdot4^{a-1}B(a,a)$, and $|X-\frac12|=\frac{|u|}2$, so
\[
  \E\Big|X-\tfrac12\Big|\ =\ \frac1{4^{a}\,B(a,a)}\int_{-1}^{1}|u|\,(1-u^{2})^{a-1}\,du
  \ =\ \frac1{a\,4^{a}\,B(a,a)} ,
\]
the last step because $w=u^2$ turns the integral into $\int_0^1(1-w)^{a-1}dw=\frac1a$.

Now clear the Beta function.  Legendre duplication,
$\Gamma(2a)=2^{2a-1}\pi^{-1/2}\Gamma(a)\Gamma(a+\frac12)$, turns
\[
  B(a,a)\ =\ \frac{\Gamma(a)^2}{\Gamma(2a)}
  \qquad\text{into}\qquad
  B(a,a)\ =\ \frac{\sqrt\pi\,\Gamma(a)}{2^{2a-1}\,\Gamma(a+\frac12)} ,
\]
and substituting that gives
\begin{equation}
  \E\Big|X-\tfrac12\Big|\ =\ \frac{\Gamma(a+\frac12)}{2a\sqrt\pi\,\Gamma(a)} .
  \label{eq:beta-exact}
\end{equation}
At $a=1$ this returns $\frac14$, the value for the uniform law on $[0,1]$, which is a useful check.

One tool controls the Gamma ratio in \eqref{eq:beta-exact} from both sides, namely log-convexity of
$\Gamma$:
\[
  \Gamma\big(x+\tfrac12\big)^2\ \le\ \Gamma(x)\,\Gamma(x+1)\ =\ x\,\Gamma(x)^2
  \qquad (x>0) .
\]
Read it at $x=a$ and it gives the upper bound at once,
$\Gamma(a+\frac12)/\Gamma(a)\le\sqrt a$.  Read it instead at $x=a+\frac12$ and it gives
\[
  \Gamma(a+1)^2\ \le\ \big(a+\tfrac12\big)\,\Gamma\big(a+\tfrac12\big)^2 ,
  \qquad\text{that is}\qquad
  \frac{\Gamma(a+\frac12)}{\Gamma(a)}\ \ge\ \frac a{\sqrt{a+\frac12}}\ \ge\ \sqrt{a-\tfrac12} ,
\]
the last step because $a^2\ge(a-\frac12)(a+\frac12)$.  The two readings sandwich the ratio between
$\sqrt{a-\frac12}$ and $\sqrt a$, and hence
\begin{equation}
  \frac{\sqrt{a-1/2}}{2a\sqrt\pi}\ \le\ \E\Big|X-\tfrac12\Big|\ \le\ \frac1{2\sqrt{\pi a}} .
  \label{eq:beta-bounds}
\end{equation}
It remains to identify the Beta parameter in each of the three ensembles.

\emph{(a)} By rotation invariance we may instead fix $H$ and take $v$ uniform on the sphere.  Then
$\bra vP_H\ket v=\sum_{i\le n/2}v_i^2$ has the form $U/(U+V)$, where $U$ and $V$ are independent
chi-square variables with $n/2$ degrees of freedom each; hence
$\bra vP_H\ket v\sim\mathrm{Beta}(n/4,n/4)$.  Setting $a=n/4$ in \eqref{eq:beta-bounds}, the upper
bound is exactly $1/\sqrt{\pi n}$ and the lower bound is
\[
  \frac{\sqrt{n-2}}{n\sqrt\pi}\ =\ \frac1{\sqrt{\pi n}}\sqrt{1-\frac2n} ,
\]
which is at least $1/\sqrt{2\pi n}$ precisely when $n\ge4$; this is \eqref{eq:beta-mad}.  For the asymptotic statement we go back to the exact value \eqref{eq:beta-exact}, the two-sided
bounds \eqref{eq:beta-bounds} being too coarse.  The standard expansion of the Gamma ratio,
\[
  \frac{\Gamma(a+\frac12)}{\Gamma(a)}\ =\ \sqrt a\,\Big(1-\frac1{8a}+O(a^{-2})\Big),
\]
turns \eqref{eq:beta-exact} into
$\E|X-\frac12|=\tfrac12(\pi a)^{-1/2}\big(1-\frac1{8a}+O(a^{-2})\big)$, which at $a=n/4$ reads
$(\pi n)^{-1/2}\big(1-\frac1{2n}+O(n^{-2})\big)$.

\emph{(b)} The complex ensemble is the same computation with the degrees of freedom doubled.  The
numerator now has $n$ real degrees of freedom out of $2n$ rather than $n/2$ out of $n$, so the Beta
parameter is $a=n/2$, and the stated bounds follow from \eqref{eq:beta-bounds} exactly as in part~(a).

\emph{(c)} A complex vector tested against a real subspace has no Beta law of the kind used in
part~(a), and convexity replaces it.  Write $v=x+iy$ with $x,y\in\bbR^n$ and $\norm x^2+\norm y^2=1$.
Since $P_H$ is real symmetric the cross terms cancel, so that
\[
  \bra{v}P_H\ket{v}\ =\ x^\top P_Hx+y^\top P_Hy ,
  \qquad\text{and hence}\qquad
  \xi_v\ =\ \norm x^2\,\xi_{\hat x}+\norm y^2\,\xi_{\hat y}
\]
for the unit vectors $\hat x=x/\norm x$ and $\hat y=y/\norm y$, where we omit any term of weight
zero.  Thus $\xi_v$ is a convex combination of two instances of the real case, and the triangle inequality together with rotation invariance gives
\[
  \E|\xi_v|\ \le\ \norm x^2\,\E|\xi_{\hat x}|+\norm y^2\,\E|\xi_{\hat y}|\ =\ \E|\xi_{e_1}| ,
\]
which part~(a) bounds by $1/\sqrt{\pi n}$, the upper bound of \eqref{eq:beta-mad}.  Only that bound
survives the convexity step, and it is all that \Cref{prop:exact-rate}(i) uses.
\end{proof}

\begin{proof}[Proof of \Cref{prop:exact-rate}]
\emph{(i)} Fix a POVM $\{E_k\}$.  The adversary sees the outcome $k$ together with every value
$\tr(E_{k'}P_H)$, so its optimal rule is the likelihood-ratio test between the two outcome
distributions, whose probabilities are $\frac2n\tr(E_kP_H)$ on branch $b=0$ and
$\frac2n\tr(E_kP_{H^\perp})$ on branch $b=1$.  The advantage of that test is the total variation
distance between them, which by $R_H=P_H-P_{H^\perp}$ equals
\[
  \frac1n\sum_{k}\big|\tr(E_kP_H)-\tr(E_kP_{H^\perp})\big|\ =\ \frac1n\sum_k\big|\tr(E_kR_H)\big| .
\]

We bound each summand separately, using nothing about the POVM beyond its trace.
Diagonalize the effect as $E_k=\sum_iw_{k,i}\proj{v_{k,i}}$ with $w_{k,i}\ge0$, noting that
the eigenvectors may be complex, since nothing constrains the POVM to be real.  Expanding the trace and
applying the triangle inequality,
\[
  \big|\tr(E_kR_H)\big|\ =\ \Big|\sum_iw_{k,i}\bra{v_{k,i}}R_H\ket{v_{k,i}}\Big|
  \ \le\ \sum_iw_{k,i}\,\big|\bra{v_{k,i}}R_H\ket{v_{k,i}}\big| .
\]
Part~(c) of \Cref{lem:mad} covers those complex eigenvectors, and this is the only place where we
use it.  Since $\bra vR_H\ket v=2\xi_v$, it gives $\E_H|\bra vR_H\ket v|\le2/\sqrt{\pi n}$ for every
unit vector $v$.  Using $\sum_iw_{k,i}=\tr E_k$ and $\sum_k\tr E_k=\tr\Id=n$,
\[
  \E_H\Big[\frac1n\sum_k\big|\tr(E_kR_H)\big|\Big]
  \ \le\ \frac1n\cdot\frac2{\sqrt{\pi n}}\sum_k\tr E_k\ =\ \frac2{\sqrt{\pi n}} ,
\]
which is the bound of~(i).

\emph{(ii)} Write $d_j\eqdef\bra{e_j}P_H\ket{e_j}$ and $\xi_j\eqdef d_j-\frac12$.  Measuring $\rho_b$
in the computational basis returns the outcome $j$ with probability $2d_j/n$ if $b=0$ and $2(1-d_j)/n$
if $b=1$.  Suppose for the moment that the adversary could apply the exact rule
$\mathbf 1[d_j<\frac12]$.  Its advantage would then be
\begin{equation}
  \Big|\frac2n\sum_j\big((1-d_j)-d_j\big)\,\mathbf 1[\xi_j<0]\Big|
  \ =\ \frac4n\sum_j|\xi_j|\,\mathbf 1[\xi_j<0] .
  \label{eq:exact-rule}
\end{equation}
Three facts evaluate the expectation of the right-hand side: there are $n$ indices; each $\xi_j$ has
the same law as $\xi_{e_1}$, by rotation invariance; and that law is symmetric about $0$, because
$\mathrm{Beta}(a,a)$ is symmetric about $\frac12$, so that
$\E\big[|\xi_j|\,\mathbf 1[\xi_j<0]\big]=\frac12\E|\xi_j|$.  Together they give
\[
  \E_H\Big[\frac4n\sum_j|\xi_j|\,\mathbf 1[\xi_j<0]\Big]
  \ =\ \frac4n\cdot n\cdot\frac12\,\E|\xi_{e_1}|\ =\ 2\,\E|\xi_{e_1}|
  \ \ge\ \frac2{\sqrt{\pi n}}\sqrt{1-\tfrac2n}
\]
by \Cref{lem:mad}(a).  The same lemma gives $\E|\xi_{e_1}|=(\pi n)^{-1/2}\big(1-O(n^{-1})\big)$,
while part~(i) is $2\sup_v\E|\xi_v|$ up to the POVM normalization; so the exact rule
attains the upper bound of~(i) to within a factor $1-O(n^{-1})$.  That is the asymptotic claim, once we check
that the truncation below costs less.

Now we pay for the finite precision, since the adversary in the statement runs a truncated rule
rather than the exact one.  It does not learn $p_j\eqdef2d_j/n$ exactly, but only its truncation
$\widetilde p_j$ to $t$ bits, and with $\eta\eqdef2^{-t}$ that truncation satisfies
$\widetilde p_j\in[p_j-\eta,p_j]$.  The rule in the statement outputs $1$ exactly when
$\widetilde p_j<\frac1n-\eta$.

The two rules agree outside a narrow window.  Indeed, if $p_j<\frac1n-2\eta$ then
$\widetilde p_j\le p_j<\frac1n-\eta$, so the truncated rule outputs $1$; and if $p_j\ge\frac1n$ then
$\widetilde p_j\ge p_j-\eta\ge\frac1n-\eta$, so it outputs $0$.  They can differ only in between, at
the indices with
\[
  \tfrac1n-2\eta\ \le\ p_j\ <\ \tfrac1n ,\qquad\text{that is}\qquad |\xi_j|\ \le\ n\eta .
\]
Each such index changes the sum in \eqref{eq:exact-rule} by at most $\frac4n\cdot n\eta=4\eta$, and
there are at most $n$ indices altogether, so the total loss is at most
\[
  4n\eta\ \le\ 4n\cdot n^{-2}2^{-\lambda}\ =\ \frac4n\,2^{-\lambda} .
\]
Subtracting that loss from the advantage of the exact rule gives the bound of~(ii).  Relative to the
$2(\pi n)^{-1/2}$ of part~(i) this loss is $2\sqrt\pi\,2^{-\lambda}n^{-1/2}$, which at
$n=2^{\lambda}$ is $O(n^{-3/2})$ and so is dominated by the $O(n^{-1})$ of the previous paragraph;
that is the asymptotic claim in full.
\end{proof}

\subsection{The one-way rate: harmonic analysis on the sphere}
\label{app:harmonics}

The space $L^2(S^{n-1},\sigma)$
decomposes as an orthogonal direct sum $\bigoplus_{\ell\ge0}\cH_\ell$ of spaces of \emph{spherical
harmonics} of degree $\ell$, the restrictions to the sphere of harmonic polynomials homogeneous of
degree $\ell$.  Each $\cH_\ell$ is invariant and irreducible under the action of $\On$ by rotation.  We write $P_\ell$ for the orthogonal projection onto $\cH_\ell$ and
\[
  W_\ell(f)\ \eqdef\ \norm{P_\ell f}_2^2
\]
for the mass of $f$ at degree $\ell$, so that Parseval reads $\sum_\ell W_\ell(f)=\norm f_2^2$.  For
an indicator $\mathbf 1_A$ that mass is $\sigma(A)$.  Two standard facts are used.

The first is the Funk--Hecke theorem~\cite{Helgason}.  If an operator on $L^2(S^{n-1})$ commutes with all rotations,
then by Schur's lemma it acts on each $\cH_\ell$ as a scalar.  When the operator has the form
$f\mapsto\E[f(u')]$ for $u'$ drawn from a law depending on $u$ only through the inner product
$t=\ip u{u'}$, that scalar is $\E[G_\ell(t)]$, where
\[
  G_\ell\ \eqdef\ C^{\zeta}_{\ell}/C^{\zeta}_{\ell}(1),\qquad \zeta=\tfrac{n-2}2 ,
\]
is the normalized Gegenbauer (ultraspherical) polynomial of degree $\ell$.  We write the index
$\zeta$ rather than the customary $\lambda$, which in this paper is the security parameter.  We use
the classical values
$C^{\zeta}_{2j}(0)=(-1)^j\binom{\zeta+j-1}{j}$ and $C^{\zeta}_{2j}(1)=\binom{2\zeta+2j-1}{2j}$,
the binomial coefficients being the generalized ones since $\zeta$ need not be an integer, together
with the expansion
\[
  C^\zeta_{2j}(t)=\sum_{k=0}^{j}(-1)^{j-k}\frac{\Gamma(\zeta+j+k)}{\Gamma(\zeta)\,(j-k)!\,(2k)!}(2t)^{2k}.
\]
Throughout, $(a)_k\eqdef a(a+1)\cdots(a+k-1)$ denotes the rising factorial and ${}_2F_1$ the Gauss
hypergeometric function.

The second is the heat semigroup $e^{t\Delta}$ generated by the Laplace--Beltrami operator, which acts
on $\cH_\ell$ as multiplication by $e^{-t\ell(\ell+n-2)}$.  The round sphere has Ricci curvature $n-2$
in every direction, so by the Bakry--\'Emery criterion $\sigma$ satisfies a logarithmic Sobolev
inequality with constant $n-2$, and Gross's theorem makes the semigroup
hypercontractive~\cite{Ledoux01,MW82}:
\[
  \norm{e^{t\Delta}f}_{2}\ \le\ \norm f_{p},\qquad p=1+e^{-2(n-2)t} .
\]

\subsection{The second-moment operator}
\label{app:second-moment}

Throughout, $\mu_H\eqdef\nu_H-\nu_{H^\perp}$ denotes the difference of the uniform probability
measures on the unit spheres of $H$ and $H^\perp$.

\begin{lemma}[Spectral form of the second-moment operator]
\label{lem:second-moment}
Let $n\ge4$ be even and let $H\subseteq\bbR^n$ be Haar of dimension $n/2$. Then for every real-valued
$f\in L^2(S^{n-1})$,
\[
  \E_H\Big[\Big(\int f\,d\mu_H\Big)^{2}\Big]
  \;=\;4\sum_{j\ \mathrm{odd}}\ \omega_{2j}\,W_{2j}(f),
  \qquad
  \omega_{2j}\eqdef\prod_{i=1}^{j}\frac{2i-1}{\,n-3+2i\,},
\]
and the degrees $\ell\not\equiv2\pmod4$ contribute nothing. In particular $\omega_2=\frac1{n-1}$, the
$\omega_{2j}$ are strictly decreasing in $j$, and $\omega_{2j}\le\big(\tfrac{2j}{n-1}\big)^{j}$.
\end{lemma}

\begin{proof}
Expand the square and average over $H$, so that the whole expectation becomes a quadratic form in two
rotation-invariant two-point operators.  The argument then has three steps: Funk--Hecke diagonalizes
both operators, the Chu--Vandermonde identity evaluates the resulting Gegenbauer scalars, and at half
dimension the two scalars cancel at every degree $\ell\equiv0\pmod4$.

One preliminary.  For bounded Borel $f$ the averaged quadratic forms
$f\mapsto\E_H[(\int f\,d\nu_H)^2]$ and $f\mapsto\E_H[(\int f\,d\nu_{H^\perp})^2]$ are at most
$\norm f_2^2$ by Jensen, so both sides of the claimed identity extend to $L^2$.

Now expand.  The two square terms agree by the symmetry $H\leftrightarrow H^\perp$, so that
\[
  \E_H\Big[\Big(\int f\,d\mu_H\Big)^{2}\Big]=2\,\langle f,(Q_s-Q_o)f\rangle .
\]
Here $Q_s$ is the two-point operator of a \emph{common} subspace: $Q_sf(u)=\E[f(u')]$, for $u'$
drawn uniformly from the sphere of a Haar $H$ conditioned on $u\in H$, and $Q_o$ is the two-point
operator of \emph{opposite} subspaces, defined in the same way but with $u'$ drawn from the sphere of
$H^\perp$.  Both are self-adjoint and commute with rotations, so by Funk--Hecke each acts on
$\cH_\ell$ as the scalar $\E[G_\ell(t)]$, where $t=\ip u{u'}$.  What is left is to evaluate those two
scalars.

Consider first $Q_o$, where $u'\perp u$ exactly and, given $u$, the vector $u'$ is uniform on the
equator sphere $S^{n-1}\cap u^\perp$. Its scalar is $G_\ell(0)$, which vanishes for odd
$\ell$ and equals $(-1)^{j}\omega_{2j}$ for $\ell=2j$, by the Gegenbauer values recorded in
\Cref{app:harmonics} together with the product form of $\omega_{2j}$ derived at the end of this proof.

Consider next $Q_s$. Given $u$ the conditional law of $u'$ is that of $tu+\sqrt{1-t^2}\,w$ with $w$
uniform on the equator of $u$ and, independently, $t$ distributed as the first coordinate of a Haar
unit vector in dimension $m=n/2$, whose density is proportional to $(1-t^2)^{(m-3)/2}$. So the scalar is
$\E[G_\ell(t)]$, again $0$ for odd $\ell$. For $\ell=2j$ we insert the expansion of $C^\zeta_{2j}$ together with the moments
$\E[t^{2k}]=\frac{(1/2)_k}{(m/2)_k}$.  Using $(2k)!=4^k\,k!\,(1/2)_k$ and $(-j)_k=(-1)^k j!/(j-k)!$,
this gives
\[
  \E\big[C^\zeta_{2j}(t)\big]
  =\frac{(-1)^j\Gamma(\zeta+j)}{\Gamma(\zeta)\,j!}\;
   {}_2F_1\!\Big(-j,\ \zeta+j;\ \frac m2;\ 1\Big)
  =\frac{(-1)^j\Gamma(\zeta+j)}{\Gamma(\zeta)\,j!}\cdot
   \frac{(\frac m2-\zeta-j)_j}{(\frac m2)_j}\ ,
\]
by the Chu--Vandermonde identity, which evaluates ${}_2F_1$ at argument $1$ when its first parameter is
a negative integer.

The half-dimension now enters, and the result rests on the cancellation it produces. At $m=n/2$
we have $\frac m2-\zeta-j=1-\frac n4-j$, so
$(\frac m2-\zeta-j)_j=(-1)^j(\frac n4)_j=(-1)^j(\frac m2)_j$, the ${}_2F_1$ equals $(-1)^j$ exactly,
and
\[
  \E\big[G_{2j}(t)\big]=\frac{\Gamma(\zeta+j)}{\Gamma(\zeta)\,j!\,C^\zeta_{2j}(1)}
  =\big|G_{2j}(0)\big| = \omega_{2j}.
\]
The scalar of $Q_s-Q_o$ at degree $2j$ is $\omega_{2j}-(-1)^j\omega_{2j}$, which is zero for
even $j$ and $2\omega_{2j}$ for odd $j$, and Parseval gives the stated identity.

Now for the product form. Legendre duplication turns the two binomial coefficients into
$(2\zeta)_{2j}=4^{j}(\zeta)_j(\zeta+\frac12)_j$ and $(2j)!=4^{j}j!\,(\frac12)_j$, so that the
ratio collapses to
\[
  \omega_{2j}\ =\ \frac{(\frac12)_j}{(\zeta+\frac12)_j}\ =\ \prod_{i=1}^{j}\frac{i-\frac12}{\zeta+i-\frac12}
  \ =\ \prod_{i=1}^{j}\frac{2i-1}{n-3+2i},
\]
using $\zeta=\frac{n-2}2$. Each factor is less than $1$ and at most $\frac{2j}{n-1}$, which gives the
monotonicity and the last bound, as claimed.
\end{proof}

\paragraph{Two elementary checks on \Cref{lem:second-moment}.}  Only two features of \Cref{lem:second-moment} are used downstream.  The first is the value
$\omega_2=\frac1{n-1}$ at degree two, which dominates every estimate below.  The second is the
vanishing at degrees $\ell\equiv0\pmod4$, which is where the half dimension enters.  We can confirm both directly, without any harmonic analysis.

For degree two take $f(x)=x_1^2-x_2^2$.  For $u$ uniform on the unit sphere of an $m$-dimensional
subspace with projector $P$ we have $\E[uu^\top]=P/m$, so with $m=n/2$,
\[
  \int f\,d\mu_H=\frac{4\,(P_{11}-P_{22})}{n},
  \qquad \norm f_2^2=\frac4{n(n+2)} .
\]
Now $P_{11}\sim\mathrm{Beta}(n/4,n/4)$ has variance $\frac1{2n+4}$, and $\sum_iP_{ii}=\frac n2$ forces
$\E[P_{11}P_{22}]=\frac14-\frac1{(n-1)(2n+4)}$, so that
$\E[(P_{11}-P_{22})^2]=\frac n{(n-1)(n+2)}$ and
\[
  \E_H\Big[\Big(\int f\,d\mu_H\Big)^2\Big]=\frac{16}{(n-1)\,n\,(n+2)}
  \ =\ 4\,\omega_2\,\norm f_2^2\quad\text{at }\omega_2=\tfrac1{n-1} .
\]

For degree four take the harmonic part of $x_1^4$, which on the sphere is
$h(x)=x_1^4-\frac6{n+4}x_1^2+\frac3{(n+2)(n+4)}$.  The same identity for $\E[uu^\top]$, together with
$\E[u_1^4]=\frac{3P_{11}^2}{m(m+2)}$, gives for $d\eqdef P_{11}-\frac12$
\[
  \int h\,d\mu_H=\frac{6d}{m}\Big(\frac1{m+2}-\frac2{n+4}\Big) ,
\]
and at $m=n/2$ the bracket vanishes identically, since then $m+2=\frac{n+4}2$.  So $\int h\,d\mu_H=0$ for \emph{every} $H$, not merely in expectation, whereas the bracket is
nonzero at every other dimension.  The vanishing at degrees $\ell\equiv0\pmod4$ is not an
artifact of the Gegenbauer bookkeeping.  It is the same cancellation that the proof above locates in
the Chu--Vandermonde step.

\subsection{A level-\texorpdfstring{$\ell$}{ell} inequality, and the leakage of one cell}
\label{app:level-ell}

\begin{lemma}[Level-$\ell$ inequality on the sphere]
\label{lem:level-ell}
There are universal $C_0,c_0'>0$ such that for every measurable $A\subseteq S^{n-1}$ with
$\alpha=\sigma(A)\le\frac12$ and every $1\le\ell\le2\ln(e/\alpha)$ with
$\ell\,\ln\!\big(2\ln(e/\alpha)\big)\le c_0'\,n$,
\[
  W_\ell(\mathbf 1_A)\ \le\ \alpha^2\Big(\frac{C_0\ln(e/\alpha)}{\ell}\Big)^{\ell}.
\]
\end{lemma}

\begin{proof}
The argument is hypercontractivity of the heat semigroup, optimized over the diffusion time.  The
semigroup acts on degree-$\ell$ harmonics as multiplication by $e^{-t\ell(\ell+n-2)}$, so for
$f=\mathbf 1_A$ we may undo that damping at level $\ell$ and then apply the hypercontractive estimate
of \Cref{app:harmonics},
\begin{equation}
  W_\ell(\mathbf 1_A)\ \le\ e^{2t\ell(\ell+n-2)}\,\norm{e^{t\Delta}\mathbf 1_A}_2^2
  \ \le\ e^{2t\ell(\ell+n-2)}\,\alpha^{2/p} .
  \label{eq:level-ell-raw}
\end{equation}
This holds at every diffusion time, and it remains to choose one.

Reparameterize by $\eps\eqdef e^{-2(n-2)t}\in(0,1]$, so that $p=1+\eps$, the prefactor becomes
$e^{2t\ell(\ell+n-2)}=\eps^{-\ell(1+\frac{\ell}{n-2})}$, and
$\alpha^{2/p}\le\alpha^2e^{2\eps\ln(1/\alpha)}$. Take
\[
  \eps\ \eqdef\ \frac{\ell}{2\ln(e/\alpha)}\ \le\ 1 ,
\]
which is admissible by the hypothesis $\ell\le2\ln(e/\alpha)$.

Two error factors must then be controlled. The first is $e^{2\eps\ln(1/\alpha)}\le e^{\ell}$,
immediate from the choice of $\eps$. The second is the excess produced by the term $\frac{\ell}{n-2}$
in the exponent of the prefactor,
\[
  \eps^{-\ell^2/(n-2)}\ =\ \exp\Big(\frac{\ell^2}{n-2}\,\ln\frac{2\ln(e/\alpha)}{\ell}\Big)\ \le\ e^{\ell} ,
\]
the last inequality holding by the second hypothesis $\ell\,\ln\!\big(2\ln(e/\alpha)\big)\le c_0'n$.
Substituting both into \eqref{eq:level-ell-raw} and absorbing the resulting factor $e^{2\ell}$ into
$C_0^{\ell}$ gives the claim.
\end{proof}

\begin{lemma}[The leakage of a single cell]
\label{lem:cell-leak}
Let $n\ge4$ be even, let $H\subseteq\bbR^n$ be Haar of dimension $n/2$, let $A\subseteq S^{n-1}$ be
measurable, and put $\alpha\eqdef\sigma(A)$ and $s\eqdef\ln(e/\alpha)$. Then
\[
  \E_H\big|\mu_H(A)\big|\ \le\ 2\alpha
\]
for every $A$, and there are universal $C,c_0>0$ such that if $\alpha\le\frac12$ and $s\le c_0\sqrt n$
then also
\begin{equation}
  \E_H\big|\mu_H(A)\big|\ \le\ C\,\frac{\alpha\,s}{\sqrt n}\ +\ 2\sqrt\alpha\,\theta(s),
  \qquad \theta(s)\eqdef\Big(\frac{4s}{n-1}\Big)^{\lceil s\rceil/2} .
  \label{eq:oneway-cell}
\end{equation}
\end{lemma}

\begin{proof}
The first bound needs no spectral information. By rotation invariance $\E_H\nu_H(A)$ and
$\E_H\nu_{H^\perp}(A)$ both equal $\sigma(A)=\alpha$, so the triangle inequality gives
$\E_H|\mu_H(A)|\le\E_H\nu_H(A)+\E_H\nu_{H^\perp}(A)=2\alpha$.

For the second bound, start from Cauchy--Schwarz and the identity of \Cref{lem:second-moment},
\[
  \E_H\big|\mu_H(A)\big|\ \le\ \Big(\E_H\big[\mu_H(A)^{2}\big]\Big)^{1/2}
  \ =\ 2\Big(\sum_{j\ \mathrm{odd}}\omega_{2j}\,W_{2j}(\mathbf 1_A)\Big)^{1/2},
\]
and split the sum at $j^*\eqdef\lceil s\rceil$, which is at least $2$ because $\alpha\le\frac12$
forces $s>1$. The two ranges are controlled by different means, the level-$\ell$ inequality below
$j^*$ and the decay of $\omega$ above it.

Consider first the terms with $j<j^*$. There $2j<2s$, so \Cref{lem:level-ell} applies at $\ell=2j$.  Its second hypothesis reads
$2j\ln(2s)\le c_0'n$, where $c_0'$ is the constant of \Cref{lem:level-ell}, and this holds since
$j<s\le c_0\sqrt n$ once $c_0$ is small in terms of $c_0'$. Inserting it
together with $\omega_{2j}\le(\frac{2j}{n-1})^j$,
\[
  \omega_{2j}W_{2j}(\mathbf 1_A)\ \le\ \alpha^2\Big(\frac{2j}{n-1}\Big)^{j}\Big(\frac{C_0s}{2j}\Big)^{2j}
  \ =\ \alpha^2\Big(\frac{C_1s^2}{j\,(n-1)}\Big)^{j},\qquad C_1\eqdef\tfrac{C_0^2}2,
\]
and summing over $j\ge1$ gives at most twice the first term, namely $2C_1\alpha^2s^2/(n-1)$, because
the ratio of consecutive terms is at most $C_1s^2/(n-1)\le\frac12$ once $c_0$ is small.

Consider next the terms with $j\ge j^*$, where monotonicity of $\omega$ and Parseval give
\[
  \sum_{j\ge j^*}\omega_{2j}W_{2j}(\mathbf 1_A)\ \le\ \omega_{2j^*}\alpha\ \le\ \alpha\,\theta(s)^2 ,
\]
the last step because $\lceil s\rceil\le2s$ for $s\ge1$. Adding the two ranges, taking square roots and using
$\sqrt{x+y}\le\sqrt x+\sqrt y$ gives \eqref{eq:oneway-cell}.
\end{proof}

\subsection{Proof of the one-way rate}
\label{app:oneway-proof}

\begin{proof}[Proof of the upper bound in \Cref{prop:oneway-rate}]
Fix the shared randomness, so that the protocol is deterministic. Alice's message then realizes a
measurable partition $\{A_m\}$ of $S^{n-1}$ into at most $2^{L+1}$ cells (\Cref{sec:prelim-cc}), and
Bob, knowing $H$ and $m$, answers optimally.  Write $\alpha_m\eqdef\sigma(A_m)$ and
$s_m\eqdef\ln(e/\alpha_m)$ for the parameters of \Cref{lem:cell-leak} at the cell $A_m$. Since the two challenge distributions on $(m,H)$ have the
same $H$-marginal and conditional laws $\nu_H(A_m)$ and $\nu_{H^\perp}(A_m)$, the advantage satisfies
\[
  \beta_{\mathrm{ow}}\ \le\ \tfrac14\,\E_H\sum_m\big|\mu_H(A_m)\big| .
\]

Three reductions dispose of the degenerate ranges.  Set $c\eqdef c_0/(4\ln2)$.
\begin{enumerate}[label=(\arabic*),leftmargin=2.2em,itemsep=3pt,topsep=3pt]
\item \emph{Long messages.}  For $L\ge c\sqrt n$ the right-hand side of the proposition exceeds
  $\frac12$ once its constant $C$ is large enough, while $\beta_{\mathrm{ow}}\le\frac12$ always.  The
  bound holds there, so we may assume $L\le c\sqrt n$.
\item \emph{Small $n$.}  We may assume $n\ge n_0$ for a universal $n_0$ fixed at the end of the proof,
  since for $n<n_0$ the proposition holds by enlarging its constant.
\item \emph{Large cells.}  Because $\mu_H(S^{n-1})=0$ we may replace any cell of measure above
  $\frac12$ by its complement, and since at most one cell has such a measure this costs one further
  application of the bounds below.
\end{enumerate}
What is left is to sum \Cref{lem:cell-leak} over the cells, which we do in two ranges of cell size,
the very small cells being disposed of by the cruder of its two bounds.

Consider first the cells with $s_m>c_0\sqrt n$, for which the first bound of \Cref{lem:cell-leak}
suffices. Each such cell has $\alpha_m<e^{1-c_0\sqrt n}$, and there are at most $2^{L+1}$ of them, so
their total contribution is at most $2^{L+2}e^{1-c_0\sqrt n}\le n^{-1/2}$, the last step because
$L\le c_0\sqrt n/(4\ln2)$ and $n\ge n_0$.

Consider next the cells with $s_m\le c_0\sqrt n$, where \eqref{eq:oneway-cell} applies and its two
terms are summed separately. The first term sums by the entropy bound, with $C$ the constant of \Cref{lem:cell-leak}, since $\sum_m\alpha_m\le1$ and
there are at most $2^{L+1}$ cells,
\[
  \frac C{\sqrt n}\sum_m\alpha_m\ln\frac e{\alpha_m}\ \le\ \frac{C\big(1+(L+1)\ln2\big)}{\sqrt n}.
\]

The second term sums by counting cells at each scale. Fewer than $e^{k}$ cells have $s_m\in[k,k+1)$,
since each of those has $\alpha_m>e^{-k}$, and for such a cell
$\theta(s_m)\le(\frac{4(k+1)}{n-1})^{k/2}$, the base being at most one for $n\ge n_0$. The cells in this range have
$s_m\le c_0\sqrt n$, so only the scales $1\le k\le k_{\max}\eqdef\lceil c_0\sqrt n\rceil$ occur, and
\[
  \sum_m2\sqrt{\alpha_m}\,\theta(s_m)
  \ \le\ 2\sum_{k=1}^{k_{\max}}e^{k}\Big(\frac{4(k+1)}{n-1}\Big)^{k/2}
  \ =\ 2\sum_{k=1}^{k_{\max}}\Big(\frac{4e^{2}(k+1)}{n-1}\Big)^{k/2} .
\]
Over $k\le k_{\max}$ the base is at most
\[
  \vartheta\ \eqdef\ \frac{4e^{2}(k_{\max}+1)}{n-1}\ =\ O\!\Big(\frac{c_0}{\sqrt n}\Big),
\]
so there is a universal $n_0$ beyond which $\vartheta\le\frac14$.  We take $n_0$ large enough that
the two earlier steps that assumed it hold as well. Separating the first scale from the rest
and bounding every base by $\vartheta$,
\[
  2\sum_{k=1}^{k_{\max}}\Big(\frac{4e^{2}(k+1)}{n-1}\Big)^{k/2}
  \ \le\ 2\Big(\frac{8e^{2}}{n-1}\Big)^{1/2}+2\sum_{k\ge2}\vartheta^{k/2}
  \ \le\ 2\Big(\frac{8e^{2}}{n-1}\Big)^{1/2}+\frac{2\vartheta}{1-\sqrt\vartheta}
  \ \le\ \frac{C'}{\sqrt n},
\]
both terms being $O(n^{-1/2})$.

Putting everything together, the three contributions just bounded are each $O\big((L+1)/\sqrt n\big)$,
so adding them and dividing by $4$ gives the claim.
\end{proof}

\section{Deferred proofs}
\label{app:deferred}

This appendix collects two statements that we use only locally: the counting argument behind the
comparison drawn in \Cref{sec:framework}, and the side result on reference copies quoted in
\Cref{sec:advice-is-a-query}.

\subsection{Classical queries do not synthesize}
\label{app:def-synth}

The following counting argument is folklore. It is the argument that \Cref{sec:intro-lmw} and
\Cref{sec:framework} appeal to when they call classical queries the easy case, and it should
be contrasted with~\cite[Thm.~4.1]{Rosenthal23}, where a single \emph{coherent} query to a classical
oracle synthesizes an arbitrary state.

\begin{proposition}[Classical queries do not synthesize]
\label{obs:classical-synthesis}
Fix $m$, and fix a circuit on $m$ output qubits and any number of ancillas that makes $T$ classical
queries to an oracle $g\colon\{0,1\}^*\to\{0,1\}$, each query string being computed from an
oracle-independent random seed $s$ and from the previous answers.  Write $\rho_g$ for the reduced
state of its output on the $m$ output qubits. If $T\le2^m-m-3$, there is an $m$-qubit pure state $\ket\psi$
such that $\bra\psi\rho_g\ket\psi\le\frac34$, and hence $\td(\rho_g,\proj\psi)\ge\frac14$, for every $g$. In
particular this holds for every $T=\poly(m)$.
\end{proposition}

\begin{proof}
Fix the seed $s$. The first query string is then fixed, and by induction, once the first $i$ answers
are fixed so is the $(i+1)$-st query string. Hence the transcript is a function of the answer vector
$(a_1,\dots,a_T)\in\{0,1\}^T$, and the conditional output $\rho_{g,s}$ ranges, as $g$ varies, over a
set $F_s$ of at most $2^T$ states that does not depend on $g$. For any pure $\ket\psi$,
\[
  \bra\psi\rho_g\ket\psi\ =\ \E_s\bra\psi\rho_{g,s}\ket\psi\ \le\ \E_s\,m_s(\psi),\qquad
  m_s(\psi)\eqdef\max_{\sigma\in F_s}\bra\psi\sigma\ket\psi .
\]
So it is enough to exhibit a single $\ket\psi$ whose overlap with $F_s$ is small on average over
$s$.  A Haar-random $\ket\psi$ will do, and the rest of the proof verifies this.

Put $N=2^m$ and draw $\ket\psi$ Haar-random in $\bbC^N$.  Against a fixed unit vector $\ket\phi$,
\[
  \Pr_\psi\big[\,|\ip\phi\psi|^2\ge\tfrac12\,\big]\ =\ 2^{-(N-1)} .
\]
A state $\sigma$ overlaps $\ket\psi$ by at most its largest overlap with one of its $N$
eigenvectors, so the same estimate survives a factor $N$:
\[
  \Pr_\psi\big[\,\bra\psi\sigma\ket\psi\ge\tfrac12\,\big]\ \le\ N\,2^{-(N-1)} .
\]
Now union bound over the at most $2^T$ members of $F_s$.  Since $T\le2^m-m-3$, the exponent is at
most $-2$, and so for each fixed seed $s$,
\[
  \Pr_\psi\big[\,m_s(\psi)\ge\tfrac12\,\big]\ \le\ N\,2^{\,T-N+1}\ \le\ \tfrac14 .
\]
Split the expectation of $m_s$ on that event.  Off it $m_s<\frac12$, and on it $m_s\le1$ while the
event itself has probability at most $\frac14$, so $\E_\psi\E_sm_s(\psi)\le\frac12+\frac14$.  Some
fixed $\ket\psi$ must then have $\E_sm_s(\psi)\le\frac34$, and for that $\ket\psi$ the first display
gives $\bra\psi\rho_g\ket\psi\le\frac34$ for every $g$.

Testing with the effect $\proj\psi$ converts an overlap into a trace distance:
$\td(\rho,\proj\psi)\ge1-\bra\psi\rho\ket\psi$ for every $\rho$.  That is the claim.
\end{proof}

\subsection{Reference copies accumulate as a random walk}
\label{app:def-copies}

This subsection states and proves the side result quoted in \Cref{sec:advice-is-a-query}: $q$
reference copies of one of the two states are worth $\sqrt q$ rather than $q$, up to a logarithm,
both as an attack and as an upper bound.  Nothing in the separation depends on it.

\begin{proposition}[Reference copies accumulate as a random walk]
\label{prop:copies-sqrt}
Let $n\ge4$ be even and $1\le q\le n^2$, and for $b\in\{0,1\}$ put
$\mu_b\eqdef\rho_0^{\otimes q}\otimes\rho_b$ on $(\bbC^n)^{\otimes(q+1)}$: a distinguisher holding
$q$ reference copies of the $b=0$ state alongside its challenge. Then:
\begin{enumerate}[label=(\roman*),leftmargin=2.2em]
\item \emph{(Attack, pointwise in $H$.)} There is a single effect $\Pi_q$ on
  $(\bbC^n)^{\otimes(q+1)}$, depending only on $n$ and $q$, such that for \emph{every} half-dimensional
  subspace $H$,
  \[
    \tr(\Pi_q\,\mu_0)-\tr(\Pi_q\,\mu_1)\ \ge\ \frac{1}{40}\,\frac{\sqrt{q}}{n\sqrt{\log_2(4n)}}\ ,
  \]
  so the bias of the corresponding distinguisher is at least $\tfrac1{80}\sqrt{q/\log_2(4n)}\,/\,n$.
\item \emph{(Matching upper bound, on average.)} Every distinguisher measuring
  $\mu_b$ with an $H$-independent effect has average bias
  $\E_H[\mathrm{bias}]\le C\sqrt{2q\log(2n)}\,/\,n$, with $C$ the constant of
  \Cref{prop:advice-sharp}.
\end{enumerate}
Part~(i) holds for every fixed half-dimensional $H$, real or complex; part~(ii) holds in either
ensemble, with $C$ the corresponding constant of \Cref{lem:khintchine}.
\end{proposition}

The proof rests on an elementary counting bound, which we state on its own because it is the only
combinatorial input and is independent of everything quantum. Throughout, the copies occupy registers $1,\dots,q$ and the challenge occupies register $q+1$.  For
$\pi\in S_{q+1}$ we write $|\pi|\eqdef(q+1)-\#\mathrm{cycles}(\pi)$ for the minimal number of
transpositions expressing $\pi$.

\begin{lemma}[Walk counting for the swap sum]
\label{lem:walk-count}
For a word $w=(i_1,\dots,i_{2k})\in[q]^{2k}$ put
\[
  \pi_w\ \eqdef\ (i_1\ q{+}1)\,(i_2\ q{+}1)\cdots(i_{2k}\ q{+}1)\ \in\ S_{q+1}.
\]
Then for every integer $k\ge2$ and every $x\in(0,1]$ with $x^2q\le2k$,
\[
  \sum_{w\in[q]^{2k}}x^{\,|\pi_w|}\ \le\ (24\,kq)^k .
\]
\end{lemma}

\begin{proof}
Suppose first that $q\le2k$. Then $|\pi_w|\ge0$ and $x\le1$ give
\[
  \sum_{w}x^{|\pi_w|}\ \le\ q^{2k}\ =\ (q^2)^k\ \le\ (2kq)^k ,
\]
which suffices.

Suppose instead that $q>2k$. Only a repeated index can move the walk back toward the identity. Each
fresh index pushes $\pi_w$ one step further from it, so a word using many distinct indices carries a
small weight $x^{|\pi_w|}$, while a word using few of them has few realizations, and the whole estimate
is a trade between these two effects.

Call step $\ell$ \emph{fresh} if $i_\ell\notin\{i_1,\dots,i_{\ell-1}\}$ and \emph{old} otherwise,
and let $f$ be the number of fresh steps.  Multiplying by a transposition changes $|\cdot|$ by
exactly $\pm1$, and a fresh step always contributes $+1$.  To see the latter, note that the support
of the permutation built so far lies in the challenge slot together with the indices already used, so
a fresh $i_\ell$ is a fixed point of it; joining its cycle to the cycle of the challenge
increases the length.  Since the remaining $2k-f$ steps contribute at least $-1$ each,
\[
  |\pi_w|\ \ge\ \max(0,\,2f-2k).
\]

It remains to count the words at each value of $f$. Those with exactly $f$ fresh steps number at most
$\binom{2k}{f}q^ff^{\,2k-f}\le4^kq^f(2k)^{2k-f}$, by choosing the positions of the fresh steps, then
their values, and then the old values from among the at most $f$ indices already used. Splitting the
sum at $f=k$ gives
\[
  \sum_wx^{|\pi_w|}\ \le\ 4^k\Bigg[\sum_{f\le k}q^f(2k)^{2k-f}
  \;+\;\sum_{g=1}^{k}q^{k+g}(2k)^{k-g}x^{2g}\Bigg].
\]
Each sum is now a geometric comparison against the single quantity $(2kq)^k$.  In the first,
\[
  q^f(2k)^{2k-f}\ =\ (2kq)^k\Big(\frac{2k}q\Big)^{k-f}\ \le\ (2kq)^k ,
\]
since $2k\le q$ and $f\le k$, and there are at most $k+1$ terms.  In the second, write $f=k+g$; then
\[
  q^{k+g}(2k)^{k-g}x^{2g}\ =\ (2kq)^k\Big(\frac{x^2q}{2k}\Big)^{g}\ \le\ (2kq)^k
\]
by the hypothesis $x^2q\le2k$, and there are $k$ terms.  The whole sum is thus at most
\[
  4^k(2k+1)(2kq)^k\ \le\ 4^k3^k(2kq)^k\ =\ (24\,kq)^k ,
\]
using $2k+1\le3^k$, and the lemma follows.
\end{proof}

\begin{proof}[Proof of \Cref{prop:copies-sqrt}]
The mechanism is a random walk.  Each copy couples to the challenge only through a swap, and the $q$
swaps sum to $X=\sum_{i=1}^q\mathrm{SWAP}_{i,\,q+1}$.  Against $\mu_b$ every permutation observable
evaluates exactly, with no dependence on $H$, so that
$\langle X\rangle_{\mu_0}-\langle X\rangle_{\mu_1}=2q/n$ exactly.  Meanwhile $X$ fluctuates at scale
$\sqrt q$ under both states, and truncating $X$ at that scale and thresholding gives~(i).  Part~(ii)
is \Cref{prop:advice-sharp} applied to a purification of the copies.  Throughout we use the
notation fixed above, writing $P_\pi$ for the operator permuting the tensor factors.

\emph{(i).} Start from an exact trace identity. For operators $\sigma_1,\dots,\sigma_{q+1}$ on
$\bbC^n$,
\[
  \tr\big[P_\pi(\sigma_1\otimes\cdots\otimes\sigma_{q+1})\big]
  \;=\;\prod_{\text{cycles }c\text{ of }\pi}\ \tr\Big(\prod_{j\in c}\sigma_j\Big),
\]
where the product inside each trace is taken along the cycle.  The order is immaterial here, because
all our factors are polynomials in the single operator $P_H$ and hence commute.  Now take
$\sigma_j=\rho_0$ for $j\le q$ and $\sigma_{q+1}=\rho_b$.  Since $\rho_0=2P_H/n$ has rank $n/2$ with
all nonzero eigenvalues $2/n$,
\[
  \tr(\rho_0^{\,\ell})=(2/n)^{\ell-1},\qquad
  \tr(\rho_0^{\,\ell-1}\rho_1)=0\ \ (\ell\ge2),\qquad \tr(\rho_1)=1 ,
\]
the middle identity holding because $\rho_0\rho_1=\frac{4}{n^2}P_HP_{H^\perp}=0$.

Read cycle by cycle, this says the following.  A cycle of length $\ell$ that avoids the challenge slot
contributes $(2/n)^{\ell-1}$.  A cycle through the challenge contributes $(2/n)^{\ell-1}$ when $b=0$;
when $b=1$ it contributes $0$, unless $\ell=1$.  Multiplying over cycles,
\begin{equation}
  \big\langle P_\pi\big\rangle_{\mu_0}=\Big(\frac2n\Big)^{|\pi|},
  \qquad
  \big\langle P_\pi\big\rangle_{\mu_1}=\mathbf 1\big[\pi(q{+}1)=q{+}1\big]\cdot\Big(\frac2n\Big)^{|\pi|},
  \label{eq:perm-values}
\end{equation}
\emph{for every fixed $H$}, and in both the real and the complex ensemble; the dimension $n/2$ is the
only feature of $H$ that enters anywhere.  Note that every quantity below is a linear combination of
the values \eqref{eq:perm-values}, so the entire argument is pointwise in $H$.

We state the content of \eqref{eq:perm-values} separately, since it is the reason the rest of
the argument stays elementary.  The maximally mixed state on $(\bbC^{d})^{\otimes(q+1)}$ with $d=n/2$
also has $\langle P_\pi\rangle=d^{-|\pi|}$, since $\tr(P_\pi)=d^{\#\mathrm{cycles}(\pi)}$.  So on
permutation observables $\mu_0$ is indistinguishable from that state, and $\mu_1$ never exceeds it.
In other words, the moment bound and the truncation below are statements about the maximally mixed
state on $q+1$ registers of dimension $d$, with no hidden subspace in them.

The two means are immediate.  Put
\[
  X\ \eqdef\ \sum_{i=1}^q\mathrm{SWAP}_{i,\,q+1}\ =\ \sum_{i=1}^qP_{(i\ q+1)} .
\]
Each transposition has $|\pi|=1$ and moves the challenge slot, so \eqref{eq:perm-values} gives
\[
  \langle X\rangle_{\mu_0}\ =\ q\cdot\frac2n ,\qquad \langle X\rangle_{\mu_1}\ =\ 0 .
\]

Next we bound the even moments of $X$, and this is the only place where the walk counting is used. We
claim that for every integer $k\ge2$ and both $b$,
\begin{equation}
  \big\langle X^{2k}\big\rangle_{\mu_b}\ \le\ (24\,kq)^k .
  \label{eq:walk-moments}
\end{equation}
Expanding $X^{2k}=\sum_wP_{\pi_w}$ over words $w\in[q]^{2k}$, every term is nonnegative and at most
$(2/n)^{|\pi_w|}$ by \eqref{eq:perm-values}, so the claim is exactly \Cref{lem:walk-count} applied
with $x=2/n$. Its hypotheses hold, since $x\le1$ for $n\ge2$, and $x^2q=4q/n^2\le4\le2k$ because
$q\le n^2$ and $k\ge2$.

Finally, truncate.  Set
\[
  k\eqdef\lceil\log_2(4n)\rceil\ \ (\ge4),\qquad \Lambda\eqdef e\sqrt{24\,kq},\qquad
  A\eqdef X\cdot\mathbf 1\big[\,|X|\le \Lambda\,\big] .
\]
Then $A$ is Hermitian with $\norm A_{\mathrm{op}}\le \Lambda$, so $\Pi_q\eqdef\frac12(\Id+A/\Lambda)$ is
an effect.  What the truncation discards is controlled by the moment bound.  Indeed
$|x|\,\mathbf 1[|x|>\Lambda]\le x^{2k}/\Lambda^{2k-1}$ for every real $x$, and
$e^{2k}\ge e^{2\ln(4n)}=16n^2$, so \eqref{eq:walk-moments} gives, for both $b$,
\[
  \big|\langle X\,\mathbf 1[|X|>\Lambda]\rangle_{\mu_b}\big|
  \ \le\ \frac{(24kq)^k}{\Lambda^{2k-1}}\ =\ \Lambda\,e^{-2k}
  \ \le\ \frac{e\sqrt{24kq}}{16n^2}\ \le\ \frac{q}{2n},
\]
the last step because $e\sqrt{24k}\le8n$ for $n\ge4$ and $q\ge1$.

So truncation costs at most half of the gap $2q/n$ between the two means, and what is left is still of
that order:
\[
  \tr\big(A(\mu_0-\mu_1)\big)\ \ge\ \frac{2q}{n}-2\cdot\frac{q}{2n}\ =\ \frac qn,
  \qquad\text{so}\qquad
  \tr\big(\Pi_q(\mu_0-\mu_1)\big)\ \ge\ \frac{q}{2\Lambda n}\ =\ \frac{\sqrt q}{2e\sqrt{24k}\,n}\ ,
\]
where the second equality uses $\Lambda=e\sqrt{24kq}$.  Since $k\le2\log_2(4n)$ we have
$2e\sqrt{24k}\le2e\sqrt{48\log_2(4n)}\le40\sqrt{\log_2(4n)}$, and together these give~(i), for every
$H$.

\emph{(ii).} Fix an effect $\Pi$ on the $q+1$ registers, and purify the copies, writing
$\rho_0^{\otimes q}=\tr_E\proj{\Psi_H}$ with $\ket{\Psi_H}\in\cV\otimes E$ for
$\cV=(\bbC^n)^{\otimes q}$ and an environment $E$ of dimension $n^q$. Then
\[
  \mathrm{bias}(\Pi)=\tfrac12\big|\tr\big(\Pi(\mu_0-\mu_1)\big)\big|
  =\tfrac1n\big|\tr\big[(\Pi\otimes\Id_E)\big(\proj{\Psi_H}\otimes R\big)\big]\big| ,
\]
after reordering the registers so that the copies and $E$ together form the advice space, of dimension
$K=n^{2q}$. This is exactly the quantity of \Cref{prop:advice-sharp}, whose bound is uniform over effects.  Hence
\[
  \E_H[\mathrm{bias}]\ \le\ \frac{C\sqrt{\log(2n^{2q})}}n\ \le\ \frac{C\sqrt{2q\log(2n)}}n ,
\]
which is~(ii).
\end{proof}

\paragraph{The origin of the $\sqrt q$ rate.}  The operator being controlled is the Jucys--Murphy
element $X_{q+1}$ of $\bbC[S_{q+1}]$~\cite{Jucys74,Murphy81}, whose spectrum consists of the integer
contents of Young-diagram boxes.  \Cref{lem:walk-count} is an elementary substitute for that spectral
information.  At $q=1$ the value is exactly $1/n$, attained for every $H$ by the symmetric-subspace
effect $\frac12(\Id+\mathrm{SWAP})$.

The proposition completes an account of what each resource about $H$ is worth against the pair.
Fresh samples of the secret's own output are worth $\widetilde\Theta(\sqrt q)/n$ for $q$ of them.  A
classical description of their span is worth $\Theta(q/n)$.  One coherent query, of any width
$M=2^{\poly(\lambda)}$, is worth $n^{-1/2}$ up to the factor $\sqrt{\log(2M)}$ (\Cref{thm:one-query},
\Cref{cor:one-query-tight}).  The pair survives relative to $\cO$ because the oracle provides only
the weakest of these resources: fresh samples from $\cR$, and classical bits from $\Count$ that the
communication bound already accounts for.

\section{The sharp rate for the half-subspace pair}
\label{app:sharp}

The bounds of \Cref{sec:security} and \Cref{sec:hybrid} pass through
\Cref{thm:transcript-security}, and a worst-case communication bound enters them through
\Cref{lem:vsp-small}, whose amplification step costs a square root; the rate they return is
$\poly(\lambda)\cdot n^{-1/6}$.

This appendix proves the same statements at $\poly(\lambda)\cdot n^{-1/2}$, which is optimal up to
the polynomial by \Cref{prop:exact-rate}.  The argument uses no communication bound.  In place of
one it uses the linearity of $\rho_0-\rho_1=2R_H/n$ in the reflection, and so, unlike
\Cref{thm:transcript-security}, it does not extend to an arbitrary hidden object.

\begin{theorem}[Sharp hybrid security for the half-subspace pair]
\label{thm:hybrid-sharp}
There is a universal $C>0$ such that for all $q<n/4$, all $L,a\ge0$ and all $M\ge1$, every hybrid
adversary of \Cref{model:hybrid} satisfies
\[
  \E_H\Big[\max_{\alpha\in\{0,1\}^{\le a}}\Adv\Big]\ \le\ \frac{2q}{n}
  \;+\;C\sqrt{\frac{\log(2M)+L+a}{n-2q}} .
\]
For $n=2^{\Omega(\lambda)}$ and all parameters $\poly(\lambda)$ the bound is
$\poly(\lambda)\cdot n^{-1/2}$.  At $M=1$ the final query is trivial, and the statement bounds an
adversary with $q$ reference copies per side, $L$ bits of adaptive classical queries and $a$ bits of
advice, with no coherent query, by $\frac{2q}n+C\sqrt{(L+a+1)/(n-2q)}$.  At $L=a=q=0$ it is
\Cref{thm:one-query}, up to the factor two between bias and advantage.
\end{theorem}

\begin{proof}
The argument is Steps~3--6 of the proof of \Cref{thm:hybrid}, applied to the whole signed advantage
rather than to its increment.  The split of Step~1 there, and the communication bound of Step~2, are
not needed: \Cref{lem:subnorm} bounds $\max_f|Z_H(f)|$ over \emph{all} truth tables, so it bounds
the contribution of an $H$-independent effect exactly as it bounds that of $F_H$, and the record
identity of \Cref{lem:record} turns the whole advantage into a convex combination of such terms.

Fix the adversary.  Every object below depends implicitly on the advice string $\alpha$, and we keep
$\max_\alpha$ outside all inequalities, each of which holds for every fixed $\alpha$.

\emph{Step 1: reference copies cost $\frac{2q}n$ and a loss of dimension.}  Give the adversary the
classical spans $(A,B)$ of its reference copies.  This only increases its power, since from a basis
of the spans it resamples the copies with the correct joint law, by the argument that opens the
proof of \Cref{lem:reference}.  Condition on $(A,B)$.  By \Cref{lem:residual} each challenge branch
is a mixture, with weight $\frac{2q}n$ on the maximally mixed state of the known subspace, $A$ or
$B$, and weight $1-\frac{2q}n$ on the maximally mixed state of the residual, $H_A$ or $H_B$, inside
$K=(A\oplus B)^\perp$ of dimension $n'\eqdef n-2q$.  For a fixed strategy the acceptance probability
is affine in the challenge, so the signed advantage splits along the mixture.  On the known branch
it is a difference of two probabilities, hence at most $1$ in absolute value, and that branch
contributes at most $\frac{2q}n$.  This is where the present argument is cheaper than Step~3 of
\Cref{thm:hybrid}, which bounds a difference of \emph{effects} and pays $\frac{4q}n$.

On the residual branch the adversary faces the half-subspace pair of $H_A$ inside $K$, with $H_A$
Haar of dimension $n'/2$ in $K$ by \Cref{lem:residual}.  Given $(A,B)$, the $\Count$ answers and the
truth table of the final query are Boolean functions of $H_A$ alone, since $H=A\oplus H_A$, and the
resampled copies are functions of $(A,B)$ and private randomness.  Hard-wiring these leaves a hybrid
adversary with no reference copies against a half-subspace pair in dimension $n'$.  It remains to
bound its advantage uniformly in $(A,B)$.  From here on we write $n$ for $n'$, $R$ for the
reflection $2P_{H_A}-P_K$ on $K$, and $\rho_b$ for the residual pair, so that $\rho_0-\rho_1=2R/n$.

\emph{Step 2: the whole advantage is a convex combination of one-query expressions.}  Let $F_H$ be
the accept effect of step~(iii): $F_H=W^\dagger(D_{s(H)}\otimes\Id)\,\Pi\,(D_{s(H)}\otimes\Id)W$ for
the pre-query isometry $W$, the truth table $s(H)$ and the final effect $\Pi$.  By
\Cref{lem:record}, with the operators $V_\pi$ of the classical phase and the masses
$w_\pi\eqdef\frac1n\tr(P_KV_\pi^\dagger V_\pi P_K)$ of the instrument restricted to $K$,
\[
  \Delta(H)\ \eqdef\ \tr\big[F_H\,\mathcal E_H(\rho_0-\rho_1)\big]
  \ =\ \frac2n\sum_\pi\chi_\pi(H)\,\tr\big[F_H\,V_\pi RV_\pi^\dagger\big],
\]
and the masses satisfy $\sum_\pi\chi_\pi(H)\,w_\pi=1$ for every $H$, the identity surviving the
restriction to $K$ exactly as in Step~4 of the proof of \Cref{thm:hybrid}.  Fix a record $\pi$.
Then $\tr[F_HV_\pi RV_\pi^\dagger]=n\,Z_H(s(H))$ for the quadratic form \eqref{eq:quadform} built
from the pair $(WV_\pi,\Pi)$, and $WV_\pi$ has normalized trace $w_\pi$ because $W$ is an isometry.
\Cref{lem:subnorm} applies to that pair and gives, for the weighted norm $\phi_\pi(H)$ it defines,
\[
  \big|\tr[F_HV_\pi RV_\pi^\dagger]\big|\ \le\ n\max_f|Z_H(f)|\ \le\ n\,w_\pi\,\phi_\pi(H).
\]
Hence, pointwise in $H$ and for every fixed $\alpha$,
\[
  |\Delta(H)|\ \le\ 2\sum_\pi\chi_\pi(H)\,w_\pi\,\phi_\pi(H)\ \le\ 2\max_\pi\phi_\pi(H),
\]
the last step because the numbers $\chi_\pi(H)w_\pi$ form a probability vector and each
$\phi_\pi\ge0$.

\emph{Step 3: passing to the means.}  This is Step~6 of the proof of \Cref{thm:hybrid}.  Write
$m_\pi=\E_H\phi_\pi\le C\sqrt{\log(2M)/n}$.  Each $\phi_\pi$ is $2$-Lipschitz along geodesics in the
Haar unitary defining $R$, hence subgaussian with variance proxy $\sigma^2\le C'/n$, and there are
at most $2^{a+1}\cdot2^{L+1}$ pairs $(\alpha,\pi)$.  The maximal inequality for subgaussian
variables, which uses only a union bound on their tails, gives
\[
  \E_H\max_\alpha|\Delta(H)|\ \le\ 2\max_{\alpha,\pi}m_\pi+2\,\E_H\max_{\alpha,\pi}|\phi_\pi-m_\pi|
  \ \le\ C\sqrt{\frac{\log(2M)+L+a}n},
\]
the first term being at most $2C\sqrt{\log(2M)/n}$ and the second at most
$2\sigma\sqrt{2\ln(2^{a+L+3})}\le C\sqrt{(L+a+3)/n}$; the constants are absorbed, since
$\log(2M)\ge\log2$.

\emph{Step 4: reassembling.}  The bound of Step~3 holds for every $(A,B)$, so averaging over
$(A,B)$ preserves it.  Adding the $\frac{2q}n$ of Step~1 and restoring $n-2q$ for $n$ gives the
theorem.
\end{proof}

The net argument of \Cref{cor:hybrid-qadvice} is unchanged by the substitution, and it improves with
the bound it is fed.

\begin{corollary}[Quantum advice at the sharp rate]
\label{cor:qadvice-sharp}
In the setting of \Cref{cor:hybrid-qadvice}, with $a'\eqdef a+\lceil2d_\cA\log_29\rceil$ and $C$ the
constant of \Cref{thm:hybrid-sharp},
\[
  \E_H\Big[\sup_{\sigma_H}\ \max_{\alpha\in\{0,1\}^{\le a}}\Adv\Big]\ \le\
  2\Big[\frac{2q}{n}+C\sqrt{\frac{\log(2M)+L+a'}{n-2q}}\Big].
\]
For $n=2^{\lambda}$, all of $q,L,a,\log M$ polynomial in $\lambda$, and $d_\cA\le n^{1-\eps}$ with
$\eps>0$ fixed, the bound is $2^{-\Omega(\lambda)}$: quantum advice of up to $(1-\eps)\lambda$
qubits is admissible, against $(\tfrac13-\eps)\lambda$ from \Cref{cor:hybrid-qadvice}.
\end{corollary}

\begin{proof}
The net argument of \Cref{cor:hybrid-qadvice} is unchanged: an index into a $\frac14$-net of the
unit sphere of $\cA$ costs $\lceil2d_\cA\log_29\rceil$ classical bits and loses a factor two.  Applying
\Cref{thm:hybrid-sharp} to the resulting adversary gives the display.  For the final claim,
$a'=O(K)\le O(n^{1-\eps})$, so the square root is $O(n^{-\eps/2})=2^{-\Omega(\lambda)}$, and
$\frac{2q}n$ is $2^{-\Omega(\lambda)}$ as well.
\end{proof}

\begin{remark}[The enlarged hidden object of \Cref{sec:single}]
\label{rem:sharp-enlarged}
\Cref{thm:hybrid-sharp} uses three properties of the hidden object: the record identity of
\Cref{lem:record} holds pointwise; the weighted norms $\phi_\pi$ depend on the object only through
the reflection $R$; and the truth tables queried are covered by the maximum over all Boolean
functions.  All three hold for $X=(H,S)$ of \Cref{model:single} in the fresh experiment of
\Cref{lem:fresh}.  The first is an identity about the instrument realized by any fixed oracle.  For
the second, $\phi_\pi$ is built from the adversary's fixed operators and $R_{H_A}$, so $S$ does not
enter it, and the expectation of Step~3 is over $H_A$ given $(A,B)$, which is Haar in $K$ by
\Cref{lem:residual}, the fresh reference copies being drawn from $\rho_b(H)$ independently of $S$.
The third holds by construction.  So \Cref{thm:single} holds with $\poly(\lambda)\cdot n^{-1/2}$ in
place of $\poly(\lambda)\cdot n^{-1/6}$, and the clause of \Cref{lem:side} that transfers the
communication term is not needed on this route.
\end{remark}

\end{document}